\documentclass[conference]{IEEEtran}
\pdfoutput=1
\usepackage{macros}

\usepackage[leftmargin=1em,rightmargin=1ex,vskip=1ex]{quoting}
\usepackage{lscape}
\usepackage{framed}
\usepackage{subcaption}
\usepackage{scalerel}
\usepackage{mathrsfs}
\usepackage{newtxmath}
\usepackage{newtxtext}
\usepackage{hyphenat}
\usepackage{cuted}
\usepackage[pass]{geometry}
\usepackage{xparse}
\newcommand{\biobj}[2]{{\protect\substack{ #1 \\ #2 }}}

\newcommand{\nhphantom}[1]{\sbox0{#1}\hspace{-\the\wd0}}

\newcommand\smallmath[2]{#1{\raisebox{\dimexpr \fontdimen 22 \textfont 2
      - \fontdimen 22 \scriptfont 2 \relax}{$\scriptstyle #2$}}}

\newcommand\smallotimes{\smallmath\mathbin\otimes}

\usepackage{unicode-mario}
\renewcommand{\unicodecalC}{\hyperlink{linkContour}{\mathcal{C}}}
\renewcommand{\unicodecalD}{\hyperlink{linkMonoidalContour}{\mathcal{D}}}
\renewcommand{\unicodecalS}{\hyperlink{linkSplice}{\mathcal{S}}}
\renewcommand{\unicodecalT}{\hyperlink{linkMonoidalSplice}{\mathcal{T}}}
\renewcommand{\unicodecalN}{\hyperlink{linkNormalization}{\mathcal{N}}}

\allowdisplaybreaks
\title{The Produoidal Algebra of Process Decomposition}
\author{Anonymous author(s)}
\author{Matt Earnshaw, James Hefford and Mario Rom\'an}

\makeatletter
\def\@copyrightspace{\relax}
\makeatother

\begin{document}

\maketitle

\begin{abstract}
  We introduce the normal produoidal category of monoidal contexts over an arbitrary monoidal category.
  In the same sense that a monoidal morphism represents a process, a monoidal context represents an incomplete process: a piece of a decomposition, possibly containing missing parts.
  We characterize monoidal contexts in terms of universal properties.
  In particular, symmetric monoidal contexts coincide with monoidal lenses, endowing them with a novel universal property.
  We apply this algebraic structure to the analysis of multi-party interaction protocols in arbitrary theories of processes.
\end{abstract}

\section{Introduction}

Theories of processes, such as \emph{stochastic}, \emph{partial} or \emph{linear} functions, are a foundational tool in computer science. They help us model how systems interact in terms of a solid mathematical foundation. Any theory of processes involving operations for \emph{sequential composition} and \emph{parallel composition}, satisfying reasonable axioms, forms a \emph{monoidal category}.

\MonoidalCategories{} are versatile: they can be used in the description of quantum circuits \cite{abramsky2009categorical}, stochastic processes \cite{cho:jacobs:disintegration2019,fritz:markov2020}, relational queries \cite{bonchi18} and non-terminating processes \cite{cockett02}, among many other applications \cite{coeckeFS16}.

At the same time, monoidal categories have two intuitive, sound and complete calculi: the first in terms of \emph{string diagrams} \cite{joyal91}, and the second in terms of their \emph{linear type theory} \cite{shulman2016categorical}. String diagrams are a 2\hyp{}dimensional syntax in which processes are represented by boxes, and their inputs and outputs are connected by wires. The type theory of symmetric monoidal categories is the basis of the more specialized \emph{arrow do-notation} used in functional programming languages \cite{hughes00,paterson01:arrows}, which becomes \emph{do-notation} for Kleisli categories of commutative monads \cite{moggi91,guitart1980tenseurs}.
Let us showcase \monoidalCategories{}, their string diagrams and the use of do-notation in the description of a protocol.

\subsection{Protocol Description}

The Transmission Control Protocol (TCP) is a connection-based communication protocol. Every connection begins with a \emph{three-way handshake}: an exchange of messages that synchronizes the state of both parties. This handshake is defined in RFC793 to have three steps: \SYN{}, \SYN{}-\ACK{} and \ACK{} \cite{rfc793}.

The client initiates the communication by sending a synchronization packet (\SYN{}) to the server. The synchronization packet contains a pseudorandom number associated to the session, the Initial Sequence Number of the client (\CLI{}).

The server acknowledges this packet and sends a message (\ACK{}) containing its own sequence number (\SRV{}) together with the client's sequence number plus one (\CLI{}$+1$). These two form the \SYN{}-\ACK{} message.
Finally, the client sends a final \ACK{} message with the server's sequence number plus one, $\SRV{}+1$. When the protocol works correctly, both client and server end up with the pair $(\CLI{}+1, \SRV{}+1)$.
\vspace{-1.5em}
\begin{figure}[ht]
  \centering

\tikzset{every picture/.style={line width=0.75pt}} %

\begin{tikzpicture}[x=0.75pt,y=0.75pt,yscale=-1,xscale=1]
\draw  [color={rgb, 255:red, 0; green, 0; blue, 0 }  ,draw opacity=1 ][fill={rgb, 255:red, 255; green, 255; blue, 255 }  ,fill opacity=1 ] (45,70) -- (85,70) -- (85,85) -- (45,85) -- cycle ;
\draw [color={rgb, 255:red, 0; green, 0; blue, 0 }  ,draw opacity=1 ]   (55,50) -- (55,70) ;
\draw  [color={rgb, 255:red, 0; green, 0; blue, 0 }  ,draw opacity=1 ][fill={rgb, 255:red, 255; green, 255; blue, 255 }  ,fill opacity=1 ] (95,115) -- (135,115) -- (135,130) -- (95,130) -- cycle ;
\draw    (75,85) .. controls (75.33,104.83) and (105,95.5) .. (105,115) ;
\draw  [color={rgb, 255:red, 0; green, 0; blue, 0 }  ,draw opacity=1 ][fill={rgb, 255:red, 255; green, 255; blue, 255 }  ,fill opacity=1 ] (145.5,160) -- (185.5,160) -- (185.5,175) -- (145.5,175) -- cycle ;
\draw    (125,130) .. controls (125.33,149.83) and (155,140.5) .. (155,160) ;
\draw [color={rgb, 255:red, 0; green, 0; blue, 0 }  ,draw opacity=1 ]   (175.5,50) -- (175.5,160) ;
\draw  [color={rgb, 255:red, 0; green, 0; blue, 0 }  ,draw opacity=1 ][fill={rgb, 255:red, 255; green, 255; blue, 255 }  ,fill opacity=1 ] (95,205) -- (135,205) -- (135,220) -- (95,220) -- cycle ;
\draw    (105,220) .. controls (105.33,239.83) and (75,225.5) .. (75,245) ;
\draw    (155,175) .. controls (155.33,194.83) and (125,185.5) .. (125,205) ;
\draw [color={rgb, 255:red, 0; green, 0; blue, 0 }  ,draw opacity=1 ]   (55,85) -- (55,245) ;
\draw  [color={rgb, 255:red, 0; green, 0; blue, 0 }  ,draw opacity=1 ][fill={rgb, 255:red, 255; green, 255; blue, 255 }  ,fill opacity=1 ] (45,245) -- (85,245) -- (85,260) -- (45,260) -- cycle ;
\draw [color={rgb, 255:red, 0; green, 0; blue, 0 }  ,draw opacity=1 ]   (175.5,175) -- (175.5,335) ;
\draw  [color={rgb, 255:red, 0; green, 0; blue, 0 }  ,draw opacity=1 ][fill={rgb, 255:red, 255; green, 255; blue, 255 }  ,fill opacity=1 ] (95,290) -- (135,290) -- (135,305) -- (95,305) -- cycle ;
\draw    (75,260) .. controls (75.33,279.83) and (105,270.5) .. (105,290) ;
\draw    (125,305) .. controls (125.33,324.83) and (155,315.5) .. (155,335) ;
\draw  [color={rgb, 255:red, 0; green, 0; blue, 0 }  ,draw opacity=1 ][fill={rgb, 255:red, 255; green, 255; blue, 255 }  ,fill opacity=1 ] (145.5,335) -- (185.5,335) -- (185.5,350) -- (145.5,350) -- cycle ;
\draw [color={rgb, 255:red, 0; green, 0; blue, 0 }  ,draw opacity=1 ]   (55,260) -- (55,385) ;
\draw [color={rgb, 255:red, 0; green, 0; blue, 0 }  ,draw opacity=1 ]   (175,350) -- (175,385) ;
\draw (55,20) node  {\includegraphics[width=22.5pt,height=22.5pt]{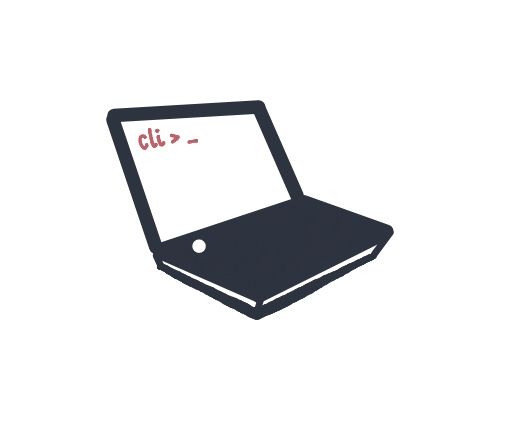}};
\draw (175,20) node  {\includegraphics[width=22.5pt,height=22.5pt]{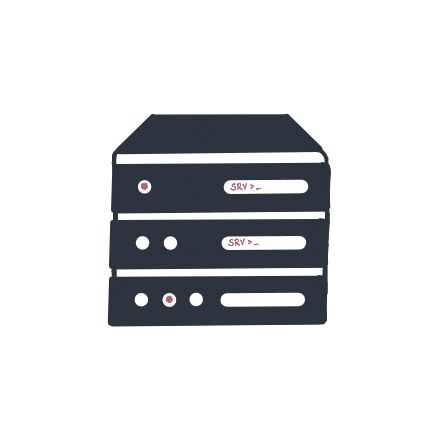}};
\draw  [draw opacity=0][fill={rgb, 255:red, 255; green, 255; blue, 255 }  ,fill opacity=1 ] (35,30) -- (75,30) -- (75,50) -- (35,50) -- cycle ;
\draw  [draw opacity=0][fill={rgb, 255:red, 255; green, 255; blue, 255 }  ,fill opacity=1 ] (155,30) -- (195,30) -- (195,50) -- (155,50) -- cycle ;
\draw  [draw opacity=0][fill={rgb, 255:red, 255; green, 255; blue, 255 }  ,fill opacity=1 ] (100,70) -- (145,70) -- (145,85) -- (100,85) -- cycle ;
\draw  [draw opacity=0][fill={rgb, 255:red, 255; green, 255; blue, 255 }  ,fill opacity=1 ] (100,85) -- (145,85) -- (145,100) -- (100,100) -- cycle ;
\draw  [draw opacity=0][fill={rgb, 255:red, 255; green, 255; blue, 255 }  ,fill opacity=1 ] (90,159) -- (135,159) -- (135,174) -- (90,174) -- cycle ;
\draw  [draw opacity=0][fill={rgb, 255:red, 255; green, 255; blue, 255 }  ,fill opacity=1 ] (90,174) -- (135,174) -- (135,189) -- (90,189) -- cycle ;
\draw  [draw opacity=0][fill={rgb, 255:red, 255; green, 255; blue, 255 }  ,fill opacity=1 ] (105,249) -- (150,249) -- (150,264) -- (105,264) -- cycle ;
\draw  [draw opacity=0][fill={rgb, 255:red, 255; green, 255; blue, 255 }  ,fill opacity=1 ] (105,264) -- (150,264) -- (150,279) -- (105,279) -- cycle ;
\draw  [draw opacity=0][fill={rgb, 255:red, 255; green, 255; blue, 255 }  ,fill opacity=1 ] (5,266) -- (50,266) -- (50,281) -- (5,281) -- cycle ;
\draw  [draw opacity=0][fill={rgb, 255:red, 255; green, 255; blue, 255 }  ,fill opacity=1 ] (5,281) -- (50,281) -- (50,296) -- (5,296) -- cycle ;
\draw  [draw opacity=0][fill={rgb, 255:red, 255; green, 255; blue, 255 }  ,fill opacity=1 ] (180,356) -- (225,356) -- (225,371) -- (180,371) -- cycle ;
\draw  [draw opacity=0][fill={rgb, 255:red, 255; green, 255; blue, 255 }  ,fill opacity=1 ] (180,371) -- (225,371) -- (225,386) -- (180,386) -- cycle ;
\draw  [draw opacity=0][fill={rgb, 255:red, 255; green, 255; blue, 255 }  ,fill opacity=1 ] (180,179) -- (225,179) -- (225,194) -- (180,194) -- cycle ;
\draw  [draw opacity=0][fill={rgb, 255:red, 255; green, 255; blue, 255 }  ,fill opacity=1 ] (180,194) -- (225,194) -- (225,209) -- (180,209) -- cycle ;
\draw  [draw opacity=0][fill={rgb, 255:red, 255; green, 255; blue, 255 }  ,fill opacity=1 ] (5,91) -- (50,91) -- (50,106) -- (5,106) -- cycle ;
\draw  [draw opacity=0][fill={rgb, 255:red, 255; green, 255; blue, 255 }  ,fill opacity=1 ] (5,106) -- (50,106) -- (50,121) -- (5,121) -- cycle ;

\draw (55.5,42.5) node  [font=\footnotesize]  {\textbf{Client}};
\draw (175,42.5) node  [font=\footnotesize]  {\textbf{Server}};
\draw (65,77.5) node  [font=\footnotesize] [align=left] {{\SYN}};
\draw (115,122.5) node  [font=\footnotesize] [align=left] {{\NOISE}};
\draw (122.5,77.5) node  [font=\footnotesize] [align=left] {{\fontfamily{pcr}\selectfont SYN:10}};
\draw (122.5,92.5) node  [font=\footnotesize] [align=left] {{\fontfamily{pcr}\selectfont ACK:00}};
\draw (112.5,166.5) node  [font=\footnotesize] [align=left] {{\fontfamily{pcr}\selectfont SYN:11}};
\draw (112.5,181.5) node  [font=\footnotesize] [align=left] {{\fontfamily{pcr}\selectfont ACK:20}};
\draw (127.5,256.5) node  [font=\footnotesize] [align=left] {{\fontfamily{pcr}\selectfont SYN:11}};
\draw (127.5,271.5) node  [font=\footnotesize] [align=left] {{\fontfamily{pcr}\selectfont ACK:21}};
\draw (165.5,167.5) node  [font=\footnotesize] [align=left] {{\SYN-\ACK}};
\draw (27.5,273.5) node  [font=\footnotesize] [align=left] {{\fontfamily{pcr}\selectfont CLI:11}};
\draw (27.5,288.5) node  [font=\footnotesize] [align=left] {{\fontfamily{pcr}\selectfont SRV:21}};
\draw (202.5,363.5) node  [font=\footnotesize] [align=left] {{\fontfamily{pcr}\selectfont CLI:11}};
\draw (202.5,378.5) node  [font=\footnotesize] [align=left] {{\fontfamily{pcr}\selectfont SRV:21}};
\draw (202.5,186.5) node  [font=\footnotesize] [align=left] {{\fontfamily{pcr}\selectfont CLI:11}};
\draw (202.5,201.5) node  [font=\footnotesize] [align=left] {{\fontfamily{pcr}\selectfont SRV:20}};
\draw (27.5,98.5) node  [font=\footnotesize] [align=left] {{\fontfamily{pcr}\selectfont CLI:10}};
\draw (27.5,113.5) node  [font=\footnotesize] [align=left] {{\fontfamily{pcr}\selectfont SRV:00}};
\draw (115,212.5) node  [font=\footnotesize] [align=left] {{\NOISE}};
\draw (115,297.5) node  [font=\footnotesize] [align=left] {{\NOISE}};
\draw (65,252.5) node  [font=\footnotesize] [align=left] {{\ACK}};
\draw (165.5,343.5) node  [font=\footnotesize] [align=left] {{\RCV}};

\end{tikzpicture}
   \caption{TCP Three way handshake.}
  \label{diagram:tcp}
\end{figure}
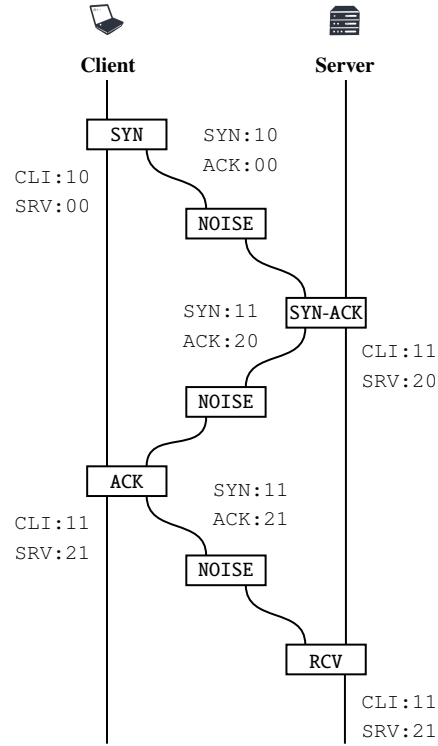

This protocol is traditionally described in terms of a communication diagram (\Cref{diagram:tcp}). This diagram can be taken seriously as a formal mathematical object: it is a string diagram describing a \emph{morphism} in a \monoidalCategory{}.
\vspace{-1em}
\begin{figure}[ht]
  \centering
  \begin{verbatim}
    syn :: Client ~> (Client, Syn, Ack)
    syn(client) = do
      client <- random
      return (client, client, 0)\end{verbatim}
  \caption{Implementation of the \SYN{} component.}
  \label{diagram:syn}
\end{figure}

The implementation of each component of the protocol is traditionally written as pseudocode. This pseudocode can also be taken seriously as the expression of a morphism in the same \monoidalCategory{}, possibly with extra structure: in this case, a commutative \emph{Freyd category} (\Cref{diagram:syn}, see Appendix \Cref{sec:threewayhandshake} \cite{moggi91}). %
That is, \symmetricMonoidalCategories{} admit two different internal languages, and we can use both to interpret formally the traditional description of a protocol in terms of string diagrams and pseudocode.

\subsection{Types for Message Passing}
The last part in formalizing a multi-party protocol in terms of \monoidalCategories{} is to actually separate its component parties.
For instance, the three-way handshake can be split into the client, the server and a channel.
Here is where the existing literature in \monoidalCategories{} seems to fall short: the parts resulting from the decomposition of a monoidal morphism are not necessarily monoidal morphisms themselves (see \Cref{diagram:tcpsplit} for the diagrammatic representation). We say that these are only \emph{monoidal contexts}.
\vspace{-0.2em}
\begin{figure}[ht]
  \centering

\tikzset{every picture/.style={line width=0.75pt}} %

\begin{tikzpicture}[x=0.75pt,y=0.75pt,yscale=-1,xscale=1]
\draw    (165,305) .. controls (165.33,324.83) and (195,315.5) .. (195,335) ;
\draw  [draw opacity=0][fill={rgb, 255:red, 255; green, 255; blue, 255 }  ,fill opacity=1 ] (165,335) .. controls (165,314.63) and (194.75,324.88) .. (195,305) -- (200,305) -- (200,340) -- (165.13,340) .. controls (165.04,338.38) and (165,336.71) .. (165,335) -- cycle ;
\draw    (115,260) .. controls (115.33,279.83) and (145,270.5) .. (145,290) ;
\draw  [draw opacity=0][fill={rgb, 255:red, 255; green, 255; blue, 255 }  ,fill opacity=1 ] (145,260) .. controls (145,280.38) and (115.25,270.13) .. (115,290) -- (110,290) -- (110,255) -- (144.87,255) .. controls (144.96,256.62) and (145,258.29) .. (145,260) -- cycle ;
\draw    (115,85) .. controls (115.33,104.83) and (145,95.5) .. (145,115) ;
\draw  [draw opacity=0][fill={rgb, 255:red, 255; green, 255; blue, 255 }  ,fill opacity=1 ] (145,85) .. controls (145,105.38) and (115.25,95.13) .. (115,115) -- (110,115) -- (110,80) -- (144.87,80) .. controls (144.96,81.62) and (145,83.29) .. (145,85) -- cycle ;
\draw    (165,130) .. controls (165.33,149.83) and (195,140.5) .. (195,160) ;
\draw  [draw opacity=0][fill={rgb, 255:red, 255; green, 255; blue, 255 }  ,fill opacity=1 ] (165,160) .. controls (165,139.63) and (194.75,149.88) .. (195,130) -- (200,130) -- (200,165) -- (165.13,165) .. controls (165.04,163.38) and (165,161.71) .. (165,160) -- cycle ;
\draw    (195,175) .. controls (195.33,194.83) and (165,185.5) .. (165,205) ;
\draw  [draw opacity=0][fill={rgb, 255:red, 255; green, 255; blue, 255 }  ,fill opacity=1 ] (165,175) .. controls (165,195.38) and (194.75,185.13) .. (195,205) -- (200,205) -- (200,170) -- (165.13,170) .. controls (165.04,171.62) and (165,173.29) .. (165,175) -- cycle ;
\draw    (210,305) .. controls (210.33,324.83) and (240,315.5) .. (240,335) ;
\draw  [draw opacity=0][fill={rgb, 255:red, 255; green, 255; blue, 255 }  ,fill opacity=1 ] (240,305) .. controls (240,325.38) and (210.25,315.13) .. (210,335) -- (205,335) -- (205,300) -- (239.87,300) .. controls (239.96,301.62) and (240,303.29) .. (240,305) -- cycle ;
\draw    (240,175) .. controls (240.33,194.83) and (210,185.5) .. (210,205) ;
\draw  [draw opacity=0][fill={rgb, 255:red, 255; green, 255; blue, 255 }  ,fill opacity=1 ] (240,205) .. controls (240,184.63) and (210.25,194.88) .. (210,175) -- (205,175) -- (205,210) -- (239.87,210) .. controls (239.96,208.38) and (240,206.71) .. (240,205) -- cycle ;
\draw    (210,130) .. controls (210.33,149.83) and (240,140.5) .. (240,160) ;
\draw  [draw opacity=0][fill={rgb, 255:red, 255; green, 255; blue, 255 }  ,fill opacity=1 ] (240,130) .. controls (240,150.38) and (210.25,140.13) .. (210,160) -- (205,160) -- (205,125) -- (239.87,125) .. controls (239.96,126.62) and (240,128.29) .. (240,130) -- cycle ;
\draw  [color={rgb, 255:red, 0; green, 0; blue, 0 }  ,draw opacity=1 ][fill={rgb, 255:red, 255; green, 255; blue, 255 }  ,fill opacity=1 ] (45,70) -- (85,70) -- (85,85) -- (45,85) -- cycle ;
\draw [color={rgb, 255:red, 0; green, 0; blue, 0 }  ,draw opacity=1 ]   (55,50) -- (55,70) ;
\draw  [color={rgb, 255:red, 0; green, 0; blue, 0 }  ,draw opacity=1 ][fill={rgb, 255:red, 255; green, 255; blue, 255 }  ,fill opacity=1 ] (135,115) -- (175,115) -- (175,130) -- (135,130) -- cycle ;
\draw    (75,85) .. controls (75.33,104.83) and (105,95.5) .. (105,115) ;
\draw  [color={rgb, 255:red, 0; green, 0; blue, 0 }  ,draw opacity=1 ][fill={rgb, 255:red, 255; green, 255; blue, 255 }  ,fill opacity=1 ] (230.5,160) -- (270.5,160) -- (270.5,175) -- (230.5,175) -- cycle ;
\draw [color={rgb, 255:red, 0; green, 0; blue, 0 }  ,draw opacity=1 ]   (260.5,50) -- (260.5,160) ;
\draw  [color={rgb, 255:red, 0; green, 0; blue, 0 }  ,draw opacity=1 ][fill={rgb, 255:red, 255; green, 255; blue, 255 }  ,fill opacity=1 ] (135,205) -- (175,205) -- (175,220) -- (135,220) -- cycle ;
\draw    (105,220) .. controls (105.33,239.83) and (75,225.5) .. (75,245) ;
\draw [color={rgb, 255:red, 0; green, 0; blue, 0 }  ,draw opacity=1 ]   (55,85) -- (55,245) ;
\draw  [color={rgb, 255:red, 0; green, 0; blue, 0 }  ,draw opacity=1 ][fill={rgb, 255:red, 255; green, 255; blue, 255 }  ,fill opacity=1 ] (45,245) -- (85,245) -- (85,260) -- (45,260) -- cycle ;
\draw [color={rgb, 255:red, 0; green, 0; blue, 0 }  ,draw opacity=1 ]   (260.5,175) -- (260.5,335) ;
\draw  [color={rgb, 255:red, 0; green, 0; blue, 0 }  ,draw opacity=1 ][fill={rgb, 255:red, 255; green, 255; blue, 255 }  ,fill opacity=1 ] (135,290) -- (175,290) -- (175,305) -- (135,305) -- cycle ;
\draw    (75,260) .. controls (75.33,279.83) and (105,270.5) .. (105,290) ;
\draw  [color={rgb, 255:red, 0; green, 0; blue, 0 }  ,draw opacity=1 ][fill={rgb, 255:red, 255; green, 255; blue, 255 }  ,fill opacity=1 ] (230.5,335) -- (270.5,335) -- (270.5,350) -- (230.5,350) -- cycle ;
\draw [color={rgb, 255:red, 0; green, 0; blue, 0 }  ,draw opacity=1 ]   (55,260) -- (55,385) ;
\draw [color={rgb, 255:red, 0; green, 0; blue, 0 }  ,draw opacity=1 ]   (260,350) -- (260,385) ;
\draw (55,20) node  {\includegraphics[width=22.5pt,height=22.5pt]{laptop-logo.jpg}};
\draw (260,20) node  {\includegraphics[width=22.5pt,height=22.5pt]{server-logo.jpg}};
\draw  [draw opacity=0][fill={rgb, 255:red, 255; green, 255; blue, 255 }  ,fill opacity=1 ] (35,30) -- (75,30) -- (75,50) -- (35,50) -- cycle ;
\draw  [draw opacity=0][fill={rgb, 255:red, 255; green, 255; blue, 255 }  ,fill opacity=1 ] (240,30) -- (280,30) -- (280,50) -- (240,50) -- cycle ;
\draw [color={rgb, 255:red, 191; green, 97; blue, 106 }  ,draw opacity=1 ] [dash pattern={on 4.5pt off 4.5pt}]  (105,5) -- (105,85) ;
\draw [color={rgb, 255:red, 191; green, 97; blue, 106 }  ,draw opacity=1 ] [dash pattern={on 4.5pt off 4.5pt}]  (195,5) -- (195,130) ;
\draw [color={rgb, 255:red, 191; green, 97; blue, 106 }  ,draw opacity=1 ] [dash pattern={on 4.5pt off 4.5pt}]  (75,115) -- (75,220) ;
\draw [color={rgb, 255:red, 191; green, 97; blue, 106 }  ,draw opacity=1 ] [dash pattern={on 4.5pt off 4.5pt}]  (195,130) .. controls (195.33,149.83) and (165,140.5) .. (165,160) ;
\draw [color={rgb, 255:red, 191; green, 97; blue, 106 }  ,draw opacity=1 ] [dash pattern={on 4.5pt off 4.5pt}]  (165,175) .. controls (165.33,194.83) and (195,185.5) .. (195,205) ;
\draw [color={rgb, 255:red, 191; green, 97; blue, 106 }  ,draw opacity=1 ] [dash pattern={on 4.5pt off 4.5pt}]  (165,160) -- (165,175) ;
\draw [color={rgb, 255:red, 191; green, 97; blue, 106 }  ,draw opacity=1 ] [dash pattern={on 4.5pt off 4.5pt}]  (105,245) -- (105,260) ;
\draw [color={rgb, 255:red, 191; green, 97; blue, 106 }  ,draw opacity=1 ] [dash pattern={on 4.5pt off 4.5pt}]  (75,290) -- (75,385) ;
\draw [color={rgb, 255:red, 191; green, 97; blue, 106 }  ,draw opacity=1 ] [dash pattern={on 4.5pt off 4.5pt}]  (195,305) .. controls (195.33,324.83) and (165,315.5) .. (165,335) ;
\draw [color={rgb, 255:red, 191; green, 97; blue, 106 }  ,draw opacity=1 ] [dash pattern={on 4.5pt off 4.5pt}]  (195,200) -- (195,305) ;
\draw [color={rgb, 255:red, 191; green, 97; blue, 106 }  ,draw opacity=1 ] [dash pattern={on 4.5pt off 4.5pt}]  (165,330) -- (165,380) ;
\draw (170,20) node  {\includegraphics[width=22.5pt,height=22.5pt]{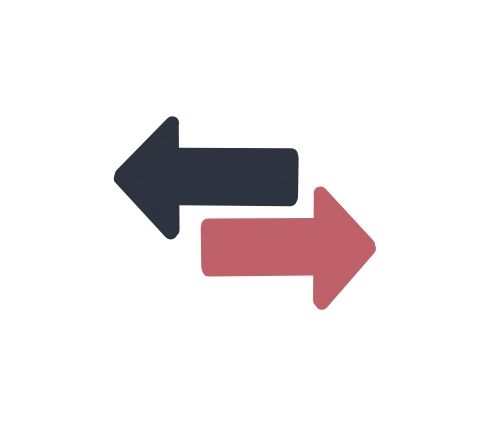}};
\draw  [draw opacity=0][fill={rgb, 255:red, 255; green, 255; blue, 255 }  ,fill opacity=1 ] (150,30) -- (190,30) -- (190,50) -- (150,50) -- cycle ;
\draw    (145,220) .. controls (145.33,239.83) and (115,225.5) .. (115,245) ;
\draw [color={rgb, 255:red, 191; green, 97; blue, 106 }  ,draw opacity=1 ] [dash pattern={on 4.5pt off 4.5pt}]  (145,5) -- (145,85) ;
\draw [color={rgb, 255:red, 191; green, 97; blue, 106 }  ,draw opacity=1 ] [dash pattern={on 4.5pt off 4.5pt}]  (145,85) .. controls (145.33,104.83) and (115,95.5) .. (115,115) ;
\draw [color={rgb, 255:red, 191; green, 97; blue, 106 }  ,draw opacity=1 ] [dash pattern={on 4.5pt off 4.5pt}]  (145,260) .. controls (145.33,279.83) and (115,270.5) .. (115,290) ;
\draw [color={rgb, 255:red, 191; green, 97; blue, 106 }  ,draw opacity=1 ] [dash pattern={on 4.5pt off 4.5pt}]  (115,115) -- (115,220) ;
\draw [color={rgb, 255:red, 191; green, 97; blue, 106 }  ,draw opacity=1 ] [dash pattern={on 4.5pt off 4.5pt}]  (115,290) -- (115,385) ;
\draw [color={rgb, 255:red, 191; green, 97; blue, 106 }  ,draw opacity=1 ] [dash pattern={on 4.5pt off 4.5pt}]  (240,5) -- (240,130) ;
\draw [color={rgb, 255:red, 191; green, 97; blue, 106 }  ,draw opacity=1 ] [dash pattern={on 4.5pt off 4.5pt}]  (240,130) .. controls (240.33,149.83) and (210,140.5) .. (210,160) ;
\draw [color={rgb, 255:red, 191; green, 97; blue, 106 }  ,draw opacity=1 ] [dash pattern={on 4.5pt off 4.5pt}]  (210,160) -- (210,175) ;
\draw [color={rgb, 255:red, 191; green, 97; blue, 106 }  ,draw opacity=1 ] [dash pattern={on 4.5pt off 4.5pt}]  (210,175) .. controls (210.33,194.83) and (240,185.5) .. (240,205) ;
\draw [color={rgb, 255:red, 191; green, 97; blue, 106 }  ,draw opacity=1 ] [dash pattern={on 4.5pt off 4.5pt}]  (240,305) .. controls (240.33,324.83) and (210,315.5) .. (210,335) ;
\draw [color={rgb, 255:red, 191; green, 97; blue, 106 }  ,draw opacity=1 ] [dash pattern={on 4.5pt off 4.5pt}]  (240,200) -- (240,305) ;
\draw [color={rgb, 255:red, 191; green, 97; blue, 106 }  ,draw opacity=1 ] [dash pattern={on 4.5pt off 4.5pt}]  (210,330) -- (210,380) ;
\draw  [draw opacity=0][fill={rgb, 255:red, 255; green, 255; blue, 255 }  ,fill opacity=1 ] (145,250) .. controls (145,229.63) and (115.25,239.88) .. (115,220) -- (110,220) -- (110,255) -- (144.87,255) .. controls (144.96,253.38) and (145,251.71) .. (145,250) -- cycle ;
\draw [color={rgb, 255:red, 191; green, 97; blue, 106 }  ,draw opacity=1 ] [dash pattern={on 4.5pt off 4.5pt}]  (115,220) .. controls (115.33,239.83) and (145,230.5) .. (145,250) ;
\draw [color={rgb, 255:red, 191; green, 97; blue, 106 }  ,draw opacity=1 ] [dash pattern={on 4.5pt off 4.5pt}]  (145,250) -- (145,260) ;
\draw  [draw opacity=0][fill={rgb, 255:red, 255; green, 255; blue, 255 }  ,fill opacity=1 ] (75,290) .. controls (75,269.63) and (104.75,279.88) .. (105,260) -- (110,260) -- (110,295) -- (75.13,295) .. controls (75.04,293.38) and (75,291.71) .. (75,290) -- cycle ;
\draw [color={rgb, 255:red, 191; green, 97; blue, 106 }  ,draw opacity=1 ] [dash pattern={on 4.5pt off 4.5pt}]  (105,260) .. controls (105.33,279.83) and (75,270.5) .. (75,290) ;
\draw  [draw opacity=0][fill={rgb, 255:red, 255; green, 255; blue, 255 }  ,fill opacity=1 ] (75,115) .. controls (75,94.63) and (104.75,104.88) .. (105,85) -- (110,85) -- (110,120) -- (75.13,120) .. controls (75.04,118.38) and (75,116.71) .. (75,115) -- cycle ;
\draw [color={rgb, 255:red, 191; green, 97; blue, 106 }  ,draw opacity=1 ] [dash pattern={on 4.5pt off 4.5pt}]  (105,85) .. controls (105.33,104.83) and (75,95.5) .. (75,115) ;
\draw  [draw opacity=0][fill={rgb, 255:red, 255; green, 255; blue, 255 }  ,fill opacity=1 ] (75,220) .. controls (75,240.38) and (104.75,230.13) .. (105,250) -- (110,250) -- (110,215) -- (75.13,215) .. controls (75.04,216.62) and (75,218.29) .. (75,220) -- cycle ;
\draw [color={rgb, 255:red, 191; green, 97; blue, 106 }  ,draw opacity=1 ] [dash pattern={on 4.5pt off 4.5pt}]  (75,220) .. controls (75.33,239.83) and (105,230.5) .. (105,250) ;

\draw (55.5,42.5) node  [font=\footnotesize]  {\textbf{Client}};
\draw (260,42.5) node  [font=\footnotesize]  {\textbf{Server}};
\draw (65,77.5) node  [font=\footnotesize] [align=left] {{\SYN}};
\draw (155,122.5) node  [font=\footnotesize] [align=left] {\NOISE};
\draw (250.5,167.5) node  [font=\footnotesize] [align=left] {{\SYN-\ACK}};
\draw (155,212.5) node  [font=\footnotesize] [align=left] {{\NOISE}};
\draw (155,297.5) node  [font=\footnotesize] [align=left] {{\NOISE}};
\draw (65,252.5) node  [font=\footnotesize] [align=left] {{\ACK}};
\draw (250.5,343.5) node  [font=\footnotesize] [align=left] {{\RCV}};
\draw (170.5,42.5) node  [font=\footnotesize]  {\textbf{Channel}};

\end{tikzpicture}
   \caption{Parties in the TCP Three-way handshake.}
  \label{diagram:tcpsplit}
\end{figure}

Contrary to monoidal morphisms, which only need to declare their input and output types, monoidal contexts need \emph{behavioural types} \cite{pierce93:subtyping,huttel16:behaviouraltypes} that specify the order and type of the exchange of information along their boundary.

A monoidal context may declare intermediate \emph{send} ($\Send{A}$) and \emph{receive} ($\Get{A}$) types, separated by a sequencing operator $(◁)$.
For instance, the channel is a monoidal morphism just declaring that it takes an input message (\textbf{Msg}) and produces another output message;
but the client is a monoidal context that transforms its memory type $\mathbf{Client} \to \mathbf{Client}$ at the same time it \emph{sends}, \emph{receives} and then \emph{sends} a message; and the server transforms its type $\mathbf{Server} \to \mathbf{Server}$ while, dually to the client, it \emph{receives}, \emph{sends} and then \emph{receives} a message.
\begin{align*}
\ClientLogo{} &\in
𝓛ℂ\left( \biobj{\mathbf{Client}}{\mathbf{Client}} \mathbin{;}
  \Send{\mathbf{Msg}} \triangleleft
  \Get{\mathbf{Msg}} \triangleleft
  \Send{\mathbf{Msg}} \right); \\[-0.5em]
\ServerLogo{} &\in
  𝓛ℂ\left( \biobj{\mathbf{Server}}{\mathbf{Server}} \mathbin{;}
    \Get{\mathbf{Msg}} \triangleleft
    \Send{\mathbf{Msg}} \triangleleft
    \Get{\mathbf{Msg}} \right); \\[-0.4em]
\mbox{\texttt{NOISE}} &\in ℂ\left(\mathbf{Msg}; \mathbf{Msg} \right);
\end{align*}

Session types \cite{honda08:sessionTypes}, including the send $(\Send{A})$ and receive $(\Get{A})$ polarized types, have been commonplace in logics of message passing. Cockett and Pastro \cite{cockettPastro09:messagepassing} already proposed a categorical semantics for message-passing which, however, needs to go beyond monoidal categories, into \emph{linear actegories} and \emph{polyactegories}.

Our claim is that, perhaps surprisingly, \monoidalCategories{} already have the necessary algebraic structure to define \emph{monoidal contexts} and their send-receive polarized types. Latent to any monoidal category, there exists a universal category of contexts with polarized types $(\Send{}/\Get{})$ and parallel/sequence operators $(\otimes/◁)$.

\subsection{Reasoning with Contexts}

This manuscript introduces the notion of \monoidalContext{} and \symmetricMonoidalContext{}; and it explains how \dinaturality{} allows us to reason with them.
In the same way that we reason with monoidal morphisms using string diagrams, we can reason about \monoidalContexts{} using \emph{incomplete string diagrams} \cite{bartlettvicary15:modularcategories,roman21:opendiagrams}.

For instance, consider the following fact about the TCP three-way handshake: the client does not need to store a starting $\SRV{}$ number for the server, as it will be overwritten as soon as the real one arrives. This fact only concerns the actions of the client, and it is independent of the server and the channel. We would like to reason about it preserving this modularity, and this is what the incomplete diagrams in \Cref{diagram:tcpclient} achieve.
\vspace{-1.4em}
\begin{figure}[ht]
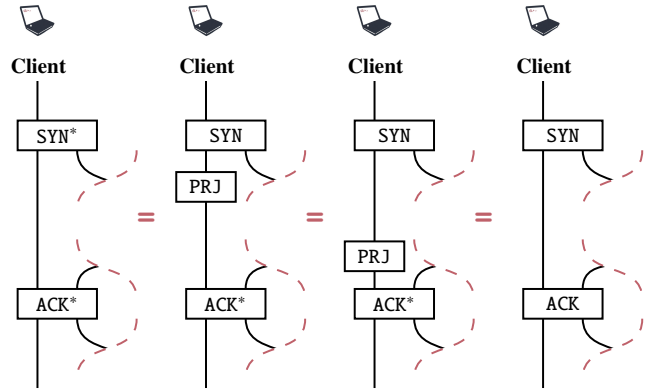

  \centering

\tikzset{every picture/.style={line width=0.75pt}} %

\begin{tikzpicture}[x=0.75pt,y=0.75pt,yscale=-1,xscale=1]
\draw  [color={rgb, 255:red, 0; green, 0; blue, 0 }  ,draw opacity=1 ][fill={rgb, 255:red, 255; green, 255; blue, 255 }  ,fill opacity=1 ] (105,70) -- (145,70) -- (145,85) -- (105,85) -- cycle ;
\draw [color={rgb, 255:red, 0; green, 0; blue, 0 }  ,draw opacity=1 ]   (115,50) -- (115,70) ;
\draw [color={rgb, 255:red, 0; green, 0; blue, 0 }  ,draw opacity=1 ]   (115,85) -- (115,155) ;
\draw  [color={rgb, 255:red, 0; green, 0; blue, 0 }  ,draw opacity=1 ][fill={rgb, 255:red, 255; green, 255; blue, 255 }  ,fill opacity=1 ] (105,155) -- (145,155) -- (145,170) -- (105,170) -- cycle ;
\draw [color={rgb, 255:red, 0; green, 0; blue, 0 }  ,draw opacity=1 ]   (115,170) -- (115,205) ;
\draw (115,20) node  {\includegraphics[width=22.5pt,height=22.5pt]{laptop-logo.jpg}};
\draw  [draw opacity=0][fill={rgb, 255:red, 255; green, 255; blue, 255 }  ,fill opacity=1 ] (95,30) -- (135,30) -- (135,50) -- (95,50) -- cycle ;
\draw  [color={rgb, 255:red, 0; green, 0; blue, 0 }  ,draw opacity=1 ][fill={rgb, 255:red, 255; green, 255; blue, 255 }  ,fill opacity=1 ] (190,70) -- (230,70) -- (230,85) -- (190,85) -- cycle ;
\draw [color={rgb, 255:red, 0; green, 0; blue, 0 }  ,draw opacity=1 ]   (200,50) -- (200,70) ;
\draw  [color={rgb, 255:red, 0; green, 0; blue, 0 }  ,draw opacity=1 ][fill={rgb, 255:red, 255; green, 255; blue, 255 }  ,fill opacity=1 ] (190,155) -- (230,155) -- (230,170) -- (190,170) -- cycle ;
\draw (200,20) node  {\includegraphics[width=22.5pt,height=22.5pt]{laptop-logo.jpg}};
\draw  [draw opacity=0][fill={rgb, 255:red, 255; green, 255; blue, 255 }  ,fill opacity=1 ] (180,30) -- (220,30) -- (220,50) -- (180,50) -- cycle ;
\draw [color={rgb, 255:red, 0; green, 0; blue, 0 }  ,draw opacity=1 ]   (200,85) -- (200,155) ;
\draw [color={rgb, 255:red, 0; green, 0; blue, 0 }  ,draw opacity=1 ]   (200,170) -- (200,205) ;
\draw  [color={rgb, 255:red, 0; green, 0; blue, 0 }  ,draw opacity=1 ][fill={rgb, 255:red, 255; green, 255; blue, 255 }  ,fill opacity=1 ] (185,96) -- (215,96) -- (215,111) -- (185,111) -- cycle ;
\draw  [color={rgb, 255:red, 0; green, 0; blue, 0 }  ,draw opacity=1 ][fill={rgb, 255:red, 255; green, 255; blue, 255 }  ,fill opacity=1 ] (275,70) -- (315,70) -- (315,85) -- (275,85) -- cycle ;
\draw [color={rgb, 255:red, 0; green, 0; blue, 0 }  ,draw opacity=1 ]   (285,50) -- (285,70) ;
\draw  [color={rgb, 255:red, 0; green, 0; blue, 0 }  ,draw opacity=1 ][fill={rgb, 255:red, 255; green, 255; blue, 255 }  ,fill opacity=1 ] (275,155) -- (315,155) -- (315,170) -- (275,170) -- cycle ;
\draw (285,20) node  {\includegraphics[width=22.5pt,height=22.5pt]{laptop-logo.jpg}};
\draw  [draw opacity=0][fill={rgb, 255:red, 255; green, 255; blue, 255 }  ,fill opacity=1 ] (265,30) -- (305,30) -- (305,50) -- (265,50) -- cycle ;
\draw [color={rgb, 255:red, 0; green, 0; blue, 0 }  ,draw opacity=1 ]   (285,85) -- (285,155) ;
\draw [color={rgb, 255:red, 0; green, 0; blue, 0 }  ,draw opacity=1 ]   (285,170) -- (285,205) ;
\draw  [color={rgb, 255:red, 0; green, 0; blue, 0 }  ,draw opacity=1 ][fill={rgb, 255:red, 255; green, 255; blue, 255 }  ,fill opacity=1 ] (360,70) -- (400,70) -- (400,85) -- (360,85) -- cycle ;
\draw [color={rgb, 255:red, 0; green, 0; blue, 0 }  ,draw opacity=1 ]   (370,50) -- (370,70) ;
\draw [color={rgb, 255:red, 0; green, 0; blue, 0 }  ,draw opacity=1 ]   (370,85) -- (370,155) ;
\draw  [color={rgb, 255:red, 0; green, 0; blue, 0 }  ,draw opacity=1 ][fill={rgb, 255:red, 255; green, 255; blue, 255 }  ,fill opacity=1 ] (360,155) -- (400,155) -- (400,170) -- (360,170) -- cycle ;
\draw [color={rgb, 255:red, 0; green, 0; blue, 0 }  ,draw opacity=1 ]   (370,170) -- (370,205) ;
\draw (370,20) node  {\includegraphics[width=22.5pt,height=22.5pt]{laptop-logo.jpg}};
\draw  [draw opacity=0][fill={rgb, 255:red, 255; green, 255; blue, 255 }  ,fill opacity=1 ] (350,30) -- (390,30) -- (390,50) -- (350,50) -- cycle ;
\draw    (135,85) .. controls (135.33,104.83) and (165,95.5) .. (165,115) ;
\draw  [draw opacity=0][fill={rgb, 255:red, 255; green, 255; blue, 255 }  ,fill opacity=1 ] (135,115) .. controls (135,94.63) and (164.75,104.88) .. (165,85) -- (170,85) -- (170,120) -- (135.13,120) .. controls (135.04,118.38) and (135,116.71) .. (135,115) -- cycle ;
\draw [color={rgb, 255:red, 191; green, 97; blue, 106 }  ,draw opacity=1 ] [dash pattern={on 4.5pt off 4.5pt}]  (165,85) .. controls (165.33,104.83) and (135,95.5) .. (135,115) ;
\draw    (165,130) .. controls (165.33,149.83) and (135,135.5) .. (135,155) ;
\draw  [draw opacity=0][fill={rgb, 255:red, 255; green, 255; blue, 255 }  ,fill opacity=1 ] (135,130) .. controls (135,150.38) and (164.75,140.13) .. (165,160) -- (170,160) -- (170,125) -- (135.13,125) .. controls (135.04,126.62) and (135,128.29) .. (135,130) -- cycle ;
\draw [color={rgb, 255:red, 191; green, 97; blue, 106 }  ,draw opacity=1 ] [dash pattern={on 4.5pt off 4.5pt}]  (135,130) .. controls (135.33,149.83) and (165,140.5) .. (165,160) ;
\draw  [draw opacity=0] (160,105) -- (180,105) -- (180,135) -- (160,135) -- cycle ;
\draw    (135,170) .. controls (135.33,189.83) and (165,180.5) .. (165,200) ;
\draw  [draw opacity=0][fill={rgb, 255:red, 255; green, 255; blue, 255 }  ,fill opacity=1 ] (135,200) .. controls (135,179.63) and (164.75,189.88) .. (165,170) -- (170,170) -- (170,205) -- (135.13,205) .. controls (135.04,203.38) and (135,201.71) .. (135,200) -- cycle ;
\draw [color={rgb, 255:red, 191; green, 97; blue, 106 }  ,draw opacity=1 ] [dash pattern={on 4.5pt off 4.5pt}]  (165,170) .. controls (165.33,189.83) and (135,180.5) .. (135,200) ;
\draw    (220,85) .. controls (220.33,104.83) and (250,95.5) .. (250,115) ;
\draw  [draw opacity=0][fill={rgb, 255:red, 255; green, 255; blue, 255 }  ,fill opacity=1 ] (220,115) .. controls (220,94.63) and (249.75,104.88) .. (250,85) -- (255,85) -- (255,120) -- (220.13,120) .. controls (220.04,118.38) and (220,116.71) .. (220,115) -- cycle ;
\draw [color={rgb, 255:red, 191; green, 97; blue, 106 }  ,draw opacity=1 ] [dash pattern={on 4.5pt off 4.5pt}]  (250,85) .. controls (250.33,104.83) and (220,95.5) .. (220,115) ;
\draw    (250,130) .. controls (250.33,149.83) and (220,135.5) .. (220,155) ;
\draw  [draw opacity=0][fill={rgb, 255:red, 255; green, 255; blue, 255 }  ,fill opacity=1 ] (220,130) .. controls (220,150.38) and (249.75,140.13) .. (250,160) -- (255,160) -- (255,125) -- (220.13,125) .. controls (220.04,126.62) and (220,128.29) .. (220,130) -- cycle ;
\draw [color={rgb, 255:red, 191; green, 97; blue, 106 }  ,draw opacity=1 ] [dash pattern={on 4.5pt off 4.5pt}]  (220,130) .. controls (220.33,149.83) and (250,140.5) .. (250,160) ;
\draw    (220,170) .. controls (220.33,189.83) and (250,180.5) .. (250,200) ;
\draw  [draw opacity=0][fill={rgb, 255:red, 255; green, 255; blue, 255 }  ,fill opacity=1 ] (220,200) .. controls (220,179.63) and (249.75,189.88) .. (250,170) -- (255,170) -- (255,205) -- (220.13,205) .. controls (220.04,203.38) and (220,201.71) .. (220,200) -- cycle ;
\draw [color={rgb, 255:red, 191; green, 97; blue, 106 }  ,draw opacity=1 ] [dash pattern={on 4.5pt off 4.5pt}]  (250,170) .. controls (250.33,189.83) and (220,180.5) .. (220,200) ;
\draw  [draw opacity=0] (245,105) -- (265,105) -- (265,135) -- (245,135) -- cycle ;
\draw  [color={rgb, 255:red, 0; green, 0; blue, 0 }  ,draw opacity=1 ][fill={rgb, 255:red, 255; green, 255; blue, 255 }  ,fill opacity=1 ] (270,131) -- (300,131) -- (300,146) -- (270,146) -- cycle ;
\draw    (305,85) .. controls (305.33,104.83) and (335,95.5) .. (335,115) ;
\draw  [draw opacity=0][fill={rgb, 255:red, 255; green, 255; blue, 255 }  ,fill opacity=1 ] (305,115) .. controls (305,94.63) and (334.75,104.88) .. (335,85) -- (340,85) -- (340,120) -- (305.13,120) .. controls (305.04,118.38) and (305,116.71) .. (305,115) -- cycle ;
\draw [color={rgb, 255:red, 191; green, 97; blue, 106 }  ,draw opacity=1 ] [dash pattern={on 4.5pt off 4.5pt}]  (335,85) .. controls (335.33,104.83) and (305,95.5) .. (305,115) ;
\draw    (335,130) .. controls (335.33,149.83) and (305,135.5) .. (305,155) ;
\draw  [draw opacity=0][fill={rgb, 255:red, 255; green, 255; blue, 255 }  ,fill opacity=1 ] (305,130) .. controls (305,150.38) and (334.75,140.13) .. (335,160) -- (340,160) -- (340,125) -- (305.13,125) .. controls (305.04,126.62) and (305,128.29) .. (305,130) -- cycle ;
\draw [color={rgb, 255:red, 191; green, 97; blue, 106 }  ,draw opacity=1 ] [dash pattern={on 4.5pt off 4.5pt}]  (305,130) .. controls (305.33,149.83) and (335,140.5) .. (335,160) ;
\draw    (305,170) .. controls (305.33,189.83) and (335,180.5) .. (335,200) ;
\draw  [draw opacity=0][fill={rgb, 255:red, 255; green, 255; blue, 255 }  ,fill opacity=1 ] (305,200) .. controls (305,179.63) and (334.75,189.88) .. (335,170) -- (340,170) -- (340,205) -- (305.13,205) .. controls (305.04,203.38) and (305,201.71) .. (305,200) -- cycle ;
\draw [color={rgb, 255:red, 191; green, 97; blue, 106 }  ,draw opacity=1 ] [dash pattern={on 4.5pt off 4.5pt}]  (335,170) .. controls (335.33,189.83) and (305,180.5) .. (305,200) ;
\draw  [draw opacity=0] (330,105) -- (350,105) -- (350,135) -- (330,135) -- cycle ;
\draw    (390,85) .. controls (390.33,104.83) and (420,95.5) .. (420,115) ;
\draw  [draw opacity=0][fill={rgb, 255:red, 255; green, 255; blue, 255 }  ,fill opacity=1 ] (390,115) .. controls (390,94.63) and (419.75,104.88) .. (420,85) -- (425,85) -- (425,120) -- (390.13,120) .. controls (390.04,118.38) and (390,116.71) .. (390,115) -- cycle ;
\draw [color={rgb, 255:red, 191; green, 97; blue, 106 }  ,draw opacity=1 ] [dash pattern={on 4.5pt off 4.5pt}]  (420,85) .. controls (420.33,104.83) and (390,95.5) .. (390,115) ;
\draw    (420,130) .. controls (420.33,149.83) and (390,135.5) .. (390,155) ;
\draw  [draw opacity=0][fill={rgb, 255:red, 255; green, 255; blue, 255 }  ,fill opacity=1 ] (390,130) .. controls (390,150.38) and (419.75,140.13) .. (420,160) -- (425,160) -- (425,125) -- (390.13,125) .. controls (390.04,126.62) and (390,128.29) .. (390,130) -- cycle ;
\draw [color={rgb, 255:red, 191; green, 97; blue, 106 }  ,draw opacity=1 ] [dash pattern={on 4.5pt off 4.5pt}]  (390,130) .. controls (390.33,149.83) and (420,140.5) .. (420,160) ;
\draw    (390,170) .. controls (390.33,189.83) and (420,180.5) .. (420,200) ;
\draw  [draw opacity=0][fill={rgb, 255:red, 255; green, 255; blue, 255 }  ,fill opacity=1 ] (390,200) .. controls (390,179.63) and (419.75,189.88) .. (420,170) -- (425,170) -- (425,205) -- (390.13,205) .. controls (390.04,203.38) and (390,201.71) .. (390,200) -- cycle ;
\draw [color={rgb, 255:red, 191; green, 97; blue, 106 }  ,draw opacity=1 ] [dash pattern={on 4.5pt off 4.5pt}]  (420,170) .. controls (420.33,189.83) and (390,180.5) .. (390,200) ;

\draw (115.5,42.5) node  [font=\footnotesize]  {\textbf{Client}};
\draw (125,77.5) node  [font=\footnotesize] [align=left] {{$\mbox{\SYN}^{\ast}$}};
\draw (125,162.5) node  [font=\footnotesize] [align=left] {{$\mbox{\ACK}^{\ast}$}};
\draw (200.5,42.5) node  [font=\footnotesize]  {\textbf{Client}};
\draw (210,77.5) node  [font=\footnotesize] [align=left] {{\SYN}};
\draw (210,162.5) node  [font=\footnotesize] [align=left] {{$\mbox{\ACK}^{\ast}$}};
\draw (200,103.5) node  [font=\footnotesize] [align=left] {{\PRJ}};
\draw (170,120) node  [font=\large,color={rgb, 255:red, 191; green, 97; blue, 106 }  ,opacity=1 ] [align=left] {{\fontfamily{pcr}\selectfont \textbf{=}}};
\draw (285.5,42.5) node  [font=\footnotesize]  {\textbf{Client}};
\draw (295,77.5) node  [font=\footnotesize] [align=left] {{\SYN}};
\draw (295,162.5) node  [font=\footnotesize] [align=left] {{$\mbox{\ACK}^{\ast}$}};
\draw (370.5,42.5) node  [font=\footnotesize]  {\textbf{Client}};
\draw (380,77.5) node  [font=\footnotesize] [align=left] {{\SYN}};
\draw (380,162.5) node  [font=\footnotesize] [align=left] {{\ACK}};
\draw (255,120) node  [font=\large,color={rgb, 255:red, 191; green, 97; blue, 106 }  ,opacity=1 ] [align=left] {{\fontfamily{pcr}\selectfont \textbf{=}}};
\draw (285,138.5) node  [font=\footnotesize] [align=left] {{\PRJ}};
\draw (340,120) node  [font=\large,color={rgb, 255:red, 191; green, 97; blue, 106 }  ,opacity=1 ] [align=left] {{\fontfamily{pcr}\selectfont \textbf{=}}};

\end{tikzpicture}
   \caption{Reasoning only with the Client.}
  \label{diagram:tcpclient}
\end{figure}

Here, we define $\SYN^{*} = \SYN ⨾ \PRJ$ to be the same as the $\SYN{}$ process but projecting out only the client $\CLI{}$ number. We also define a new $\ACK^{*}$ that ignores the server $\SRV{}$ number, so that $\ACK = \PRJ ⨾ \ACK^{*}$. These two equations are enough to complete our reasoning.

Monoidal contexts and their incomplete diagrams are defined to be convenient tuples of morphisms, e.g. $(\SYN{} | \ACK{})$ in our example; what makes them interesting is the equivalence relation we impose on them: this equivalence relation makes the pair $(\SYN{} ⨾ \PRJ{} | \ACK^{*})$ equal to $(\SYN{} | \PRJ{} ⨾ \ACK^{*})$. \emph{Dinaturality} is the name we give to this relation, and we will see how it arises canonically from the algebra of \profunctors{}.

\subsection{The Produoidal Algebra of Monoidal Context}
Despite the relative popularity of string diagrams and other forms of formal 2-dimensional syntax, the algebra of incomplete monoidal morphisms has remained obscure.
This manuscript elucidates this algebra: we show that, as monoidal morphisms together with their string diagrams form \emph{\monoidalCategories{}}, monoidal contexts together with their incomplete string diagrams form \emph{normal \produoidalCategories{}}.  Normal produoidal categories were a poorly understood categorical structure, for which we provide examples. Let us motivate ``normal produoidal categories'' by parts.

First, the \emph{``duoidal''} part. \MonoidalContexts{} can be composed sequentially and in parallel, but also nested together to fill the missing parts.
Nesting is captured by categorical composition, so we need specific tensors for both sequential $(◁)$ and parallel $(\otimes)$ composition. This is what \duoidalCategories{} provide. \DuoidalCategories{} are categories with two monoidal structures, e.g. $(◁,N)$ and $(\otimes,I)$.
These two monoidal structures are in principle independent but, whenever they share the same unit $(I ≅ N)$, they become well-suited to express process dependence \cite{shapiro22:duoidal}: they become \emph{``normal''}.

Finally, the \emph{``pro-''} prefix. It is not that we want to impose this structure on top of the monoidal one, but we want to capture the structure morphisms already form. The two tensors $(◁,\otimes)$ do not necessarily exist in the original category; in technical terms, they are not \emph{representable} or \emph{functorial}, but \emph{virtual} or \emph{profunctorial}. This makes us turn to the \produoidalCategories{} of Booker and Street \cite{bookerstreet13}.

Not only is all of this algebra present in \monoidalContexts{}.
\MonoidalContexts{} are the \emph{canonical} such algebra; in a precise sense given by universal properties. The slogan for the main result of this manuscript (\Cref{th:monoidalContextsAreANormalization}) is that
\begingroup %
\addtolength\leftmargini{-0.05in}
\begin{leftbar} %
  \noindent \MonoidalContexts{} are the \emph{free} normalization of the \emph{cofree} \produoidalCategory{} over a \monoidalCategory{}.
\end{leftbar}
\endgroup

\subsection{Related Work}
Far from being the proposal of yet another paradigm, monoidal contexts form a novel algebraic formalization of a widespread paradigm. We argue that the idea of monoidal contexts has been recurrent in the literature, just never appearing explicitly and formally.
Our main contribution is to formalize an algebra of \monoidalContexts{}, in the form of a \emph{normal produoidal} category.

In fact, the Symposium on Logic in Computer Science has recently seen multiple implicit applications of monoidal contexts. Kissinger and Uijlen \cite{kissinger:uijlen:causalstructure:lics2017} describe higher order quantum processes using contexts with holes in compact closed monoidal categories. Ghani, Hedges, Winschel and Zahn \cite{ghani:compositionalgametheory2018} describe economic game theory in terms of \emph{lenses} and incomplete processes in cartesian monoidal categories. Bonchi, Piedeleu, \Sobocinski{} and Zanasi \cite{bonchi:graphicalaffinealgebra2019} study contextual equivalence in their monoidal category of affine signal flow graphs. Di Lavore, de Felice and Román \cite{monoidalstreams} define \emph{monoidal streams} by iterating monoidal context coalgebraically.

\emph{Language theory.}
Motivated by language theory and the Chomsky-Schützenberger theorem, Melliès and Zeilberger \cite{mellies2022parsing} were the first to present the multicategorical \emph{splice-contour} adjunction. We are indebted to their exposition, which we extend to the promonoidal and produoidal cases. Earnshaw and \Sobocinski{}~\cite{earnshaw22} described a %
congruence on formal languages of string diagrams using \monoidalContexts{}. We prove how \monoidalContexts{} arise from an extended \produoidal{} splice-contour adjunction; unifying these two threads.

\textit{Session types.}
Session types \cite{honda93,honda08:sessionTypes}
are the mainstay type formalism for communication protocols, and they have been extensively applied to the π-calculus \cite{sangiorgi01:picalculus}.
Our approach is not set up to capture all of the features of a fully fledged session type theory \cite{kobayashi96:linearitypicalculus}.
Arguably, this makes it more general in what it does: it always provides a universal way of implementing send $(\Send{A})$ and receive $(\Get{A})$ operations in an arbitrary theory of processes represented by a \monoidalCategory{}.
For instance, recursion and the internal/external choice duality \cite{gay99,pierce93:subtyping} are not discussed, although they could be considered as extensions in the same way they are to monoidal categories: via trace \cite{hasegawa97} and linear distributivity \cite{cockett1997}.

\textit{Lenses and incomplete diagrams.}
Lenses are a notion of bidirectional transformation~\cite{foster07:bidirectional} that can be cast in arbitrary monoidal categories.
The first mention of monoidal lenses separate from their classical database counterparts \cite{johnson2012lenses} is due to Pastro and Street \cite{pastro07}, who identify them as an example of a \promonoidalCategory{}. However, it was with a different monoidal structure \cite{riley2018categories} that they became popular in recent years, spawning applications not only in bidirectional transformations \cite{foster07:bidirectional} but also in functional programming \cite{pickering17:profunctoroptics,ClarkeRoman20:ProfunctorOptics}, open games \cite{ghani:compositionalgametheory2018}, polynomial functors \cite{niuspivak:polynomial} and quantum combs \cite{hefford_combs}.
Relating this monoidal category of lenses with the previous \promonoidalCategory{} of lenses was an open problem; and the \promonoidal{} structure was mostly ignored in applications.

We solve this problem, proving that lenses are a universal normal \symmetricProduoidal{} category (the symmetric \monoidalContexts{}), which endows them with a novel algebra and a novel universal property.
This also extends work on the relation between \emph{incomplete diagrams}, \emph{comb-shaped diagrams}, and \emph{lenses} \cite{roman2020,roman21:opendiagrams}.

Finally, Nester et al. have recently proposed a syntax for lenses and message-passing~\cite{nester23:processhistories,boisseaunester:corneringoptics} and lenses themselves have been applied to protocol specification~\cite{videlacapucci22}.
Spivak \cite{spivak13} also discusses the \multicategory{} of \emph{wiring diagrams}, later used for incomplete diagrams \cite{patterson21:wiringdiagrams} and related to lenses \cite{schultz20:dynamical}.
The \promonoidalCategories{} we use can be seen as \multicategories{} with an extra coherence property. In this sense, we contribute the missing algebraic structure of the universal multicategory of \emph{wiring diagrams relative to a monoidal category}.

\subsection{Contributions}
Our main contribution is the original definition of a \produoidalCategory{} of \emph{\monoidalContexts{}} over a \monoidalCategory{} (\Cref{def:monoidalcontext}) and its characterization in terms of universal properties (\Cref{th:monoidalContextsAreANormalization}).

\Cref{sec:profunctors} presents expository material on \profunctors{}, \dinaturality{} and \promonoidalCategories{}; the rest are novel contributions.
\Cref{sec:sequentialContext} constructs \splicedArrows{} as the cofree \promonoidal{} over a category (\Cref{th:catpromadj}).
\Cref{sec:parallelContext}, on top of this, constructs \splicedMonoidalArrows{} as the cofree \produoidal{} over a \monoidalCategory{} (\Cref{prop:produoidalSpliceContour}).
\Cref{sec:monoidalContexts} explicitly constructs a \produoidal{} algebra of \monoidalContexts{} (\Cref{prop:MonoidalContextProtensor}) as a free normalization. \Cref{sec:monoidallenses} constructs a symmetric \produoidal{} algebra of \monoidalLenses{} (\Cref{prop:monoidalLensesProduoidal}), universally characterizing them (\Cref{th:lensesuniversal}), and an interpretation of send/receive types ($\Send{}/\Get{}$) (\Cref{prop:sessionNotation}). \Cref{sec:normalization} introduces a novel free normalization procedure (\Cref{th:normalizationIdempotent,th:freeNormalProduoidal}) as an idempotent monad on \produoidalCategories{}, employed in \Cref{sec:monoidalContexts,sec:monoidallenses}.

\section{Profunctors and Virtual Structures}
\label{sec:profunctors}
\Profunctors{} describe families of processes indexed by the input and output types of a category.
Profunctors provide canonical notions for \emph{composition}, \emph{dinaturality} and \emph{virtual structure}. These notions are not only canonical, but also %
easy to reason with thanks to \emph{coend calculus} \cite{loregian2021}.

\begin{definition}
  A \defining{linkprofunctor}{\emph{profunctor}} $P \colon 𝔹₀ × \mydots × 𝔹ₘ ⧑ 𝔸₀ × \mydots × 𝔸ₙ$ is a functor $P \colon 𝔸₀^{op} \mydots × 𝔸ₙ^{op} × 𝔹₀ × \mydots × 𝔹ₘ \to \mathbf{Set}$.
\end{definition}

For our purposes, a profunctor $P(A₀, \mydots, Aₙ; B₀, \mydots, Bₘ)$ is a family of processes indexed by contravariant inputs $A₀, \mydots, Aₙ$ and covariant outputs $B₀, \mydots, Bₘ$. The profunctor is endowed with jointly functorial left $(≻_0, \mydots, ≻_m)$ and right $(≺_0, \mydots, ≺_n)$ actions of the morphisms of $𝔸₀, \mydots, 𝔸ₙ$ and $𝔹₀, \mydots, 𝔹ₘ$, respectively \cite{benabou00,loregian2021}.\footnote{We simply use $(≺/≻)$ without any subscript whenever the input/output is unique. See Appendix, \Cref{ax:sec:profunctors} for more details on profunctors.}

\subsection{Dinaturality}

Composing \profunctors{} is subtle: the same processes could arise as the composite of different pairs of processes and so, we need to impose a careful equivalence relation.
Fortunately, \profunctors{} come with a canonical notion of \dinaturalEquivalence{} which achieves precisely this.

Imagine we try to connect two different processes:
  $p \in P(A_0,\mydots,A_n;B_0,\dots,B_m)$, and
  $q \in Q(C_0,\mydots,C_k;D_0,\dots,D_h)$;
and we have some morphism $f \colon B_i \to C_j$ that translates the i-th output port of $p$ to the j-th input port of $q$. Let us write $(ᵢ|ⱼ)$ for this connection operation. Note that we could connect them in two different ways:
\begin{itemize}
  \item we could use $f$ to change \emph{the output of the first process} $p ≺_i f$ before connecting both, $(p ≺ᵢ f)\, {}_i|_j\, q$;
  \item and we could use $f$ to change
  \emph{the input of the second process} $f ≻_j q$ before connecting both,  $p\, {}_i|_j\, (f ≻_j q)$.
\end{itemize}
These are different descriptions, made up of two different components. However, they essentially describe the same process: they are \emph{dinaturally equal} \cite{monoidalstreams}.
Indeed, \profunctors{} are canonically endowed with a notion of \emph{dinatural equivalence} \cite{benabou00,loregian2021}, which precisely equates these two descriptions.
\Profunctors{}, and their elements, are thus composed \emph{up to dinatural equivalence}.

\begin{definition}[Dinatural equivalence]
  \defining{linkdinaturality}
  Consider two \profunctors{} $P \colon 𝔹₀×\mydots × 𝔹ₘ ⧑ 𝔸₀× \mydots ×𝔸ₙ$ and $Q \colon ℂ₀× \mydots ×ℂ_k ⧑ 𝔻₀ × \mydots × 𝔻_h$ such that $𝔹ᵢ = ℂⱼ$; and let $\mathbf{S}_{P,Q}^{i,j}(A;C)$ be the set
  $$\sum_{X \in 𝔹_i} P(A₀\mydots Aₙ; B₀ \mydots X \mydots Bₘ) × Q(C₀\mydots X \mydots C_k ; D₀\mydots D_h).$$
  \emph{Dinatural equivalence}, $(\sim)$, on the set $\mathbf{S}_{P,Q}^{i,j}(A;C)$ is the smallest equivalence relation satisfying $(p ≺ᵢ f\, {}_i|_j\, q) \sim (p\, {}_i|_j\, f ≻_j q)$.
  The \emph{coend} is defined as this coproduct quotiented by dinaturality, $\mathbf{S}_{P,Q}^{i,j}(A;C) / (\sim)$, and written as an integral.
  $$\int^{X \in ℂ} P(A₀\mydots Aₙ; B₀ \mydots X \mydots Bₘ) × Q(C₀\mydots X \mydots C_k ; D₀\mydots D_h).$$
\end{definition}

\begin{definition}[Profunctor composition]
  Consider two \profunctors{} $P \colon 𝔹₀×\mydots × 𝔹ₘ ⧑ 𝔸₀× \mydots ×𝔸ₙ$ and $Q \colon ℂ₀× \mydots ×ℂ_k ⧑ 𝔻₀ × \mydots × 𝔻_h$ such that $𝔹ᵢ = ℂⱼ$;
  their \emph{composition} along ports $i$ and $j$ is a profunctor; we write it marking this connection
  $$P(A_0\mydots A_n ; B_0 \mydots •_x \mydots B_n) \diamond
  Q(C₀\mydots \bullet_x \mydots C_k ; D₀\mydots D_h),$$
  and it is defined as the coproduct of the product of both profunctors, indexed by the common variable, and quotiented by \dinaturalEquivalence{},
  $$\int^{X \in ℂ} P(A₀\mydots Aₙ; B₀ \mydots X \mydots Bₘ) × Q(C₀\mydots X \mydots C_k ; D₀\mydots D_h).$$
\end{definition}

\begin{remark}[Representability]
  \label{rem:representability}
  \defining{linkrepresentable}{}
  Every functor $F \colon 𝔸 \to B$ gives rise to two different \profunctors{}: its representable \profunctor{} $𝔸(F•, •) \colon 𝔸 ⧒ 𝔹$, and its corepresentable \profunctor{} $𝔸(•, F•) \colon 𝔸 ⧑ 𝔹$. We say that a profunctor is \emph{representable} or \emph{corepresentable} if it arises in this way.
  Under this interpretation, functors are \profunctors{} that happen to be representable.
  This suggests that we can repeat structures based on functors, such as \monoidalCategories{}, now in terms of \profunctors{}.
\end{remark}

We justified in the introduction the importance of \monoidalCategories{}: they are the algebra of processes composing sequentially and in parallel, joining and splitting resources.
However, there exist some theories that can deal only with splitting without being necessarily full theories of processes: that is, we may be able to talk about splitting without being able to talk about joining. Such ``\monoidalCategories{} on one side'' are \emph{promonoidal categories}.

The difference between \monoidalCategories{} and \promonoidalCategories{} is that the tensor is no longer a functor but is instead a profunctor.\footnote{In more technical terms, monoidal categories are pseudomonoids in the monoidal bicategory of categories \emph{and functors}; while promonoidal categories are pseudomonoids in the monoidal bicategory of categories \emph{and profunctors}.} In other words, the tensor is no longer \representable{} – such a structure is called \emph{virtual}, as in \emph{virtual double} and \emph{virtual duoidal} categories \cite{cruttwell2010,nlab:duoidal}.

\subsection{Promonoidal Categories}
\defining{linkPromonoidalComponents}{}

\PromonoidalCategories{} are the algebra of \emph{coherent decomposition}. A category $ℂ$ contains sets of \emph{morphisms}, $ℂ(X ; Y)$. In the same way, a promonoidal category $𝕍$ contains sets of \emph{splits}, $𝕍(X; Y_0 ◁ Y_1)$, \emph{morphisms}, $𝕍(X; Y)$, and \emph{units}, $𝕍(X; N)$, where $N$ is the virtual tensor unit.
\Splits{}, $𝕍(X; Y_0 ◁ Y_1)$, represent a way of decomposing objects of type $X$ into objects of type $Y_0$ and $Y_1$. Morphisms, $𝕍(X; Y)$, as in any category, are transformations of $X$ into $Y$. \Units{}, $𝕍(X;N)$, are the atomic pieces of type $X$.

These decompositions must be coherent. For instance, imagine we want to split $X$ into $Y_0$, $Y_1$ and $Y_2$. Splitting $X$ into $Y_0$ and something $(•)$, and then splitting that something into $Y_1$ and $Y_2$ \emph{should be doable in essentially the same ways} as splitting $X$ into something $(•)$ and $Y_2$, and then splitting that something into $Y_0$ and $Y_1$. Formally, we are saying that,
$$𝕍(X ; Y_0 ◁ •) \diamond 𝕍(• ; Y_1 ◁ Y_2) \cong
𝕍(X ; • ◁ Y_2) \diamond 𝕍(• ; Y_0 ◁ Y_1),
$$
and, in fact, we just write $𝕍(X; Y_0◁ Y_1 ◁ Y_2)$ for the set of such decompositions.

\defining{linkpromonoidal}{} \label{def:promonoidal}

  \begin{definition}[Promonoidal category]
  A \emph{promonoidal category} is a category $𝕍(• ; •)$ endowed with \profunctors{}
  $$𝕍(• ;• ◁ •) \colon 𝕍 × 𝕍 ⧑ 𝕍, \mbox{ and } 𝕍(•; N) \colon 1 ⧑ 𝕍.$$
  Equivalently, these are functors
  $$𝕍(•; • ◁ •) \colon 𝕍\op × 𝕍 × 𝕍 \to \mathbf{Set}, \mbox{ and }
  𝕍(•; N) \colon 𝕍\op \to \mathbf{Set}.$$
  Moreover, \promonoidalCategories{} must be endowed with the following natural isomorphisms,
  \begin{align*}
    𝕍(X ;• ◁ Y_2) ⋄ 𝕍(• ; Y_0 ◁ Y_1) &≅
    𝕍(X ; • ◁ Y_2) \diamond 𝕍(• ; Y_0 ◁ Y_1),\\
    𝕍(X ; •◁ Y) ⋄ 𝕍(•;N) &\cong 𝕍(X; Y), \\
    𝕍(X ;Y ◁ •) ⋄ 𝕍(•;N) &\cong 𝕍(X; Y),
  \end{align*}
  called $α, \lambda, \rho$, respectively, and asked to satisfy the pentagon and triangle coherence equations, $α ⨾ α = (α⋄1)⨾α⨾(1⋄α)$, and $(ρ ⋄ 1) = α ⨾ (λ ⋄ 1)$.
  \end{definition}

  \begin{definition}[Promonoidal functor]
    \label{def:promonoidalfunctor}
    \defining{linkpromonoidalfunctor}{}
    Let $𝕍$ and $𝕎$ be \promonoidalCategories{}. A \emph{promonoidal functor} $F﹕ 𝕍(•,•) → 𝕎(•,•)$ is a functor between the two categories, together with natural transformations:
    $$\begin{aligned}
        & F_{◁}﹕ 𝕍(A;B ◁ C) → 𝕎(FA; FB ◁ FC), \mbox{ and } \\
        & F_{N} ﹕ 𝕍(A;N) → 𝕎(FA;N),
    \end{aligned}$$
    that satisfy $λ ⨾ F_\map = (F_{◁} × F_{N}) ⨾ λ$, $ρ ⨾ F_\map = (F_{◁} × F_{N}) ⨾ ρ$, and $α ⨾ (F_{◁} × F_{◁}) ⨾ i = (F_{◁} × F_{◁}) ⨾ i ⨾ α$. We denote by $\pMon$ the category of \promonoidalCategories{} and \promonoidalFunctors{}.
  \end{definition}

  \begin{remark}[Promonoidal coherence]
      As with \monoidalCategories{}, the pentagon and triangle equations imply that every formal equation written out of coherence isomorphisms holds. This means we can write $𝕍(•; • ◁ • ◁ •)$ without specifying which one of the two sides of the associator we are describing.
  \end{remark}

    \begin{remark}[Multicategories]
      \label{remark:multicategories}
    The reader may be more familiar with the algebra of not-necessarily-coherent decomposition: \emph{multicategories}.
    Every \promonoidalCategory{} $\mathbb{V}$ induces a co-\multicategory{} with morphisms given by elements of the following sets $𝕍(•; • ◁ \overset{n}\dots ◁ •)$.
    Similarly, $\mathbb{V}^\text{op}$ is a co-\promonoidalCategory{} and thus induces a \multicategory{}.
    These are special kinds of (co-)\multicategories{}, they are \emph{coherent} so that every $n$-to-1 morphism splits, in any possible shape, as 2-to-1 and 0-to-1 morphisms; moreover, they do so \emph{uniquely up to dinaturality}. Appendix \ref{ax:sec:multicategories} spells out this relation.
    \end{remark}

The next section studies how to coherently decompose morphisms of a category. Categories are an algebraic structure for sequential composition: they contain a ``sequencing'' operator $(⨾)$ and a neutral element, $\mathrm{id}$.  We present an algebra for decomposing sequential compositions in terms of \promonoidalCategories{}.

\section{Sequential context}
\label{sec:sequentialContext}

Assume a morphism factors as follows,
$$f_0 ⨾ g_0 ⨾ h ⨾ g_1 ⨾ f_1 ⨾ k ⨾ f_2.$$
We can say that this morphism came from the context $\trispliced{f_0}{f_1}{f_2}$, filled on its left side with the context $\bispliced{g_0}{g_1}$, then filled with $h$, and finally completed on its right side with the morphism $k$. \Cref{fig:contourExample} expresses this decomposition.
\begin{figure}[ht]
  \centering

\tikzset{every picture/.style={line width=0.75pt}} %

\begin{tikzpicture}[x=0.75pt,y=0.75pt,yscale=-1,xscale=1]
\draw   (135,35) -- (139.5,20) -- (200.5,20) -- (205,35) -- cycle ;
\draw    (170,20) -- (170,10) ;
\draw [color={rgb, 255:red, 191; green, 97; blue, 106 }  ,draw opacity=1 ]   (133.13,39.28) .. controls (131.96,38.82) and (130.98,38.09) .. (130,35) .. controls (128.54,30.42) and (128.57,25.58) .. (130,20) .. controls (131.43,14.42) and (139.67,15.06) .. (145,15) ;
\draw [shift={(135,40)}, rotate = 205.96] [color={rgb, 255:red, 191; green, 97; blue, 106 }  ,draw opacity=1 ][line width=0.75]    (4.37,-1.32) .. controls (2.78,-0.56) and (1.32,-0.12) .. (0,0) .. controls (1.32,0.12) and (2.78,0.56) .. (4.37,1.32)   ;
\draw [color={rgb, 255:red, 191; green, 97; blue, 106 }  ,draw opacity=1 ]   (155,40) -- (183,40) ;
\draw [shift={(185,40)}, rotate = 180] [color={rgb, 255:red, 191; green, 97; blue, 106 }  ,draw opacity=1 ][line width=0.75]    (4.37,-1.32) .. controls (2.78,-0.56) and (1.32,-0.12) .. (0,0) .. controls (1.32,0.12) and (2.78,0.56) .. (4.37,1.32)   ;
\draw [color={rgb, 255:red, 191; green, 97; blue, 106 }  ,draw opacity=1 ]   (205,40) .. controls (207.43,39.87) and (208.54,39.58) .. (210,35) .. controls (211.46,30.42) and (211.43,25.58) .. (210,20) .. controls (208.73,15.06) and (202.13,15) .. (196.94,15) ;
\draw [shift={(195,15)}, rotate = 0.66] [color={rgb, 255:red, 191; green, 97; blue, 106 }  ,draw opacity=1 ][line width=0.75]    (4.37,-1.32) .. controls (2.78,-0.56) and (1.32,-0.12) .. (0,0) .. controls (1.32,0.12) and (2.78,0.56) .. (4.37,1.32)   ;
\draw   (115,75) -- (119.5,60) -- (160.5,60) -- (165,75) -- cycle ;
\draw    (150,35) .. controls (150.17,44.98) and (140.17,49.98) .. (140,60) ;
\draw    (140,85) -- (140,75) ;
\draw [color={rgb, 255:red, 191; green, 97; blue, 106 }  ,draw opacity=1 ]   (113.13,79.28) .. controls (111.96,78.82) and (110.98,78.09) .. (110,75) .. controls (108.54,70.42) and (108.57,65.58) .. (110,60) .. controls (111.43,54.42) and (119.67,55.06) .. (125,55) ;
\draw [shift={(115,80)}, rotate = 205.96] [color={rgb, 255:red, 191; green, 97; blue, 106 }  ,draw opacity=1 ][line width=0.75]    (4.37,-1.32) .. controls (2.78,-0.56) and (1.32,-0.12) .. (0,0) .. controls (1.32,0.12) and (2.78,0.56) .. (4.37,1.32)   ;
\draw   (130.98,100.5) -- (135.48,85.5) -- (146.48,85.5) -- (150.98,100.5) -- cycle ;
\draw [color={rgb, 255:red, 191; green, 97; blue, 106 }  ,draw opacity=1 ]   (130.98,85.5) .. controls (127.31,94.29) and (125.6,105.47) .. (130.98,105.5) .. controls (136.37,105.53) and (145.6,105.47) .. (150.98,105.5) .. controls (156.05,105.52) and (156.29,96.19) .. (151.89,87.21) ;
\draw [shift={(150.98,85.5)}, rotate = 60.63] [color={rgb, 255:red, 191; green, 97; blue, 106 }  ,draw opacity=1 ][line width=0.75]    (4.37,-1.32) .. controls (2.78,-0.56) and (1.32,-0.12) .. (0,0) .. controls (1.32,0.12) and (2.78,0.56) .. (4.37,1.32)   ;
\draw    (190,35) .. controls (190.03,49.73) and (200.03,44.73) .. (200,60) ;
\draw  [draw opacity=0] (110,15) -- (130,15) -- (130,35) -- (110,35) -- cycle ;
\draw  [draw opacity=0] (160,35) -- (180,35) -- (180,55) -- (160,55) -- cycle ;
\draw  [draw opacity=0] (210,15) -- (230,15) -- (230,35) -- (210,35) -- cycle ;
\draw [color={rgb, 255:red, 191; green, 97; blue, 106 }  ,draw opacity=1 ]   (165,80) .. controls (167.43,79.87) and (168.54,79.58) .. (170,75) .. controls (171.46,70.42) and (171.43,65.58) .. (170,60) .. controls (168.73,55.06) and (162.13,55) .. (156.94,55) ;
\draw [shift={(155,55)}, rotate = 0.66] [color={rgb, 255:red, 191; green, 97; blue, 106 }  ,draw opacity=1 ][line width=0.75]    (4.37,-1.32) .. controls (2.78,-0.56) and (1.32,-0.12) .. (0,0) .. controls (1.32,0.12) and (2.78,0.56) .. (4.37,1.32)   ;
\draw  [draw opacity=0] (120.98,85.5) -- (140.98,85.5) -- (140.98,105.5) -- (120.98,105.5) -- cycle ;
\draw  [draw opacity=0] (90,50) -- (110,50) -- (110,70) -- (90,70) -- cycle ;
\draw  [draw opacity=0] (170,55) -- (190,55) -- (190,75) -- (170,75) -- cycle ;
\draw   (190,75.5) -- (194.5,60.5) -- (205.5,60.5) -- (210,75.5) -- cycle ;
\draw [color={rgb, 255:red, 191; green, 97; blue, 106 }  ,draw opacity=1 ]   (190,60.5) .. controls (186.32,69.29) and (184.61,80.47) .. (190,80.5) .. controls (195.39,80.53) and (204.61,80.47) .. (210,80.5) .. controls (215.06,80.52) and (215.31,71.19) .. (210.9,62.21) ;
\draw [shift={(210,60.5)}, rotate = 60.63] [color={rgb, 255:red, 191; green, 97; blue, 106 }  ,draw opacity=1 ][line width=0.75]    (4.37,-1.32) .. controls (2.78,-0.56) and (1.32,-0.12) .. (0,0) .. controls (1.32,0.12) and (2.78,0.56) .. (4.37,1.32)   ;
\draw  [draw opacity=0] (180,60.5) -- (200,60.5) -- (200,80.5) -- (180,80.5) -- cycle ;

\draw (170,27.5) node  [font=\tiny]  {$f_{0} ⨾ \square ⨾ f_{1} ⨾ \square ⨾ f_{2}$};
\draw (120,25) node  [font=\tiny,color={rgb, 255:red, 191; green, 97; blue, 106 }  ,opacity=1 ]  {$f_{0}$};
\draw (140,67.5) node  [font=\tiny]  {$g_{0} ⨾ \square ⨾ g_{1}$};
\draw (140.98,93) node  [font=\tiny]  {$h$};
\draw (170,45) node  [font=\tiny,color={rgb, 255:red, 191; green, 97; blue, 106 }  ,opacity=1 ]  {$f_{1}$};
\draw (220,25) node  [font=\tiny,color={rgb, 255:red, 191; green, 97; blue, 106 }  ,opacity=1 ]  {$f_{2}$};
\draw (140.98,110.5) node  [font=\tiny,color={rgb, 255:red, 191; green, 97; blue, 106 }  ,opacity=1 ]  {$h$};
\draw (100,65) node  [font=\tiny,color={rgb, 255:red, 191; green, 97; blue, 106 }  ,opacity=1 ]  {$g_{0}$};
\draw (180,65) node  [font=\tiny,color={rgb, 255:red, 191; green, 97; blue, 106 }  ,opacity=1 ]  {$g_{1}$};
\draw (200,68) node  [font=\tiny]  {$k$};
\draw (200,85.5) node  [font=\tiny,color={rgb, 255:red, 191; green, 97; blue, 106 }  ,opacity=1 ]  {$k$};

\end{tikzpicture}
   \caption{Decomposition of $f_0 ⨾ g_0 ⨾ h ⨾ g_1 ⨾ f_1 ⨾ k ⨾ f_2$.}
  \label{fig:contourExample}
\end{figure}
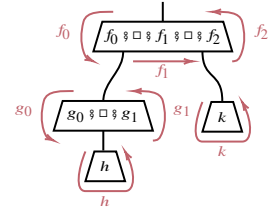

Contexts compose in a tree-like structure, and their resulting morphism is extracted by \emph{contouring} that tree. This section presents the algebra of \emph{context} and \emph{decomposition}. We then prove that they are two sides of the same coin: the two sides of an adjunction of categories.

\subsection{Contour of a Promonoidal Category}

Any \promonoidalCategory{} freely generates another category, its \emph{contour}. This can be interpreted as the category that tracks the processes of decomposition that the \promonoidalCategory{} describes. The construction is particularly pleasant from the geometric point of view: it takes its name from the fact that it can be constructed by following the contour of the shape of the decomposition.

\begin{definition}[Contour] 
  \label{defn:contour}
  \defining{linkContour}{}
    The \emph{contour} of a \promonoidalCategory{} $𝕍$ is a category $\mathcal{C}𝕍$ that has two objects, $X^L$ (left-handed) and $X^R$ (right-handed), for each object $X \in \obj{𝕍}$; and has as arrows those that arise from \emph{contouring} the decompositions of the promonoidal category.
    \begin{figure}[ht]
        \centering

\tikzset{every picture/.style={line width=0.75pt}} %

\begin{tikzpicture}[x=0.75pt,y=0.75pt,yscale=-1,xscale=1]
\draw   (55,60) -- (59.5,45) -- (70.5,45) -- (75,60) -- cycle ;
\draw    (65,45) -- (65,35) ;
\draw [color={rgb, 255:red, 191; green, 97; blue, 106 }  ,draw opacity=1 ]   (55,45) .. controls (51.32,53.79) and (49.61,64.97) .. (55,65) .. controls (60.39,65.03) and (69.61,64.97) .. (75,65) .. controls (80.06,65.02) and (80.31,55.69) .. (75.9,46.71) ;
\draw [shift={(75,45)}, rotate = 60.63] [color={rgb, 255:red, 191; green, 97; blue, 106 }  ,draw opacity=1 ][line width=0.75]    (4.37,-1.32) .. controls (2.78,-0.56) and (1.32,-0.12) .. (0,0) .. controls (1.32,0.12) and (2.78,0.56) .. (4.37,1.32)   ;
\draw  [draw opacity=0] (55,60) -- (75,60) -- (75,80) -- (55,80) -- cycle ;
\draw    (125,45) -- (125,35) ;
\draw   (115,45) -- (135,45) -- (135,60) -- (115,60) -- cycle ;
\draw    (125,70) -- (125,60) ;
\draw [color={rgb, 255:red, 191; green, 97; blue, 106 }  ,draw opacity=1 ]   (117.96,64.93) .. controls (112.06,64.68) and (111.32,64.14) .. (110,60) .. controls (108.54,55.42) and (108.57,50.58) .. (110,45) .. controls (111.43,39.42) and (114.67,40.06) .. (120,40) ;
\draw [shift={(120,65)}, rotate = 181.8] [color={rgb, 255:red, 191; green, 97; blue, 106 }  ,draw opacity=1 ][line width=0.75]    (4.37,-1.32) .. controls (2.78,-0.56) and (1.32,-0.12) .. (0,0) .. controls (1.32,0.12) and (2.78,0.56) .. (4.37,1.32)   ;
\draw [color={rgb, 255:red, 191; green, 97; blue, 106 }  ,draw opacity=1 ]   (132.04,39.86) .. controls (136.32,39.74) and (138.73,41.02) .. (140,45) .. controls (141.46,49.58) and (141.43,54.42) .. (140,60) .. controls (138.57,65.58) and (135.33,64.94) .. (130,65) ;
\draw [shift={(130,40)}, rotate = 354.07] [color={rgb, 255:red, 191; green, 97; blue, 106 }  ,draw opacity=1 ][line width=0.75]    (4.37,-1.32) .. controls (2.78,-0.56) and (1.32,-0.12) .. (0,0) .. controls (1.32,0.12) and (2.78,0.56) .. (4.37,1.32)   ;
\draw   (180,60) -- (184.5,45) -- (205.5,45) -- (210,60) -- cycle ;
\draw    (185,70) -- (185,60) ;
\draw    (205,70) -- (205,60) ;
\draw    (195,45) -- (195,35) ;
\draw [color={rgb, 255:red, 191; green, 97; blue, 106 }  ,draw opacity=1 ]   (178.13,64.28) .. controls (176.96,63.82) and (175.98,63.09) .. (175,60) .. controls (173.54,55.42) and (173.57,50.58) .. (175,45) .. controls (176.43,39.42) and (184.67,40.06) .. (190,40) ;
\draw [shift={(180,65)}, rotate = 205.96] [color={rgb, 255:red, 191; green, 97; blue, 106 }  ,draw opacity=1 ][line width=0.75]    (4.37,-1.32) .. controls (2.78,-0.56) and (1.32,-0.12) .. (0,0) .. controls (1.32,0.12) and (2.78,0.56) .. (4.37,1.32)   ;
\draw [color={rgb, 255:red, 191; green, 97; blue, 106 }  ,draw opacity=1 ]   (190,65) -- (198,65) ;
\draw [shift={(200,65)}, rotate = 180] [color={rgb, 255:red, 191; green, 97; blue, 106 }  ,draw opacity=1 ][line width=0.75]    (4.37,-1.32) .. controls (2.78,-0.56) and (1.32,-0.12) .. (0,0) .. controls (1.32,0.12) and (2.78,0.56) .. (4.37,1.32)   ;
\draw [color={rgb, 255:red, 191; green, 97; blue, 106 }  ,draw opacity=1 ]   (210,65) .. controls (212.43,64.87) and (213.54,64.58) .. (215,60) .. controls (216.46,55.42) and (216.43,50.58) .. (215,45) .. controls (213.73,40.06) and (207.13,40) .. (201.94,40) ;
\draw [shift={(200,40)}, rotate = 0.66] [color={rgb, 255:red, 191; green, 97; blue, 106 }  ,draw opacity=1 ][line width=0.75]    (4.37,-1.32) .. controls (2.78,-0.56) and (1.32,-0.12) .. (0,0) .. controls (1.32,0.12) and (2.78,0.56) .. (4.37,1.32)   ;

\draw (65,52.5) node  [font=\footnotesize]  {$a$};
\draw (65,70) node  [font=\tiny,color={rgb, 255:red, 191; green, 97; blue, 106 }  ,opacity=1 ]  {$a_{0}$};
\draw (125,52.5) node  [font=\footnotesize]  {$b$};
\draw (102.5,52.5) node  [font=\tiny,color={rgb, 255:red, 191; green, 97; blue, 106 }  ,opacity=1 ]  {$b_{0}$};
\draw (147.5,52.5) node  [font=\tiny,color={rgb, 255:red, 191; green, 97; blue, 106 }  ,opacity=1 ]  {$b_{1}$};
\draw (195,52.5) node  [font=\footnotesize]  {$c$};
\draw (167.5,52.5) node  [font=\tiny,color={rgb, 255:red, 191; green, 97; blue, 106 }  ,opacity=1 ]  {$c_{0}$};
\draw (195,72.5) node  [font=\tiny,color={rgb, 255:red, 191; green, 97; blue, 106 }  ,opacity=1 ]  {$c_{1}$};
\draw (222.5,52.5) node  [font=\tiny,color={rgb, 255:red, 191; green, 97; blue, 106 }  ,opacity=1 ]  {$c_{2}$};
\draw (65,31.6) node [anchor=south] [inner sep=0.75pt]  [font=\scriptsize]  {$A$};
\draw (125,31.6) node [anchor=south] [inner sep=0.75pt]  [font=\scriptsize]  {$B$};
\draw (195,31.6) node [anchor=south] [inner sep=0.75pt]  [font=\scriptsize]  {$C$};
\draw (125,73.4) node [anchor=north] [inner sep=0.75pt]  [font=\scriptsize]  {$X$};
\draw (185,73.4) node [anchor=north] [inner sep=0.75pt]  [font=\scriptsize]  {$Y$};
\draw (205,73.4) node [anchor=north] [inner sep=0.75pt]  [font=\scriptsize]  {$Z$};

\end{tikzpicture}
         \caption{Contour of a promonoidal.}
        \label{fig:promonoidalContour}
    \end{figure}
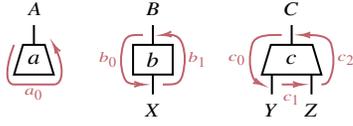

    Specifically, it is freely presented by \emph{(i)} a morphism $a_0 \in 𝓒𝕍(A^L; A^R)$, for each \unitP{} $a \in 𝕍(A;N)$; \emph{(ii)}  a pair of morphisms $b_0 \in 𝓒𝕍(B^L; X^L)$, $b_1 \in 𝓒𝕍(X^R;B^R)$, for each morphism $b \in 𝕍(B;X)$; and \emph{(iii)} a triple of morphisms $c_0 \in 𝓒𝕍(C^L; Y^L)$, $c_1 \in 𝓒𝕍(Y^R; Z^L)$, $c_2 \in 𝓒𝕍(Z^R; C^R)$ for each \splitP{} $c \in 𝕍(C;Y ◁ Z)$, see \Cref{fig:promonoidalContour}.

    For each equality $α(a \mathbin{|} b) = (c \mathbin{|} d)$, we impose the equations $a_0 = c_0 ⨾ d_0$; $a_1 ⨾ b_0 = d_1$ and $b_1 = d_2 ⨾ c_1$; $a_2 ⨾ b_2 = c_2$. For each equality $ρ(a \mathbin{|} b) = c = \lambda(d \mathbin{|} e)$, we impose $a_0 = c_0 = d_0 ⨾ e_0 ⨾ d_1$ and $a_1 ⨾ b_0 ⨾ a_2 = c_1 = d_2$.
    Graphically, these follow \Cref{fig:promonoidalContourCompose}.
    \begin{figure}[ht]
        \centering 

\tikzset{every picture/.style={line width=0.75pt}} %

\begin{tikzpicture}[x=0.75pt,y=0.75pt,yscale=-1,xscale=1]
\draw   (155,35) -- (159.5,20) -- (180.5,20) -- (185,35) -- cycle ;
\draw    (170,20) -- (170,10) ;
\draw [color={rgb, 255:red, 191; green, 97; blue, 106 }  ,draw opacity=1 ]   (153.13,39.28) .. controls (151.96,38.82) and (150.98,38.09) .. (150,35) .. controls (148.54,30.42) and (148.57,25.58) .. (150,20) .. controls (151.43,14.42) and (159.67,15.06) .. (165,15) ;
\draw [shift={(155,40)}, rotate = 205.96] [color={rgb, 255:red, 191; green, 97; blue, 106 }  ,draw opacity=1 ][line width=0.75]    (4.37,-1.32) .. controls (2.78,-0.56) and (1.32,-0.12) .. (0,0) .. controls (1.32,0.12) and (2.78,0.56) .. (4.37,1.32)   ;
\draw [color={rgb, 255:red, 191; green, 97; blue, 106 }  ,draw opacity=1 ]   (165,40) -- (173,40) ;
\draw [shift={(175,40)}, rotate = 180] [color={rgb, 255:red, 191; green, 97; blue, 106 }  ,draw opacity=1 ][line width=0.75]    (4.37,-1.32) .. controls (2.78,-0.56) and (1.32,-0.12) .. (0,0) .. controls (1.32,0.12) and (2.78,0.56) .. (4.37,1.32)   ;
\draw [color={rgb, 255:red, 191; green, 97; blue, 106 }  ,draw opacity=1 ]   (185,40) .. controls (187.43,39.87) and (188.54,39.58) .. (190,35) .. controls (191.46,30.42) and (191.43,25.58) .. (190,20) .. controls (188.73,15.06) and (182.13,15) .. (176.94,15) ;
\draw [shift={(175,15)}, rotate = 0.66] [color={rgb, 255:red, 191; green, 97; blue, 106 }  ,draw opacity=1 ][line width=0.75]    (4.37,-1.32) .. controls (2.78,-0.56) and (1.32,-0.12) .. (0,0) .. controls (1.32,0.12) and (2.78,0.56) .. (4.37,1.32)   ;
\draw   (170,75) -- (174.5,60) -- (195.5,60) -- (200,75) -- cycle ;
\draw    (175,85) -- (175,75) ;
\draw    (195,85) -- (195,75) ;
\draw [color={rgb, 255:red, 191; green, 97; blue, 106 }  ,draw opacity=1 ]   (168.13,79.28) .. controls (166.96,78.82) and (165.98,78.09) .. (165,75) .. controls (163.54,70.42) and (163.57,65.58) .. (165,60) .. controls (166.43,54.42) and (174.67,55.06) .. (180,55) ;
\draw [shift={(170,80)}, rotate = 205.96] [color={rgb, 255:red, 191; green, 97; blue, 106 }  ,draw opacity=1 ][line width=0.75]    (4.37,-1.32) .. controls (2.78,-0.56) and (1.32,-0.12) .. (0,0) .. controls (1.32,0.12) and (2.78,0.56) .. (4.37,1.32)   ;
\draw [color={rgb, 255:red, 191; green, 97; blue, 106 }  ,draw opacity=1 ]   (180,80) -- (188,80) ;
\draw [shift={(190,80)}, rotate = 180] [color={rgb, 255:red, 191; green, 97; blue, 106 }  ,draw opacity=1 ][line width=0.75]    (4.37,-1.32) .. controls (2.78,-0.56) and (1.32,-0.12) .. (0,0) .. controls (1.32,0.12) and (2.78,0.56) .. (4.37,1.32)   ;
\draw [color={rgb, 255:red, 191; green, 97; blue, 106 }  ,draw opacity=1 ]   (200,80) .. controls (202.43,79.87) and (203.54,79.58) .. (205,75) .. controls (206.46,70.42) and (206.43,65.58) .. (205,60) .. controls (203.73,55.06) and (197.13,55) .. (191.94,55) ;
\draw [shift={(190,55)}, rotate = 0.66] [color={rgb, 255:red, 191; green, 97; blue, 106 }  ,draw opacity=1 ][line width=0.75]    (4.37,-1.32) .. controls (2.78,-0.56) and (1.32,-0.12) .. (0,0) .. controls (1.32,0.12) and (2.78,0.56) .. (4.37,1.32)   ;
\draw    (180,35) .. controls (180.03,49.73) and (185.03,44.73) .. (185,60) ;
\draw    (160,35) .. controls (160.03,49.73) and (150.03,44.73) .. (150,60) ;
\draw    (150,85) -- (150,60) ;
\draw   (250,35) -- (254.5,20) -- (275.5,20) -- (280,35) -- cycle ;
\draw    (265,20) -- (265,10) ;
\draw [color={rgb, 255:red, 191; green, 97; blue, 106 }  ,draw opacity=1 ]   (248.13,39.28) .. controls (246.96,38.82) and (245.98,38.09) .. (245,35) .. controls (243.54,30.42) and (243.57,25.58) .. (245,20) .. controls (246.43,14.42) and (254.67,15.06) .. (260,15) ;
\draw [shift={(250,40)}, rotate = 205.96] [color={rgb, 255:red, 191; green, 97; blue, 106 }  ,draw opacity=1 ][line width=0.75]    (4.37,-1.32) .. controls (2.78,-0.56) and (1.32,-0.12) .. (0,0) .. controls (1.32,0.12) and (2.78,0.56) .. (4.37,1.32)   ;
\draw [color={rgb, 255:red, 191; green, 97; blue, 106 }  ,draw opacity=1 ]   (260,40) -- (268,40) ;
\draw [shift={(270,40)}, rotate = 180] [color={rgb, 255:red, 191; green, 97; blue, 106 }  ,draw opacity=1 ][line width=0.75]    (4.37,-1.32) .. controls (2.78,-0.56) and (1.32,-0.12) .. (0,0) .. controls (1.32,0.12) and (2.78,0.56) .. (4.37,1.32)   ;
\draw [color={rgb, 255:red, 191; green, 97; blue, 106 }  ,draw opacity=1 ]   (280,40) .. controls (282.43,39.87) and (283.54,39.58) .. (285,35) .. controls (286.46,30.42) and (286.43,25.58) .. (285,20) .. controls (283.73,15.06) and (277.13,15) .. (271.94,15) ;
\draw [shift={(270,15)}, rotate = 0.66] [color={rgb, 255:red, 191; green, 97; blue, 106 }  ,draw opacity=1 ][line width=0.75]    (4.37,-1.32) .. controls (2.78,-0.56) and (1.32,-0.12) .. (0,0) .. controls (1.32,0.12) and (2.78,0.56) .. (4.37,1.32)   ;
\draw   (235,75) -- (239.5,60) -- (260.5,60) -- (265,75) -- cycle ;
\draw    (240,85) -- (240,75) ;
\draw    (260,85) -- (260,75) ;
\draw [color={rgb, 255:red, 191; green, 97; blue, 106 }  ,draw opacity=1 ]   (233.13,79.28) .. controls (231.96,78.82) and (230.98,78.09) .. (230,75) .. controls (228.54,70.42) and (228.57,65.58) .. (230,60) .. controls (231.43,54.42) and (239.67,55.06) .. (245,55) ;
\draw [shift={(235,80)}, rotate = 205.96] [color={rgb, 255:red, 191; green, 97; blue, 106 }  ,draw opacity=1 ][line width=0.75]    (4.37,-1.32) .. controls (2.78,-0.56) and (1.32,-0.12) .. (0,0) .. controls (1.32,0.12) and (2.78,0.56) .. (4.37,1.32)   ;
\draw [color={rgb, 255:red, 191; green, 97; blue, 106 }  ,draw opacity=1 ]   (245,80) -- (253,80) ;
\draw [shift={(255,80)}, rotate = 180] [color={rgb, 255:red, 191; green, 97; blue, 106 }  ,draw opacity=1 ][line width=0.75]    (4.37,-1.32) .. controls (2.78,-0.56) and (1.32,-0.12) .. (0,0) .. controls (1.32,0.12) and (2.78,0.56) .. (4.37,1.32)   ;
\draw [color={rgb, 255:red, 191; green, 97; blue, 106 }  ,draw opacity=1 ]   (265,80) .. controls (267.43,79.87) and (268.54,79.58) .. (270,75) .. controls (271.46,70.42) and (271.43,65.58) .. (270,60) .. controls (268.73,55.06) and (262.13,55) .. (256.94,55) ;
\draw [shift={(255,55)}, rotate = 0.66] [color={rgb, 255:red, 191; green, 97; blue, 106 }  ,draw opacity=1 ][line width=0.75]    (4.37,-1.32) .. controls (2.78,-0.56) and (1.32,-0.12) .. (0,0) .. controls (1.32,0.12) and (2.78,0.56) .. (4.37,1.32)   ;
\draw    (255,35) .. controls (255.03,49.73) and (250.03,44.73) .. (250,60) ;
\draw    (275,35) .. controls (275.03,49.73) and (285.03,44.73) .. (285,60) ;
\draw    (285,85) -- (285,60) ;
\draw  [draw opacity=0] (205,40) -- (230,40) -- (230,60) -- (205,60) -- cycle ;
\draw  [draw opacity=0] (280,40) -- (305,40) -- (305,60) -- (280,60) -- cycle ;
\draw   (135,115) -- (139.5,100) -- (160.5,100) -- (165,115) -- cycle ;
\draw    (150,100) -- (150,90) ;
\draw [color={rgb, 255:red, 191; green, 97; blue, 106 }  ,draw opacity=1 ]   (133.13,119.28) .. controls (131.96,118.82) and (130.98,118.09) .. (130,115) .. controls (128.54,110.42) and (128.57,105.58) .. (130,100) .. controls (131.43,94.42) and (139.67,95.06) .. (145,95) ;
\draw [shift={(135,120)}, rotate = 205.96] [color={rgb, 255:red, 191; green, 97; blue, 106 }  ,draw opacity=1 ][line width=0.75]    (4.37,-1.32) .. controls (2.78,-0.56) and (1.32,-0.12) .. (0,0) .. controls (1.32,0.12) and (2.78,0.56) .. (4.37,1.32)   ;
\draw [color={rgb, 255:red, 191; green, 97; blue, 106 }  ,draw opacity=1 ]   (145,120) -- (153,120) ;
\draw [shift={(155,120)}, rotate = 180] [color={rgb, 255:red, 191; green, 97; blue, 106 }  ,draw opacity=1 ][line width=0.75]    (4.37,-1.32) .. controls (2.78,-0.56) and (1.32,-0.12) .. (0,0) .. controls (1.32,0.12) and (2.78,0.56) .. (4.37,1.32)   ;
\draw [color={rgb, 255:red, 191; green, 97; blue, 106 }  ,draw opacity=1 ]   (165,120) .. controls (167.43,119.87) and (168.54,119.58) .. (170,115) .. controls (171.46,110.42) and (171.43,105.58) .. (170,100) .. controls (168.73,95.06) and (162.13,95) .. (156.94,95) ;
\draw [shift={(155,95)}, rotate = 0.66] [color={rgb, 255:red, 191; green, 97; blue, 106 }  ,draw opacity=1 ][line width=0.75]    (4.37,-1.32) .. controls (2.78,-0.56) and (1.32,-0.12) .. (0,0) .. controls (1.32,0.12) and (2.78,0.56) .. (4.37,1.32)   ;
\draw    (160,115) .. controls (160.03,129.73) and (165.03,124.73) .. (165,140) ;
\draw    (140,115) .. controls (140.03,129.73) and (130.03,124.73) .. (130,140) ;
\draw    (130,165) -- (130,140) ;
\draw   (155,155) -- (159.5,140) -- (170.5,140) -- (175,155) -- cycle ;
\draw [color={rgb, 255:red, 191; green, 97; blue, 106 }  ,draw opacity=1 ]   (155,140) .. controls (151.32,148.79) and (149.61,159.97) .. (155,160) .. controls (160.39,160.03) and (169.61,159.97) .. (175,160) .. controls (180.06,160.02) and (180.31,150.69) .. (175.9,141.71) ;
\draw [shift={(175,140)}, rotate = 60.63] [color={rgb, 255:red, 191; green, 97; blue, 106 }  ,draw opacity=1 ][line width=0.75]    (4.37,-1.32) .. controls (2.78,-0.56) and (1.32,-0.12) .. (0,0) .. controls (1.32,0.12) and (2.78,0.56) .. (4.37,1.32)   ;
\draw  [draw opacity=0] (170,120) -- (195,120) -- (195,140) -- (170,140) -- cycle ;
\draw    (217,120) -- (217,105) ;
\draw   (207,120) -- (227,120) -- (227,135) -- (207,135) -- cycle ;
\draw    (217,150) -- (217,135) ;
\draw [color={rgb, 255:red, 191; green, 97; blue, 106 }  ,draw opacity=1 ]   (209.96,139.93) .. controls (204.06,139.68) and (203.32,139.14) .. (202,135) .. controls (200.54,130.42) and (200.57,125.58) .. (202,120) .. controls (203.43,114.42) and (206.67,115.06) .. (212,115) ;
\draw [shift={(212,140)}, rotate = 181.8] [color={rgb, 255:red, 191; green, 97; blue, 106 }  ,draw opacity=1 ][line width=0.75]    (4.37,-1.32) .. controls (2.78,-0.56) and (1.32,-0.12) .. (0,0) .. controls (1.32,0.12) and (2.78,0.56) .. (4.37,1.32)   ;
\draw [color={rgb, 255:red, 191; green, 97; blue, 106 }  ,draw opacity=1 ]   (224.04,114.86) .. controls (228.32,114.74) and (230.73,116.02) .. (232,120) .. controls (233.46,124.58) and (233.43,129.42) .. (232,135) .. controls (230.57,140.58) and (227.33,139.94) .. (222,140) ;
\draw [shift={(222,115)}, rotate = 354.07] [color={rgb, 255:red, 191; green, 97; blue, 106 }  ,draw opacity=1 ][line width=0.75]    (4.37,-1.32) .. controls (2.78,-0.56) and (1.32,-0.12) .. (0,0) .. controls (1.32,0.12) and (2.78,0.56) .. (4.37,1.32)   ;
\draw   (270,115) -- (274.5,100) -- (295.5,100) -- (300,115) -- cycle ;
\draw    (285,100) -- (285,90) ;
\draw [color={rgb, 255:red, 191; green, 97; blue, 106 }  ,draw opacity=1 ]   (268.13,119.28) .. controls (266.96,118.82) and (265.98,118.09) .. (265,115) .. controls (263.54,110.42) and (263.57,105.58) .. (265,100) .. controls (266.43,94.42) and (274.67,95.06) .. (280,95) ;
\draw [shift={(270,120)}, rotate = 205.96] [color={rgb, 255:red, 191; green, 97; blue, 106 }  ,draw opacity=1 ][line width=0.75]    (4.37,-1.32) .. controls (2.78,-0.56) and (1.32,-0.12) .. (0,0) .. controls (1.32,0.12) and (2.78,0.56) .. (4.37,1.32)   ;
\draw [color={rgb, 255:red, 191; green, 97; blue, 106 }  ,draw opacity=1 ]   (280,120) -- (288,120) ;
\draw [shift={(290,120)}, rotate = 180] [color={rgb, 255:red, 191; green, 97; blue, 106 }  ,draw opacity=1 ][line width=0.75]    (4.37,-1.32) .. controls (2.78,-0.56) and (1.32,-0.12) .. (0,0) .. controls (1.32,0.12) and (2.78,0.56) .. (4.37,1.32)   ;
\draw [color={rgb, 255:red, 191; green, 97; blue, 106 }  ,draw opacity=1 ]   (300,120) .. controls (302.43,119.87) and (303.54,119.58) .. (305,115) .. controls (306.46,110.42) and (306.43,105.58) .. (305,100) .. controls (303.73,95.06) and (297.13,95) .. (291.94,95) ;
\draw [shift={(290,95)}, rotate = 0.66] [color={rgb, 255:red, 191; green, 97; blue, 106 }  ,draw opacity=1 ][line width=0.75]    (4.37,-1.32) .. controls (2.78,-0.56) and (1.32,-0.12) .. (0,0) .. controls (1.32,0.12) and (2.78,0.56) .. (4.37,1.32)   ;
\draw    (275,115) .. controls (275.03,129.73) and (270.03,124.73) .. (270,140) ;
\draw    (295,115) .. controls (295.03,129.73) and (305.03,124.73) .. (305,140) ;
\draw    (305,165) -- (305,140) ;
\draw  [draw opacity=0] (305,120) -- (330,120) -- (330,140) -- (305,140) -- cycle ;
\draw  [draw opacity=0] (240,120) -- (260,120) -- (260,140) -- (240,140) -- cycle ;
\draw   (260,155) -- (264.5,140) -- (275.5,140) -- (280,155) -- cycle ;
\draw [color={rgb, 255:red, 191; green, 97; blue, 106 }  ,draw opacity=1 ]   (260,139.5) .. controls (256.32,148.29) and (254.61,159.47) .. (260,159.5) .. controls (265.39,159.53) and (274.61,159.47) .. (280,159.5) .. controls (285.06,159.52) and (285.31,150.19) .. (280.9,141.21) ;
\draw [shift={(280,139.5)}, rotate = 60.63] [color={rgb, 255:red, 191; green, 97; blue, 106 }  ,draw opacity=1 ][line width=0.75]    (4.37,-1.32) .. controls (2.78,-0.56) and (1.32,-0.12) .. (0,0) .. controls (1.32,0.12) and (2.78,0.56) .. (4.37,1.32)   ;

\draw (170,27.5) node  [font=\footnotesize]  {$a$};
\draw (142.5,27.5) node  [font=\tiny,color={rgb, 255:red, 191; green, 97; blue, 106 }  ,opacity=1 ]  {$a_{0}$};
\draw (170,47.5) node  [font=\tiny,color={rgb, 255:red, 191; green, 97; blue, 106 }  ,opacity=1 ]  {$a_{1}$};
\draw (197.5,27.5) node  [font=\tiny,color={rgb, 255:red, 191; green, 97; blue, 106 }  ,opacity=1 ]  {$a_{2}$};
\draw (185,67.5) node  [font=\footnotesize]  {$b$};
\draw (157.5,67.5) node  [font=\tiny,color={rgb, 255:red, 191; green, 97; blue, 106 }  ,opacity=1 ]  {$b_{0}$};
\draw (185,87.5) node  [font=\tiny,color={rgb, 255:red, 191; green, 97; blue, 106 }  ,opacity=1 ]  {$b_{1}$};
\draw (212.5,67.5) node  [font=\tiny,color={rgb, 255:red, 191; green, 97; blue, 106 }  ,opacity=1 ]  {$b_{2}$};
\draw (265,27.5) node  [font=\footnotesize]  {$c$};
\draw (237.5,27.5) node  [font=\tiny,color={rgb, 255:red, 191; green, 97; blue, 106 }  ,opacity=1 ]  {$c_{0}$};
\draw (265,47.5) node  [font=\tiny,color={rgb, 255:red, 191; green, 97; blue, 106 }  ,opacity=1 ]  {$c_{1}$};
\draw (292.5,27.5) node  [font=\tiny,color={rgb, 255:red, 191; green, 97; blue, 106 }  ,opacity=1 ]  {$c_{2}$};
\draw (250,67.5) node  [font=\footnotesize]  {$d$};
\draw (223.5,67.5) node  [font=\tiny,color={rgb, 255:red, 191; green, 97; blue, 106 }  ,opacity=1 ]  {$d_{0}$};
\draw (250,87.5) node  [font=\tiny,color={rgb, 255:red, 191; green, 97; blue, 106 }  ,opacity=1 ]  {$d_{1}$};
\draw (277.5,67.5) node  [font=\tiny,color={rgb, 255:red, 191; green, 97; blue, 106 }  ,opacity=1 ]  {$d_{2}$};
\draw (217.5,50) node  [font=\footnotesize]  {$=$};
\draw (297.5,50) node  [font=\footnotesize]  {$;$};
\draw (150,107.5) node  [font=\footnotesize]  {$a$};
\draw (122.5,107.5) node  [font=\tiny,color={rgb, 255:red, 191; green, 97; blue, 106 }  ,opacity=1 ]  {$a_{0}$};
\draw (150,127.5) node  [font=\tiny,color={rgb, 255:red, 191; green, 97; blue, 106 }  ,opacity=1 ]  {$a_{1}$};
\draw (177.5,107.5) node  [font=\tiny,color={rgb, 255:red, 191; green, 97; blue, 106 }  ,opacity=1 ]  {$a_{2}$};
\draw (165.98,165) node  [font=\tiny,color={rgb, 255:red, 191; green, 97; blue, 106 }  ,opacity=1 ]  {$b_{0}$};
\draw (165,147.5) node  [font=\footnotesize]  {$b$};
\draw (182.5,130) node  [font=\footnotesize]  {$=$};
\draw (217,127.5) node  [font=\footnotesize]  {$c$};
\draw (194.5,127.5) node  [font=\tiny,color={rgb, 255:red, 191; green, 97; blue, 106 }  ,opacity=1 ]  {$c_{0}$};
\draw (239.5,130) node  [font=\tiny,color={rgb, 255:red, 191; green, 97; blue, 106 }  ,opacity=1 ]  {$c_{1}$};
\draw (285,107.5) node  [font=\footnotesize]  {$d$};
\draw (257.5,107.5) node  [font=\tiny,color={rgb, 255:red, 191; green, 97; blue, 106 }  ,opacity=1 ]  {$d_{0}$};
\draw (285,127.5) node  [font=\tiny,color={rgb, 255:red, 191; green, 97; blue, 106 }  ,opacity=1 ]  {$d_{1}$};
\draw (312.5,107.5) node  [font=\tiny,color={rgb, 255:red, 191; green, 97; blue, 106 }  ,opacity=1 ]  {$d_{2}$};
\draw (317.5,130) node  [font=\footnotesize]  {$;$};
\draw (250,130) node  [font=\footnotesize]  {$=$};
\draw (270.98,164.5) node  [font=\tiny,color={rgb, 255:red, 191; green, 97; blue, 106 }  ,opacity=1 ]  {$e_{0}$};
\draw (270,147.5) node  [font=\footnotesize]  {$e$};

\end{tikzpicture}
         \caption{Equations between contours from $α, ρ,$ and $λ$ in $𝕍$.}
        \label{fig:promonoidalContourCompose}
      \end{figure}
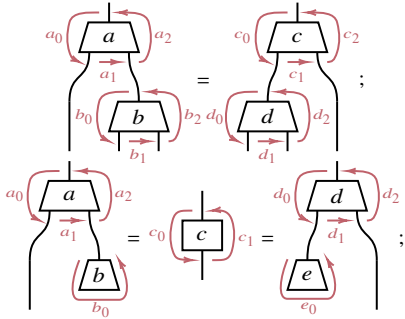
    \end{definition}
  
  \begin{proposition} \label{prop:contourFunctor}
    Contour gives a functor $𝓒 : \pMon → \Cat$.
  \end{proposition}
  \begin{proof}
    See Appendix, \Cref{ax:prop:contourFunctor}.
  \end{proof}

  \begin{remark} \label{rem:mellies}
    The \contour{} of a multicategory was first introduced by Melliès and Zeilberger \cite{mellies2022parsing}. \Cref{defn:contour} and the following \Cref{th:catpromadj} closely follow their work; although the promonoidal version we introduce does involve fewer equations due to the extra coherence (\Cref{remark:multicategories}).
  \end{remark}

\subsection{The Promonoidal Category of Spliced Arrows}

  We described a category tracking the process of decomposing in a given \promonoidalCategory{}.
  However, we want to go the other way around: given a category, what is the \promonoidalCategory{} describing decomposition in that category?
  This subsection finds a right adjoint to the \contour{} construction: the \splicedArrows{} \promonoidalCategory{}.
  \SplicedArrows{} have already been used to describe context in parsing
  \cite{mellies2022parsing}.

  \begin{definition}[Spliced arrows] \label{defn:splicedArrows}
    \defining{linkSplice}{}
    Let $ℂ$ be a category. The \promonoidalCategory{} of \emph{spliced arrows}, $\Splice{ℂ}$, has as objects pairs of objects of $ℂ$. It uses the following \profunctors{} to define morphisms, \splits{} and \units{}.
    \begin{align*}
      \Splice{ℂ} \left(
      \biobj{A}{B};
      \biobj{X}{Y} \right) &= ℂ(A;X) × ℂ(Y,B); \\
      \Splice{ℂ}(\biobj{A}{B} ; \biobj{X}{Y} ◁ \biobj{X'}{Y'}) &= 
      ℂ(A;X) × ℂ(Y;X') × ℂ(Y';B); \\      
      \Splice{ℂ}(\biobj{A}{B} ; N) &= ℂ(A;B).
    \end{align*}
  \end{definition}
    In other words, morphisms are \emph{pairs of arrows} $f \colon A \to X$ and $g \colon Y \to B$. 
    \Splits{} are \emph{triples of arrows} $f \colon A \to X$, $g \colon Y \to X'$ and $h \colon Y' \to B$. \Units{} are simply \emph{arrows} $f \colon A \to B$. We use the following notation for
    \begin{align*}
      & \mbox{morphisms,}&f& ⨾ \square ⨾ g &&\in \Splice{ℂ} \left(
        \biobj{A}{B};
        \biobj{X}{Y} \right); \\
      & \mbox{\splits{},}&f& ⨾ \square ⨾ g ⨾ \square ⨾ h &&\in \Splice{ℂ} \left( \biobj{A}{B}; \biobj{X}{Y} \triangleleft \biobj{X'}{Y'} \right); \\
      & \mbox{and \units{},}&f& &&\in \Splice{ℂ} \left(
        \biobj{A}{B}; N \right).
    \end{align*}
    The profunctor actions, associativity and unitality of the promonoidal category are defined in a straightforward way by \emph{filling the holes}. For instance,
    $$\begin{aligned}(f ⨾ \square ⨾ g ⨾ \square ⨾ h) ≺_1 (u ⨾ \square ⨾ v) =
    (f ⨾ u ⨾ \square ⨾ v ⨾ g ⨾ \square ⨾ h), \\
    (f ⨾ \square ⨾ g ⨾ \square ⨾ h) ≺_2 (u ⨾ \square ⨾ v) =
    (f ⨾ \square ⨾ g ⨾ u ⨾ \square ⨾ v ⨾ h).\end{aligned}$$
    See the Appendix, \Cref{sec:ax:sequentialContext} for details.

  \begin{proposition}
    \label{prop:spliceIsPromonoidal}
    Spliced arrows form a \promonoidalCategory{} with their splits, units, and suitable coherence morphisms.
  \end{proposition}
  \begin{proof}
    See Appendix, \Cref{ax:prop:spliceIsPromonoidal}.
  \end{proof}

  As a consequence, we can talk about spliced arrows with an arbitrary number of holes: for instance, a three-way split arises as a split filled by another split, in either position. For instance,
  $$\quasplice{f_0}{f_1}{f_2}{f_3}$$
  can be written in two different ways,
  $$\begin{aligned}
    & \trisplice{f_0}{f_2}{f_3} \mathbin{≺_1} \trisplice{id}{f_1}{id} & \mbox{ or } \\
    & \trisplice{f_0}{f_1}{f_3} \mathbin{≺_2} \trisplice{id}{f_2}{id}.
  \end{aligned}$$

  \begin{proposition} \label{prop:spliceFunctor}
    Splice gives a functor $\Splice{} : \Cat → \pMon$.
  \end{proposition}
  \begin{proof}
    See Appendix, \Cref{ax:prop:spliceFunctor}.
  \end{proof}

\begin{theorem} \label{th:catpromadj}
  There exists an adjunction between categories and \promonoidalCategories{}, where the contour of a \promonoidal{} is the left adjoint, and the splice category is the right adjoint.
\end{theorem}
\begin{proof}
  See Appendix, \Cref{ax:th:catpromadj}.
\end{proof}

\SplicedArrows{} can be computed for \emph{any} category, including monoidal categories. However, we expect the \splicedArrows{} of a monoidal category to have a richer algebraic structure. This extra structure is the subject of the next section.

\section{Parallel-Sequential Context}
\label{sec:parallelContext}

Monoidal categories are an algebraic structure for sequential and parallel composition: they contain a ``tensoring'' operator on morphisms, $(⊗)$, apart from the usual sequencing, $(⨾)$, and identities $(\mathrm{id})$.

Assume a monoidal morphism factors as follows,
$$f_0 ⨾ (g ⊗ (h ⨾ (k ⊗ (l_0 ⨾ l_1)))) ⨾ f_1.$$
We can say that this morphism came from dividing everything between $f_0$ and $f_1$ by a tensor. That is, from a context $f_0 ⨾ (\square ⊗ \square) ⨾ f_1$. We filled the first hole of this context with a $g$, and then proceeded to split the second part as $h ⨾ (\square ⊗ \square) ⨾ \im$. Finally, we filled the first part with $k$ and the second one we left disconnected by filling it with $l_0$, $\im_I$, and $l_1$.

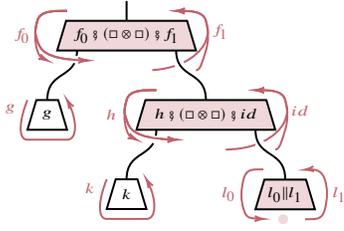
\begin{figure}[ht]
  \centering

\tikzset{every picture/.style={line width=0.75pt}} %

\begin{tikzpicture}[x=0.75pt,y=0.75pt,yscale=-1,xscale=1]
\draw [color={rgb, 255:red, 191; green, 97; blue, 106 }  ,draw opacity=1 ]   (200,45) .. controls (220.36,44.69) and (228.54,34.58) .. (230,30) .. controls (231.46,25.42) and (231.43,25.58) .. (230,20) .. controls (228.73,15.06) and (222.13,15) .. (216.94,15) ;
\draw [shift={(215,15)}, rotate = 0.66] [color={rgb, 255:red, 191; green, 97; blue, 106 }  ,draw opacity=1 ][line width=0.75]    (4.37,-1.32) .. controls (2.78,-0.56) and (1.32,-0.12) .. (0,0) .. controls (1.32,0.12) and (2.78,0.56) .. (4.37,1.32)   ;
\draw [color={rgb, 255:red, 191; green, 97; blue, 106 }  ,draw opacity=1 ]   (167.82,39.94) .. controls (148.45,39.26) and (146.4,34.42) .. (145,30) .. controls (143.54,25.42) and (143.57,25.58) .. (145,20) .. controls (146.43,14.42) and (154.67,15.06) .. (160,15) ;
\draw [shift={(170,40)}, rotate = 181.03] [color={rgb, 255:red, 191; green, 97; blue, 106 }  ,draw opacity=1 ][line width=0.75]    (4.37,-1.32) .. controls (2.78,-0.56) and (1.32,-0.12) .. (0,0) .. controls (1.32,0.12) and (2.78,0.56) .. (4.37,1.32)   ;
\draw    (165,35) .. controls (165.03,49.73) and (150.03,44.73) .. (150,60) ;
\draw  [color={rgb, 255:red, 0; green, 0; blue, 0 }  ,draw opacity=1 ][fill={rgb, 255:red, 191; green, 97; blue, 106 }  ,fill opacity=0.25 ] (155,35) -- (159.5,20) -- (220.5,20) -- (225,35) -- cycle ;
\draw [color={rgb, 255:red, 191; green, 97; blue, 106 }  ,draw opacity=1 ]   (149.17,39.65) .. controls (147.5,39.04) and (146.26,38.95) .. (145,35) .. controls (143.54,30.42) and (143.57,25.58) .. (145,20) .. controls (146.43,14.42) and (154.67,15.06) .. (160,15) ;
\draw [shift={(150,40)}, rotate = 205.96] [color={rgb, 255:red, 191; green, 97; blue, 106 }  ,draw opacity=1 ][line width=0.75]    (4.37,-1.32) .. controls (2.78,-0.56) and (1.32,-0.12) .. (0,0) .. controls (1.32,0.12) and (2.78,0.56) .. (4.37,1.32)   ;
\draw [color={rgb, 255:red, 191; green, 97; blue, 106 }  ,draw opacity=1 ]   (225,45) .. controls (227.43,44.87) and (228.54,39.58) .. (230,35) .. controls (231.46,30.42) and (231.43,25.58) .. (230,20) .. controls (228.73,15.06) and (222.13,15) .. (216.94,15) ;
\draw [shift={(215,15)}, rotate = 0.66] [color={rgb, 255:red, 191; green, 97; blue, 106 }  ,draw opacity=1 ][line width=0.75]    (4.37,-1.32) .. controls (2.78,-0.56) and (1.32,-0.12) .. (0,0) .. controls (1.32,0.12) and (2.78,0.56) .. (4.37,1.32)   ;
\draw  [draw opacity=0][fill={rgb, 255:red, 255; green, 255; blue, 255 }  ,fill opacity=1 ] (214.13,41.88) .. controls (214.13,40.49) and (215.24,39.38) .. (216.63,39.38) .. controls (218.01,39.38) and (219.13,40.49) .. (219.13,41.88) .. controls (219.13,43.26) and (218.01,44.38) .. (216.63,44.38) .. controls (215.24,44.38) and (214.13,43.26) .. (214.13,41.88) -- cycle ;
\draw    (190,20) -- (190,10) ;
\draw  [draw opacity=0][fill={rgb, 255:red, 255; green, 255; blue, 255 }  ,fill opacity=1 ] (162,39.64) .. controls (162,38.53) and (163.34,37.64) .. (165,37.64) .. controls (166.66,37.64) and (168,38.53) .. (168,39.64) .. controls (168,40.74) and (166.66,41.64) .. (165,41.64) .. controls (163.34,41.64) and (162,40.74) .. (162,39.64) -- cycle ;
\draw  [draw opacity=0][fill={rgb, 255:red, 255; green, 255; blue, 255 }  ,fill opacity=1 ] (197,45) .. controls (197,43.9) and (198.34,43) .. (200,43) .. controls (201.66,43) and (203,43.9) .. (203,45) .. controls (203,46.1) and (201.66,47) .. (200,47) .. controls (198.34,47) and (197,46.1) .. (197,45) -- cycle ;
\draw [color={rgb, 255:red, 191; green, 97; blue, 106 }  ,draw opacity=1 ]   (167.82,39.94) .. controls (148.45,39.26) and (146.4,34.42) .. (145,30) .. controls (143.54,25.42) and (143.57,25.58) .. (145,20) .. controls (146.43,14.42) and (154.67,15.06) .. (160,15) ;
\draw [shift={(170,40)}, rotate = 181.03] [color={rgb, 255:red, 191; green, 97; blue, 106 }  ,draw opacity=1 ][line width=0.75]    (4.37,-1.32) .. controls (2.78,-0.56) and (1.32,-0.12) .. (0,0) .. controls (1.32,0.12) and (2.78,0.56) .. (4.37,1.32)   ;
\draw   (140,75) -- (144.5,60) -- (155.5,60) -- (160,75) -- cycle ;
\draw [color={rgb, 255:red, 191; green, 97; blue, 106 }  ,draw opacity=1 ]   (140,60) .. controls (136.32,68.79) and (134.61,79.97) .. (140,80) .. controls (145.39,80.03) and (154.61,79.97) .. (160,80) .. controls (165.06,80.02) and (165.31,70.69) .. (160.9,61.71) ;
\draw [shift={(160,60)}, rotate = 60.63] [color={rgb, 255:red, 191; green, 97; blue, 106 }  ,draw opacity=1 ][line width=0.75]    (4.37,-1.32) .. controls (2.78,-0.56) and (1.32,-0.12) .. (0,0) .. controls (1.32,0.12) and (2.78,0.56) .. (4.37,1.32)   ;
\draw  [draw opacity=0] (130,60) -- (150,60) -- (150,80) -- (130,80) -- cycle ;
\draw  [color={rgb, 255:red, 0; green, 0; blue, 0 }  ,draw opacity=1 ][fill={rgb, 255:red, 191; green, 97; blue, 106 }  ,fill opacity=0.25 ] (195,75) -- (199.5,60) -- (260.5,60) -- (265,75) -- cycle ;
\draw [color={rgb, 255:red, 191; green, 97; blue, 106 }  ,draw opacity=1 ]   (212.82,79.94) .. controls (193.45,79.26) and (191.4,74.42) .. (190,70) .. controls (188.54,65.42) and (188.57,65.58) .. (190,60) .. controls (191.43,54.42) and (199.67,55.06) .. (205,55) ;
\draw [shift={(215,80)}, rotate = 181.03] [color={rgb, 255:red, 191; green, 97; blue, 106 }  ,draw opacity=1 ][line width=0.75]    (4.37,-1.32) .. controls (2.78,-0.56) and (1.32,-0.12) .. (0,0) .. controls (1.32,0.12) and (2.78,0.56) .. (4.37,1.32)   ;
\draw [color={rgb, 255:red, 191; green, 97; blue, 106 }  ,draw opacity=1 ]   (194.17,79.65) .. controls (192.5,79.04) and (191.26,78.95) .. (190,75) .. controls (188.54,70.42) and (188.57,65.58) .. (190,60) .. controls (191.43,54.42) and (199.67,55.06) .. (205,55) ;
\draw [shift={(195,80)}, rotate = 205.96] [color={rgb, 255:red, 191; green, 97; blue, 106 }  ,draw opacity=1 ][line width=0.75]    (4.37,-1.32) .. controls (2.78,-0.56) and (1.32,-0.12) .. (0,0) .. controls (1.32,0.12) and (2.78,0.56) .. (4.37,1.32)   ;
\draw [color={rgb, 255:red, 191; green, 97; blue, 106 }  ,draw opacity=1 ]   (240,85) .. controls (260.36,84.69) and (268.54,74.58) .. (270,70) .. controls (271.46,65.42) and (271.43,65.58) .. (270,60) .. controls (268.73,55.06) and (262.13,55) .. (256.94,55) ;
\draw [shift={(255,55)}, rotate = 0.66] [color={rgb, 255:red, 191; green, 97; blue, 106 }  ,draw opacity=1 ][line width=0.75]    (4.37,-1.32) .. controls (2.78,-0.56) and (1.32,-0.12) .. (0,0) .. controls (1.32,0.12) and (2.78,0.56) .. (4.37,1.32)   ;
\draw [color={rgb, 255:red, 191; green, 97; blue, 106 }  ,draw opacity=1 ]   (265,85) .. controls (267.43,84.87) and (268.54,79.58) .. (270,75) .. controls (271.46,70.42) and (271.43,65.58) .. (270,60) .. controls (268.73,55.06) and (262.13,55) .. (256.94,55) ;
\draw [shift={(255,55)}, rotate = 0.66] [color={rgb, 255:red, 191; green, 97; blue, 106 }  ,draw opacity=1 ][line width=0.75]    (4.37,-1.32) .. controls (2.78,-0.56) and (1.32,-0.12) .. (0,0) .. controls (1.32,0.12) and (2.78,0.56) .. (4.37,1.32)   ;
\draw    (205,74.99) .. controls (205.03,89.72) and (190.03,84.72) .. (190,99.99) ;
\draw   (180,114.99) -- (184.5,99.99) -- (195.5,99.99) -- (200,114.99) -- cycle ;
\draw [color={rgb, 255:red, 191; green, 97; blue, 106 }  ,draw opacity=1 ]   (180,99.99) .. controls (176.32,108.78) and (174.61,119.97) .. (180,119.99) .. controls (185.39,120.02) and (194.61,119.97) .. (200,119.99) .. controls (205.06,120.02) and (205.31,110.68) .. (200.9,101.7) ;
\draw [shift={(200,99.99)}, rotate = 60.63] [color={rgb, 255:red, 191; green, 97; blue, 106 }  ,draw opacity=1 ][line width=0.75]    (4.37,-1.32) .. controls (2.78,-0.56) and (1.32,-0.12) .. (0,0) .. controls (1.32,0.12) and (2.78,0.56) .. (4.37,1.32)   ;
\draw  [draw opacity=0] (170,99.99) -- (190,99.99) -- (190,119.99) -- (170,119.99) -- cycle ;
\draw  [draw opacity=0][fill={rgb, 255:red, 255; green, 255; blue, 255 }  ,fill opacity=1 ] (202,80) .. controls (202,78.9) and (203.12,78) .. (204.5,78) .. controls (205.88,78) and (207,78.9) .. (207,80) .. controls (207,81.1) and (205.88,82) .. (204.5,82) .. controls (203.12,82) and (202,81.1) .. (202,80) -- cycle ;
\draw [color={rgb, 255:red, 191; green, 97; blue, 106 }  ,draw opacity=1 ]   (212.82,79.94) .. controls (193.45,79.26) and (191.4,74.42) .. (190,70) .. controls (188.54,65.42) and (188.57,65.58) .. (190,60) .. controls (191.43,54.42) and (199.67,55.06) .. (205,55) ;
\draw [shift={(215,80)}, rotate = 181.03] [color={rgb, 255:red, 191; green, 97; blue, 106 }  ,draw opacity=1 ][line width=0.75]    (4.37,-1.32) .. controls (2.78,-0.56) and (1.32,-0.12) .. (0,0) .. controls (1.32,0.12) and (2.78,0.56) .. (4.37,1.32)   ;
\draw  [fill={rgb, 255:red, 191; green, 97; blue, 106 }  ,fill opacity=0.25 ] (255,115.15) -- (259.5,100.15) -- (280.5,100.15) -- (285,115.15) -- cycle ;
\draw  [draw opacity=0] (259,115.15) -- (279,115.15) -- (279,135.15) -- (259,135.15) -- cycle ;
\draw [color={rgb, 255:red, 191; green, 97; blue, 106 }  ,draw opacity=1 ]   (256.96,119.93) .. controls (251.06,119.68) and (250.32,119.14) .. (249,115) .. controls (247.54,110.42) and (247.57,105.58) .. (249,100) .. controls (250.43,94.42) and (253.67,95.06) .. (259,95) ;
\draw [shift={(259,120)}, rotate = 181.8] [color={rgb, 255:red, 191; green, 97; blue, 106 }  ,draw opacity=1 ][line width=0.75]    (4.37,-1.32) .. controls (2.78,-0.56) and (1.32,-0.12) .. (0,0) .. controls (1.32,0.12) and (2.78,0.56) .. (4.37,1.32)   ;
\draw [color={rgb, 255:red, 191; green, 97; blue, 106 }  ,draw opacity=1 ]   (282.04,95.01) .. controls (286.32,94.89) and (288.73,96.17) .. (290,100.15) .. controls (291.46,104.73) and (291.43,109.57) .. (290,115.15) .. controls (288.57,120.73) and (285.33,120.09) .. (280,120.15) ;
\draw [shift={(280,95.15)}, rotate = 354.07] [color={rgb, 255:red, 191; green, 97; blue, 106 }  ,draw opacity=1 ][line width=0.75]    (4.37,-1.32) .. controls (2.78,-0.56) and (1.32,-0.12) .. (0,0) .. controls (1.32,0.12) and (2.78,0.56) .. (4.37,1.32)   ;
\draw  [draw opacity=0][fill={rgb, 255:red, 191; green, 97; blue, 106 }  ,fill opacity=0.25 ] (266.5,120.15) .. controls (266.5,118.77) and (267.62,117.65) .. (269,117.65) .. controls (270.38,117.65) and (271.5,118.77) .. (271.5,120.15) .. controls (271.5,121.53) and (270.38,122.65) .. (269,122.65) .. controls (267.62,122.65) and (266.5,121.53) .. (266.5,120.15) -- cycle ;
\draw  [draw opacity=0][fill={rgb, 255:red, 255; green, 255; blue, 255 }  ,fill opacity=1 ] (255,82.5) .. controls (255,81.12) and (256.12,80) .. (257.5,80) .. controls (258.88,80) and (260,81.12) .. (260,82.5) .. controls (260,83.88) and (258.88,85) .. (257.5,85) .. controls (256.12,85) and (255,83.88) .. (255,82.5) -- cycle ;
\draw    (255,75) .. controls (255.03,89.73) and (270.03,84.73) .. (270,100) ;
\draw    (215,35) .. controls (215.03,49.73) and (230.03,44.73) .. (230,60) ;

\draw (190,27.5) node  [font=\tiny]  {$f_{0} ⨾ ( \square \otimes \square ) ⨾ f_{1}$};
\draw (137.5,27.5) node  [font=\tiny,color={rgb, 255:red, 191; green, 97; blue, 106 }  ,opacity=1 ]  {$f_{0}$};
\draw (237.5,25.5) node  [font=\tiny,color={rgb, 255:red, 191; green, 97; blue, 106 }  ,opacity=1 ]  {$f_{1}$};
\draw (131.5,64.5) node  [font=\tiny,color={rgb, 255:red, 191; green, 97; blue, 106 }  ,opacity=1 ]  {$g$};
\draw (150,67.5) node  [font=\tiny]  {$g$};
\draw (230,67.5) node  [font=\tiny]  {$h ⨾ ( \square \otimes \square ) ⨾ id$};
\draw (182.5,67.5) node  [font=\tiny,color={rgb, 255:red, 191; green, 97; blue, 106 }  ,opacity=1 ]  {$h$};
\draw (277.5,65.5) node  [font=\tiny,color={rgb, 255:red, 191; green, 97; blue, 106 }  ,opacity=1 ]  {$id$};
\draw (171.5,104.49) node  [font=\tiny,color={rgb, 255:red, 191; green, 97; blue, 106 }  ,opacity=1 ]  {$k$};
\draw (190,107.49) node  [font=\tiny]  {$k$};
\draw (297.5,107.65) node  [font=\tiny,color={rgb, 255:red, 191; green, 97; blue, 106 }  ,opacity=1 ]  {$l_{1}$};
\draw (241.5,107.5) node  [font=\tiny,color={rgb, 255:red, 191; green, 97; blue, 106 }  ,opacity=1 ]  {$l_{0}$};
\draw (270,107.65) node  [font=\tiny]  {$l_{0} \| l_{1}$};

\end{tikzpicture}

   \caption{Decomposition of $f_0 ⨾ (g ⊗ (h ⨾ (k ⊗ (l_0 ⨾ l_1)))) ⨾ f_1$.}
  \label{fig:monoidalContourExample}
\end{figure}

This section studies decomposition of morphisms in a \emph{monoidal} category, in the same way we studied decomposition of morphisms in a category before.  We present an algebraic structure for decomposing both sequential and parallel compositions: \emph{produoidal categories}.

\subsection{Produoidal Categories}
\defining{linkProduoidalComponents}
Produoidal categories, first defined by Booker and Street \cite{bookerstreet13}, provide an algebraic structure for the interaction of sequential and parallel decomposition. A \produoidalCategory{} $𝕍$ not only contains \emph{morphisms}, $𝕍(X; Y)$, \emph{sequential splits}, $𝕍(X; Y_0 ◁ Y_1)$,  and \emph{sequential units}, $𝕍(X; N)$, as a \promonoidalCategory{} does; it also contains \emph{parallel splits}, $𝕍(X; Y_0 ⊗ Y_1)$ and \emph{parallel units}, $𝕍(X; I)$.

\begin{remark}[Nesting virtual structures]
  Notation for nesting functorial structures, say $(◁)$ and $(\otimes)$, is straightforward: we use expressions like $(X_1 \otimes Y_1) ◁ (X_2 \otimes Y_2)$ without a second thought. Nesting the virtual structures $(◁)$ and $(\otimes)$ is more subtle: defining $𝕍(•; X \otimes Y)$ and $𝕍(•; X ◁ Y)$ for each pair of objects $X$ and $Y$ does not itself define what something like $𝕍(•; (X_1 \otimes Y_1) ◁ (X_2 \otimes Y_2))$ means. Recall that, in the virtual case, $X_1 ◁ Y_1$ and $X_1 \otimes Y_1$ are not objects themselves: they are just names for the profunctors $𝕍(•; X_1 ◁ Y_1)$ and $𝕍(•; X_1 \otimes Y_1)$.

  Instead, when we write $𝕍(•; (X_1 \otimes Y_1) ◁ (X_2 \otimes Y_2))$, we formally mean
  $𝕍(•; •_1 ◁ •_2) \diamond 𝕍(•_1; X_1 \otimes Y_1) \diamond 𝕍(•_2; X_2 \otimes Y_2)$. By convention, nesting virtual structures means profunctor composition in this text.
\end{remark}

\begin{definition}[Produoidal category]
  \defining{linkproduoidal}{}
  \label{def:produoidal}
  A \emph{produoidal category} is a category $𝕍$ endowed with two \promonoidal{} structures,
  $$\begin{gathered}
    𝕍(• ; • ⊗ •) \colon 𝕍 × 𝕍 ⧑ 𝕍, \mbox{ and } 𝕍(•; I) \colon 1 ⧑ 𝕍, \\
    𝕍(• ;• ◁ •) \colon 𝕍 × 𝕍 ⧑ 𝕍, \mbox{ and } 𝕍(•; N) \colon 1 ⧑ 𝕍,
  \end{gathered}$$
  such that one laxly distributes over the other.
  This is to say that it is endowed with the following natural \emph{laxators},
  \begin{align*}
  ψ_2 \colon 𝕍(•;(X◁Y)⊗(Z◁W)) &→ 𝕍(•;(X⊗Z)◁(Y⊗W)),\\
  ψ_0 \colon 𝕍(•;I) &→ 𝕍(•;I◁I),\\
  φ_2 \colon 𝕍(•;N⊗N) &→ 𝕍(•;N),\\
  φ_0 \colon 𝕍(•;I) &→ 𝕍(•;N).
  \end{align*}
  Laxators, together with unitors and associators, must satisfy coherence conditions (see Appendix, \Cref{ax:def:produoidal}).
\end{definition}

\begin{definition}[Produoidal functor]
  \defining{linkproduoidalfunctor}{}\label{def:produoidalfunctor}
  Let $𝕍_{⊗,I,◁,N}$ and $𝕎_{⊘,J,◀,M}$ be produoidal categories. A \emph{produoidal functor} $F$ is a functor between the two categories $F : 𝕍(•,•) \to 𝕎(•,•)$ together with natural transformations
  \begin{align*}
    & F_{⊗} : 𝕍(A; B ⊗ C) \to 𝕎(FA; FB ⊘ FC), \\
    & F_{I} : 𝕍(A; I) \to 𝕎(FA; J), \\
    & F_{◁} : 𝕍(A; B ◁ C) \to 𝕎(FA; FB ◀ FC), \mbox{ and } \\
    & F_{N} : 𝕍(A; N) \to 𝕎(FA; M),
  \end{align*}
preserving coherence isomorphisms for each \promonoidal{} structure, and the laxators. Denote by $\pDuo$ the category of \produoidalCategories{} and \produoidalFunctors{}.
\end{definition}

\subsection{Monoidal Contour of a Produoidal Category}

Any \produoidalCategory{} freely generates a monoidal category, its \emph{monoidal contour}. Again, this is interpreted as a monoidal category tracking the processes of parallel and sequential decomposition described by the \produoidalCategory{}.
And again, the construction follows a pleasant geometric pattern, where we follow the shape of the decomposition, now in both the parallel and sequential dimensions.

\begin{definition}[Monoidal contour]
  \label{def:monoidalContour}
  \defining{linkMonoidalContour}{}
    The \emph{contour} of a \produoidalCategory{} $𝔹$ is the monoidal category $𝓓𝔹$ that has two objects, $X^L$ (left-handed) and $X^R$ (right-handed), for each object $X \in 𝔹_{\text{obj}}$; and has arrows those that arise from \emph{contouring} both sequential and parallel decompositions of the promonoidal category.
    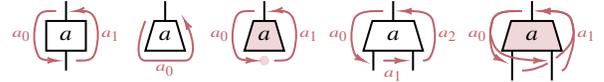
\begin{figure}[ht]
      \centering

\tikzset{every picture/.style={line width=0.75pt}} %

\begin{tikzpicture}[x=0.75pt,y=0.75pt,yscale=-1,xscale=1]
\draw  [fill={rgb, 255:red, 191; green, 97; blue, 106 }  ,fill opacity=0.25 ] (190,35) -- (194.5,20) -- (205.5,20) -- (210,35) -- cycle ;
\draw    (200,20) -- (200,10) ;
\draw  [draw opacity=0] (190,35) -- (210,35) -- (210,55) -- (190,55) -- cycle ;
\draw    (100,20) -- (100,10) ;
\draw   (90,20) -- (110,20) -- (110,35) -- (90,35) -- cycle ;
\draw    (100,45) -- (100,35) ;
\draw [color={rgb, 255:red, 191; green, 97; blue, 106 }  ,draw opacity=1 ]   (92.96,39.93) .. controls (87.06,39.68) and (86.32,39.14) .. (85,35) .. controls (83.54,30.42) and (83.57,25.58) .. (85,20) .. controls (86.43,14.42) and (89.67,15.06) .. (95,15) ;
\draw [shift={(95,40)}, rotate = 181.8] [color={rgb, 255:red, 191; green, 97; blue, 106 }  ,draw opacity=1 ][line width=0.75]    (4.37,-1.32) .. controls (2.78,-0.56) and (1.32,-0.12) .. (0,0) .. controls (1.32,0.12) and (2.78,0.56) .. (4.37,1.32)   ;
\draw [color={rgb, 255:red, 191; green, 97; blue, 106 }  ,draw opacity=1 ]   (107.04,14.86) .. controls (111.32,14.74) and (113.73,16.02) .. (115,20) .. controls (116.46,24.58) and (116.43,29.42) .. (115,35) .. controls (113.57,40.58) and (110.33,39.94) .. (105,40) ;
\draw [shift={(105,15)}, rotate = 354.07] [color={rgb, 255:red, 191; green, 97; blue, 106 }  ,draw opacity=1 ][line width=0.75]    (4.37,-1.32) .. controls (2.78,-0.56) and (1.32,-0.12) .. (0,0) .. controls (1.32,0.12) and (2.78,0.56) .. (4.37,1.32)   ;
\draw  [color={rgb, 255:red, 0; green, 0; blue, 0 }  ,draw opacity=1 ][fill={rgb, 255:red, 191; green, 97; blue, 106 }  ,fill opacity=0.25 ] (320,35) -- (324.5,20) -- (345.5,20) -- (350,35) -- cycle ;
\draw    (325,50) -- (325,35) ;
\draw    (335,20) -- (335,10) ;
\draw [color={rgb, 255:red, 191; green, 97; blue, 106 }  ,draw opacity=1 ]   (319.17,39.65) .. controls (317.5,39.04) and (316.26,38.95) .. (315,35) .. controls (313.54,30.42) and (313.57,25.58) .. (315,20) .. controls (316.43,14.42) and (324.67,15.06) .. (330,15) ;
\draw [shift={(320,40)}, rotate = 205.96] [color={rgb, 255:red, 191; green, 97; blue, 106 }  ,draw opacity=1 ][line width=0.75]    (4.37,-1.32) .. controls (2.78,-0.56) and (1.32,-0.12) .. (0,0) .. controls (1.32,0.12) and (2.78,0.56) .. (4.37,1.32)   ;
\draw [color={rgb, 255:red, 191; green, 97; blue, 106 }  ,draw opacity=1 ]   (350,45) .. controls (352.43,44.87) and (353.54,39.58) .. (355,35) .. controls (356.46,30.42) and (356.43,25.58) .. (355,20) .. controls (353.73,15.06) and (347.13,15) .. (341.94,15) ;
\draw [shift={(340,15)}, rotate = 0.66] [color={rgb, 255:red, 191; green, 97; blue, 106 }  ,draw opacity=1 ][line width=0.75]    (4.37,-1.32) .. controls (2.78,-0.56) and (1.32,-0.12) .. (0,0) .. controls (1.32,0.12) and (2.78,0.56) .. (4.37,1.32)   ;
\draw [color={rgb, 255:red, 191; green, 97; blue, 106 }  ,draw opacity=1 ]   (337.82,39.94) .. controls (318.45,39.26) and (316.4,34.42) .. (315,30) .. controls (313.54,25.42) and (313.57,25.58) .. (315,20) .. controls (316.43,14.42) and (324.67,15.06) .. (330,15) ;
\draw [shift={(340,40)}, rotate = 181.03] [color={rgb, 255:red, 191; green, 97; blue, 106 }  ,draw opacity=1 ][line width=0.75]    (4.37,-1.32) .. controls (2.78,-0.56) and (1.32,-0.12) .. (0,0) .. controls (1.32,0.12) and (2.78,0.56) .. (4.37,1.32)   ;
\draw [color={rgb, 255:red, 191; green, 97; blue, 106 }  ,draw opacity=1 ]   (330,45) .. controls (341.21,45.37) and (353.54,34.58) .. (355,30) .. controls (356.46,25.42) and (356.43,25.58) .. (355,20) .. controls (353.73,15.06) and (347.13,15) .. (341.94,15) ;
\draw [shift={(340,15)}, rotate = 0.66] [color={rgb, 255:red, 191; green, 97; blue, 106 }  ,draw opacity=1 ][line width=0.75]    (4.37,-1.32) .. controls (2.78,-0.56) and (1.32,-0.12) .. (0,0) .. controls (1.32,0.12) and (2.78,0.56) .. (4.37,1.32)   ;
\draw  [draw opacity=0][fill={rgb, 255:red, 255; green, 255; blue, 255 }  ,fill opacity=1 ] (342.5,40) .. controls (342.5,38.62) and (343.62,37.5) .. (345,37.5) .. controls (346.38,37.5) and (347.5,38.62) .. (347.5,40) .. controls (347.5,41.38) and (346.38,42.5) .. (345,42.5) .. controls (343.62,42.5) and (342.5,41.38) .. (342.5,40) -- cycle ;
\draw [color={rgb, 255:red, 191; green, 97; blue, 106 }  ,draw opacity=1 ]   (192.96,39.93) .. controls (187.06,39.68) and (186.32,39.14) .. (185,35) .. controls (183.54,30.42) and (183.57,25.58) .. (185,20) .. controls (186.43,14.42) and (189.67,15.06) .. (195,15) ;
\draw [shift={(195,40)}, rotate = 181.8] [color={rgb, 255:red, 191; green, 97; blue, 106 }  ,draw opacity=1 ][line width=0.75]    (4.37,-1.32) .. controls (2.78,-0.56) and (1.32,-0.12) .. (0,0) .. controls (1.32,0.12) and (2.78,0.56) .. (4.37,1.32)   ;
\draw [color={rgb, 255:red, 191; green, 97; blue, 106 }  ,draw opacity=1 ]   (207.04,14.86) .. controls (211.32,14.74) and (213.73,16.02) .. (215,20) .. controls (216.46,24.58) and (216.43,29.42) .. (215,35) .. controls (213.57,40.58) and (210.33,39.94) .. (205,40) ;
\draw [shift={(205,15)}, rotate = 354.07] [color={rgb, 255:red, 191; green, 97; blue, 106 }  ,draw opacity=1 ][line width=0.75]    (4.37,-1.32) .. controls (2.78,-0.56) and (1.32,-0.12) .. (0,0) .. controls (1.32,0.12) and (2.78,0.56) .. (4.37,1.32)   ;
\draw  [draw opacity=0][fill={rgb, 255:red, 191; green, 97; blue, 106 }  ,fill opacity=0.25 ] (197.5,40) .. controls (197.5,38.62) and (198.62,37.5) .. (200,37.5) .. controls (201.38,37.5) and (202.5,38.62) .. (202.5,40) .. controls (202.5,41.38) and (201.38,42.5) .. (200,42.5) .. controls (198.62,42.5) and (197.5,41.38) .. (197.5,40) -- cycle ;
\draw    (345,50) -- (345,35) ;
\draw   (140,35) -- (144.5,20) -- (155.5,20) -- (160,35) -- cycle ;
\draw    (150,20) -- (150,10) ;
\draw [color={rgb, 255:red, 191; green, 97; blue, 106 }  ,draw opacity=1 ]   (140,20) .. controls (136.32,28.79) and (134.61,39.97) .. (140,40) .. controls (145.39,40.03) and (154.61,39.97) .. (160,40) .. controls (165.06,40.02) and (165.31,30.69) .. (160.9,21.71) ;
\draw [shift={(160,20)}, rotate = 60.63] [color={rgb, 255:red, 191; green, 97; blue, 106 }  ,draw opacity=1 ][line width=0.75]    (4.37,-1.32) .. controls (2.78,-0.56) and (1.32,-0.12) .. (0,0) .. controls (1.32,0.12) and (2.78,0.56) .. (4.37,1.32)   ;
\draw  [draw opacity=0] (140,35) -- (160,35) -- (160,55) -- (140,55) -- cycle ;
\draw   (250,35) -- (254.5,20) -- (275.5,20) -- (280,35) -- cycle ;
\draw    (255,45) -- (255,35) ;
\draw    (275,45) -- (275,35) ;
\draw    (265,20) -- (265,10) ;
\draw [color={rgb, 255:red, 191; green, 97; blue, 106 }  ,draw opacity=1 ]   (248.13,39.28) .. controls (246.96,38.82) and (245.98,38.09) .. (245,35) .. controls (243.54,30.42) and (243.57,25.58) .. (245,20) .. controls (246.43,14.42) and (254.67,15.06) .. (260,15) ;
\draw [shift={(250,40)}, rotate = 205.96] [color={rgb, 255:red, 191; green, 97; blue, 106 }  ,draw opacity=1 ][line width=0.75]    (4.37,-1.32) .. controls (2.78,-0.56) and (1.32,-0.12) .. (0,0) .. controls (1.32,0.12) and (2.78,0.56) .. (4.37,1.32)   ;
\draw [color={rgb, 255:red, 191; green, 97; blue, 106 }  ,draw opacity=1 ]   (260,40) -- (268,40) ;
\draw [shift={(270,40)}, rotate = 180] [color={rgb, 255:red, 191; green, 97; blue, 106 }  ,draw opacity=1 ][line width=0.75]    (4.37,-1.32) .. controls (2.78,-0.56) and (1.32,-0.12) .. (0,0) .. controls (1.32,0.12) and (2.78,0.56) .. (4.37,1.32)   ;
\draw [color={rgb, 255:red, 191; green, 97; blue, 106 }  ,draw opacity=1 ]   (280,40) .. controls (282.43,39.87) and (283.54,39.58) .. (285,35) .. controls (286.46,30.42) and (286.43,25.58) .. (285,20) .. controls (283.73,15.06) and (277.13,15) .. (271.94,15) ;
\draw [shift={(270,15)}, rotate = 0.66] [color={rgb, 255:red, 191; green, 97; blue, 106 }  ,draw opacity=1 ][line width=0.75]    (4.37,-1.32) .. controls (2.78,-0.56) and (1.32,-0.12) .. (0,0) .. controls (1.32,0.12) and (2.78,0.56) .. (4.37,1.32)   ;

\draw (200,27.5) node  [font=\footnotesize]  {$a$};
\draw (100,27.5) node  [font=\footnotesize]  {$a$};
\draw (77.5,27.5) node  [font=\tiny,color={rgb, 255:red, 191; green, 97; blue, 106 }  ,opacity=1 ]  {$a_{0}$};
\draw (122.5,27.5) node  [font=\tiny,color={rgb, 255:red, 191; green, 97; blue, 106 }  ,opacity=1 ]  {$a_{1}$};
\draw (335,27.5) node  [font=\footnotesize]  {$a$};
\draw (307.5,27.5) node  [font=\tiny,color={rgb, 255:red, 191; green, 97; blue, 106 }  ,opacity=1 ]  {$a_{0}$};
\draw (362.5,27.5) node  [font=\tiny,color={rgb, 255:red, 191; green, 97; blue, 106 }  ,opacity=1 ]  {$a_{1}$};
\draw (222.5,27.5) node  [font=\tiny,color={rgb, 255:red, 191; green, 97; blue, 106 }  ,opacity=1 ]  {$a_{1}$};
\draw (177.5,27.5) node  [font=\tiny,color={rgb, 255:red, 191; green, 97; blue, 106 }  ,opacity=1 ]  {$a_{0}$};
\draw (150,27.5) node  [font=\footnotesize]  {$a$};
\draw (150,45) node  [font=\tiny,color={rgb, 255:red, 191; green, 97; blue, 106 }  ,opacity=1 ]  {$a_{0}$};
\draw (265,27.5) node  [font=\footnotesize]  {$a$};
\draw (237.5,27.5) node  [font=\tiny,color={rgb, 255:red, 191; green, 97; blue, 106 }  ,opacity=1 ]  {$a_{0}$};
\draw (265,47.5) node  [font=\tiny,color={rgb, 255:red, 191; green, 97; blue, 106 }  ,opacity=1 ]  {$a_{1}$};
\draw (292.5,27.5) node  [font=\tiny,color={rgb, 255:red, 191; green, 97; blue, 106 }  ,opacity=1 ]  {$a_{2}$};

\end{tikzpicture}
       \caption{Generators of the monoidal category of contours.}
      \label{fig:monoidal-contour}
    \end{figure}

    Specifically, it is freely presented by
    \emph{(i)} a pair of morphisms $a_0 \in 𝓓𝔹(A^L; X^L)$, $a_1 \in 𝓓𝔹(X^R;A^R)$ for each morphism $a \in 𝔹(A;X)$;
    \emph{(ii)} a morphism $a_0 \in 𝓓𝔹(A^L; A^R)$, for each \sequentialUnit{} $a \in ℂ(A;N)$;
    \emph{(iii)} a pair of morphisms $a_0 \in 𝓓𝔹(A^L; I)$ and $a_0 \in 𝓓𝔹(I; A^R)$, for each \parallelUnit{} $a \in 𝔹(A;I)$;
    \emph{(iv)} a triple of morphisms
    $a_0 \in 𝓓𝔹(A^L; X^L)$, $a_1 \in 𝓓𝔹(X^R; Y^L)$, $a_2 \in 𝓓𝔹(Y^R; A^R)$ for each \sequentialSplit{} $a \in 𝔹(A;X ◁ Y)$; and
    \emph{(v)} a pair of morphisms $a_0 \in 𝓓𝔹(A^L; X^L ⊗ Y^L)$ and $a_1 \in 𝓓𝔹(X^R ⊗ Y^R; A^R)$ for each \parallelSplit{} $a \in 𝔹(A; X ⊗ Y)$, see \Cref{fig:monoidal-contour}.

    We impose the same equations as in the categorical contour coming from the associator and unitor of the $◁$ structure; but moreover, we impose the following new equations, coming from the $⊗$ structure:
    For each application of associativity, $α(a ⨾_1 b) = c ⨾_2 d$, we impose the equations $a_0 ⨾ (b_0 ⊗ \im) = c_0 ⨾ (\im ⊗ d_0)$ and $(b_1 ⊗ \im) ⨾ a_1 = (\im ⊗ d_1) ⨾ c_1$. These follow from \Cref{fig:monoidalContourAssociativity}.

    \begin{figure}[ht]
      \centering

\tikzset{every picture/.style={line width=0.75pt}} %

\begin{tikzpicture}[x=0.75pt,y=0.75pt,yscale=-1,xscale=1]
\draw  [color={rgb, 255:red, 0; green, 0; blue, 0 }  ,draw opacity=1 ][fill={rgb, 255:red, 191; green, 97; blue, 106 }  ,fill opacity=0.25 ] (35,75) -- (39.5,60) -- (60.5,60) -- (65,75) -- cycle ;
\draw    (40,90) -- (40,75) ;
\draw [color={rgb, 255:red, 191; green, 97; blue, 106 }  ,draw opacity=1 ]   (34.17,79.65) .. controls (32.5,79.04) and (31.26,78.95) .. (30,75) .. controls (28.54,70.42) and (28.57,65.58) .. (30,60) .. controls (31.43,54.42) and (39.67,55.06) .. (45,55) ;
\draw [shift={(35,80)}, rotate = 205.96] [color={rgb, 255:red, 191; green, 97; blue, 106 }  ,draw opacity=1 ][line width=0.75]    (4.37,-1.32) .. controls (2.78,-0.56) and (1.32,-0.12) .. (0,0) .. controls (1.32,0.12) and (2.78,0.56) .. (4.37,1.32)   ;
\draw [color={rgb, 255:red, 191; green, 97; blue, 106 }  ,draw opacity=1 ]   (65,85) .. controls (67.43,84.87) and (68.54,79.58) .. (70,75) .. controls (71.46,70.42) and (71.43,65.58) .. (70,60) .. controls (68.73,55.06) and (62.13,55) .. (56.94,55) ;
\draw [shift={(55,55)}, rotate = 0.66] [color={rgb, 255:red, 191; green, 97; blue, 106 }  ,draw opacity=1 ][line width=0.75]    (4.37,-1.32) .. controls (2.78,-0.56) and (1.32,-0.12) .. (0,0) .. controls (1.32,0.12) and (2.78,0.56) .. (4.37,1.32)   ;
\draw [color={rgb, 255:red, 191; green, 97; blue, 106 }  ,draw opacity=1 ]   (45,85) .. controls (56.21,85.37) and (68.54,74.58) .. (70,70) .. controls (71.46,65.42) and (71.43,65.58) .. (70,60) .. controls (68.73,55.06) and (62.13,55) .. (56.94,55) ;
\draw [shift={(55,55)}, rotate = 0.66] [color={rgb, 255:red, 191; green, 97; blue, 106 }  ,draw opacity=1 ][line width=0.75]    (4.37,-1.32) .. controls (2.78,-0.56) and (1.32,-0.12) .. (0,0) .. controls (1.32,0.12) and (2.78,0.56) .. (4.37,1.32)   ;
\draw  [draw opacity=0][fill={rgb, 255:red, 255; green, 255; blue, 255 }  ,fill opacity=1 ] (57.5,80) .. controls (57.5,78.62) and (58.62,77.5) .. (60,77.5) .. controls (61.38,77.5) and (62.5,78.62) .. (62.5,80) .. controls (62.5,81.38) and (61.38,82.5) .. (60,82.5) .. controls (58.62,82.5) and (57.5,81.38) .. (57.5,80) -- cycle ;
\draw    (60,90) -- (60,75) ;
\draw  [draw opacity=0][fill={rgb, 255:red, 255; green, 255; blue, 255 }  ,fill opacity=1 ] (37,78) .. controls (37,76.9) and (38.34,76) .. (40,76) .. controls (41.66,76) and (43,76.9) .. (43,78) .. controls (43,79.1) and (41.66,80) .. (40,80) .. controls (38.34,80) and (37,79.1) .. (37,78) -- cycle ;
\draw [color={rgb, 255:red, 191; green, 97; blue, 106 }  ,draw opacity=1 ]   (52.82,79.94) .. controls (33.45,79.26) and (31.4,74.42) .. (30,70) .. controls (28.54,65.42) and (28.57,65.58) .. (30,60) .. controls (31.43,54.42) and (39.67,55.06) .. (45,55) ;
\draw [shift={(55,80)}, rotate = 181.03] [color={rgb, 255:red, 191; green, 97; blue, 106 }  ,draw opacity=1 ][line width=0.75]    (4.37,-1.32) .. controls (2.78,-0.56) and (1.32,-0.12) .. (0,0) .. controls (1.32,0.12) and (2.78,0.56) .. (4.37,1.32)   ;
\draw [color={rgb, 255:red, 191; green, 97; blue, 106 }  ,draw opacity=1 ]   (67.82,39.94) .. controls (48.45,39.26) and (46.4,34.42) .. (45,30) .. controls (43.54,25.42) and (43.57,25.58) .. (45,20) .. controls (46.43,14.42) and (54.67,15.06) .. (60,15) ;
\draw [shift={(70,40)}, rotate = 181.03] [color={rgb, 255:red, 191; green, 97; blue, 106 }  ,draw opacity=1 ][line width=0.75]    (4.37,-1.32) .. controls (2.78,-0.56) and (1.32,-0.12) .. (0,0) .. controls (1.32,0.12) and (2.78,0.56) .. (4.37,1.32)   ;
\draw    (55,36) .. controls (55.03,50.73) and (50.03,44.73) .. (50,60) ;
\draw  [color={rgb, 255:red, 0; green, 0; blue, 0 }  ,draw opacity=1 ][fill={rgb, 255:red, 191; green, 97; blue, 106 }  ,fill opacity=0.25 ] (50,35) -- (54.5,20) -- (75.5,20) -- (80,35) -- cycle ;
\draw [color={rgb, 255:red, 191; green, 97; blue, 106 }  ,draw opacity=1 ]   (49.17,39.65) .. controls (47.5,39.04) and (46.26,38.95) .. (45,35) .. controls (43.54,30.42) and (43.57,25.58) .. (45,20) .. controls (46.43,14.42) and (54.67,15.06) .. (60,15) ;
\draw [shift={(50,40)}, rotate = 205.96] [color={rgb, 255:red, 191; green, 97; blue, 106 }  ,draw opacity=1 ][line width=0.75]    (4.37,-1.32) .. controls (2.78,-0.56) and (1.32,-0.12) .. (0,0) .. controls (1.32,0.12) and (2.78,0.56) .. (4.37,1.32)   ;
\draw [color={rgb, 255:red, 191; green, 97; blue, 106 }  ,draw opacity=1 ]   (80,45) .. controls (82.43,44.87) and (83.54,39.58) .. (85,35) .. controls (86.46,30.42) and (86.43,25.58) .. (85,20) .. controls (83.73,15.06) and (77.13,15) .. (71.94,15) ;
\draw [shift={(70,15)}, rotate = 0.66] [color={rgb, 255:red, 191; green, 97; blue, 106 }  ,draw opacity=1 ][line width=0.75]    (4.37,-1.32) .. controls (2.78,-0.56) and (1.32,-0.12) .. (0,0) .. controls (1.32,0.12) and (2.78,0.56) .. (4.37,1.32)   ;
\draw [color={rgb, 255:red, 191; green, 97; blue, 106 }  ,draw opacity=1 ]   (55,45) .. controls (75.36,44.69) and (83.54,34.58) .. (85,30) .. controls (86.46,25.42) and (86.43,25.58) .. (85,20) .. controls (83.73,15.06) and (77.13,15) .. (71.94,15) ;
\draw [shift={(70,15)}, rotate = 0.66] [color={rgb, 255:red, 191; green, 97; blue, 106 }  ,draw opacity=1 ][line width=0.75]    (4.37,-1.32) .. controls (2.78,-0.56) and (1.32,-0.12) .. (0,0) .. controls (1.32,0.12) and (2.78,0.56) .. (4.37,1.32)   ;
\draw    (65,20) -- (65,10) ;
\draw  [draw opacity=0][fill={rgb, 255:red, 255; green, 255; blue, 255 }  ,fill opacity=1 ] (52,38) .. controls (52,36.9) and (53.34,36) .. (55,36) .. controls (56.66,36) and (58,36.9) .. (58,38) .. controls (58,39.1) and (56.66,40) .. (55,40) .. controls (53.34,40) and (52,39.1) .. (52,38) -- cycle ;
\draw [color={rgb, 255:red, 191; green, 97; blue, 106 }  ,draw opacity=1 ]   (67.82,39.94) .. controls (48.45,39.26) and (46.4,34.42) .. (45,30) .. controls (43.54,25.42) and (43.57,25.58) .. (45,20) .. controls (46.43,14.42) and (54.67,15.06) .. (60,15) ;
\draw [shift={(70,40)}, rotate = 181.03] [color={rgb, 255:red, 191; green, 97; blue, 106 }  ,draw opacity=1 ][line width=0.75]    (4.37,-1.32) .. controls (2.78,-0.56) and (1.32,-0.12) .. (0,0) .. controls (1.32,0.12) and (2.78,0.56) .. (4.37,1.32)   ;
\draw    (75,35) .. controls (75.03,49.73) and (80.03,44.73) .. (80,60) ;
\draw    (80,90) -- (80,60) ;
\draw  [color={rgb, 255:red, 0; green, 0; blue, 0 }  ,draw opacity=1 ][fill={rgb, 255:red, 191; green, 97; blue, 106 }  ,fill opacity=0.25 ] (140,75) -- (144.5,60) -- (165.5,60) -- (170,75) -- cycle ;
\draw    (145,90) -- (145,75) ;
\draw [color={rgb, 255:red, 191; green, 97; blue, 106 }  ,draw opacity=1 ]   (139.17,79.65) .. controls (137.5,79.04) and (136.26,78.95) .. (135,75) .. controls (133.54,70.42) and (133.57,65.58) .. (135,60) .. controls (136.43,54.42) and (144.67,55.06) .. (150,55) ;
\draw [shift={(140,80)}, rotate = 205.96] [color={rgb, 255:red, 191; green, 97; blue, 106 }  ,draw opacity=1 ][line width=0.75]    (4.37,-1.32) .. controls (2.78,-0.56) and (1.32,-0.12) .. (0,0) .. controls (1.32,0.12) and (2.78,0.56) .. (4.37,1.32)   ;
\draw [color={rgb, 255:red, 191; green, 97; blue, 106 }  ,draw opacity=1 ]   (170,85) .. controls (172.43,84.87) and (173.54,79.58) .. (175,75) .. controls (176.46,70.42) and (176.43,65.58) .. (175,60) .. controls (173.73,55.06) and (167.13,55) .. (161.94,55) ;
\draw [shift={(160,55)}, rotate = 0.66] [color={rgb, 255:red, 191; green, 97; blue, 106 }  ,draw opacity=1 ][line width=0.75]    (4.37,-1.32) .. controls (2.78,-0.56) and (1.32,-0.12) .. (0,0) .. controls (1.32,0.12) and (2.78,0.56) .. (4.37,1.32)   ;
\draw [color={rgb, 255:red, 191; green, 97; blue, 106 }  ,draw opacity=1 ]   (150,85) .. controls (161.21,85.37) and (173.54,74.58) .. (175,70) .. controls (176.46,65.42) and (176.43,65.58) .. (175,60) .. controls (173.73,55.06) and (167.13,55) .. (161.94,55) ;
\draw [shift={(160,55)}, rotate = 0.66] [color={rgb, 255:red, 191; green, 97; blue, 106 }  ,draw opacity=1 ][line width=0.75]    (4.37,-1.32) .. controls (2.78,-0.56) and (1.32,-0.12) .. (0,0) .. controls (1.32,0.12) and (2.78,0.56) .. (4.37,1.32)   ;
\draw  [draw opacity=0][fill={rgb, 255:red, 255; green, 255; blue, 255 }  ,fill opacity=1 ] (162.5,80) .. controls (162.5,78.62) and (163.62,77.5) .. (165,77.5) .. controls (166.38,77.5) and (167.5,78.62) .. (167.5,80) .. controls (167.5,81.38) and (166.38,82.5) .. (165,82.5) .. controls (163.62,82.5) and (162.5,81.38) .. (162.5,80) -- cycle ;
\draw    (165,90) -- (165,75) ;
\draw  [draw opacity=0][fill={rgb, 255:red, 255; green, 255; blue, 255 }  ,fill opacity=1 ] (142,78) .. controls (142,76.9) and (143.34,76) .. (145,76) .. controls (146.66,76) and (148,76.9) .. (148,78) .. controls (148,79.1) and (146.66,80) .. (145,80) .. controls (143.34,80) and (142,79.1) .. (142,78) -- cycle ;
\draw [color={rgb, 255:red, 191; green, 97; blue, 106 }  ,draw opacity=1 ]   (157.82,79.94) .. controls (138.45,79.26) and (136.4,74.42) .. (135,70) .. controls (133.54,65.42) and (133.57,65.58) .. (135,60) .. controls (136.43,54.42) and (144.67,55.06) .. (150,55) ;
\draw [shift={(160,80)}, rotate = 181.03] [color={rgb, 255:red, 191; green, 97; blue, 106 }  ,draw opacity=1 ][line width=0.75]    (4.37,-1.32) .. controls (2.78,-0.56) and (1.32,-0.12) .. (0,0) .. controls (1.32,0.12) and (2.78,0.56) .. (4.37,1.32)   ;
\draw [color={rgb, 255:red, 191; green, 97; blue, 106 }  ,draw opacity=1 ]   (142.82,39.94) .. controls (123.45,39.26) and (121.4,34.42) .. (120,30) .. controls (118.54,25.42) and (118.57,25.58) .. (120,20) .. controls (121.43,14.42) and (129.67,15.06) .. (135,15) ;
\draw [shift={(145,40)}, rotate = 181.03] [color={rgb, 255:red, 191; green, 97; blue, 106 }  ,draw opacity=1 ][line width=0.75]    (4.37,-1.32) .. controls (2.78,-0.56) and (1.32,-0.12) .. (0,0) .. controls (1.32,0.12) and (2.78,0.56) .. (4.37,1.32)   ;
\draw    (130,36) .. controls (130.03,50.73) and (125.03,44.73) .. (125,60) ;
\draw  [color={rgb, 255:red, 0; green, 0; blue, 0 }  ,draw opacity=1 ][fill={rgb, 255:red, 191; green, 97; blue, 106 }  ,fill opacity=0.25 ] (125,35) -- (129.5,20) -- (150.5,20) -- (155,35) -- cycle ;
\draw [color={rgb, 255:red, 191; green, 97; blue, 106 }  ,draw opacity=1 ]   (124.17,39.65) .. controls (122.5,39.04) and (121.26,38.95) .. (120,35) .. controls (118.54,30.42) and (118.57,25.58) .. (120,20) .. controls (121.43,14.42) and (129.67,15.06) .. (135,15) ;
\draw [shift={(125,40)}, rotate = 205.96] [color={rgb, 255:red, 191; green, 97; blue, 106 }  ,draw opacity=1 ][line width=0.75]    (4.37,-1.32) .. controls (2.78,-0.56) and (1.32,-0.12) .. (0,0) .. controls (1.32,0.12) and (2.78,0.56) .. (4.37,1.32)   ;
\draw [color={rgb, 255:red, 191; green, 97; blue, 106 }  ,draw opacity=1 ]   (155,45) .. controls (157.43,44.87) and (158.54,39.58) .. (160,35) .. controls (161.46,30.42) and (161.43,25.58) .. (160,20) .. controls (158.73,15.06) and (152.13,15) .. (146.94,15) ;
\draw [shift={(145,15)}, rotate = 0.66] [color={rgb, 255:red, 191; green, 97; blue, 106 }  ,draw opacity=1 ][line width=0.75]    (4.37,-1.32) .. controls (2.78,-0.56) and (1.32,-0.12) .. (0,0) .. controls (1.32,0.12) and (2.78,0.56) .. (4.37,1.32)   ;
\draw [color={rgb, 255:red, 191; green, 97; blue, 106 }  ,draw opacity=1 ]   (130,45) .. controls (150.36,44.69) and (158.54,34.58) .. (160,30) .. controls (161.46,25.42) and (161.43,25.58) .. (160,20) .. controls (158.73,15.06) and (152.13,15) .. (146.94,15) ;
\draw [shift={(145,15)}, rotate = 0.66] [color={rgb, 255:red, 191; green, 97; blue, 106 }  ,draw opacity=1 ][line width=0.75]    (4.37,-1.32) .. controls (2.78,-0.56) and (1.32,-0.12) .. (0,0) .. controls (1.32,0.12) and (2.78,0.56) .. (4.37,1.32)   ;
\draw    (140,20) -- (140,10) ;
\draw  [draw opacity=0][fill={rgb, 255:red, 255; green, 255; blue, 255 }  ,fill opacity=1 ] (127,38) .. controls (127,36.9) and (128.34,36) .. (130,36) .. controls (131.66,36) and (133,36.9) .. (133,38) .. controls (133,39.1) and (131.66,40) .. (130,40) .. controls (128.34,40) and (127,39.1) .. (127,38) -- cycle ;
\draw [color={rgb, 255:red, 191; green, 97; blue, 106 }  ,draw opacity=1 ]   (142.82,39.94) .. controls (123.45,39.26) and (121.4,34.42) .. (120,30) .. controls (118.54,25.42) and (118.57,25.58) .. (120,20) .. controls (121.43,14.42) and (129.67,15.06) .. (135,15) ;
\draw [shift={(145,40)}, rotate = 181.03] [color={rgb, 255:red, 191; green, 97; blue, 106 }  ,draw opacity=1 ][line width=0.75]    (4.37,-1.32) .. controls (2.78,-0.56) and (1.32,-0.12) .. (0,0) .. controls (1.32,0.12) and (2.78,0.56) .. (4.37,1.32)   ;
\draw    (150,35) .. controls (150.03,49.73) and (155.03,44.73) .. (155,60) ;
\draw    (125,90) -- (125,60) ;
\draw  [draw opacity=0] (90,40) -- (115,40) -- (115,60) -- (90,60) -- cycle ;

\draw (50,67.5) node  [font=\footnotesize]  {$b$};
\draw (22.5,67.5) node  [font=\tiny,color={rgb, 255:red, 191; green, 97; blue, 106 }  ,opacity=1 ]  {$b_{0}$};
\draw (72.5,54.5) node  [font=\tiny,color={rgb, 255:red, 191; green, 97; blue, 106 }  ,opacity=1 ]  {$b_{1}$};
\draw (65,27.5) node  [font=\footnotesize]  {$a$};
\draw (37.5,27.5) node  [font=\tiny,color={rgb, 255:red, 191; green, 97; blue, 106 }  ,opacity=1 ]  {$a_{0}$};
\draw (92.5,25.5) node  [font=\tiny,color={rgb, 255:red, 191; green, 97; blue, 106 }  ,opacity=1 ]  {$a_{1}$};
\draw (155,67.5) node  [font=\footnotesize]  {$d$};
\draw (130.5,55.5) node  [font=\tiny,color={rgb, 255:red, 191; green, 97; blue, 106 }  ,opacity=1 ]  {$d_{0}$};
\draw (177.5,54.5) node  [font=\tiny,color={rgb, 255:red, 191; green, 97; blue, 106 }  ,opacity=1 ]  {$d_{1}$};
\draw (140,27.5) node  [font=\footnotesize]  {$c$};
\draw (112.5,27.5) node  [font=\tiny,color={rgb, 255:red, 191; green, 97; blue, 106 }  ,opacity=1 ]  {$c_{0}$};
\draw (167.5,25.5) node  [font=\tiny,color={rgb, 255:red, 191; green, 97; blue, 106 }  ,opacity=1 ]  {$c_{1}$};
\draw (102.5,50) node  [font=\footnotesize]  {$=$};

\end{tikzpicture}
       \caption{Equation between contours from $⊗$ associator.}
      \label{fig:monoidalContourAssociativity}
    \end{figure}
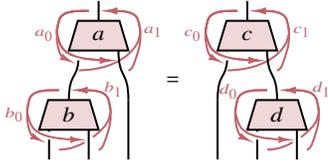

    For each application of unitality, $λ(a ⨾_1 b) = c = \rho(d ⨾_2 e)$, we impose the equations $a_0 ⨾ (b_0 ⊗ \im) = c_0 = d_0 ⨾ (\im ⊗ e_0)$ and $(b_1 ⊗ \im) ⨾ a_1 = c_1 = (\im ⊗ e_1) ⨾ d_1$. These follow from \Cref{fig:monoidalContourUnitality}.

    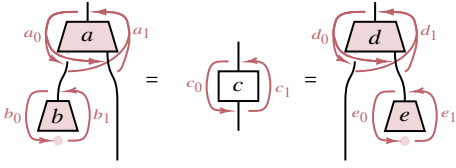
\begin{figure}[ht]
      \centering

\tikzset{every picture/.style={line width=0.75pt}} %

\begin{tikzpicture}[x=0.75pt,y=0.75pt,yscale=-1,xscale=1]
\draw [color={rgb, 255:red, 191; green, 97; blue, 106 }  ,draw opacity=1 ]   (67.82,39.94) .. controls (48.45,39.26) and (46.4,34.42) .. (45,30) .. controls (43.54,25.42) and (43.57,25.58) .. (45,20) .. controls (46.43,14.42) and (54.67,15.06) .. (60,15) ;
\draw [shift={(70,40)}, rotate = 181.03] [color={rgb, 255:red, 191; green, 97; blue, 106 }  ,draw opacity=1 ][line width=0.75]    (4.37,-1.32) .. controls (2.78,-0.56) and (1.32,-0.12) .. (0,0) .. controls (1.32,0.12) and (2.78,0.56) .. (4.37,1.32)   ;
\draw    (55,36) .. controls (55.03,50.73) and (50.03,44.73) .. (50,60) ;
\draw  [color={rgb, 255:red, 0; green, 0; blue, 0 }  ,draw opacity=1 ][fill={rgb, 255:red, 191; green, 97; blue, 106 }  ,fill opacity=0.25 ] (50,35) -- (54.5,20) -- (75.5,20) -- (80,35) -- cycle ;
\draw [color={rgb, 255:red, 191; green, 97; blue, 106 }  ,draw opacity=1 ]   (49.17,39.65) .. controls (47.5,39.04) and (46.26,38.95) .. (45,35) .. controls (43.54,30.42) and (43.57,25.58) .. (45,20) .. controls (46.43,14.42) and (54.67,15.06) .. (60,15) ;
\draw [shift={(50,40)}, rotate = 205.96] [color={rgb, 255:red, 191; green, 97; blue, 106 }  ,draw opacity=1 ][line width=0.75]    (4.37,-1.32) .. controls (2.78,-0.56) and (1.32,-0.12) .. (0,0) .. controls (1.32,0.12) and (2.78,0.56) .. (4.37,1.32)   ;
\draw [color={rgb, 255:red, 191; green, 97; blue, 106 }  ,draw opacity=1 ]   (80,45) .. controls (82.43,44.87) and (83.54,39.58) .. (85,35) .. controls (86.46,30.42) and (86.43,25.58) .. (85,20) .. controls (83.73,15.06) and (77.13,15) .. (71.94,15) ;
\draw [shift={(70,15)}, rotate = 0.66] [color={rgb, 255:red, 191; green, 97; blue, 106 }  ,draw opacity=1 ][line width=0.75]    (4.37,-1.32) .. controls (2.78,-0.56) and (1.32,-0.12) .. (0,0) .. controls (1.32,0.12) and (2.78,0.56) .. (4.37,1.32)   ;
\draw [color={rgb, 255:red, 191; green, 97; blue, 106 }  ,draw opacity=1 ]   (55,45) .. controls (75.36,44.69) and (83.54,34.58) .. (85,30) .. controls (86.46,25.42) and (86.43,25.58) .. (85,20) .. controls (83.73,15.06) and (77.13,15) .. (71.94,15) ;
\draw [shift={(70,15)}, rotate = 0.66] [color={rgb, 255:red, 191; green, 97; blue, 106 }  ,draw opacity=1 ][line width=0.75]    (4.37,-1.32) .. controls (2.78,-0.56) and (1.32,-0.12) .. (0,0) .. controls (1.32,0.12) and (2.78,0.56) .. (4.37,1.32)   ;
\draw    (65,20) -- (65,10) ;
\draw  [draw opacity=0][fill={rgb, 255:red, 255; green, 255; blue, 255 }  ,fill opacity=1 ] (52,38) .. controls (52,36.9) and (53.34,36) .. (55,36) .. controls (56.66,36) and (58,36.9) .. (58,38) .. controls (58,39.1) and (56.66,40) .. (55,40) .. controls (53.34,40) and (52,39.1) .. (52,38) -- cycle ;
\draw [color={rgb, 255:red, 191; green, 97; blue, 106 }  ,draw opacity=1 ]   (67.82,39.94) .. controls (48.45,39.26) and (46.4,34.42) .. (45,30) .. controls (43.54,25.42) and (43.57,25.58) .. (45,20) .. controls (46.43,14.42) and (54.67,15.06) .. (60,15) ;
\draw [shift={(70,40)}, rotate = 181.03] [color={rgb, 255:red, 191; green, 97; blue, 106 }  ,draw opacity=1 ][line width=0.75]    (4.37,-1.32) .. controls (2.78,-0.56) and (1.32,-0.12) .. (0,0) .. controls (1.32,0.12) and (2.78,0.56) .. (4.37,1.32)   ;
\draw    (75,35) .. controls (75.03,49.73) and (80.03,44.73) .. (80,60) ;
\draw    (80,90) -- (80,60) ;
\draw [color={rgb, 255:red, 191; green, 97; blue, 106 }  ,draw opacity=1 ]   (212.82,39.94) .. controls (193.45,39.26) and (191.4,34.42) .. (190,30) .. controls (188.54,25.42) and (188.57,25.58) .. (190,20) .. controls (191.43,14.42) and (199.67,15.06) .. (205,15) ;
\draw [shift={(215,40)}, rotate = 181.03] [color={rgb, 255:red, 191; green, 97; blue, 106 }  ,draw opacity=1 ][line width=0.75]    (4.37,-1.32) .. controls (2.78,-0.56) and (1.32,-0.12) .. (0,0) .. controls (1.32,0.12) and (2.78,0.56) .. (4.37,1.32)   ;
\draw    (200,36) .. controls (200.03,50.73) and (195.03,44.73) .. (195,60) ;
\draw  [color={rgb, 255:red, 0; green, 0; blue, 0 }  ,draw opacity=1 ][fill={rgb, 255:red, 191; green, 97; blue, 106 }  ,fill opacity=0.25 ] (195,35) -- (199.5,20) -- (220.5,20) -- (225,35) -- cycle ;
\draw [color={rgb, 255:red, 191; green, 97; blue, 106 }  ,draw opacity=1 ]   (194.17,39.65) .. controls (192.5,39.04) and (191.26,38.95) .. (190,35) .. controls (188.54,30.42) and (188.57,25.58) .. (190,20) .. controls (191.43,14.42) and (199.67,15.06) .. (205,15) ;
\draw [shift={(195,40)}, rotate = 205.96] [color={rgb, 255:red, 191; green, 97; blue, 106 }  ,draw opacity=1 ][line width=0.75]    (4.37,-1.32) .. controls (2.78,-0.56) and (1.32,-0.12) .. (0,0) .. controls (1.32,0.12) and (2.78,0.56) .. (4.37,1.32)   ;
\draw [color={rgb, 255:red, 191; green, 97; blue, 106 }  ,draw opacity=1 ]   (225,45) .. controls (227.43,44.87) and (228.54,39.58) .. (230,35) .. controls (231.46,30.42) and (231.43,25.58) .. (230,20) .. controls (228.73,15.06) and (222.13,15) .. (216.94,15) ;
\draw [shift={(215,15)}, rotate = 0.66] [color={rgb, 255:red, 191; green, 97; blue, 106 }  ,draw opacity=1 ][line width=0.75]    (4.37,-1.32) .. controls (2.78,-0.56) and (1.32,-0.12) .. (0,0) .. controls (1.32,0.12) and (2.78,0.56) .. (4.37,1.32)   ;
\draw [color={rgb, 255:red, 191; green, 97; blue, 106 }  ,draw opacity=1 ]   (200,45) .. controls (220.36,44.69) and (228.54,34.58) .. (230,30) .. controls (231.46,25.42) and (231.43,25.58) .. (230,20) .. controls (228.73,15.06) and (222.13,15) .. (216.94,15) ;
\draw [shift={(215,15)}, rotate = 0.66] [color={rgb, 255:red, 191; green, 97; blue, 106 }  ,draw opacity=1 ][line width=0.75]    (4.37,-1.32) .. controls (2.78,-0.56) and (1.32,-0.12) .. (0,0) .. controls (1.32,0.12) and (2.78,0.56) .. (4.37,1.32)   ;
\draw    (210,20) -- (210,10) ;
\draw  [draw opacity=0][fill={rgb, 255:red, 255; green, 255; blue, 255 }  ,fill opacity=1 ] (197,38) .. controls (197,36.9) and (198.34,36) .. (200,36) .. controls (201.66,36) and (203,36.9) .. (203,38) .. controls (203,39.1) and (201.66,40) .. (200,40) .. controls (198.34,40) and (197,39.1) .. (197,38) -- cycle ;
\draw [color={rgb, 255:red, 191; green, 97; blue, 106 }  ,draw opacity=1 ]   (212.82,39.94) .. controls (193.45,39.26) and (191.4,34.42) .. (190,30) .. controls (188.54,25.42) and (188.57,25.58) .. (190,20) .. controls (191.43,14.42) and (199.67,15.06) .. (205,15) ;
\draw [shift={(215,40)}, rotate = 181.03] [color={rgb, 255:red, 191; green, 97; blue, 106 }  ,draw opacity=1 ][line width=0.75]    (4.37,-1.32) .. controls (2.78,-0.56) and (1.32,-0.12) .. (0,0) .. controls (1.32,0.12) and (2.78,0.56) .. (4.37,1.32)   ;
\draw    (220,35) .. controls (220.03,49.73) and (225.03,44.73) .. (225,60) ;
\draw    (195,90) -- (195,60) ;
\draw  [draw opacity=0] (165,40) -- (190,40) -- (190,60) -- (165,60) -- cycle ;
\draw  [fill={rgb, 255:red, 191; green, 97; blue, 106 }  ,fill opacity=0.25 ] (40,75) -- (44.5,60) -- (55.5,60) -- (60,75) -- cycle ;
\draw  [draw opacity=0] (40,75) -- (60,75) -- (60,95) -- (40,95) -- cycle ;
\draw [color={rgb, 255:red, 191; green, 97; blue, 106 }  ,draw opacity=1 ]   (42.96,79.93) .. controls (37.06,79.68) and (36.32,79.14) .. (35,75) .. controls (33.54,70.42) and (33.57,65.58) .. (35,60) .. controls (36.43,54.42) and (39.67,55.06) .. (45,55) ;
\draw [shift={(45,80)}, rotate = 181.8] [color={rgb, 255:red, 191; green, 97; blue, 106 }  ,draw opacity=1 ][line width=0.75]    (4.37,-1.32) .. controls (2.78,-0.56) and (1.32,-0.12) .. (0,0) .. controls (1.32,0.12) and (2.78,0.56) .. (4.37,1.32)   ;
\draw [color={rgb, 255:red, 191; green, 97; blue, 106 }  ,draw opacity=1 ]   (57.04,54.86) .. controls (61.32,54.74) and (63.73,56.02) .. (65,60) .. controls (66.46,64.58) and (66.43,69.42) .. (65,75) .. controls (63.57,80.58) and (60.33,79.94) .. (55,80) ;
\draw [shift={(55,55)}, rotate = 354.07] [color={rgb, 255:red, 191; green, 97; blue, 106 }  ,draw opacity=1 ][line width=0.75]    (4.37,-1.32) .. controls (2.78,-0.56) and (1.32,-0.12) .. (0,0) .. controls (1.32,0.12) and (2.78,0.56) .. (4.37,1.32)   ;
\draw  [draw opacity=0][fill={rgb, 255:red, 191; green, 97; blue, 106 }  ,fill opacity=0.25 ] (47.5,80) .. controls (47.5,78.62) and (48.62,77.5) .. (50,77.5) .. controls (51.38,77.5) and (52.5,78.62) .. (52.5,80) .. controls (52.5,81.38) and (51.38,82.5) .. (50,82.5) .. controls (48.62,82.5) and (47.5,81.38) .. (47.5,80) -- cycle ;
\draw  [fill={rgb, 255:red, 191; green, 97; blue, 106 }  ,fill opacity=0.25 ] (215,75.15) -- (219.5,60.15) -- (230.5,60.15) -- (235,75.15) -- cycle ;
\draw  [draw opacity=0] (215,75.15) -- (235,75.15) -- (235,95.15) -- (215,95.15) -- cycle ;
\draw [color={rgb, 255:red, 191; green, 97; blue, 106 }  ,draw opacity=1 ]   (217.96,80.08) .. controls (212.06,79.83) and (211.32,79.29) .. (210,75.15) .. controls (208.54,70.57) and (208.57,65.73) .. (210,60.15) .. controls (211.43,54.57) and (214.67,55.21) .. (220,55.15) ;
\draw [shift={(220,80.15)}, rotate = 181.8] [color={rgb, 255:red, 191; green, 97; blue, 106 }  ,draw opacity=1 ][line width=0.75]    (4.37,-1.32) .. controls (2.78,-0.56) and (1.32,-0.12) .. (0,0) .. controls (1.32,0.12) and (2.78,0.56) .. (4.37,1.32)   ;
\draw [color={rgb, 255:red, 191; green, 97; blue, 106 }  ,draw opacity=1 ]   (232.04,55.01) .. controls (236.32,54.89) and (238.73,56.17) .. (240,60.15) .. controls (241.46,64.73) and (241.43,69.57) .. (240,75.15) .. controls (238.57,80.73) and (235.33,80.09) .. (230,80.15) ;
\draw [shift={(230,55.15)}, rotate = 354.07] [color={rgb, 255:red, 191; green, 97; blue, 106 }  ,draw opacity=1 ][line width=0.75]    (4.37,-1.32) .. controls (2.78,-0.56) and (1.32,-0.12) .. (0,0) .. controls (1.32,0.12) and (2.78,0.56) .. (4.37,1.32)   ;
\draw  [draw opacity=0][fill={rgb, 255:red, 191; green, 97; blue, 106 }  ,fill opacity=0.25 ] (222.5,80.15) .. controls (222.5,78.77) and (223.62,77.65) .. (225,77.65) .. controls (226.38,77.65) and (227.5,78.77) .. (227.5,80.15) .. controls (227.5,81.53) and (226.38,82.65) .. (225,82.65) .. controls (223.62,82.65) and (222.5,81.53) .. (222.5,80.15) -- cycle ;
\draw    (141.08,45) -- (141.08,30) ;
\draw   (131.08,45) -- (151.08,45) -- (151.08,60) -- (131.08,60) -- cycle ;
\draw    (141.08,75) -- (141.08,60) ;
\draw [color={rgb, 255:red, 191; green, 97; blue, 106 }  ,draw opacity=1 ]   (134.05,64.93) .. controls (128.14,64.68) and (127.4,64.14) .. (126.08,60) .. controls (124.63,55.42) and (124.65,50.58) .. (126.08,45) .. controls (127.52,39.42) and (130.76,40.06) .. (136.08,40) ;
\draw [shift={(136.08,65)}, rotate = 181.8] [color={rgb, 255:red, 191; green, 97; blue, 106 }  ,draw opacity=1 ][line width=0.75]    (4.37,-1.32) .. controls (2.78,-0.56) and (1.32,-0.12) .. (0,0) .. controls (1.32,0.12) and (2.78,0.56) .. (4.37,1.32)   ;
\draw [color={rgb, 255:red, 191; green, 97; blue, 106 }  ,draw opacity=1 ]   (148.13,39.86) .. controls (152.4,39.74) and (154.82,41.02) .. (156.08,45) .. controls (157.54,49.58) and (157.52,54.42) .. (156.08,60) .. controls (154.65,65.58) and (151.41,64.94) .. (146.08,65) ;
\draw [shift={(146.08,40)}, rotate = 354.07] [color={rgb, 255:red, 191; green, 97; blue, 106 }  ,draw opacity=1 ][line width=0.75]    (4.37,-1.32) .. controls (2.78,-0.56) and (1.32,-0.12) .. (0,0) .. controls (1.32,0.12) and (2.78,0.56) .. (4.37,1.32)   ;
\draw  [draw opacity=0] (85,40) -- (110,40) -- (110,60) -- (85,60) -- cycle ;

\draw (65,27.5) node  [font=\footnotesize]  {$a$};
\draw (37.5,27.5) node  [font=\tiny,color={rgb, 255:red, 191; green, 97; blue, 106 }  ,opacity=1 ]  {$a_{0}$};
\draw (92.5,25.5) node  [font=\tiny,color={rgb, 255:red, 191; green, 97; blue, 106 }  ,opacity=1 ]  {$a_{1}$};
\draw (210,27.5) node  [font=\footnotesize]  {$d$};
\draw (182.5,27.5) node  [font=\tiny,color={rgb, 255:red, 191; green, 97; blue, 106 }  ,opacity=1 ]  {$d_{0}$};
\draw (237.5,25.5) node  [font=\tiny,color={rgb, 255:red, 191; green, 97; blue, 106 }  ,opacity=1 ]  {$d_{1}$};
\draw (177.5,50) node  [font=\footnotesize]  {$=$};
\draw (50,67.5) node  [font=\footnotesize]  {$b$};
\draw (72.5,67.5) node  [font=\tiny,color={rgb, 255:red, 191; green, 97; blue, 106 }  ,opacity=1 ]  {$b_{1}$};
\draw (27.5,67.5) node  [font=\tiny,color={rgb, 255:red, 191; green, 97; blue, 106 }  ,opacity=1 ]  {$b_{0}$};
\draw (225,67.65) node  [font=\footnotesize]  {$e$};
\draw (247.5,67.65) node  [font=\tiny,color={rgb, 255:red, 191; green, 97; blue, 106 }  ,opacity=1 ]  {$e_{1}$};
\draw (202.5,67.65) node  [font=\tiny,color={rgb, 255:red, 191; green, 97; blue, 106 }  ,opacity=1 ]  {$e_{0}$};
\draw (141.08,52.5) node  [font=\footnotesize]  {$c$};
\draw (119,52.5) node  [font=\tiny,color={rgb, 255:red, 191; green, 97; blue, 106 }  ,opacity=1 ]  {$c_{0}$};
\draw (164,55) node  [font=\tiny,color={rgb, 255:red, 191; green, 97; blue, 106 }  ,opacity=1 ]  {$c_{1}$};
\draw (97.5,50) node  [font=\footnotesize]  {$=$};

\end{tikzpicture}
       \caption{Equations from $⊗$ unitor.}
      \label{fig:monoidalContourUnitality}
    \end{figure}

    For each application of the laxator, $ψ_2(a \mathbin{|} b \mathbin{|} c) = (d \mathbin{|} e \mathbin{|} f)$, we impose the equation $a_0 ⨾ (b_0 ⊗ c_0) = d_0 ⨾ e_0$, the middle equation $b_1 ⊗ c_1 = e_1 ⨾ d_1 ⨾ f_0$, and $(b_2 ⊗ c_2) ⨾ a_1 = f_1 ⨾ d_2$. These follow \Cref{fig:monoidal-contour-equation1}. We finally impose similar equations for the rest of the laxators, see \Cref{ax:def:monoidalContour} for details.
    \begin{figure}[ht!]
      \centering

\tikzset{every picture/.style={line width=0.75pt}} %

\begin{tikzpicture}[x=0.75pt,y=0.75pt,yscale=-1,xscale=1]
\draw [color={rgb, 255:red, 191; green, 97; blue, 106 }  ,draw opacity=1 ]   (187.82,39.94) .. controls (168.45,39.26) and (166.4,34.42) .. (165,30) .. controls (163.54,25.42) and (163.57,25.58) .. (165,20) .. controls (166.43,14.42) and (174.67,15.06) .. (180,15) ;
\draw [shift={(190,40)}, rotate = 181.03] [color={rgb, 255:red, 191; green, 97; blue, 106 }  ,draw opacity=1 ][line width=0.75]    (4.37,-1.32) .. controls (2.78,-0.56) and (1.32,-0.12) .. (0,0) .. controls (1.32,0.12) and (2.78,0.56) .. (4.37,1.32)   ;
\draw    (175,35) .. controls (175.03,49.73) and (160.03,44.73) .. (160,60) ;
\draw  [color={rgb, 255:red, 0; green, 0; blue, 0 }  ,draw opacity=1 ][fill={rgb, 255:red, 191; green, 97; blue, 106 }  ,fill opacity=0.25 ] (170,35) -- (174.5,20) -- (195.5,20) -- (200,35) -- cycle ;
\draw [color={rgb, 255:red, 191; green, 97; blue, 106 }  ,draw opacity=1 ]   (169.17,39.65) .. controls (167.5,39.04) and (166.26,38.95) .. (165,35) .. controls (163.54,30.42) and (163.57,25.58) .. (165,20) .. controls (166.43,14.42) and (174.67,15.06) .. (180,15) ;
\draw [shift={(170,40)}, rotate = 205.96] [color={rgb, 255:red, 191; green, 97; blue, 106 }  ,draw opacity=1 ][line width=0.75]    (4.37,-1.32) .. controls (2.78,-0.56) and (1.32,-0.12) .. (0,0) .. controls (1.32,0.12) and (2.78,0.56) .. (4.37,1.32)   ;
\draw [color={rgb, 255:red, 191; green, 97; blue, 106 }  ,draw opacity=1 ]   (200,45) .. controls (202.43,44.87) and (203.54,39.58) .. (205,35) .. controls (206.46,30.42) and (206.43,25.58) .. (205,20) .. controls (203.73,15.06) and (197.13,15) .. (191.94,15) ;
\draw [shift={(190,15)}, rotate = 0.66] [color={rgb, 255:red, 191; green, 97; blue, 106 }  ,draw opacity=1 ][line width=0.75]    (4.37,-1.32) .. controls (2.78,-0.56) and (1.32,-0.12) .. (0,0) .. controls (1.32,0.12) and (2.78,0.56) .. (4.37,1.32)   ;
\draw [color={rgb, 255:red, 191; green, 97; blue, 106 }  ,draw opacity=1 ]   (175,45) .. controls (195.36,44.69) and (203.54,34.58) .. (205,30) .. controls (206.46,25.42) and (206.43,25.58) .. (205,20) .. controls (203.73,15.06) and (197.13,15) .. (191.94,15) ;
\draw [shift={(190,15)}, rotate = 0.66] [color={rgb, 255:red, 191; green, 97; blue, 106 }  ,draw opacity=1 ][line width=0.75]    (4.37,-1.32) .. controls (2.78,-0.56) and (1.32,-0.12) .. (0,0) .. controls (1.32,0.12) and (2.78,0.56) .. (4.37,1.32)   ;
\draw  [draw opacity=0][fill={rgb, 255:red, 255; green, 255; blue, 255 }  ,fill opacity=1 ] (192.5,40) .. controls (192.5,38.62) and (193.62,37.5) .. (195,37.5) .. controls (196.38,37.5) and (197.5,38.62) .. (197.5,40) .. controls (197.5,41.38) and (196.38,42.5) .. (195,42.5) .. controls (193.62,42.5) and (192.5,41.38) .. (192.5,40) -- cycle ;
\draw    (185,20) -- (185,10) ;
\draw   (195,75) -- (199.5,60) -- (220.5,60) -- (225,75) -- cycle ;
\draw [color={rgb, 255:red, 191; green, 97; blue, 106 }  ,draw opacity=1 ]   (193.13,79.28) .. controls (191.96,78.82) and (190.98,78.09) .. (190,75) .. controls (188.54,70.42) and (188.57,65.58) .. (190,60) .. controls (191.43,54.42) and (199.67,55.06) .. (205,55) ;
\draw [shift={(195,80)}, rotate = 205.96] [color={rgb, 255:red, 191; green, 97; blue, 106 }  ,draw opacity=1 ][line width=0.75]    (4.37,-1.32) .. controls (2.78,-0.56) and (1.32,-0.12) .. (0,0) .. controls (1.32,0.12) and (2.78,0.56) .. (4.37,1.32)   ;
\draw [color={rgb, 255:red, 191; green, 97; blue, 106 }  ,draw opacity=1 ]   (205,80) -- (213,80) ;
\draw [shift={(215,80)}, rotate = 180] [color={rgb, 255:red, 191; green, 97; blue, 106 }  ,draw opacity=1 ][line width=0.75]    (4.37,-1.32) .. controls (2.78,-0.56) and (1.32,-0.12) .. (0,0) .. controls (1.32,0.12) and (2.78,0.56) .. (4.37,1.32)   ;
\draw [color={rgb, 255:red, 191; green, 97; blue, 106 }  ,draw opacity=1 ]   (225,80) .. controls (227.43,79.87) and (228.54,79.58) .. (230,75) .. controls (231.46,70.42) and (231.43,65.58) .. (230,60) .. controls (228.73,55.06) and (222.13,55) .. (216.94,55) ;
\draw [shift={(215,55)}, rotate = 0.66] [color={rgb, 255:red, 191; green, 97; blue, 106 }  ,draw opacity=1 ][line width=0.75]    (4.37,-1.32) .. controls (2.78,-0.56) and (1.32,-0.12) .. (0,0) .. controls (1.32,0.12) and (2.78,0.56) .. (4.37,1.32)   ;
\draw    (220,90) -- (220,75) ;
\draw   (145,75) -- (149.5,60) -- (170.5,60) -- (175,75) -- cycle ;
\draw [color={rgb, 255:red, 191; green, 97; blue, 106 }  ,draw opacity=1 ]   (143.13,79.28) .. controls (141.96,78.82) and (140.98,78.09) .. (140,75) .. controls (138.54,70.42) and (138.57,65.58) .. (140,60) .. controls (141.43,54.42) and (149.67,55.06) .. (155,55) ;
\draw [shift={(145,80)}, rotate = 205.96] [color={rgb, 255:red, 191; green, 97; blue, 106 }  ,draw opacity=1 ][line width=0.75]    (4.37,-1.32) .. controls (2.78,-0.56) and (1.32,-0.12) .. (0,0) .. controls (1.32,0.12) and (2.78,0.56) .. (4.37,1.32)   ;
\draw [color={rgb, 255:red, 191; green, 97; blue, 106 }  ,draw opacity=1 ]   (155,80) -- (163,80) ;
\draw [shift={(165,80)}, rotate = 180] [color={rgb, 255:red, 191; green, 97; blue, 106 }  ,draw opacity=1 ][line width=0.75]    (4.37,-1.32) .. controls (2.78,-0.56) and (1.32,-0.12) .. (0,0) .. controls (1.32,0.12) and (2.78,0.56) .. (4.37,1.32)   ;
\draw [color={rgb, 255:red, 191; green, 97; blue, 106 }  ,draw opacity=1 ]   (180,75) .. controls (181.46,70.42) and (181.43,65.58) .. (180,60) .. controls (178.73,55.06) and (172.13,55) .. (166.94,55) ;
\draw [shift={(165,55)}, rotate = 0.66] [color={rgb, 255:red, 191; green, 97; blue, 106 }  ,draw opacity=1 ][line width=0.75]    (4.37,-1.32) .. controls (2.78,-0.56) and (1.32,-0.12) .. (0,0) .. controls (1.32,0.12) and (2.78,0.56) .. (4.37,1.32)   ;
\draw    (150,90) -- (150,75) ;
\draw  [draw opacity=0][fill={rgb, 255:red, 255; green, 255; blue, 255 }  ,fill opacity=1 ] (172,38) .. controls (172,36.9) and (173.34,36) .. (175,36) .. controls (176.66,36) and (178,36.9) .. (178,38) .. controls (178,39.1) and (176.66,40) .. (175,40) .. controls (173.34,40) and (172,39.1) .. (172,38) -- cycle ;
\draw  [draw opacity=0][fill={rgb, 255:red, 255; green, 255; blue, 255 }  ,fill opacity=1 ] (197,45) .. controls (197,43.9) and (198.34,43) .. (200,43) .. controls (201.66,43) and (203,43.9) .. (203,45) .. controls (203,46.1) and (201.66,47) .. (200,47) .. controls (198.34,47) and (197,46.1) .. (197,45) -- cycle ;
\draw [color={rgb, 255:red, 191; green, 97; blue, 106 }  ,draw opacity=1 ]   (187.82,39.94) .. controls (168.45,39.26) and (166.4,34.42) .. (165,30) .. controls (163.54,25.42) and (163.57,25.58) .. (165,20) .. controls (166.43,14.42) and (174.67,15.06) .. (180,15) ;
\draw [shift={(190,40)}, rotate = 181.03] [color={rgb, 255:red, 191; green, 97; blue, 106 }  ,draw opacity=1 ][line width=0.75]    (4.37,-1.32) .. controls (2.78,-0.56) and (1.32,-0.12) .. (0,0) .. controls (1.32,0.12) and (2.78,0.56) .. (4.37,1.32)   ;
\draw    (200,75) .. controls (199.96,89.32) and (170.14,79.77) .. (170,90) ;
\draw    (195,35) .. controls (195.03,49.73) and (210.03,44.73) .. (210,60) ;
\draw  [draw opacity=0][fill={rgb, 255:red, 255; green, 255; blue, 255 }  ,fill opacity=1 ] (178,84) .. controls (178,81.79) and (181.13,80) .. (185,80) .. controls (188.87,80) and (192,81.79) .. (192,84) .. controls (192,86.21) and (188.87,88) .. (185,88) .. controls (181.13,88) and (178,86.21) .. (178,84) -- cycle ;
\draw    (170,75) .. controls (170.14,89.16) and (200.14,80.08) .. (200,90) ;
\draw [color={rgb, 255:red, 0; green, 0; blue, 0 }  ,draw opacity=1 ][line width=1.5]    (230,45) -- (252,45) ;
\draw [shift={(255,45)}, rotate = 180] [color={rgb, 255:red, 0; green, 0; blue, 0 }  ,draw opacity=1 ][line width=1.5]    (8.53,-2.57) .. controls (5.42,-1.09) and (2.58,-0.23) .. (0,0) .. controls (2.58,0.23) and (5.42,1.09) .. (8.53,2.57)   ;
\draw  [color={rgb, 255:red, 0; green, 0; blue, 0 }  ,draw opacity=1 ][fill={rgb, 255:red, 191; green, 97; blue, 106 }  ,fill opacity=0.25 ] (260,75) -- (264.5,60) -- (285.5,60) -- (290,75) -- cycle ;
\draw    (265,90) -- (265,75) ;
\draw [color={rgb, 255:red, 191; green, 97; blue, 106 }  ,draw opacity=1 ]   (259.17,79.65) .. controls (257.5,79.04) and (256.26,78.95) .. (255,75) .. controls (253.54,70.42) and (253.57,65.58) .. (255,60) .. controls (256.43,54.42) and (264.67,55.06) .. (270,55) ;
\draw [shift={(260,80)}, rotate = 205.96] [color={rgb, 255:red, 191; green, 97; blue, 106 }  ,draw opacity=1 ][line width=0.75]    (4.37,-1.32) .. controls (2.78,-0.56) and (1.32,-0.12) .. (0,0) .. controls (1.32,0.12) and (2.78,0.56) .. (4.37,1.32)   ;
\draw [color={rgb, 255:red, 191; green, 97; blue, 106 }  ,draw opacity=1 ]   (290,85) .. controls (292.43,84.87) and (293.54,79.58) .. (295,75) .. controls (296.46,70.42) and (296.43,65.58) .. (295,60) .. controls (293.73,55.06) and (287.13,55) .. (281.94,55) ;
\draw [shift={(280,55)}, rotate = 0.66] [color={rgb, 255:red, 191; green, 97; blue, 106 }  ,draw opacity=1 ][line width=0.75]    (4.37,-1.32) .. controls (2.78,-0.56) and (1.32,-0.12) .. (0,0) .. controls (1.32,0.12) and (2.78,0.56) .. (4.37,1.32)   ;
\draw [color={rgb, 255:red, 191; green, 97; blue, 106 }  ,draw opacity=1 ]   (270,85) .. controls (281.21,85.37) and (293.54,74.58) .. (295,70) .. controls (296.46,65.42) and (296.43,65.58) .. (295,60) .. controls (293.73,55.06) and (287.13,55) .. (281.94,55) ;
\draw [shift={(280,55)}, rotate = 0.66] [color={rgb, 255:red, 191; green, 97; blue, 106 }  ,draw opacity=1 ][line width=0.75]    (4.37,-1.32) .. controls (2.78,-0.56) and (1.32,-0.12) .. (0,0) .. controls (1.32,0.12) and (2.78,0.56) .. (4.37,1.32)   ;
\draw  [draw opacity=0][fill={rgb, 255:red, 255; green, 255; blue, 255 }  ,fill opacity=1 ] (282.5,80) .. controls (282.5,78.62) and (283.62,77.5) .. (285,77.5) .. controls (286.38,77.5) and (287.5,78.62) .. (287.5,80) .. controls (287.5,81.38) and (286.38,82.5) .. (285,82.5) .. controls (283.62,82.5) and (282.5,81.38) .. (282.5,80) -- cycle ;
\draw    (285,90) -- (285,75) ;
\draw   (285,35) -- (289.5,20) -- (310.5,20) -- (315,35) -- cycle ;
\draw    (300,20) -- (300,10) ;
\draw [color={rgb, 255:red, 191; green, 97; blue, 106 }  ,draw opacity=1 ]   (283.13,39.28) .. controls (281.96,38.82) and (280.98,38.09) .. (280,35) .. controls (278.54,30.42) and (278.57,25.58) .. (280,20) .. controls (281.43,14.42) and (289.67,15.06) .. (295,15) ;
\draw [shift={(285,40)}, rotate = 205.96] [color={rgb, 255:red, 191; green, 97; blue, 106 }  ,draw opacity=1 ][line width=0.75]    (4.37,-1.32) .. controls (2.78,-0.56) and (1.32,-0.12) .. (0,0) .. controls (1.32,0.12) and (2.78,0.56) .. (4.37,1.32)   ;
\draw [color={rgb, 255:red, 191; green, 97; blue, 106 }  ,draw opacity=1 ]   (295,40) -- (303,40) ;
\draw [shift={(305,40)}, rotate = 180] [color={rgb, 255:red, 191; green, 97; blue, 106 }  ,draw opacity=1 ][line width=0.75]    (4.37,-1.32) .. controls (2.78,-0.56) and (1.32,-0.12) .. (0,0) .. controls (1.32,0.12) and (2.78,0.56) .. (4.37,1.32)   ;
\draw [color={rgb, 255:red, 191; green, 97; blue, 106 }  ,draw opacity=1 ]   (315,40) .. controls (317.43,39.87) and (318.54,39.58) .. (320,35) .. controls (321.46,30.42) and (321.43,25.58) .. (320,20) .. controls (318.73,15.06) and (312.13,15) .. (306.94,15) ;
\draw [shift={(305,15)}, rotate = 0.66] [color={rgb, 255:red, 191; green, 97; blue, 106 }  ,draw opacity=1 ][line width=0.75]    (4.37,-1.32) .. controls (2.78,-0.56) and (1.32,-0.12) .. (0,0) .. controls (1.32,0.12) and (2.78,0.56) .. (4.37,1.32)   ;
\draw    (290,35) .. controls (290.03,49.73) and (275.03,44.73) .. (275,60) ;
\draw  [color={rgb, 255:red, 0; green, 0; blue, 0 }  ,draw opacity=1 ][fill={rgb, 255:red, 191; green, 97; blue, 106 }  ,fill opacity=0.25 ] (310,75) -- (314.5,60) -- (335.5,60) -- (340,75) -- cycle ;
\draw    (315,90) -- (315,75) ;
\draw [color={rgb, 255:red, 191; green, 97; blue, 106 }  ,draw opacity=1 ]   (309.17,79.65) .. controls (307.5,79.04) and (306.26,78.95) .. (305,75) .. controls (303.54,70.42) and (303.57,65.58) .. (305,60) .. controls (306.43,54.42) and (314.67,55.06) .. (320,55) ;
\draw [shift={(310,80)}, rotate = 205.96] [color={rgb, 255:red, 191; green, 97; blue, 106 }  ,draw opacity=1 ][line width=0.75]    (4.37,-1.32) .. controls (2.78,-0.56) and (1.32,-0.12) .. (0,0) .. controls (1.32,0.12) and (2.78,0.56) .. (4.37,1.32)   ;
\draw [color={rgb, 255:red, 191; green, 97; blue, 106 }  ,draw opacity=1 ]   (340,85) .. controls (342.43,84.87) and (343.54,79.58) .. (345,75) .. controls (346.46,70.42) and (346.43,65.58) .. (345,60) .. controls (343.73,55.06) and (337.13,55) .. (331.94,55) ;
\draw [shift={(330,55)}, rotate = 0.66] [color={rgb, 255:red, 191; green, 97; blue, 106 }  ,draw opacity=1 ][line width=0.75]    (4.37,-1.32) .. controls (2.78,-0.56) and (1.32,-0.12) .. (0,0) .. controls (1.32,0.12) and (2.78,0.56) .. (4.37,1.32)   ;
\draw [color={rgb, 255:red, 191; green, 97; blue, 106 }  ,draw opacity=1 ]   (320,85) .. controls (331.21,85.37) and (343.54,74.58) .. (345,70) .. controls (346.46,65.42) and (346.43,65.58) .. (345,60) .. controls (343.73,55.06) and (337.13,55) .. (331.94,55) ;
\draw [shift={(330,55)}, rotate = 0.66] [color={rgb, 255:red, 191; green, 97; blue, 106 }  ,draw opacity=1 ][line width=0.75]    (4.37,-1.32) .. controls (2.78,-0.56) and (1.32,-0.12) .. (0,0) .. controls (1.32,0.12) and (2.78,0.56) .. (4.37,1.32)   ;
\draw  [draw opacity=0][fill={rgb, 255:red, 255; green, 255; blue, 255 }  ,fill opacity=1 ] (332.5,80) .. controls (332.5,78.62) and (333.62,77.5) .. (335,77.5) .. controls (336.38,77.5) and (337.5,78.62) .. (337.5,80) .. controls (337.5,81.38) and (336.38,82.5) .. (335,82.5) .. controls (333.62,82.5) and (332.5,81.38) .. (332.5,80) -- cycle ;
\draw    (310,35) .. controls (310.03,49.73) and (325.03,44.73) .. (325,60) ;
\draw    (335,90) -- (335,75) ;
\draw  [draw opacity=0][fill={rgb, 255:red, 255; green, 255; blue, 255 }  ,fill opacity=1 ] (262,78) .. controls (262,76.9) and (263.34,76) .. (265,76) .. controls (266.66,76) and (268,76.9) .. (268,78) .. controls (268,79.1) and (266.66,80) .. (265,80) .. controls (263.34,80) and (262,79.1) .. (262,78) -- cycle ;
\draw  [draw opacity=0][fill={rgb, 255:red, 255; green, 255; blue, 255 }  ,fill opacity=1 ] (312,78) .. controls (312,76.9) and (313.34,76) .. (315,76) .. controls (316.66,76) and (318,76.9) .. (318,78) .. controls (318,79.1) and (316.66,80) .. (315,80) .. controls (313.34,80) and (312,79.1) .. (312,78) -- cycle ;
\draw [color={rgb, 255:red, 191; green, 97; blue, 106 }  ,draw opacity=1 ]   (277.82,79.94) .. controls (258.45,79.26) and (256.4,74.42) .. (255,70) .. controls (253.54,65.42) and (253.57,65.58) .. (255,60) .. controls (256.43,54.42) and (264.67,55.06) .. (270,55) ;
\draw [shift={(280,80)}, rotate = 181.03] [color={rgb, 255:red, 191; green, 97; blue, 106 }  ,draw opacity=1 ][line width=0.75]    (4.37,-1.32) .. controls (2.78,-0.56) and (1.32,-0.12) .. (0,0) .. controls (1.32,0.12) and (2.78,0.56) .. (4.37,1.32)   ;
\draw [color={rgb, 255:red, 191; green, 97; blue, 106 }  ,draw opacity=1 ]   (327.82,79.94) .. controls (308.45,79.26) and (306.4,74.42) .. (305,70) .. controls (303.54,65.42) and (303.57,65.58) .. (305,60) .. controls (306.43,54.42) and (314.67,55.06) .. (320,55) ;
\draw [shift={(330,80)}, rotate = 181.03] [color={rgb, 255:red, 191; green, 97; blue, 106 }  ,draw opacity=1 ][line width=0.75]    (4.37,-1.32) .. controls (2.78,-0.56) and (1.32,-0.12) .. (0,0) .. controls (1.32,0.12) and (2.78,0.56) .. (4.37,1.32)   ;

\draw (185,27.5) node  [font=\footnotesize]  {$a$};
\draw (157.5,27.5) node  [font=\tiny,color={rgb, 255:red, 191; green, 97; blue, 106 }  ,opacity=1 ]  {$a_{0}$};
\draw (212.5,25.5) node  [font=\tiny,color={rgb, 255:red, 191; green, 97; blue, 106 }  ,opacity=1 ]  {$a_{1}$};
\draw (210,67.5) node  [font=\footnotesize]  {$c$};
\draw (192.5,49.5) node  [font=\tiny,color={rgb, 255:red, 191; green, 97; blue, 106 }  ,opacity=1 ]  {$c_{0}$};
\draw (210,85) node  [font=\tiny,color={rgb, 255:red, 191; green, 97; blue, 106 }  ,opacity=1 ]  {$c_{1}$};
\draw (237.5,67.5) node  [font=\tiny,color={rgb, 255:red, 191; green, 97; blue, 106 }  ,opacity=1 ]  {$c_{2}$};
\draw (160,67.5) node  [font=\footnotesize]  {$b$};
\draw (132.5,67.5) node  [font=\tiny,color={rgb, 255:red, 191; green, 97; blue, 106 }  ,opacity=1 ]  {$b_{0}$};
\draw (160,85) node  [font=\tiny,color={rgb, 255:red, 191; green, 97; blue, 106 }  ,opacity=1 ]  {$b_{1}$};
\draw (177.5,49.5) node  [font=\tiny,color={rgb, 255:red, 191; green, 97; blue, 106 }  ,opacity=1 ]  {$b_{2}$};
\draw (275,67.5) node  [font=\footnotesize]  {$e$};
\draw (247.5,67.5) node  [font=\tiny,color={rgb, 255:red, 191; green, 97; blue, 106 }  ,opacity=1 ]  {$e_{0}$};
\draw (297.5,54.5) node  [font=\tiny,color={rgb, 255:red, 191; green, 97; blue, 106 }  ,opacity=1 ]  {$e_{1}$};
\draw (300,27.5) node  [font=\footnotesize]  {$d$};
\draw (272.5,27.5) node  [font=\tiny,color={rgb, 255:red, 191; green, 97; blue, 106 }  ,opacity=1 ]  {$d_{0}$};
\draw (300,45) node  [font=\tiny,color={rgb, 255:red, 191; green, 97; blue, 106 }  ,opacity=1 ]  {$d_{1}$};
\draw (327.5,27.5) node  [font=\tiny,color={rgb, 255:red, 191; green, 97; blue, 106 }  ,opacity=1 ]  {$d_{2}$};
\draw (325,67.5) node  [font=\footnotesize]  {$f$};
\draw (307.5,85.5) node  [font=\tiny,color={rgb, 255:red, 191; green, 97; blue, 106 }  ,opacity=1 ]  {$f_{0}$};
\draw (352.5,67.5) node  [font=\tiny,color={rgb, 255:red, 191; green, 97; blue, 106 }  ,opacity=1 ]  {$f_{1}$};

\end{tikzpicture}
       \caption{Equations from the laxator $ψ_2$.}
      \label{fig:monoidal-contour-equation1}
    \end{figure}
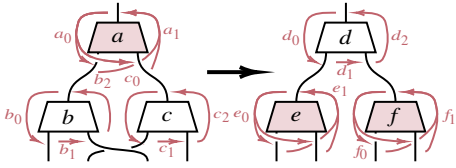
\end{definition}

\begin{proposition} \label{prop:monoidalContourFunctor}
  \MonoidalContour{} extends to a functor $𝓓 : \pDuo → \Mon$.
\end{proposition}
\begin{proof}
  See Appendix, \Cref{ax:prop:monoidalContourFunctor}.
\end{proof}

\subsection{Produoidal Category of Spliced Monoidal Arrows}
\defining{linkProduoidalSplice}{}
Again, we want to go the other way around: given a \monoidalCategory{}, what is the \produoidalCategory{} that tracks decomposition of arrows in that \monoidalCategory{}?
This subsection finds a right adjoint to the \monoidalContour{} construction: the \produoidalCategory{} of \emph{spliced monoidal arrows}.

\begin{definition}
\label{def:monoidalSplice}
Let $(ℂ,⊗,I)$ be a monoidal category.\defining{linkMonoidalSplice}{}
The \produoidalCategory{} of \emph{spliced monoidal arrows}, $\mathcal{T}{ℂ}$, has as objects pairs of objects of $ℂ$. It uses the following \profunctors{} to define \sequentialSplits{}, \parallelSplits{}, \sequentialUnits{}, \parallelUnits{} and morphisms.
\begin{align*}
   𝓣{ℂ} \left(\biobj{A}{B}; \biobj{X}{Y} \right) &= ℂ(A;X) × ℂ(Y,B); \\
  𝓣{ℂ}(\biobj{A}{B};\biobj{X}{Y} ◁ \biobj{X'}{Y'}) &= ℂ(A;X) × ℂ(Y;X') × ℂ(Y';B); \\
  𝓣{ℂ}(\biobj{A}{B}; \biobj{X}{Y} ⊗ \biobj{X'}{Y'}) &= ℂ(A;X ⊗ X') × ℂ(Y ⊗ Y';B); \\
  𝓣{ℂ}(\biobj{A}{B};N) &= ℂ(A;B); \\
  𝓣{ℂ}(\biobj{A}{B};I) &= ℂ(A;I) × ℂ(I;B).
\end{align*}
In other words, morphisms are \emph{pairs of arrows} $f \colon A \to X$ and $g \colon Y \to B$. \SequentialSplits{} are \emph{triples of arrows} $f \colon A \to X$, $g \colon Y \to X'$ and $h \colon Y' \to B$. \ParallelSplits{} are \emph{pairs of arrows} $f \colon A \to X \otimes X'$ and $g \colon Y \otimes Y' \to B$. \SequentialUnits{} are \emph{arrows} $f \colon A \to B$. \ParallelUnits{} are pairs of arrows $f \colon A \to I$ and $g \colon I \to B$.
In summary, we have
\begin{align*}
  & \mbox{morphisms, } && f ⨾ \square ⨾ g
  && \in  𝓣{ℂ} \left(
    \biobj{A}{B};
    \biobj{X}{Y} \right); \\
  & \mbox{\sequentialSplits{}, }
  && f ⨾ \square ⨾ g ⨾ \square ⨾ h
  && \in 𝓣{ℂ} \left( \biobj{A}{B}; \biobj{X}{Y} \triangleleft \biobj{X'}{Y'} \right); \\
  & \mbox{\parallelSplits{}, } && f ⨾ (\square ⊗ \square) ⨾ h
  && \in 𝓣{ℂ} \left( \biobj{A}{B}; \biobj{X}{Y} ⊗ \biobj{X'}{Y'} \right); \\
  & \mbox{\sequentialUnits{}, } && f
  && \in 𝓣{ℂ} \left(\biobj{A}{B}; N \right). \\
  & \mbox{and \parallelUnits{}, } && f \mathbin{\Vert} g
  && \in 𝓣{ℂ} \left(\biobj{A}{B}; I \right).
\end{align*}

Finally, the laxators unite two different connections between two gaps into a single one. For instance, the last laxator takes parallel sequences of holes,
$$f_0 ⨾ ((h_0 ⨾ \square ⨾ h_1 ⨾ \square ⨾ h_2) ⊗ (k_0 ⨾ \square ⨾ k_1 ⨾ \square ⨾ k_2)) ⨾ f_1$$
into sequences of parallel holes,
$$f_0 ⨾ (h_0 ⊗ k_0) ⨾ (\square ⊗  \square) ⨾ (h_1 ⊗ k_1) ⨾ (\square ⊗ \square) ⨾ (h_2 ⊗ k_2) ⨾ f_1.$$
See Appendix, \Cref{ax:sec:produoidalSplice} for details.
\end{definition}

\begin{proposition}
  \label{prop:spliceIsProduoidal}
  \SplicedMonoidalArrows{} form a \produoidalCategory{} with their \hyperlink{linkProduoidalComponents}{sequential and parallel splits, units}, and suitable coherence morphisms and laxators.
\end{proposition}
\begin{proof}
  See Appendix, \Cref{ax:prop:spliceIsProduoidal}.
\end{proof}

\begin{proposition} \label{prop:monoidalSpliceFunctor}
  \SplicedMonoidalArrows{} extends to a functor $𝓣{} : \Mon → \pDuo$. %
\end{proposition}
\begin{proof}
  See Appendix, \Cref{ax:prop:monoidalSpliceFunctor}.
\end{proof}

As in the categorical case, \splicedMonoidalArrows{} and \monoidalContour{} again form an adjunction. This adjunction characterizes spliced monoidal arrows as a cofree construction.

\begin{theorem}
  \label{prop:produoidalSpliceContour}
    There exists an adjunction between \monoidalCategories{} and \produoidalCategories{}, where the \monoidalContour{} is the left adjoint, and the \produoidalSplice{} category is the right adjoint.
\end{theorem}
\begin{proof}
  See Appendix, \Cref{ax:prop:produoidalSpliceContour}.
\end{proof}

\subsection{Representable Parallel Structure}

A \produoidalCategory{} has two tensors, and neither is, in principle, representable. However, the cofree \produoidalCategory{} over a category we have just constructed happens also to have a representable tensor, $(⊗)$. Spliced monoidal arrows form a \monoidalCategory{}.

\begin{proposition}
  \ParallelSplits{} and \parallelUnits{} of spliced monoidal arrows are representable \profunctors{}. Explicitly,
  $$𝓣{ℂ}\left(\biobj{A}{B} ; \biobj{X}{Y} ⊗ \biobj{X'}{Y'}\right) ≅ 𝓣{ℂ}\left(\biobj{A}{B} ; \biobj{X ⊗ X'}{Y ⊗ Y'}\right),
  \mbox{ and }
  𝓣{ℂ}\left(\biobj{A}{B} ; I\right) ≅ 𝓣{ℂ}\left(\biobj{A}{B} ; \biobj{I}{I} \right).$$
\end{proposition}

In fact, these sets are equal by definition. However, there is a good reason to work in the full generality of \produoidalCategories{}: every produoidal category, representable or not, has an associated \emph{normal} \produoidalCategory{}, which may be again representable or not. Normalization is a canonical procedure to mix both tensors, $(⊗)$ and $(\triangleleft)$; and it will allow us to write \emph{\monoidalContexts} in \Cref{sec:monoidalContexts}, which form a \produoidalCategory{} without representable structure.

\begin{remark} %
  This means $𝓣{ℂ}$ has the structure of a \emph{virtual duoidal category} \cite{nlab:duoidal} or \emph{monoidal multicategory}, defined by Aguiar, Haim and López Franco \cite{aguiar18} as a pseudomonoid in the cartesian monoidal 2-category of multicategories.
\end{remark}

\section{Interlude: Normalization}
\label{sec:normalization}

\ProduoidalCategories{} seem to contain too much structure: of course, we want to split things in two different ways, sequentially $(◁)$ and in parallel $(⊗)$; but that does not necessarily mean that we want to keep track of two different types of units, parallel $(I)$ and sequential $(N)$. The atomic components of our decomposition algebra should be the same, without having to care if they are \emph{atomic for sequential composition} or \emph{atomic for parallel composition}.

Fortunately, there exists an abstract procedure that, starting from any \produoidalCategory{}, constructs a new produoidal category where both units are identified. This procedure is known as \emph{normalization}, and the resulting produoidal categories are called \emph{normal}.

\begin{definition}[Normal produoidal category]
  \defining{linkNormalProduoidalCategory}{}
  A \emph{normal produoidal category} is a \produoidalCategory{} where the laxator $\varphi_0 \colon 𝕍(•;I) \to 𝕍(•;N)$ is an isomorphism.

  \NormalProduoidalCategories{} form a category $\npDuo$ with produoidal functors between them and endowed with fully faithful forgetful functor $𝓤 \colon \npDuo \to \pDuo$.
\end{definition}

\begin{theorem}
  \label{th:normalizationProduoidal}
  \defining{linkNormalization}
    Let $𝕍_{⊗,I,◁,N}$ be a \produoidal{} category. The profunctor
    $𝓝𝕍(•; •) = 𝕍(•; N ⊗ • ⊗ N)$
    forms a \promonad{}. Moreover, the Kleisli category of this \promonad{} is a \normalProduoidalCategory{} with the following splits and units.
    $$\begin{aligned}
        𝓝𝕍(A;B) &= 𝕍(A ; N ⊗ B ⊗ N); \\
        𝓝𝕍(A;B ⊗_N C) &= 𝕍(A ; N ⊗ B ⊗ N ⊗ C ⊗ N); \\
        𝓝𝕍(A;B ◁_N C) &= 𝕍(A ; (N ⊗ B ⊗ N) ◁ (N ⊗ C ⊗ N)); \\
        𝓝𝕍(A;I_N) &= 𝕍(A ; N); \\ 
        𝓝𝕍(A;N_N) &= 𝕍(A ; N).
    \end{aligned}$$
\end{theorem}
\begin{proof}
  See Appendix, \Cref{ax:th:normalizationProduoidal}.
\end{proof}

A normalization procedure for duoidal categories was given by Garner and López Franco \cite{garner16}; our contribution is its produoidal counterpart. This novel produoidal normalization is better behaved than the duoidal one \cite{garner16}: the latter does not always exist, but we show produoidal normalization does.
Indeed, we prove that produoidal normalization forms an idempotent monad.
The technical reason for this improvement is that the original required the existence of certain coequalizers in $𝕍$; produoidal normalization uses coequalizers in $\Set$. Appendix \ref{ax:sec:garnernorm} outlines a relation between the two procedures.

\begin{theorem} \label{th:normalizationIdempotent}
  Normalization extends to an idempotent monad.%
\end{theorem}
\begin{proof}
  See Appendix, \Cref{ax:th:normalizationIdempotent}.
\end{proof}

\begin{theorem}[Free normal produoidal]
  \label{th:freeNormalProduoidal}
  Normalization determines an adjunction between produoidal categories and normal produoidal categories,
  $𝓝 \colon \Produo \rightleftharpoons \npDuo \colon 𝓤.$
  That is, $𝓝𝕍$ is the free produoidal category over $𝕍$.
\end{theorem}
\begin{proof}
  See Appendix, \Cref{ax:th:freeNormalProduoidal}.
\end{proof}

In the previous \Cref{sec:parallelContext}, we constructed the \produoidalCategory{} of \splicedMonoidalArrows{}, which distinguishes between morphisms and morphisms with a hole in the monoidal unit.
This is because the latter hole splits the morphism in two parts.   
Normalization equates both; it sews these two parts.
In \Cref{sec:monoidalContexts}, we explicitly construct \monoidalContexts{}, the normalization of \splicedMonoidalArrows{}.

\subsection{Symmetric Normalization}

Normalization is a generic procedure that applies to any \produoidalCategory{}, it does not matter if the parallel split $(⊗)$ is symmetric or not. However, when $⊗$ happens to be symmetric, we can also apply a more specialized normalization procedure.

\begin{definition}[Symmetric produoidal category]
  \defining{linkSymmetricProduoidal}{}
  \label{def:symmetricProduoidal}
  A \emph{symmetric produoidal category} is a \produoidalCategory{} $𝕍_{◁,N,⊗,I}$ endowed with a natural isomorphism
  $σ \colon 𝕍(A;B ⊗ C) ≅ 𝕍(A;C ⊗ B)$ satisfying the symmetry and hexagon equations.
\end{definition}

\begin{theorem}
  \label{th:symNormalizationProduoidal}
    Let $𝕍_{⊗,I,◁,N}$ be a \symmetricProduoidal{} category. The profunctor
    $𝓝_{σ}𝕍(• ; •) =𝕍(• ; N ⊗ •)$
    forms a \promonad{}. Moreover, the Kleisli category of this \promonad{} is a normal \symmetricProduoidal{} category with the following splits and units.
    $$\begin{aligned}
        𝓝_{σ}𝕍(A;B) &= 𝕍(A ; N ⊗ B); \\
        𝓝_{σ}𝕍(A;B ⊗_N C) &= 𝕍(A ; N ⊗ B ⊗ C); \\
        𝓝_{σ}𝕍(A;B ◁_N C) &= 𝕍(A ; (N ⊗ B) ◁ (N ⊗ C)); \\
        𝓝_{σ}𝕍(A ; I_N) &= 𝕍(A ; N); \\
        𝓝_{σ}𝕍(A ; N_N) &= 𝕍(A ; N).
    \end{aligned}$$
\end{theorem}
\begin{proof}
  See Appendix, \Cref{ax:th:symNormalizationProduoidal}.
\end{proof}

\begin{theorem}
  \label{th:sym:freeNormalProduoidal}
  Normalization determines an adjunction between \symmetricProduoidal{} and normal \symmetricProduoidal{} categories,
  $𝓝_σ \colon \symProduo \rightleftharpoons \nSymProduo \colon 𝓤.$
  That is, $𝓝_σ𝕍$ is the free \symmetricProduoidal{} category over $𝕍$.
\end{theorem}
\begin{proof}
  See Appendix, \Cref{ax:sym:th:freeNormalProduoidal}.
\end{proof}

\section{Monoidal Context: Mixing $◁$ and $⊗$ by normalization}
\label{sec:monoidalContexts}

\MonoidalContexts{} formalize the notion of an incomplete morphism in a monoidal category. The category of monoidal contexts will have a rich algebraic structure: we shall be able to still compose contexts sequentially and in parallel and, at the same time, we shall be able to fill a context using another monoidal context. Perhaps surprisingly, then, the category of monoidal contexts is not even monoidal.

We justify this apparent contradiction in terms of profunctorial structure: the category is not monoidal, but it does have two \promonoidal{} structures that precisely represent sequential and parallel composition. These structures form a \normalProduoidalCategory{}.
In fact, we show it to be the normalization of the \produoidal{} category of \splicedMonoidalArrows{}.

This section constructs explicitly the \normalProduoidalCategory{} of \monoidalContexts{}.

\subsection{The Category of Monoidal Contexts}

A \monoidalContext{}, $\MC{\biobj{A}{B}}{\biobj{X}{Y}}$, represents a process from $A$ to $B$ with a hole admitting a process from $X$ to $Y$. In this sense, \monoidalContexts{} are similar to \splicedMonoidalArrows{}. The difference with \splicedMonoidalArrows{} is that \monoidalContexts{} allow for communication to happen to the left and to the right of this hole.

\begin{definition}[Monoidal context]
  \defining{linkmonoidalcontext}{}
  \label{def:monoidalcontext}
  Let $(ℂ,⊗,I)$ be a monoidal category.
  \emph{Monoidal contexts} are the elements of the following profunctor,
  $$\MC{\biobj{A}{B}}{\biobj{X}{Y}} = ℂ(A;•_1⊗X⊗•_2)⋄ℂ(•_1⊗Y⊗•_2;B).$$
\end{definition}

In other words, a \emph{monoidal context} from $A$ to $B$, \emph{with a hole} from $X$ to $Y$, is an equivalence class consisting of a pair of objects $M, N ∈ \obj{ℂ}$ and a pair of morphisms $f ∈ ℂ(A; M⊗X⊗N)$  and $g ∈ ℂ(M⊗Y⊗N;B)$, quotiented by dinaturality of $M$ and $N$ (\Cref{fig:contextprotensors}).
We write \monoidalContexts{} as
$$(f ⨾ (\im_M ⊗ \blacksquare ⊗ \im_N) ⨾ g) \in
 \MC{\biobj{A}{B}}{\biobj{X}{Y}}.$$
In this notation, dinaturality explicitly means that
\begin{align*}
  & \nmc{f ⨾ (m ⊗ \im_X ⊗ n)}{W}{H}{g} & = \\
  & \nmc{f}{M}{N}{(m ⊗ \im_Y ⊗ n) ⨾ g}.
\end{align*}

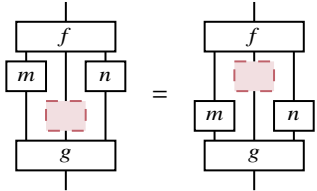
\begin{figure}[ht]
    \centering%

\tikzset{every picture/.style={line width=0.75pt}} %

\begin{tikzpicture}[x=0.75pt,y=0.75pt,yscale=-1,xscale=1]
\draw [color={rgb, 255:red, 0; green, 0; blue, 0 }  ,draw opacity=1 ]   (49.99,110) -- (49.99,120) ;
\draw  [color={rgb, 255:red, 191; green, 97; blue, 106 }  ,draw opacity=1 ][fill={rgb, 255:red, 191; green, 97; blue, 106 }  ,fill opacity=0.2 ][dash pattern={on 4.5pt off 4.5pt}] (40,75) -- (60,75) -- (60,90) -- (40,90) -- cycle ;
\draw  [color={rgb, 255:red, 0; green, 0; blue, 0 }  ,draw opacity=1 ][fill={rgb, 255:red, 255; green, 255; blue, 255 }  ,fill opacity=1 ] (25,35) -- (75,35) -- (75,50) -- (25,50) -- cycle ;
\draw [color={rgb, 255:red, 0; green, 0; blue, 0 }  ,draw opacity=1 ]   (50,25) -- (50,35) ;
\draw  [color={rgb, 255:red, 0; green, 0; blue, 0 }  ,draw opacity=1 ][fill={rgb, 255:red, 255; green, 255; blue, 255 }  ,fill opacity=1 ] (25,95) -- (75,95) -- (75,110) -- (25,110) -- cycle ;
\draw [color={rgb, 255:red, 0; green, 0; blue, 0 }  ,draw opacity=1 ]   (29.99,50) -- (30,55) ;
\draw [color={rgb, 255:red, 0; green, 0; blue, 0 }  ,draw opacity=1 ]   (50,50) -- (50,75) ;
\draw [color={rgb, 255:red, 0; green, 0; blue, 0 }  ,draw opacity=1 ]   (70,70) -- (70,95) ;
\draw  [draw opacity=0] (80,60) -- (115,60) -- (115,85) -- (80,85) -- cycle ;
\draw  [color={rgb, 255:red, 0; green, 0; blue, 0 }  ,draw opacity=1 ][fill={rgb, 255:red, 255; green, 255; blue, 255 }  ,fill opacity=1 ] (20,55) -- (40,55) -- (40,70) -- (20,70) -- cycle ;
\draw [color={rgb, 255:red, 0; green, 0; blue, 0 }  ,draw opacity=1 ]   (69.99,50) -- (70,55) ;
\draw  [color={rgb, 255:red, 0; green, 0; blue, 0 }  ,draw opacity=1 ][fill={rgb, 255:red, 255; green, 255; blue, 255 }  ,fill opacity=1 ] (60,55) -- (80,55) -- (80,70) -- (60,70) -- cycle ;
\draw [color={rgb, 255:red, 0; green, 0; blue, 0 }  ,draw opacity=1 ]   (30,70) -- (30,95) ;
\draw [color={rgb, 255:red, 0; green, 0; blue, 0 }  ,draw opacity=1 ]   (49.99,90) -- (50,95) ;
\draw [color={rgb, 255:red, 0; green, 0; blue, 0 }  ,draw opacity=1 ]   (144.99,110) -- (144.99,120) ;
\draw  [color={rgb, 255:red, 191; green, 97; blue, 106 }  ,draw opacity=1 ][fill={rgb, 255:red, 191; green, 97; blue, 106 }  ,fill opacity=0.2 ][dash pattern={on 4.5pt off 4.5pt}] (135,55) -- (155,55) -- (155,70) -- (135,70) -- cycle ;
\draw  [color={rgb, 255:red, 0; green, 0; blue, 0 }  ,draw opacity=1 ][fill={rgb, 255:red, 255; green, 255; blue, 255 }  ,fill opacity=1 ] (120,35) -- (170,35) -- (170,50) -- (120,50) -- cycle ;
\draw [color={rgb, 255:red, 0; green, 0; blue, 0 }  ,draw opacity=1 ]   (145,25) -- (145,35) ;
\draw  [color={rgb, 255:red, 0; green, 0; blue, 0 }  ,draw opacity=1 ][fill={rgb, 255:red, 255; green, 255; blue, 255 }  ,fill opacity=1 ] (120,95) -- (170,95) -- (170,110) -- (120,110) -- cycle ;
\draw [color={rgb, 255:red, 0; green, 0; blue, 0 }  ,draw opacity=1 ]   (124.99,90) -- (125,95) ;
\draw [color={rgb, 255:red, 0; green, 0; blue, 0 }  ,draw opacity=1 ]   (145,70) -- (145,95) ;
\draw [color={rgb, 255:red, 0; green, 0; blue, 0 }  ,draw opacity=1 ]   (165,50) -- (165,75) ;
\draw  [color={rgb, 255:red, 0; green, 0; blue, 0 }  ,draw opacity=1 ][fill={rgb, 255:red, 255; green, 255; blue, 255 }  ,fill opacity=1 ] (115,75) -- (135,75) -- (135,90) -- (115,90) -- cycle ;
\draw [color={rgb, 255:red, 0; green, 0; blue, 0 }  ,draw opacity=1 ]   (165,90) -- (165.01,95) ;
\draw  [color={rgb, 255:red, 0; green, 0; blue, 0 }  ,draw opacity=1 ][fill={rgb, 255:red, 255; green, 255; blue, 255 }  ,fill opacity=1 ] (155,75) -- (175,75) -- (175,90) -- (155,90) -- cycle ;
\draw [color={rgb, 255:red, 0; green, 0; blue, 0 }  ,draw opacity=1 ]   (125,50) -- (125,75) ;
\draw [color={rgb, 255:red, 0; green, 0; blue, 0 }  ,draw opacity=1 ]   (144.99,50) -- (145,55) ;

\draw (50,42.5) node  [font=\footnotesize]  {$f$};
\draw (50,102.5) node  [font=\footnotesize]  {$g$};
\draw (97.5,72.5) node    {$=$};
\draw (30,62.5) node  [font=\footnotesize]  {$m$};
\draw (70,62.5) node  [font=\footnotesize]  {$n$};
\draw (145,42.5) node  [font=\footnotesize]  {$f$};
\draw (145,102.5) node  [font=\footnotesize]  {$g$};
\draw (125,82.5) node  [font=\footnotesize]  {$m$};
\draw (165,82.5) node  [font=\footnotesize]  {$n$};

\end{tikzpicture}
     \caption{Dinaturality for monoidal contexts.}
    \label{fig:contextprotensors}
\end{figure}

\begin{proposition}
  \label{prop:contextCategory}
  \MonoidalContexts{} form a category.
\end{proposition}
\begin{proof}[Proof]
  We define composition of \monoidalContexts{} by the following formula (illustrated in \Cref{fig:monoidalcontexts}, iii).
  \begin{align*}
    & \nmc{f}{M}{N}{g} ≺
    \nmc{h}{M'}{N'}{k} & = \\
    & f ⨾ (\im_M ⊗ h ⊗ \im_N) ⨾ (\im_{M ⊗ M'} ⊗ \blacksquare ⊗ \im_{N ⊗ N'}) \\
    & \quad ⨾ (\im_M ⊗ k ⊗ \im_N) ⨾ g
  \end{align*}
  For each pair of objects, we define the identity \monoidalContext{} as $\im_A ⨾ \blacksquare ⨾ \im_B$ (illustrated in \Cref{fig:monoidalcontexts}, ii).
  We check that this composition is associative and unital in the Appendix, \Cref{ax:prop:contextCategory}.
\end{proof}

\begin{figure}[ht]
  \centering

\tikzset{every picture/.style={line width=0.75pt}} %

\begin{tikzpicture}[x=0.75pt,y=0.75pt,yscale=-1,xscale=1]
\draw  [color={rgb, 255:red, 0; green, 0; blue, 0 }  ,draw opacity=1 ][fill={rgb, 255:red, 255; green, 255; blue, 255 }  ,fill opacity=1 ] (100,45) -- (140,45) -- (140,60) -- (100,60) -- cycle ;
\draw  [draw opacity=0] (345,65) -- (365,65) -- (365,90) -- (345,90) -- cycle ;
\draw [color={rgb, 255:red, 0; green, 0; blue, 0 }  ,draw opacity=1 ]   (120,110) -- (120,120) ;
\draw  [color={rgb, 255:red, 191; green, 97; blue, 106 }  ,draw opacity=1 ][fill={rgb, 255:red, 191; green, 97; blue, 106 }  ,fill opacity=0.2 ][dash pattern={on 4.5pt off 4.5pt}] (110,70) -- (130,70) -- (130,85) -- (110,85) -- cycle ;
\draw [color={rgb, 255:red, 0; green, 0; blue, 0 }  ,draw opacity=1 ]   (120,35) -- (120,45) ;
\draw  [color={rgb, 255:red, 0; green, 0; blue, 0 }  ,draw opacity=1 ][fill={rgb, 255:red, 255; green, 255; blue, 255 }  ,fill opacity=1 ] (100,95) -- (140,95) -- (140,110) -- (100,110) -- cycle ;
\draw [color={rgb, 255:red, 0; green, 0; blue, 0 }  ,draw opacity=1 ]   (120,60) -- (120,70) ;
\draw [color={rgb, 255:red, 0; green, 0; blue, 0 }  ,draw opacity=1 ]   (120,85) -- (120,95) ;
\draw [color={rgb, 255:red, 0; green, 0; blue, 0 }  ,draw opacity=1 ]   (135,60) -- (135,95) ;
\draw [color={rgb, 255:red, 0; green, 0; blue, 0 }  ,draw opacity=1 ]   (105,60) -- (105,95) ;
\draw [color={rgb, 255:red, 0; green, 0; blue, 0 }  ,draw opacity=1 ]   (330,60) -- (330,70) ;
\draw [color={rgb, 255:red, 0; green, 0; blue, 0 }  ,draw opacity=1 ]   (330,85) -- (330,95) ;
\draw  [draw opacity=0] (140,65) -- (160,65) -- (160,90) -- (140,90) -- cycle ;
\draw  [color={rgb, 255:red, 191; green, 97; blue, 106 }  ,draw opacity=1 ][fill={rgb, 255:red, 191; green, 97; blue, 106 }  ,fill opacity=0.2 ][dash pattern={on 4.5pt off 4.5pt}] (175,50) -- (195,50) -- (195,65) -- (175,65) -- cycle ;
\draw  [color={rgb, 255:red, 0; green, 0; blue, 0 }  ,draw opacity=1 ][fill={rgb, 255:red, 255; green, 255; blue, 255 }  ,fill opacity=1 ] (165,30) -- (205,30) -- (205,45) -- (165,45) -- cycle ;
\draw [color={rgb, 255:red, 0; green, 0; blue, 0 }  ,draw opacity=1 ]   (185,20) -- (185,30) ;
\draw [color={rgb, 255:red, 0; green, 0; blue, 0 }  ,draw opacity=1 ]   (170,45) -- (170,70) ;
\draw  [color={rgb, 255:red, 0; green, 0; blue, 0 }  ,draw opacity=1 ][fill={rgb, 255:red, 255; green, 255; blue, 255 }  ,fill opacity=1 ] (165,70) -- (205,70) -- (205,85) -- (165,85) -- cycle ;
\draw [color={rgb, 255:red, 0; green, 0; blue, 0 }  ,draw opacity=1 ]   (185,45) -- (185,50) ;
\draw [color={rgb, 255:red, 0; green, 0; blue, 0 }  ,draw opacity=1 ]   (185,65) -- (185,70) ;
\draw  [color={rgb, 255:red, 0; green, 0; blue, 0 }  ,draw opacity=1 ][fill={rgb, 255:red, 255; green, 255; blue, 255 }  ,fill opacity=1 ] (165,109.5) -- (205,109.5) -- (205,124.5) -- (165,124.5) -- cycle ;
\draw [color={rgb, 255:red, 0; green, 0; blue, 0 }  ,draw opacity=1 ]   (185,125) -- (185,135) ;
\draw [color={rgb, 255:red, 0; green, 0; blue, 0 }  ,draw opacity=1 ]   (200,45) -- (200,70) ;
\draw  [color={rgb, 255:red, 191; green, 97; blue, 106 }  ,draw opacity=1 ][fill={rgb, 255:red, 191; green, 97; blue, 106 }  ,fill opacity=0.2 ][dash pattern={on 4.5pt off 4.5pt}] (175,90) -- (195,90) -- (195,105) -- (175,105) -- cycle ;
\draw [color={rgb, 255:red, 0; green, 0; blue, 0 }  ,draw opacity=1 ]   (170,85) -- (170,110) ;
\draw [color={rgb, 255:red, 0; green, 0; blue, 0 }  ,draw opacity=1 ]   (184.99,85) -- (185,90) ;
\draw [color={rgb, 255:red, 0; green, 0; blue, 0 }  ,draw opacity=1 ]   (185,105) -- (185,110) ;
\draw [color={rgb, 255:red, 0; green, 0; blue, 0 }  ,draw opacity=1 ]   (200,85) -- (200,110) ;
\draw  [draw opacity=0] (205,65) -- (225,65) -- (225,90) -- (205,90) -- cycle ;
\draw  [color={rgb, 255:red, 0; green, 0; blue, 0 }  ,draw opacity=1 ][fill={rgb, 255:red, 255; green, 255; blue, 255 }  ,fill opacity=1 ] (315,70) -- (345,70) -- (345,85) -- (315,85) -- cycle ;
\draw [color={rgb, 255:red, 0; green, 0; blue, 0 }  ,draw opacity=1 ]   (260.01,110) -- (260.01,120) ;
\draw  [color={rgb, 255:red, 191; green, 97; blue, 106 }  ,draw opacity=1 ][fill={rgb, 255:red, 191; green, 97; blue, 106 }  ,fill opacity=0.2 ][dash pattern={on 4.5pt off 4.5pt}] (235.01,70) -- (255.01,70) -- (255.01,85) -- (235.01,85) -- cycle ;
\draw  [color={rgb, 255:red, 0; green, 0; blue, 0 }  ,draw opacity=1 ][fill={rgb, 255:red, 255; green, 255; blue, 255 }  ,fill opacity=1 ] (230,45) -- (290.01,45) -- (290.01,60) -- (230,60) -- cycle ;
\draw [color={rgb, 255:red, 0; green, 0; blue, 0 }  ,draw opacity=1 ]   (260.01,35) -- (260.01,45) ;
\draw [color={rgb, 255:red, 0; green, 0; blue, 0 }  ,draw opacity=1 ]   (225,75) -- (225,80) ;
\draw  [color={rgb, 255:red, 0; green, 0; blue, 0 }  ,draw opacity=1 ][fill={rgb, 255:red, 255; green, 255; blue, 255 }  ,fill opacity=1 ] (230,95) -- (290.01,95) -- (290.01,110) -- (230,110) -- cycle ;
\draw [color={rgb, 255:red, 0; green, 0; blue, 0 }  ,draw opacity=1 ]   (245,60) -- (245.01,70) ;
\draw [color={rgb, 255:red, 0; green, 0; blue, 0 }  ,draw opacity=1 ]   (245.01,85) -- (245.01,95) ;
\draw [color={rgb, 255:red, 0; green, 0; blue, 0 }  ,draw opacity=1 ]   (235,60) .. controls (234.6,69.2) and (225.2,65.8) .. (225,75) ;
\draw [color={rgb, 255:red, 0; green, 0; blue, 0 }  ,draw opacity=1 ]   (285.01,60) .. controls (284.61,69.2) and (295.21,65.8) .. (295.01,75) ;
\draw [color={rgb, 255:red, 0; green, 0; blue, 0 }  ,draw opacity=1 ]   (285.01,95) .. controls (284.61,85.8) and (295.21,89.2) .. (295.01,80) ;
\draw [color={rgb, 255:red, 0; green, 0; blue, 0 }  ,draw opacity=1 ]   (235,95) .. controls (234.6,85.8) and (225.2,89.2) .. (225,80) ;
\draw [color={rgb, 255:red, 0; green, 0; blue, 0 }  ,draw opacity=1 ]   (260.01,60) -- (260.02,95) ;
\draw  [color={rgb, 255:red, 191; green, 97; blue, 106 }  ,draw opacity=1 ][fill={rgb, 255:red, 191; green, 97; blue, 106 }  ,fill opacity=0.2 ][dash pattern={on 4.5pt off 4.5pt}] (265.01,70) -- (285.01,70) -- (285.01,85) -- (265.01,85) -- cycle ;
\draw [color={rgb, 255:red, 0; green, 0; blue, 0 }  ,draw opacity=1 ]   (275,60) -- (275.01,70) ;
\draw [color={rgb, 255:red, 0; green, 0; blue, 0 }  ,draw opacity=1 ]   (275.01,85) -- (275.02,95) ;
\draw [color={rgb, 255:red, 0; green, 0; blue, 0 }  ,draw opacity=1 ]   (295,75) -- (295.01,80) ;
\draw  [draw opacity=0] (295,65) -- (315,65) -- (315,90) -- (295,90) -- cycle ;

\draw (355,77.5) node    {$;$};
\draw (120,52.5) node  [font=\footnotesize]  {$f$};
\draw (120,102.5) node  [font=\footnotesize]  {$g$};
\draw (150.01,77.5) node    {$;$};
\draw (185,37.5) node  [font=\footnotesize]  {$f$};
\draw (185,77.5) node  [font=\footnotesize]  {$g$};
\draw (185,117) node  [font=\footnotesize]  {$h$};
\draw (215,77.5) node    {$;$};
\draw (330,77.5) node  [font=\footnotesize]  {$f$};
\draw (260.01,52.5) node  [font=\footnotesize]  {$f$};
\draw (260.01,102.5) node  [font=\footnotesize]  {$g$};
\draw (305,77.5) node    {$;$};

\end{tikzpicture}
   \caption{Morphisms, sequential and parallel splits, and units of the splice monoidal arrow produoidal category.}
  \label{fig:monoidalsplicecomponents}
\end{figure}
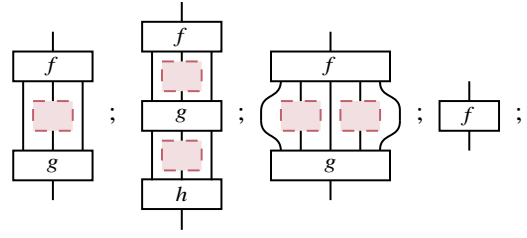

\begin{remark}
  Even when we introduce $(\im \otimes \blacksquare \otimes \im)$ as a piece of suggestive notation, we can still write $(g \otimes \blacksquare \otimes h)$ unambiguously, because of \dinaturality{},
  \begin{align*}
    (g \otimes \im \otimes h) ⨾
    (\im \otimes \blacksquare \otimes \im) & =
    (\im \otimes \blacksquare \otimes \im) ⨾
    (g \otimes \im \otimes h). & \\
  \end{align*}
\end{remark}

\subsection{The Normal Produoidal Algebra of Monoidal Contexts}

Let us endow \monoidalContexts{} with their normal produoidal structure.

\begin{definition}
  \label{defn:contextSeqProtensor}
  The category of \monoidalContexts{}, $\Mctx{ℂ}$, has as objects pairs of objects of $ℂ$.
  We use the following \profunctors{} to define \sequentialSplits{}, \parallelSplits{}, units and morphisms.
  \begin{align*}
   \MC{\biobj{A}{B}}{\biobj{X}{Y}} =\
  & ℂ(A;•_1⊗X⊗•_2) ⋄ ℂ(•_1⊗Y⊗•_2;B); \\
  \MC{\biobj{A}{B}}{\biobj{X}{Y} \triangleleft \biobj{X'}{Y'}} =\ &
      ℂ(A ; •_1 ⊗ X ⊗ •_2)\ \diamond \\ &
      ℂ(•_1 ⊗ Y ⊗ •_2; •_3 ⊗ X' ⊗ •_4)\ \diamond \\ &
      ℂ(•_3 ⊗ Y' ⊗ •_4 ; B); \\
  \MC{\biobj{A}{B}}{\biobj{X}{Y} ⊗ \biobj{X'}{Y'}} =\ &
      ℂ(A ; •_1 ⊗ X ⊗ •_2 ⊗ X' ⊗ •_3)\ \diamond \\ &
      ℂ(•_1 ⊗ Y ⊗ •_2 ⊗ Y' ⊗ •_3; B); \\
  \MC{\biobj{A}{B}}{N} =\ &ℂ(A ; B).
  \end{align*}
\end{definition}

In other words, \sequentialSplits{} are triples of arrows $f \colon A \to M \otimes X \otimes N$, $g \colon M \otimes Y \otimes N \to M' \otimes X' \otimes N'$ and $h \colon M' \otimes Y' \otimes N' \to B$, quotiented by \dinaturality{} of $M,M',N,N'$.
\ParallelSplits{} are pairs of arrows $f \colon A \to M \otimes X \otimes N \otimes X' \otimes O$ and $g \colon M \otimes Y \otimes N \otimes Y' \otimes O → B$, quotiented by \dinaturality{} of $M,N,O$. Units are simply arrows $f \colon A \to B$.
In summary, we have
\begin{align*}
  & \mbox{morphisms, } && f ⨾ (\im \otimes \blacksquare \otimes \im) ⨾ g \\
  & \mbox{\sequentialSplits{}, }
  && f ⨾ (\im \otimes \blacksquare \otimes \im) ⨾ g ⨾ (\im \otimes \blacksquare \otimes \im) ⨾ h; \\
  & \mbox{\parallelSplits{}, } && f ⨾ (\im \otimes \blacksquare \otimes \im ⊗ \blacksquare \otimes \im) ⨾ g; \\
  & \mbox{\sequentialUnits{}, } && f.
\end{align*}

  \Dinaturality{} for \sequentialSplits{} and \parallelSplits{} is depicted the Appendix, \Cref{fig:sequentialprotensor,fig:parallelprotensor}.

\begin{proposition}
  \label{prop:MonoidalContextProtensor}
  The category of monoidal contexts forms a normal \produoidalCategory{} with its units, sequential and parallel splits.
\end{proposition}
\begin{proof}
  See Appendix, \Cref{ax:prop:MonoidalContextProtensor}.
\end{proof}

\begin{theorem}
  \label{th:monoidalContextsAreANormalization}
  \MonoidalContexts{} are the free normalization of the cofree \produoidalCategory{} over a category.
  In other words, \monoidalContexts{} are the normalization of \splicedMonoidalArrows{},
  $𝓝𝓣ℂ \cong \Mctx{ℂ}$.
\end{theorem}
\begin{proof}
  See Appendix, \Cref{ax:th:monoidalContextsAreANormalization}.
\end{proof}

\section{Monoidal Lenses}
\defining{linksymmetricmonoidalcontext}{}
\label{sec:monoidallenses}

\MonoidalLenses{} are \emph{symmetric} \monoidalContexts{}.
Again, the category of monoidal lenses has a rich algebraic structure; and again, most of this structure exists only virtually in terms of profunctors. In this case, though, the monoidal tensor \emph{does} indeed exist: contrary to \monoidalContexts{}, \monoidalLenses{} form also a monoidal category.

This is perhaps why applications of \monoidalLenses{} have grown popular in recent years~\cite{riley2018categories}, with applications in decision theory \cite{ghani:compositionalgametheory2018}, supervised learning \cite{cruttwell22:learning,fong19:lenses} and most notably in functional data accessing \cite{kmett12:lenslibrary,pickering17:profunctoroptics,boisseau2018you,ClarkeRoman20:ProfunctorOptics}.
The \promonoidal{} structure of optics was ignored, even when, after now identifying for the first time its relation to the monoidal structure of optics, we argue that it could be potentially useful in these applications: e.g. in multi-stage decision problems, or in multi-stage data accessors.

This section explicitly constructs the normal symmetric \produoidalCategory{} of \emph{monoidal lenses}.
We describe it for the first time by a universal property: it is the free symmetric normalization of the cofree \produoidalCategory{}.

\subsection{The Category of Monoidal Lenses}

A \monoidalLens{} of type $𝓛ℂ(\biobj{A}{B}, \biobj{X}{Y})$ represents a process in a symmetric \monoidalCategory{} with a hole admitting a process from $X$ to $Y$.

\begin{figure}[ht]
  \centering

\tikzset{every picture/.style={line width=0.75pt}} %

\begin{tikzpicture}[x=0.75pt,y=0.75pt,yscale=-1,xscale=1]
\draw    (200,120) .. controls (200.33,100.17) and (225.43,103) .. (225,90) ;
\draw  [color={rgb, 255:red, 0; green, 0; blue, 0 }  ,draw opacity=1 ][fill={rgb, 255:red, 255; green, 255; blue, 255 }  ,fill opacity=1 ] (10,30) -- (50,30) -- (50,45) -- (10,45) -- cycle ;
\draw [color={rgb, 255:red, 0; green, 0; blue, 0 }  ,draw opacity=1 ]   (20,10) -- (20,30) ;
\draw [color={rgb, 255:red, 0; green, 0; blue, 0 }  ,draw opacity=1 ]   (20,45) -- (20,120) ;
\draw  [color={rgb, 255:red, 0; green, 0; blue, 0 }  ,draw opacity=1 ][fill={rgb, 255:red, 255; green, 255; blue, 255 }  ,fill opacity=1 ] (10,120) -- (50,120) -- (50,135) -- (10,135) -- cycle ;
\draw [color={rgb, 255:red, 0; green, 0; blue, 0 }  ,draw opacity=1 ]   (20,135) -- (20,155) ;
\draw    (40,45) .. controls (40.33,64.83) and (70,55.5) .. (70,75) ;
\draw  [draw opacity=0][fill={rgb, 255:red, 255; green, 255; blue, 255 }  ,fill opacity=1 ] (40,75) .. controls (40,54.63) and (69.75,64.88) .. (70,45) -- (75,45) -- (75,80) -- (40.13,80) .. controls (40.04,78.38) and (40,76.71) .. (40,75) -- cycle ;
\draw [color={rgb, 255:red, 208; green, 2; blue, 27 }  ,draw opacity=1 ] [dash pattern={on 4.5pt off 4.5pt}]  (70,45) .. controls (70.33,64.83) and (40,55.5) .. (40,75) ;
\draw    (70,95) .. controls (70.33,114.83) and (40,100.5) .. (40,120) ;
\draw  [draw opacity=0][fill={rgb, 255:red, 255; green, 255; blue, 255 }  ,fill opacity=1 ] (40,95) .. controls (40,115.38) and (69.75,105.13) .. (70,125) -- (75,125) -- (75,90) -- (40.13,90) .. controls (40.04,91.62) and (40,93.29) .. (40,95) -- cycle ;
\draw [color={rgb, 255:red, 208; green, 2; blue, 27 }  ,draw opacity=1 ] [dash pattern={on 4.5pt off 4.5pt}]  (40,95) .. controls (40.33,114.83) and (70,105.5) .. (70,125) ;
\draw  [draw opacity=0] (75,70) -- (95,70) -- (95,100) -- (75,100) -- cycle ;
\draw  [color={rgb, 255:red, 0; green, 0; blue, 0 }  ,draw opacity=1 ][fill={rgb, 255:red, 255; green, 255; blue, 255 }  ,fill opacity=1 ] (180,30) -- (220,30) -- (220,45) -- (180,45) -- cycle ;
\draw [color={rgb, 255:red, 0; green, 0; blue, 0 }  ,draw opacity=1 ]   (190,10) -- (190,30) ;
\draw [color={rgb, 255:red, 0; green, 0; blue, 0 }  ,draw opacity=1 ]   (190,45) -- (190,120) ;
\draw  [color={rgb, 255:red, 0; green, 0; blue, 0 }  ,draw opacity=1 ][fill={rgb, 255:red, 255; green, 255; blue, 255 }  ,fill opacity=1 ] (180,120) -- (220,120) -- (220,135) -- (180,135) -- cycle ;
\draw [color={rgb, 255:red, 0; green, 0; blue, 0 }  ,draw opacity=1 ]   (190,135) -- (190,155) ;
\draw  [draw opacity=0][fill={rgb, 255:red, 255; green, 255; blue, 255 }  ,fill opacity=1 ] (200,95) .. controls (200,115.38) and (229.75,105.13) .. (230,125) -- (235,125) -- (235,90) -- (200.13,90) .. controls (200.04,91.62) and (200,93.29) .. (200,95) -- cycle ;
\draw [color={rgb, 255:red, 208; green, 2; blue, 27 }  ,draw opacity=1 ] [dash pattern={on 4.5pt off 4.5pt}]  (200,95) .. controls (200.33,114.83) and (230,105.5) .. (230,125) ;
\draw  [color={rgb, 255:red, 0; green, 0; blue, 0 }  ,draw opacity=1 ][fill={rgb, 255:red, 255; green, 255; blue, 255 }  ,fill opacity=1 ] (95,20) -- (135,20) -- (135,35) -- (95,35) -- cycle ;
\draw [color={rgb, 255:red, 0; green, 0; blue, 0 }  ,draw opacity=1 ]   (105,10) -- (105,20) ;
\draw [color={rgb, 255:red, 0; green, 0; blue, 0 }  ,draw opacity=1 ]   (105,35) -- (105,130) ;
\draw  [color={rgb, 255:red, 0; green, 0; blue, 0 }  ,draw opacity=1 ][fill={rgb, 255:red, 255; green, 255; blue, 255 }  ,fill opacity=1 ] (95,130) -- (135,130) -- (135,145) -- (95,145) -- cycle ;
\draw [color={rgb, 255:red, 0; green, 0; blue, 0 }  ,draw opacity=1 ]   (105,145) -- (105,155) ;
\draw  [draw opacity=0] (250,70) -- (270,70) -- (270,100) -- (250,100) -- cycle ;
\draw  [color={rgb, 255:red, 0; green, 0; blue, 0 }  ,draw opacity=1 ][fill={rgb, 255:red, 255; green, 255; blue, 255 }  ,fill opacity=1 ] (95,75) -- (135,75) -- (135,90) -- (95,90) -- cycle ;
\draw    (125,35) .. controls (125.33,46.33) and (145,43.86) .. (145,55) ;
\draw  [draw opacity=0][fill={rgb, 255:red, 255; green, 255; blue, 255 }  ,fill opacity=1 ] (120,52.14) .. controls (120,40.5) and (149.75,46.36) .. (150,35) -- (155,35) -- (155,55) -- (120.13,55) .. controls (120.04,54.07) and (120,53.12) .. (120,52.14) -- cycle ;
\draw [color={rgb, 255:red, 208; green, 2; blue, 27 }  ,draw opacity=1 ] [dash pattern={on 4.5pt off 4.5pt}]  (150,35) .. controls (150.33,46.33) and (120,41) .. (120,52.14) ;
\draw    (125,75) .. controls (125.33,63.67) and (145,66.14) .. (145,55) ;
\draw  [draw opacity=0][fill={rgb, 255:red, 255; green, 255; blue, 255 }  ,fill opacity=1 ] (120,57.86) .. controls (120,69.5) and (149.75,63.64) .. (150,75) -- (155,75) -- (155,55) -- (120.13,55) .. controls (120.04,55.93) and (120,56.88) .. (120,57.86) -- cycle ;
\draw [color={rgb, 255:red, 208; green, 2; blue, 27 }  ,draw opacity=1 ] [dash pattern={on 4.5pt off 4.5pt}]  (150,75) .. controls (150.33,63.67) and (120,69) .. (120,57.86) ;
\draw    (125,90) .. controls (125.33,101.33) and (145,98.86) .. (145,110) ;
\draw  [draw opacity=0][fill={rgb, 255:red, 255; green, 255; blue, 255 }  ,fill opacity=1 ] (120,107.14) .. controls (120,95.5) and (149.75,101.36) .. (150,90) -- (155,90) -- (155,110) -- (120.13,110) .. controls (120.04,109.07) and (120,108.12) .. (120,107.14) -- cycle ;
\draw [color={rgb, 255:red, 208; green, 2; blue, 27 }  ,draw opacity=1 ] [dash pattern={on 4.5pt off 4.5pt}]  (150,90) .. controls (150.33,101.33) and (120,96) .. (120,107.14) ;
\draw    (125,130) .. controls (125.33,118.67) and (145,121.14) .. (145,110) ;
\draw  [draw opacity=0][fill={rgb, 255:red, 255; green, 255; blue, 255 }  ,fill opacity=1 ] (120,112.86) .. controls (120,124.5) and (149.75,118.64) .. (150,130) -- (155,130) -- (155,110) -- (120.13,110) .. controls (120.04,110.93) and (120,111.88) .. (120,112.86) -- cycle ;
\draw [color={rgb, 255:red, 208; green, 2; blue, 27 }  ,draw opacity=1 ] [dash pattern={on 4.5pt off 4.5pt}]  (150,130) .. controls (150.33,118.67) and (120,124) .. (120,112.86) ;
\draw    (240,95) .. controls (240.33,114.83) and (210,100.5) .. (210,120) ;
\draw  [draw opacity=0][fill={rgb, 255:red, 255; green, 255; blue, 255 }  ,fill opacity=1 ] (215,95) .. controls (215,115.38) and (244.75,105.13) .. (245,125) -- (250,125) -- (250,90) -- (215.13,90) .. controls (215.04,91.62) and (215,93.29) .. (215,95) -- cycle ;
\draw [color={rgb, 255:red, 208; green, 2; blue, 27 }  ,draw opacity=1 ] [dash pattern={on 4.5pt off 4.5pt}]  (215,95) .. controls (215.33,114.83) and (245,105.5) .. (245,125) ;
\draw    (200,45) .. controls (200.33,64.83) and (225.43,62) .. (225,75) ;
\draw  [draw opacity=0][fill={rgb, 255:red, 255; green, 255; blue, 255 }  ,fill opacity=1 ] (200,70) .. controls (200,49.63) and (229.75,59.88) .. (230,40) -- (235,40) -- (235,75) -- (200.13,75) .. controls (200.04,73.38) and (200,71.71) .. (200,70) -- cycle ;
\draw [color={rgb, 255:red, 208; green, 2; blue, 27 }  ,draw opacity=1 ] [dash pattern={on 4.5pt off 4.5pt}]  (200,70) .. controls (200.33,50.17) and (230,59.5) .. (230,40) ;
\draw    (240,70) .. controls (240.33,50.17) and (210,64.5) .. (210,45) ;
\draw  [draw opacity=0][fill={rgb, 255:red, 255; green, 255; blue, 255 }  ,fill opacity=1 ] (215,70) .. controls (215,49.63) and (244.75,59.88) .. (245,40) -- (250,40) -- (250,75) -- (215.13,75) .. controls (215.04,73.38) and (215,71.71) .. (215,70) -- cycle ;
\draw [color={rgb, 255:red, 208; green, 2; blue, 27 }  ,draw opacity=1 ] [dash pattern={on 4.5pt off 4.5pt}]  (215,70) .. controls (215.33,50.17) and (245,59.5) .. (245,40) ;

\draw (30,37.5) node  [font=\footnotesize]  {$f$};
\draw (30,127.5) node  [font=\footnotesize]  {$g$};
\draw (85,85) node  [font=\small,color={rgb, 255:red, 208; green, 2; blue, 27 }  ,opacity=1 ] [align=left] {$\displaystyle ;$};
\draw (200,37.5) node  [font=\footnotesize]  {$f$};
\draw (200,127.5) node  [font=\footnotesize]  {$g$};
\draw (170,85) node  [font=\small,color={rgb, 255:red, 208; green, 2; blue, 27 }  ,opacity=1 ] [align=left] {$\displaystyle ;$};
\draw (115,27.5) node  [font=\footnotesize]  {$f$};
\draw (115,137.5) node  [font=\footnotesize]  {$h$};
\draw (260,85) node  [font=\small,color={rgb, 255:red, 208; green, 2; blue, 27 }  ,opacity=1 ] [align=left] {$\displaystyle ;$};
\draw (115,82.5) node  [font=\footnotesize]  {$g$};

\end{tikzpicture}
   \caption{Generic monoidal lens, sequential and parallel split.}
  \label{fig:monoidallens}
\end{figure}
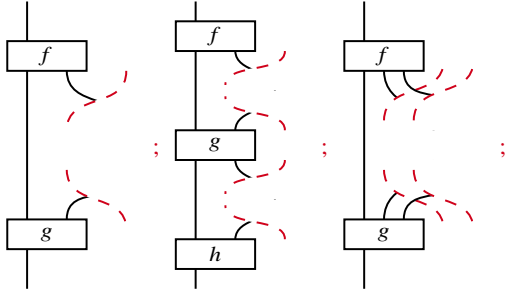

\begin{definition}[Monoidal Lens]
  \label{def:monoidallens}
  \defining{linkMonoidalLens}{}
  Let $(ℂ,⊗,I)$ be a symmetric \monoidalCategory{}.
  \emph{Monoidal lenses} are the elements of the following profunctor,
  $$\LC{\biobj{A}{B}}{\biobj{X}{Y}} = ℂ(A ; • ⊗ X) \diamond ℂ(• ⊗ Y ; B).$$
\end{definition}

In other words, a \emph{monoidal lens} from $A$ to $B$, \emph{with a hole} from $X$ to $Y$, is an equivalence class consisting of a pair of objects $M, N \in \obj{ℂ}$ and a pair of morphisms $f \in ℂ(A; M⊗X)$  and $g ∈ ℂ(M⊗Y;B)$, quotiented by \dinaturality{} of $M$.
We write \monoidalLenses{} as
$$f ⨾ (\im_M ⊗ \blacksquare) ⨾ g \in
 \LC{\biobj{A}{B}}{\biobj{X}{Y}}.$$

\begin{proposition}
  \label{prop:monoidalLensesProduoidal}
  \MonoidalLenses{} form a normal symmetric produoidal category with the following morphisms, units, sequential and parallel splits.
  $$\begin{aligned}
    \LC{\biobj{A}{B}}{\biobj{X}{Y}} =\ &ℂ(A ; • ⊗ X) \diamond ℂ(• ⊗ Y ; B); \\
    \LC{\biobj{A}{B}}{N} =\ &ℂ(A ; B); \\
    \LC{\biobj{A}{B}}{\biobj{X}{Y} \triangleleft \biobj{X'}{Y'}} =\ &
      ℂ(A ; •_1 ⊗ X)\ \diamond \\ &
      ℂ(•_1 ⊗ Y; •_2 ⊗ X') \diamond
      ℂ(•_2 ⊗ Y' ; B); \\
    \LC{\biobj{A}{B}}{\biobj{X}{Y} ⊗ \biobj{X'}{Y'}} =\ &
      ℂ(A ; •_1 ⊗ X ⊗ X')\diamond
      ℂ(•_1 ⊗ Y ⊗ Y'; B).
  \end{aligned}$$
\end{proposition}
\begin{proof}
  See Appendix, \Cref{ax:prop:monoidalLensesProduoidal}.
\end{proof}

\begin{theorem}
  \label{th:lensesuniversal}
  Monoidal lenses are the free symmetric normalization of the cofree symmetric produoidal category over a monoidal category.
\end{theorem}
\begin{proof}
  See Appendix, \Cref{ax:th:lensesuniversal}.
\end{proof}

\begin{remark}[Representable parallel structure] \label{rem:lens-rep}
  \defining{linklensrep}{}
  The parallel splitting structure of monoidal lenses is representable,
  $$\LC{\biobj{A}{B}}{\biobj{X}{Y} ⊗ \biobj{X'}{Y'}} = \LC{\biobj{A}{B}}{\biobj{X ⊗ X'}{Y ⊗ Y'}}.$$
  Lenses over a symmetric monoidal category are known to be monoidal \cite{riley2018categories,hedges2017coherence}, but it remained unexplained why a similar structure was not present in non-symmetric lenses. The contradiction can be solved by noting that both symmetric and non-symmetric lenses are indeed \emph{\promonoidal{}}, even if only symmetric optics provide a representable tensor.
\end{remark}

\begin{remark}[Session notation for lenses]
  We will write $\Send{A} = \left(\biobj{A}{I}\right)$ and $\Get{B} = \left(\biobj{I}{B}\right)$ for the objects of the \produoidalCategory{} of lenses that have a monoidal unit as one of its objects. These are enough to express all objects because $\Send{A} \otimes \Get{B} = \left(\biobj{A}{B}\right)$; and,  moreover, they satisfy the following properties definitionally.
  \begin{align*}
    & ℂ(• ; \Get{A} \triangleleft \Get{B}) ≅
    ℂ(• ; \Get{A} \otimes \Get{B}); && \Send{(A \otimes B)} = \Send{A} \otimes \Send{B}; \\
    & ℂ(• ; \Send{A} \triangleleft \Send{B}) ≅
    ℂ(• ; \Send{A} \otimes \Send{B}); && \Get{(A \otimes B)} = \Get{A} \otimes \Get{B}; \\
    & ℂ(• ; \Send{A} \triangleleft \Get{B}) ≅ ℂ(• ; \Send{A} \otimes \Get{B}).
  \end{align*}
\end{remark}

\begin{proposition}
  \label{prop:sessionNotation}
  Let $(ℂ,\otimes,I)$ be a symmetric \monoidalCategory{}. There exist monoidal functors $(\Send{}) \colon ℂ \to 𝓛ℂ$ and $(\Get{}) \colon ℂ^{op} \to 𝓛ℂ$.
\end{proposition}
\begin{proof}
  See Appendix, \Cref{ax:prop:sessionNotation}.
\end{proof}

\subsection{Protocol Analysis}

Let us go back to our running example (\Cref{diagram:tcp}).
We can now declare that the client and server have the following types, representing the order in which they communicate,
$$\begin{aligned}
& \ClientLogo{} \in
𝓛ℂ\left( \biobj{\mbox{Client}}{\mbox{Client}} \mathbin{;}
  \Send{\mbox{Msg}} \triangleleft
  \Get{\mbox{Msg}} \triangleleft
  \Send{\mbox{Msg}} \right); \\[-0.5em]
& \ServerLogo{} \in
𝓛ℂ\left( \biobj{\mbox{Server}}{\mbox{Server}} \mathbin{;}
  \Get{\mbox{Msg}} \triangleleft
  \Send{\mbox{Msg}} \triangleleft
  \Get{\mbox{Msg}} \right).
\end{aligned}$$
Moreover, we can use the duoidal algebra to compose them. Indeed, tensoring client and server, we get the following codomain type,
$$
(\Send{\mbox{Msg}} ◁
     \Get{\mbox{Msg}} ◁
     \Send{\mbox{Msg}}) ⊗
     (\Get{\mbox{Msg}} ◁
     \Send{\mbox{Msg}} ◁
     \Get{\mbox{Msg}}).$$
We then apply the laxators to mix inputs and outputs, obtaining
$$
(\Send{\mbox{Msg}} ⊗ \Get{\mbox{Msg}}) ◁
         (\Get{\mbox{Msg}} ⊗ \Send{\mbox{Msg}}) ◁
         (\Send{\mbox{Msg}} ⊗ \Get{\mbox{Msg}}),
$$
and we finally apply the unitors to fill the communication holes with noisy channels.
$$\begin{aligned}
  & ψ_2
      \left( \ClientLogo{} \otimes \ServerLogo{} \right) ≺^3_\lambda
      \NOISE^3
  \in
  𝓛ℂ\left( \biobj{\mbox{Client} \otimes \mbox{Server}}{\mbox{Client} \otimes \mbox{Server}} \right).
  \end{aligned}$$

We end up obtaining the protocol as a single morphism $\mbox{Client} \otimes \mbox{Server} \to \mbox{Client} \otimes \mbox{Server}$ in whatever category we are using to program.  Assuming the category of finite stochastic maps, this single morphism represents the distribution over the possible outcomes of the protocol.
Finally, by \dinaturality{}, we can reason over independent parts of the protocol.

\begin{proposition}
  Let $(\scalebox{0.8}{\ClientLogo{}}) = (\mbox{\SYN} ⨾ (\im \otimes \blacksquare) ⨾ \mbox{\ACK} ⨾ (\im \otimes \blacksquare))$.
  The equalities in \Cref{diagram:tcp} are a consequence of the dinaturality of a monoidal lens.
\end{proposition}
\begin{proof}
  We recognize the diagram in \Cref{diagram:tcp} as representing the elements in the following equation.
  \begin{align*}
    & \mbox{\SYN} ⨾ (\im \otimes \blacksquare) ⨾ \mbox{\ACK} ⨾ (\im \otimes \blacksquare) & = \\
    & \mbox{\SYN}^{\ast} ⨾ (\PRJ \otimes \im) ⨾ \blacksquare ⨾ \mbox{\ACK} ⨾ (\im ⊗ \blacksquare) & = \\
    & \mbox{\SYN}^{\ast} ⨾ (\im ⊗ \blacksquare) ⨾ (\PRJ{} ⊗
    \im) ⨾ \mbox{\ACK} ⨾ (\im ⊗ \blacksquare) & = \\
    & \mbox{\SYN}^{\ast} ⨾ (\im ⊗ \blacksquare) ⨾ \mbox{\ACK}^{\ast} ⨾ (\im ⊗ \blacksquare).
  \end{align*}
  In the same way we would apply the \emph{interchange law} in completed morphisms, we have applied dinaturality over $\PRJ{}$.
\end{proof}

\subsection{Cartesian Lenses}
We have worked in full generality, but cartesian lenses are particularly important to applications in game theory \cite{ghani:compositionalgametheory2018} and functional programming \cite{kmett12:lenslibrary,pickering17:profunctoroptics}. We introduce their newly constructed \produoidal{} structure.

\begin{proposition}[Cartesian Lenses]
  \label{prop:cartesianlenses}
  Let $(ℂ,⋅,1)$ be a cartesian monoidal category.
  Its \produoidalCategory{} of lenses is given by the following \profunctors{}.
  \begin{align*}
    𝓛ℂ\left(\biobj{A}{B}; \biobj{X}{Y} \right) &≅
    ℂ(A;X) × ℂ(AY; B), \\
    𝓛ℂ\left(\biobj{A}{B}; \biobj{X}{Y} ◁ \biobj{X'}{Y'} \right) &≅
    ℂ(A;X) × ℂ(AY; X') × ℂ(AYY' ; B), \\
    𝓛ℂ\left(\biobj{A}{B}; \biobj{X}{Y} ⊗ \biobj{X'}{Y'} \right) &≅
    ℂ(A;XX') × ℂ(AYY';B), \\
    𝓛ℂ\left(\biobj{A}{B}\right) &≅
    ℂ(A;B).
  \end{align*}
\end{proposition}

\begin{proof}
  See Appendix, \Cref{ax:prop:cartesianlenses}.
\end{proof}

\section{Conclusions}

\MonoidalContexts{} are an algebra of incomplete processes, commonly generalizing lenses \cite{riley2018categories} and spliced arrows \cite{mellies2022parsing}.
In the same way that the $\pi$-calculus allows input/output channels of an abstract model of computation, monoidal contexts allow input/output communication on arbitrary theories of processes, such as stochastic or partial functions, quantum processes or relational queries.

\MonoidalContexts{} form a normal \produoidal{} category: a highly structured and rich categorical algebra. Moreover, they are the universal such algebra on a monoidal category.
This is good news for applications: the literature on concurrency is rich in frameworks; but the lack of \emph{canonicity} may get us confused when trying to choose, design, or compare among them, as Abramsky \cite{abramsky06:concurrency} has pointed out. %
Precisely characterizing the universal property of a model addresses this concern.
This is also good news for the category theorist: not only is this an example shedding light on a relatively obscure structure; it is a paradigmatic such one.

We rely on two mathematical ideas: \emph{monoidal} and \emph{duoidal} categories on one hand, and \emph{dinaturality} and \emph{profunctorial} structures on the other.
\emph{Monoidal categories}, which could be accidentally dismissed as a toy version of cartesian categories, show that their string diagrams can bootstrap our conceptual understanding of new fundamental process structures, while keeping an abstraction over their implementation that cartesian categories cannot afford.
Duoidal categories are such an example: starting to appear insistently in computer science \cite{shapiro22:duoidal,sigal:23duoidally},
they capture the posetal structure of process dependency and communication.
\emph{Dinaturality}, virtual structures and profunctors, even if sometimes judged arcane, show again that they can canonically capture a notion as concrete as process composition.

\subsection{Further Work}

\textbf{Dependencies.} Shapiro and Spivak \cite{shapiro22:duoidal} prove that normal symmetric duoidal categories with certain limits additionally have the structure of \emph{dependence categories}: they can not only express dependence structures generated by $(◁)$ and $(\otimes)$, but arbitrary poset-mediated dependence structures. \ProduoidalCategories{} are better behaved: the limits always exist, and we only require these are preserved by the coend.

\begin{proposition}
  \label{th:virtualdependence}
  Let $𝕍$ be a normal and $⊗$-symmetric \produoidalCategory{} with coends over $𝕍$ commuting with finite connected limits.
  Then, $[𝕍\op, \mathbf{Set}]$ is a dependence category in the sense of Shapiro and Spivak \cite{shapiro22:duoidal}.
\end{proposition}
\begin{proof}[Proof sketch]
  See Appendix, \Cref{ax:th:virtualdependence}.
\end{proof}

Weakening dependence categories in this way combines the ideas of Shapiro and Spivak \cite{shapiro22:duoidal} with those of Hefford and Kissinger \cite{hefford_spacetime}, who employ virtual objects to deal with the non-existence of tensor products in models of spacetime.

\textbf{Language theory.}
Melliès and Zeilberger \cite{mellies2022parsing} used a multicategorical form of splice-contour adjunction (\Cref{rem:mellies}) to give a novel proof of the Chomsky-Schüt\-zen\-ber\-ger representation theorem, generalized to context-free languages in categories. Our produoidal splice-contour adjunction (\Cref{sec:parallelContext}), combined with recent work on languages of morphisms in monoidal categories \cite{earnshaw22} opens the way for a vertical categorification of the Chomsky-Schüt\-zen\-ber\-ger theorem, which we plan to elaborate in future work.

\textbf{String diagrams for concurrency.}
Nester et al. \cite{nester23:processhistories,boisseaunester:corneringoptics} have recently introduced an alternative description of lenses in terms of \emph{proarrow equipments}, which have a good 2-dimensional syntax \cite{myers16} we can use for send/receive types $(\Send{}/\Get{})$.
We have shown how this structure arises universally in symmetric \monoidalCategories{}. It remains as further work to determine a good 2-dimensional syntax for concurrent programs with \emph{iteration} and \emph{internal/external choice}.

\section{Acknowledgements}

We thank Pawel Sobocinski, Fosco Loregian, Chad Nester and David Spivak for discussion.

Matt Earnshaw and Mario Román were supported by the European Social Fund Estonian IT Academy research measure (project 2014-2020.4.05.19-0001).
James Hefford is supported by University College London and the EPSRC [grant number EP/L015242/1].

\bibliographystyle{alpha}
\bibliography{bibliography.bib}
\newpage

\appendices
\newgeometry{top=1in,bottom=1.25in, left=1.25in, right = 1.25in}
\onecolumn

\ExplSyntaxOn
\clist_map_inline:nn {A,B,C,D,E,F,G,H,I,J,K,L,M,N,O,P,Q,R,S,T,U,V,W,X,Y,Z} {
  \exp_args:Nc \DeclareDocumentCommand {#1} {} {\ensuremath{\textcolor{mygrey}{#1}}}
}
\ExplSyntaxOff

\section{Introduction}

\subsection{Three Way handshake Implementation}
\label{sec:threewayhandshake}
The following can be interpreted as pseudocode using the linear type theory of symmetric monoidal categories~\cite{shulman2016categorical}.
The type theory of symmetric monoidal categories (\Cref{fig:typetheory}) uses declarations such as \verb|(x , y) <- f(a, b, c)| to represent morphisms such as  $f \colon A \otimes B \otimes C \to X \otimes Y$.

\begin{figure}[H]
  \begin{mathpar}
    \infer[\textsc{Gen}]
      {f \in \mathcal{G}(A_1,\dots,A_n;B) \and
      \Gamma_1 \entails x_1 : A_1 \dots \Gamma_n \entails x_n : A_n}
      {\Shuf(\Gamma_1,\dots,\Gamma_n) \entails f(x_1,\dots,x_n) : B} \\
    \infer[\textsc{Pair}]
    {\Gamma_1 \entails x_1 : A_1\ \dots\ \Gamma_n \entails x_n : A_n}
    {\Shuf(\Gamma_1,\dots,\Gamma_n) \entails [x_1,\mydots,x_n] : A_1 \otimes \mydots \otimes A_n}  \and
    \infer[\textsc{Var}]{ }{x : A\entails x : A} \\
    \infer[\textsc{Split}]{\Delta \entails m : A_1 \otimes \dots \otimes A_n \and \Gamma, x_1 : A_1, \dots , x_n : A_n \entails z : C}{ \Shuf(\Gamma,\Delta) \entails [x_1,\dots,x_n]\ \gets\ m\ \textsc{;}\ z : C} \\
  \end{mathpar}
  \label{fig:typetheory}
  \caption{Type theory of symmetric monoidal categories~\cite{shulman2016categorical}.}
  \end{figure}

We can interpret pseudocode as talking about the type theory of monoidal categories. Usually, we will need some extra structure: such as \verb|if-then-else| or explicit functions. 
It has been found in programming that a good level of concreteness for monoidal categories is given by the Kleisli categories of commutative monads, sometimes abstracted by Freyd categories \cite{moggi91,hughes00}, see \cite{roman22} for a comparison with plain monoidal categories and string diagrams.
For convenience, we assume this setting in the following code, but note that it is not strictly necessary, and that a type-theoretic implementation of monoidal categories would work just the same.

The following code inspired by Haskell's do-notation~\cite{hughes00} and it has been tested in the Glasgow Haskell Compiler, version 9.2.5. 
\begin{quote}
  \begin{verbatim}
    syn :: Client ~> (Client, Syn, Ack)
    syn(client) = do
      client <- random
      return (client, client, 0)
  \end{verbatim}
  \begin{verbatim}
    synack :: (Syn, Ack, Server) ~> (Syn, Ack, Server)
    synack(syn, ack, server) = do
      server <- random
      return (if syn == 0 then (0,0,0) else (server, ack+1, server))
  \end{verbatim}
  \begin{verbatim}
    noise :: Noise -> (Syn, Ack) ~> (Syn, Ack)
    noise k (syn,ack) = do
      noise <- binomial k
      return (if noise then (0,0) else (syn,ack))
  \end{verbatim}
  \begin{verbatim}
    ack :: (Client, Syn, Ack) ~> (Client, Syn, Ack)
    ack(client, syn, ack) = do
      return (if client+1 /= ack then (0,0,0) else (client+1, syn+1, client))
  \end{verbatim}
  \begin{verbatim}
    receive :: (Syn, Ack, Server) ~> Server
    receive(syn, ack, server) = do
      return (if server+1 /= ack then 0 else server)
  \end{verbatim}
\end{quote}

We can use the \produoidalCategory{} of lenses to provide a modular description of this protocol.

The programmer will not need to know about produoidal categories: they will be able to define \emph{splits} of a process; they will be able to read the type of the \emph{split} in terms of the send-receive steps of the protocol; they will be able to combine them, and the typechecker should produce an error whenever \dinaturality{} is not respected.
In fact, in the following code, naively combining client and server in a way that does not preserve dinaturality will produce a type error because GHC will not be able to match the types.
We present the description of the protocol, encoding send/receive types.

\begin{quote}
  \begin{verbatim}

  protocol :: 
    Split (Kleisli Distribution) Client Client 
      (Syn, Ack) -- ! 
      (Syn, Ack) -- ?
      (Syn, Ack) -- !
      ()         -- ? 
    -> Split (Kleisli Distribution) Server Server 
      ()         -- ! 
      (Syn, Ack) -- ?
      (Syn, Ack) -- !
      (Syn, Ack) -- ?
    -> (Client, Server) ~> (Client, Server)
  protocol 
    (Split (Kleisli client1) (Kleisli client2) (Kleisli client3)) 
    (Split (Kleisli server1) (Kleisli server2) (Kleisli server3)) 
    (client , server) = do
      (server, ())    <- server1(server)
      (client, (s,a)) <- client1(client)
      (s, a) <- noise 0.1 (s,a)
      (server, (s,a)) <- server2(server, (s,a))
      (s, a) <- noise 0.1 (s,a)
      (client, (s,a)) <- client2(client, (s,a))
      (s, a) <- noise 0.1 (s,a)
      (server)        <- server3(server, (s,a))
      (client)        <- client3(client, ())
      return (client, server)
  \end{verbatim}
\end{quote}

The following \Cref{fig:clientModuleCode} and \Cref{fig:serverModuleCode} show the separate Haskell code for the client and server modules.

\newpage
\begin{figure}
\begin{quote}
  \begin{verbatim}
  client :: Split (Kleisli Distribution) Client Client 
    (Syn, Ack) -- !
    (Syn, Ack) -- ? 
    (Syn, Ack) -- !
    ()         -- ? 
  client = Split {
       
      -- Part 1: Send a SYN message.
      part1 = Kleisli $ \client -> do
          client <- pure 10
          return (client, (client, 0))
      
      -- Part 2: Receive ACK, send ACK.
      , part2 = Kleisli $ \(client, (syn, ack)) -> do
          return (if client+1 /= syn then (0,(0,0)) else (client, (client+1, ack+1)))
  
      -- Part 3: Close protocol.
      , part3 = Kleisli $ \(client, ()) -> do
          return client
        
      }   
  \end{verbatim}
\end{quote}
\caption{Haskell code for the client module.}
\label{fig:clientModuleCode}
\end{figure}

\begin{figure}
  \begin{quote}
    \begin{verbatim}
    server :: Split (Kleisli Distribution) Server Server 
      ()         -- send    ==>
      (Syn, Ack) -- receive <== 
      (Syn, Ack) -- send    ==>
      (Syn, Ack) -- receive <==
    server = Split
    
        -- Part 1: Open protocol.
        { part1 = Kleisli $ \server -> do
            return (server, ())
    
        -- Part 2: Receive SYN and send ACK.
        , part2 = Kleisli $ \(server, (syn, ack)) -> do
            server <- pure 20
            return (if syn == 0 then (0,(0,0)) else (server, (syn+1, server)))
    
        -- Part 3: Receive ACK.
        , part3 = Kleisli $ \(server, (syn, ack)) -> do
            return (if server+1 /= ack then 0 else server)
        }
    \end{verbatim}
  \end{quote}
  \caption{Code for the server module.}
  \label{fig:serverModuleCode}
\end{figure}

\clearpage

\begin{figure}
  \begin{quote}
    \begin{verbatim}
  data Split c a b x y s t where
      Split :: { part1 :: c a       (m , x)
               , part2 :: c (m , y) (n , s)
               , part3 :: c (n , t) b 
               } -> Split c a b x y s t
  
  data Unit c a b where
      Unit :: { unit :: c a b } -> Unit c a b
  
  data Context c a b x y where
      Context :: { partA :: c a (m , x) 
                 , partB :: c (m , y) (m , b) 
                 } -> Context c a b x y

  type (a ~> b) = (a -> Distribution b) 
    \end{verbatim}
  \end{quote}
  \caption{Code describing the profunctors of monoidal lenses.}
  \label{fig:monoidalModuleCode}
\end{figure} \clearpage
\section{Profunctors and virtual structures}
\label{ax:sec:profunctors}
\begin{definition}
  A \emph{profunctor} $(P,≺,≻)$ between two categories $𝔸$ and $𝔹$, written $P(•; •) \colon 𝔸 ⧑ 𝔹$, is a family of sets $P(B; A)$ indexed by objects $𝔸$ and $𝔹$, and endowed with jointly functorial left and right actions of the morphisms of $𝔸$ and $𝔹$, respectively \cite{benabou00,loregian2021}.

  Explicitly, the types of these actions are $(≻) \colon \mathbb{B}(B',B) \times P(B,A) \to P(B',A)$, and $(≺) \colon P(B,A) \times \mathbb{A}(A,A') \to P(B,A')$. These must
  \begin{itemize}
      \item satisfy compatibility, $(f ≻ p) ≺ g = f ≻ (p ≺ g)$,
      \item preserve identities, $id ≻ p = p$, and $p ≺ id = p$,
      \item and preserve compositions, $(p ≺ f) ≺ g = p ≺  (f ⨾ g)$ and $f ≻ (g ≻ p) = (f ⨾ g) ≻ p$.
  \end{itemize}
\end{definition}

\begin{remark}
More succinctly, a \profunctor{} $P \colon 𝔸 ⧑ 𝔹$ is a functor $P \colon 𝔹\op × 𝔸 \to \mathbf{Set}$.
Analogously, a \profunctor{} $P \colon 𝔸 ⧒ 𝔹$ is a functor $P \colon 𝔸\op × 𝔹 \to \mathbf{Set}$, or a \profunctor{} $P \colon 𝔹 ⧑ 𝔸$.\footnote{Notation for profunctors conflicts in the literature. To side-step this problem, we use the symbols $(⧒)$ and $(⧑)$, where $\circ$ marks the contravariant ($\text{op}$) argument. This idea we take from Mike Shulman.} When presented as a family of sets with a pair of actions, profunctors are sometimes called bimodules.
\end{remark}

\begin{theorem}[Yoneda isomorphisms]
  \defining{linkyoneda}{}
  Let $ℂ$ be a category. There exist bijections between the following sets defined by coends. These are natural in the copresheaf $F \colon ℂ \to \mathbf{Set}$, the presheaf $G \colon ℂ\op \to \mathbf{Set}$ and $A \in ℂ$,
  $$∫^{X} ℂ(X;A) × F(X) \overset{y_1}{≅} F(A); \quad
  ∫^{X} ℂ(A;X) × G(X) \overset{y_2}{≅} G(A);$$
  and they are defined by $y_1(f｜α) = F(f)(α)$ and $y_2(g｜β) = G(g)(β)$. These are called \emph{Yoneda reductions} or \emph{Yoneda isomorphisms}, because they appear in the proof of Yoneda lemma. Moreover, any formal diagram constructed out of these reductions, products, identities and compositions commutes.
\end{theorem}

\subsection{Promonads}
\begin{definition}
  \defining{linkpromonad}{}
  A \emph{promonad} $(P, ★, {}^{\circ})$ over a category $ℂ$ is a \profunctor{} $P \colon ℂ ⧒ ℂ$ together with two natural transformations representing inclusion $({}^{\circ}) \colon ℂ(X; Y) \to P(X; Y)$ and multiplication $(★) \colon P(X; Y) × P(Y ; Z) \to P(X; Z)$, and such that
  \begin{itemize}
    \item the left action is premultiplication, $f^{\circ} ★ p = f ≻ p$,
    \item the right action is postmultiplication, $p ★ f^{\circ} = p ≺ f$,
    \item multiplication is dinatural, $p ★ (f ≻ q) = (p ≺ f) ★ q$,
    \item and multiplication is associative, $(p_1 ★ p_2) ★ p_3 = p_1 ★ (p_2 ★ p_3)$.
  \end{itemize}
  
  Equivalently, promonads are monoids in the category of endoprofunctors. Every promonad induces a category, its \emph{Kleisli} category, with the same objects as the original $ℂ$, but with hom-sets given by the promonad, $P(\bullet; \bullet)$.~\cite{roman22}
\end{definition}

\subsection{Multicategories} \label{ax:sec:multicategories}

\textbf{Multicategories.} 
We can explain \promonoidal{} categories in terms of their better-known relatives: \emph{multicategories}. Multicategories can be used to describe (non-necessarily-coherent) decomposition. They contain \emph{multimorphisms}, $X \to Y_0,\dots,Y_n$ that represent a way of decomposing an object $X$ into a list of objects $Y_0, \dots , Y_n$.

\begin{definition}[Multicategory]
\label{def:multicategory}\defining{linkmulticategory}{}
A \emph{multicategory} is a category $ℂ$ endowed  with a set of multimorphisms, $ℂ(X;Y_0, \dots, Y_n)$ for each list of objects $X_0,\dots,X_n,Y$ in $\obj{ℂ}$, and a composition \Cref{fig:multicategoricalcomposition} operation
$$\begin{aligned}
  (⨾)^{n,m}_{Y_k}  \colon
  ℂ(X;Y_0,\dots,Y_n) ×
  ℂ(Y_i;Z_0, \dots, Z_m )
  → ℂ(Z; Y_0,\dots,X_0,\dots,X_m,\dots,Y_m).
\end{aligned}$$
Composition is unital, meaning $id_{X_i} ⨾ f = f ⨾ id_Y$ for any $f$ making the equation formally well-typed. Composition is also \emph{associative}, meaning
$(h ⨾ g) ⨾ f = h ⨾ (g ⨾ f); \mbox{ and } g ⨾ (h ⨾ f) = h ⨾ (g ⨾ f)$
holds whenever it is formally well-typed.
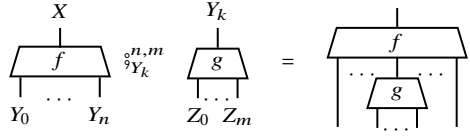
\begin{figure}[ht]
  \centering
\tikzset{every picture/.style={line width=0.75pt}} %

\begin{tikzpicture}[x=0.75pt,y=0.75pt,yscale=-1,xscale=1]
\draw   (10,49.5) -- (14.5,34.5) -- (55.5,34.5) -- (60,49.5) -- cycle ;
\draw    (35,35) -- (35,25) ;
\draw    (55,60) -- (55,50) ;
\draw    (15,60) -- (15,50) ;
\draw  [draw opacity=0] (20,70) -- (26,50) -- (44,50) -- (50,70) -- cycle ;
\draw   (99.5,50.5) -- (104,35.5) -- (125,35.5) -- (129.5,50.5) -- cycle ;
\draw    (104.5,60.5) -- (104.5,50.5) ;
\draw    (124.5,60.5) -- (124.5,50.5) ;
\draw  [draw opacity=0] (99.5,70) -- (105.5,50) -- (123.5,50) -- (129.5,70) -- cycle ;
\draw    (114.5,35) -- (114.5,25) ;
\draw  [draw opacity=0] (70,25) -- (85,25) -- (85,60) -- (70,60) -- cycle ;
\draw  [draw opacity=0] (140,25) -- (160,25) -- (160,60) -- (140,60) -- cycle ;
\draw  [draw opacity=0] (170.5,55) -- (175,40) -- (200.5,40) -- (205,55) -- cycle ;
\draw    (175,75) -- (175,40) ;
\draw   (170,40) -- (174.5,25) -- (235.5,25) -- (240,40) -- cycle ;
\draw    (205,25) -- (205,15) ;
\draw   (190,65.5) -- (194.5,50.5) -- (215.5,50.5) -- (220,65.5) -- cycle ;
\draw    (195,75.5) -- (195,65.5) ;
\draw    (215,75.5) -- (215,65.5) ;
\draw  [draw opacity=0] (205,55) -- (209.5,40) -- (235.5,40) -- (240,55) -- cycle ;
\draw    (205,50) -- (205,40) ;
\draw    (235,75) -- (235,40) ;
\draw  [draw opacity=0] (190,85) -- (196,65) -- (214,65) -- (220,85) -- cycle ;

\draw (35,42) node  [font=\footnotesize]  {$f$};
\draw (35,60) node  [font=\footnotesize,rotate=-180,xscale=-1]  {$\cdots $};
\draw (114.5,60) node  [font=\footnotesize,rotate=-180,xscale=-1]  {$\cdots $};
\draw (114.5,43) node  [font=\footnotesize]  {$g$};
\draw (77.5,42.5) node  [font=\footnotesize]  {$\comp^{n,m}_{Y_k}$};
\draw (150,42.5) node  [font=\footnotesize]  {$=$};
\draw (187.75,47.5) node  [font=\footnotesize,rotate=-180,xscale=-1]  {$\cdots $};
\draw (222.5,47.5) node  [font=\footnotesize,rotate=-180,xscale=-1]  {$\cdots $};
\draw (205,75) node  [font=\footnotesize,rotate=-180,xscale=-1]  {$\cdots $};
\draw (205,32.5) node  [font=\footnotesize]  {$f$};
\draw (205,58) node  [font=\footnotesize]  {$g$};
\draw (35,21.6) node [anchor=south] [inner sep=0.75pt]  [font=\footnotesize]  {$X$};
\draw (15,63.4) node [anchor=north] [inner sep=0.75pt]  [font=\footnotesize]  {$Y_{0}$};
\draw (55,63.4) node [anchor=north] [inner sep=0.75pt]  [font=\footnotesize]  {$Y_{n}$};
\draw (114.5,22.5) node [anchor=south] [inner sep=0.75pt]  [font=\footnotesize]  {$Y_{k}$};
\draw (104.5,63.9) node [anchor=north] [inner sep=0.75pt]  [font=\footnotesize]  {$Z_{0}$};
\draw (124.5,63.9) node [anchor=north] [inner sep=0.75pt]  [font=\footnotesize]  {$Z_{m}$};

\end{tikzpicture}
     \caption{Multicategorical composition.}
    \label{fig:multicategoricalcomposition}
\end{figure}
\end{definition}

\begin{proposition}
Multicategorical composition is dinatural on the object we are composing along. This is to say that composition, $(⨾)^{n,m}_k$, induces a well-defined and dinatural \emph{composition operation} on the coend the variable $Y_k$ we are composing along.
  $$\begin{aligned}
    (⨾)^{n,m}_{•_k}  \colon &
    \left(
    \int^{Y_k \in ℂ}ℂ(X;Y_0,\dots,Y_n) ×
    ℂ(Y_k;Z_0, \dots, Z_m )
    \right)
   → ℂ(Z; Y_0,\dots,X_0,\dots,X_m,\dots,Y_m).
  \end{aligned}$$

  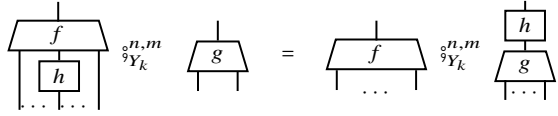
\begin{figure}[ht]
    \centering

\tikzset{every picture/.style={line width=0.75pt}} %

\begin{tikzpicture}[x=0.75pt,y=0.75pt,yscale=-1,xscale=1]
\draw   (19.5,35) -- (24,20) -- (65,20) -- (69.5,35) -- cycle ;
\draw    (44.5,20.5) -- (44.5,10.5) ;
\draw    (65,65) -- (65,35) ;
\draw    (25,65) -- (25,35) ;
\draw  [draw opacity=0] (29.5,55.5) -- (35.5,35.5) -- (53.5,35.5) -- (59.5,55.5) -- cycle ;
\draw   (110,45.5) -- (114.5,30.5) -- (135.5,30.5) -- (140,45.5) -- cycle ;
\draw    (115,55.5) -- (115,45.5) ;
\draw    (135,55.5) -- (135,45.5) ;
\draw    (125,30) -- (125,20) ;
\draw  [draw opacity=0] (79.5,20) -- (94.5,20) -- (94.5,55) -- (79.5,55) -- cycle ;
\draw   (35,40.5) -- (55,40.5) -- (55,55.5) -- (35,55.5) -- cycle ;
\draw    (45,40) -- (45,37.84) -- (45,35) ;
\draw    (45,65) -- (45,55) ;
\draw  [draw opacity=0] (15,69.5) -- (19.5,54.5) -- (45.5,54.5) -- (50,69.5) -- cycle ;
\draw  [draw opacity=0] (35,69.5) -- (39.5,54.5) -- (65.5,54.5) -- (70,69.5) -- cycle ;
\draw  [draw opacity=0] (150.5,20) -- (170.5,20) -- (170.5,55) -- (150.5,55) -- cycle ;
\draw   (180,44.5) -- (184.5,29.5) -- (225.5,29.5) -- (230,44.5) -- cycle ;
\draw    (205,30) -- (205,20) ;
\draw    (225,55) -- (225,45) ;
\draw    (185,55) -- (185,45) ;
\draw  [draw opacity=0] (190,65) -- (196,45) -- (214,45) -- (220,65) -- cycle ;
\draw   (265,50.5) -- (269.5,35.5) -- (290.5,35.5) -- (295,50.5) -- cycle ;
\draw    (270,60) -- (270,50) ;
\draw    (290,60) -- (290,50) ;
\draw  [draw opacity=0] (265,65) -- (271,45) -- (289,45) -- (295,65) -- cycle ;
\draw    (280,35) -- (280,30) ;
\draw  [draw opacity=0] (240,20) -- (255,20) -- (255,55) -- (240,55) -- cycle ;
\draw  [draw opacity=0] (310,20) -- (330,20) -- (330,55) -- (310,55) -- cycle ;
\draw   (270,15) -- (290,15) -- (290,30) -- (270,30) -- cycle ;
\draw    (280,15) -- (280,12.84) -- (280,10) ;

\draw (44.5,27.5) node  [font=\footnotesize]  {$f$};
\draw (125,38) node  [font=\footnotesize]  {$g$};
\draw (87,37.5) node  [font=\footnotesize]  {$\comp^{n,m}_{Y_k}$};
\draw (45,48) node  [font=\footnotesize]  {$h$};
\draw (32.5,62) node  [font=\footnotesize,rotate=-180,xscale=-1]  {$\cdots $};
\draw (52.5,62) node  [font=\footnotesize,rotate=-180,xscale=-1]  {$\cdots $};
\draw (160.5,37.5) node  [font=\footnotesize]  {$=$};
\draw (205,37) node  [font=\footnotesize]  {$f$};
\draw (205,55) node  [font=\footnotesize,rotate=-180,xscale=-1]  {$\cdots $};
\draw (280,55) node  [font=\footnotesize,rotate=-180,xscale=-1]  {$\cdots $};
\draw (280,43) node  [font=\footnotesize]  {$g$};
\draw (247.5,37.5) node  [font=\footnotesize]  {$\comp^{n,m}_{Y_k}$};
\draw (280,22) node  [font=\footnotesize]  {$h$};

\end{tikzpicture}
       \caption{Multicategorical composition is dinatural.}
      \label{fig:multicategoricalcompositiondinatural}
  \end{figure}
\end{proposition}

\begin{proof}
    This is a direct consequence of the associativity of composition for multicategories, inducing an isomorphism.
\end{proof}

\begin{remark}
  A promonoidal category is a multicategory where \emph{dinatural composition} is invertible.
\end{remark}

\textbf{Duomulticategories} describe the interaction between two kinds of decomposition: a \emph{sequential} one and a \emph{parallel} one. We can mix this two ways of decomposing: for instance, we can decompose $X$ sequentially and then decompose each one of its factors in parallel, finally decompose the last one of these sequentially again.
$$ℂ(X; (Y_0 \cdot Y_1),  (Y_2 \cdot (Y_3, Y_4))).$$

\begin{definition}[Duomulticategory]
  A \emph{duomulticategory} is a category $ℂ$ endowed with a set of multimorphisms,
  $ℂ(X ; E(Y_0,\dots,Y_n)),$ for each list of objects $Y_0, \dots, Y_n$ in $\obj{ℂ}$ and each expression $E$ on two monoids. Moreover, it is endowed with a dinatural composition operation
  $$\begin{aligned}
    (⨾)^{n,m}_{Y_k}  \colon &
    \int^{Y_i} ℂ(X;E_1[Y_0,\dots,Y_n]) ×
    ℂ(Y_i; E_2[Z_0, \dots, Z_m] )
    \longrightarrow ℂ(X; E_1[Y_0,\dots,E_2[X_0,\dots,X_m],\dots,Y_m),
  \end{aligned}$$
  and laxators relating sequential and parallel composition,
  $$\begin{aligned}
    ℂ(X; & E_1[Y_0,\dots,((Z_0,Z_1) \cdot (Z_2,Z_3)),\dots,Y_n]) \longrightarrow 
    ℂ(X;  E_1[Y_0,\dots,((Z_0 \cdot Z_2), (Z_1 \cdot Z_3)),\dots,Y_n]).\end{aligned}$$
\end{definition}

\begin{remark}
  In the same sense that a \promonoidalCategory{} is a category where dinatural composition is invertible in a specific sense, a \produoidalCategory{} can be conjectured to be a duomulticategory where dinatural composition is invertible, inducing an isomorphism.
\end{remark}
\section{Sequential Context}
\label{sec:ax:sequentialContext}
\begin{proposition}[From \Cref{prop:contourFunctor}] \label{ax:prop:contourFunctor}
  Contour gives a functor $𝓒 : \pMon → \Cat$.
\end{proposition}
\begin{proof}
  \Cref{defn:contour} defines the action on \promonoidalCategories{}. We define the action on \promonoidalFunctors{}. Given a \promonoidalFunctor{} $F : 𝕍 → 𝕎$, define the functor $𝓒F : 𝓒𝕍 → 𝓒𝕎$ by the following morphism of presentations:
  $$\begin{aligned}
    X^{L} ↦ F(X)^{L};& \ X^{R} ↦ F(X)^{R} \\
    \text{for each $a \in 𝕍(A;N)$,\ }& a_0 : A^L \to A^R  ↦ F_N(a)_0 \\
    \text{for each $b \in 𝕍(X;B)$,\ }& b_0 : B^L \to X^L ↦ F(b)_0; \enspace b_1 : X^R \to B^R ↦ F(b)_1 \\
    \text{for each $c \in 𝕍(C; Y ◁ Z)$,\ }& c_0 : C^L \to Y^L ↦ F_{◁}(c)_0; \enspace c_1 : Y^R \to Z^L ↦ F_{◁}(c)_1; \enspace c_2 : Z^R \to C^R ↦ F_{◁}(c)_2.
    \end{aligned}$$

  It follows from $F: 𝕍 \to 𝕎$ being a \promonoidalFunctor{} that the contour equations of \Cref{defn:contour} hold between the images of generators, so this defines a functor. In particular when $\idFun{_𝕍} : 𝕍 \to 𝕍$ is an identity, it is an identity functor. Let $G : 𝕌 → 𝕍$ be another \promonoidalFunctor{}, then $𝓒(G⨾F) = 𝓒(G)⨾𝓒(F)$ follows from the composition of \promonoidalFunctors{}.
\end{proof}

\begin{proposition}[From \Cref{prop:spliceIsPromonoidal}]
    \label{ax:prop:spliceIsPromonoidal}
    Spliced arrows form a \promonoidalCategory{} with their sequential splits, units, and suitable coherence morphisms.
\end{proposition}
\begin{proof}
    In \Cref{lemma:promonoidalAssociator}, we construct the associator out of \YonedaIsomorphisms{}. In \Cref{lemma:spliceLeftUnitor,lemma:spliceRightUnitor}, we construct both unitors. As they are all constructed with Yoneda isomorphisms, they must satisfy the coherence equations.
\end{proof}

\begin{lemma}[Promonoidal splice associator]
  \label{lemma:promonoidalAssociator}
  We can construct a natural isomorphism,
  $$α \colon \int^{\biobj{U}{V} \in \mathbb{SC}}
    \SCtensor{A}{B}{X}{Y}{U}{V} ×
    \SCtensor{U}{V}{X'}{Y'}{X''}{Y''} \cong
    \int^{\biobj{U}{V} \in \mathbb{SC}}
    \SCtensor{A}{B}{U}{V}{X''}{Y''} ×
    \SCtensor{U}{V}{X}{Y}{X'}{Y'},$$
  exclusively from \YonedaIsomorphisms{}. This isomorphism is defined by stating that $α(\trisplice{f_0}{f_1}{f_2}|\trisplice{g_0}{g_1}{g_2}) = (\trisplice{h_0}{h_1}{h_2}|\trisplice{k_0}{k_1}{k_2})$ if and only if
  $$\quasplice{f_0 ⨾ g_0}{g_1}{g_2 ⨾ f_1}{g_2} = \quasplice{h_0}{h_1 ⨾ k_0}{k_1}{k_2 ⨾ h_2}.$$
\end{lemma}
\begin{proof}
  We will show that both sides of the equation are isomorphic to $ℂ(A;X) × ℂ(Y;X') × ℂ(Y';X'') × ℂ(X'';B)$; that is, the set of quadruples of morphisms $\quasplice{p_0}{p_1}{p_2}{p_3}$ where $p_0 \colon A \to X$, $p_1 \colon Y \to X'$, $p_2 \colon Y' \to X'$ and $p_3 \colon Y'' \to B$.

  Indeed, the following coend calculus computation constructs an isomorphism,
    $$\begin{aligned}
    & \int^{\biobj{U}{V} \in 𝓢{ℂ}}
    \SCtensor{A}{B}{X}{Y}{U}{V} ×
    \SCtensor{U}{V}{X'}{Y'}{X''}{Y''} &\quad
    = & \quad\mbox{(by definition)} \\
    & \int^{\biobj{U}{V} \in 𝓢{ℂ}}
    ℂ(A; X) × ℂ(Y ; U) × ℂ(V; B) × ℂ(U;X') × ℂ(Y' ; X'') × ℂ(Y''; V) &\quad
    = & \quad\mbox{(by definition)} \\
    & \int^{\biobj{U}{V} \in 𝓢{ℂ}}
    ℂ(A; X) × \SChom{Y}{B}{U}{V} × ℂ(U;X') × ℂ(Y' ; X'') × ℂ(Y''; V) &\quad
    \cong & \quad\mbox{(\byYonedaReduction{})} \\
    & ℂ(A; X) × ℂ(Y;X') × ℂ(Y' ; X'') × ℂ(Y''; B),\phantom{\int}
    \end{aligned}$$
  that sends a pair $\trisplice{f_0}{f_1}{f_2}|\trisplice{g_0}{g_1}{g_2}$, quotiented by the equivalence relation generated by $\trisplice{f_0}{f_1 ⨾ n}{m ⨾ f_2}|\trisplice{g_0}{g_1}{g_2} = \trisplice{f_0}{f_1}{f_2}|\trisplice{n ⨾ g_0}{g_1}{m ⨾ g_2}$, to the canonical form $\quasplice{f_0}{f_1 ⨾ g_0}{g_1}{g_2 ⨾ f_2}$.

  In the same way, the following coend calculus computation constructs the second isomorphism,
    \begin{align*}
    & \int^{\biobj{U}{V} \in 𝓢{ℂ}}
    \SCtensor{A}{B}{U}{V}{X''}{Y''} ×
    \SCtensor{U}{V}{X}{Y}{X'}{Y'} 
    & \bydef \\
    & \int^{\biobj{U}{V} \in 𝓢{ℂ}}
    ℂ(A; U) × ℂ(V ; X'') × ℂ(Y''; B) × ℂ(U;X) × ℂ(Y ; X') × ℂ(Y'; V) 
    & \bydef \\
    & \int^{\biobj{U}{V} \in 𝓢{ℂ}}
    \SChom{A}{X''}{U}{V} × ℂ(Y''; B) × ℂ(U;X) × ℂ(Y ; X') × ℂ(Y'; V) 
    & \yo1 \\
    & ℂ(A; X) × ℂ(Y;X') × ℂ(Y' ; X'') × ℂ(Y''; B),\phantom{\int}
    \end{align*}
  that sends a pair $\trisplice{h_0}{h_1}{h_2}|\trisplice{k_0}{k_1}{k_2}$, quotiented by the equivalence relation generated by $\trisplice{h_0 ⨾ n}{m ⨾ h_1}{h_2}|\trisplice{k_0}{k_1}{k_2} = \trisplice{h_0}{h_1}{h_2}|\trisplice{n ⨾ k_0}{k_1}{k_2 ⨾ m}$, to the canonical form $\quasplice{h_0 ⨾ k_0}{k_1}{k_2 ⨾ h_1}{h_2}$.

  In summary, we have that $α(\trisplice{f_0}{f_1}{f_2}|\trisplice{g_0}{g_1}{g_2}) = (\trisplice{h_0}{h_1}{h_2}|\trisplice{k_0}{k_1}{k_2})$ if and only if $$\quasplice{f_0 ⨾ g_0}{g_1}{g_2 ⨾ f_1}{g_2} = \quasplice{h_0}{h_1 ⨾ k_0}{k_1}{k_2 ⨾ h_2},$$
  which is what we wanted to prove.
\end{proof}

\begin{lemma}[Promonoidal splice left unitor]
  \label{lemma:spliceLeftUnitor}
  We can construct a natural isomorphism,
  $$λ \colon \int^{\biobj{U}{V} \in \mathbb{SC}}
  \SCtensor{A}{B}{U}{V}{X}{Y} × \SCunit{U}{V} \cong \SChom{A}{B}{X}{Y},$$
  exclusively from \YonedaIsomorphisms{}. This isomorphism is defined by $λ(\trisplice{f_0}{f_1}{f_2}|g) = \bisplice{f_0 ⨾ g ⨾ f_1}{f_2}$.
\end{lemma}
\begin{proof}
  Indeed, the following coend calculus derivation constructs the isomorphism.
  $$\begin{aligned}
    & \int^{\biobj{U}{V} \in 𝓢{ℂ}} \SCtensor{A}{B}{U}{V}{X}{Y} × \SCunit{U}{V} &\quad
    = & \quad\mbox{(by definition)} \\
    & \int^{\biobj{U}{V} \in 𝓢{ℂ}}
    ℂ(A; U) × ℂ(V ; X) × ℂ(Y ; B) × ℂ(U ; V) &\quad
    = & \quad\mbox{(by definition)} \\
    & \int^{\biobj{U}{V} \in 𝓢{ℂ}}
    \SChom{A}{X}{U}{V} × ℂ(Y ; B) × ℂ(U ; V) &\quad
    \cong & \quad\mbox{(\byYonedaReduction{})} \\
    & ℂ(A; X) × ℂ(Y; B).\phantom{\int}
    \end{aligned}$$
    Thus, it is constructed by a \YonedaIsomorphism{}.
\end{proof}

\begin{lemma}[Promonoidal splice right unitor]
    \label{lemma:spliceRightUnitor}
    We can construct a natural isomorphism,
    $$ρ \colon \int^{\biobj{U}{V} \in \mathbb{SC}}
    \SCtensor{A}{B}{X}{Y}{U}{V} × \SCunit{U}{V} \cong \SChom{A}{B}{X}{Y},$$
    exclusively from \YonedaIsomorphisms{}. This isomorphism is defined by $ρ(\trisplice{f_0}{f_1}{f_2}|g) = \bisplice{f_0}{f_1 ⨾ g ⨾ f_2}$.
\end{lemma}
\begin{proof}
    Indeed, the following coend calculus derivation constructs the isomorphism.
    $$\begin{aligned}
      & \int^{\biobj{U}{V} \in 𝓢{ℂ}}
      \SCtensor{A}{B}{X}{Y}{U}{V} × \SCunit{U}{V} &\quad
      = & \quad\mbox{(by definition)} \\
      & \int^{\biobj{U}{V} \in 𝓢{ℂ}}
      ℂ(A; X) × ℂ(Y ; U) × ℂ(V ; B) × ℂ(U ; V) &\quad
      = & \quad\mbox{(by definition)} \\
      & \int^{\biobj{U}{V} \in 𝓢{ℂ}}
      ℂ(A; X) × \SChom{Y}{B}{U}{V} × ℂ(U ; V) &\quad
      \cong & \quad\mbox{(\byYonedaReduction{})} \\
      & ℂ(A; X) × ℂ(Y; B).\phantom{\int}
      \end{aligned}$$
    Thus, it is constructed by a \YonedaIsomorphism{}.
  \end{proof}

  \begin{proposition}[From \Cref{prop:spliceFunctor}] \label{ax:prop:spliceFunctor}
    Splice gives a functor $\Splice{} : \Cat → \pMon$.
  \end{proposition}
  \begin{proof}
    \Cref{defn:splicedArrows} defines the action on categories. We define the action on functors. Given a functor $F : ℂ → 𝔻$, define the \promonoidalFunctor{} $\Splice{F} : \Splice{ℂ} → \Splice{𝔻}$ by
    \begin{align*}
      &\biobj{A}{B} ↦ \biobj{FA}{FB}, \\
      &\Splice{F}  := F_{A,X} × F_{Y,B} : \Splice{ℂ}(\biobj{A}{B}, \biobj{X}{Y}) → \Splice{𝔻}(\biobj{FA}{FB}, \biobj{FX}{FY}), \\
      &\Splice{F}_{◁} := F_{A,X} × F_{Y,X'} × F_{Y',B} : \Splice{ℂ}(\biobj{A}{B}, \biobj{X}{Y} ◁ \biobj{X'}{Y'}) → \Splice{𝔻}(\biobj{FA}{FB}, \biobj{FX}{FY} ◁ \biobj{FX'}{FY'}), \\
      &\Splice{F}_{N} := F_{A,B} : \Splice{ℂ}(\biobj{A}{B}, N) → \Splice{𝔻}(\biobj{FA}{FB}, N).
    \end{align*}
    It follows from the promonoidal structure on spliced arrows (\Cref{ax:prop:spliceIsPromonoidal}) that this preserves coherence maps. If $\idFun{_ℂ} : ℂ → ℂ$ is an identity functor, then it defines the identity $\idFun{_{\Splice{ℂ}}}$, which has underlying functor the identity and identity natural transformations. If $G : 𝔹 → ℂ$ is another functor, then $\Splice{(G ⨾ F)} = \Splice{G} ⨾ \Splice{F}$ follows from composition of functors.
  \end{proof}

  \begin{theorem}[From \Cref{th:catpromadj}] \label{ax:th:catpromadj}
    There exists an adjunction between categories and \promonoidalCategories{}, where the contour of a promonoidal is the left adjoint, and the splice category is the right adjoint.
  \end{theorem}
  \begin{proof}
    Let $ℂ$ be a category and let $𝔹$ be a \promonoidalCategory{}. We will show that the \promonoidalFunctors{} $𝔹 \to 𝓢{ℂ}$ are in natural correspondence with the functors $𝓒𝔹 \to ℂ$.
    We first observe that the category $𝓒𝔹$ is freely presented; thus, a functor $𝓒𝔹 \to ℂ$ amounts to a choice of some objects and some morphisms in $ℂ$ satisfying some equations. Explicitly, by the definition of \contour{}, a functor $𝓒𝔹 \to ℂ$ amounts to
    \begin{itemize}
        \item for each $X \in \obj{𝔹}$, a choice of objects $X^L, X^R \in \obj{ℂ}$;
        \item for each element $a \in 𝔹(X)$, a choice of morphisms $a_0 \in ℂ(X^L,X^R)$;
        \item for each morphism $a \in 𝔹(A;X)$, a choice of morphisms $a_0 \in ℂ(A^L;X^L)$ and $a_1 \in ℂ(X^R;A^R)$;
        \item for each split $a \in ℂ(A; X ◁ Y)$, a choice of morphisms $a_0 \in ℂ(A^L;X^L)$, $a_1 \in ℂ(X^R;Y^L)$ and $a_2 \in ℂ(Y^R;A^R)$;
        \item the choice must be such that $α(a \mathbin{|} b) = (c \mathbin{|} d)$ implies $a_0 = c_0 ⨾ d_0$; $a_1 ⨾ b_0 = d_1$ and $b_1 = d_2 ⨾ c_1$; $a_2 ⨾ b_2 = c_2$;
        \item the choice must be such that $ρ(a \mathbin{|} b) = c = \lambda(d \mathbin{|} e)$ implies $a_0 = c_0 = d_0 ⨾ e_0 ⨾ d_1$ and $a_1 ⨾ b_0 ⨾ a_2 = c_1 = d_2$.
    \end{itemize}

    On the other hand, a promonoidal functor $𝔹 \to 𝓢{ℂ}$, also amounts to
    \begin{itemize}
        \item for each $X \in \obj{𝔹}$, an object $FX = (X^L, X^R) \in \obj{𝓢{ℂ}}$, which is a pair of objects of $\obj{ℂ}$;
        \item for each element $a \in 𝔹(X)$, a morphism $F(a) = a_0 \in 𝓢{ℂ}(FX)$;
        \item for each element $a \in 𝔹(A;X)$, a splice $F(a) = \bisplice{a_0}{a_1} \in 𝓢{ℂ}(FA;FX)$;
        \item for each element $a \in 𝔹(A;X ◁ Y)$, a splice $F(a) = \trisplice{a_0}{a_1}{a_2} \in 𝓢{ℂ}(FA;FX ◁ FY)$;
        \item preserving associativity, with $α(a \mid b) = (c \mid d)$ implying $α(F(a) \mid F(b)) = F(c) \mid F(d)$;
        \item preserving unitality, with $ρ(a \mid b) = c = \lambda(d \mid e)$ implying $ρ(F(a) \mid F(b)) = F(c) = \lambda(F(d) \mid F(e))$;
    \end{itemize}
    by the definition of \splice{}, its associativity and unitality, the structure on each one of these points is exactly equal.
  \end{proof}

\subsection{Spliced arrow multicategory}
As a consequence of the previous discussion, the n-morphisms are the sequences of arrows in $ℂ$ separated by $n$ gaps; the sequence of arrows goes from $A$ to $B$, with holes typed by $\{ X_i, Y_i \}_{i \in [1,\dots,n]}$. In other words,

$$\begin{aligned}
  \Splice{ℂ} &\left(  \biobj{A}{B}\mathrel{;}
\biobj{X_1}{Y_1}⊗ \dots ⊗ \biobj{X_n}{Y_n} \right) =
ℂ(A;X_1) × \left( \prod_{k=1}^{n-1} ℂ(Y_k, X_{k+1})\right) × ℂ(Y_n,B).
\end{aligned}$$

Composition in the multicategory is defined by substitution of a spliced arrow into one of the gaps of the second; the identity is just $\id{A} - \id{B}$, the spliced arrow with a single gap typed by $(A,B)$.

\begin{proposition}
The multicategory of spliced arrows, $\Splice{ℂ}$, is precisely the promonoidal category induced by the duality $ℂ\op \dashv ℂ$ in the monoidal bicategory of profunctors: a promonoidal category over $ℂ \times ℂ\op$.
\end{proposition}

\clearpage %
\section{Parallel-Sequential Context}
\label{sec:ax:parallelContext}

\subsection{Monoidal Contour}
\label{ax:sec:monoidalContour}

\begin{definition}[{{Monoidal contour, from \Cref{def:monoidalContour}}}]
  \label{ax:def:monoidalContour}
    The \emph{contour} of a \produoidalCategory{} $𝔹$ is the monoidal category $𝓓𝔹$ that has two objects, $X^L$ (left-handed) and $X^R$ (right-handed), for each object $X \in 𝔹_{obj}$; and has arrows those that arise from \emph{contouring} both sequential and parallel decompositions of the promonoidal category.
    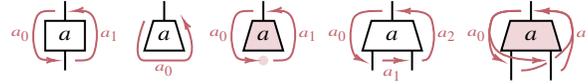
\begin{figure}[ht]
      \centering

\tikzset{every picture/.style={line width=0.75pt}} %

\begin{tikzpicture}[x=0.75pt,y=0.75pt,yscale=-1,xscale=1]
\draw  [fill={rgb, 255:red, 191; green, 97; blue, 106 }  ,fill opacity=0.25 ] (190,35) -- (194.5,20) -- (205.5,20) -- (210,35) -- cycle ;
\draw    (200,20) -- (200,10) ;
\draw  [draw opacity=0] (190,35) -- (210,35) -- (210,55) -- (190,55) -- cycle ;
\draw    (100,20) -- (100,10) ;
\draw   (90,20) -- (110,20) -- (110,35) -- (90,35) -- cycle ;
\draw    (100,45) -- (100,35) ;
\draw [color={rgb, 255:red, 191; green, 97; blue, 106 }  ,draw opacity=1 ]   (92.96,39.93) .. controls (87.06,39.68) and (86.32,39.14) .. (85,35) .. controls (83.54,30.42) and (83.57,25.58) .. (85,20) .. controls (86.43,14.42) and (89.67,15.06) .. (95,15) ;
\draw [shift={(95,40)}, rotate = 181.8] [color={rgb, 255:red, 191; green, 97; blue, 106 }  ,draw opacity=1 ][line width=0.75]    (4.37,-1.32) .. controls (2.78,-0.56) and (1.32,-0.12) .. (0,0) .. controls (1.32,0.12) and (2.78,0.56) .. (4.37,1.32)   ;
\draw [color={rgb, 255:red, 191; green, 97; blue, 106 }  ,draw opacity=1 ]   (107.04,14.86) .. controls (111.32,14.74) and (113.73,16.02) .. (115,20) .. controls (116.46,24.58) and (116.43,29.42) .. (115,35) .. controls (113.57,40.58) and (110.33,39.94) .. (105,40) ;
\draw [shift={(105,15)}, rotate = 354.07] [color={rgb, 255:red, 191; green, 97; blue, 106 }  ,draw opacity=1 ][line width=0.75]    (4.37,-1.32) .. controls (2.78,-0.56) and (1.32,-0.12) .. (0,0) .. controls (1.32,0.12) and (2.78,0.56) .. (4.37,1.32)   ;
\draw  [color={rgb, 255:red, 0; green, 0; blue, 0 }  ,draw opacity=1 ][fill={rgb, 255:red, 191; green, 97; blue, 106 }  ,fill opacity=0.25 ] (320,35) -- (324.5,20) -- (345.5,20) -- (350,35) -- cycle ;
\draw    (325,50) -- (325,35) ;
\draw    (335,20) -- (335,10) ;
\draw [color={rgb, 255:red, 191; green, 97; blue, 106 }  ,draw opacity=1 ]   (319.17,39.65) .. controls (317.5,39.04) and (316.26,38.95) .. (315,35) .. controls (313.54,30.42) and (313.57,25.58) .. (315,20) .. controls (316.43,14.42) and (324.67,15.06) .. (330,15) ;
\draw [shift={(320,40)}, rotate = 205.96] [color={rgb, 255:red, 191; green, 97; blue, 106 }  ,draw opacity=1 ][line width=0.75]    (4.37,-1.32) .. controls (2.78,-0.56) and (1.32,-0.12) .. (0,0) .. controls (1.32,0.12) and (2.78,0.56) .. (4.37,1.32)   ;
\draw [color={rgb, 255:red, 191; green, 97; blue, 106 }  ,draw opacity=1 ]   (350,45) .. controls (352.43,44.87) and (353.54,39.58) .. (355,35) .. controls (356.46,30.42) and (356.43,25.58) .. (355,20) .. controls (353.73,15.06) and (347.13,15) .. (341.94,15) ;
\draw [shift={(340,15)}, rotate = 0.66] [color={rgb, 255:red, 191; green, 97; blue, 106 }  ,draw opacity=1 ][line width=0.75]    (4.37,-1.32) .. controls (2.78,-0.56) and (1.32,-0.12) .. (0,0) .. controls (1.32,0.12) and (2.78,0.56) .. (4.37,1.32)   ;
\draw [color={rgb, 255:red, 191; green, 97; blue, 106 }  ,draw opacity=1 ]   (337.82,39.94) .. controls (318.45,39.26) and (316.4,34.42) .. (315,30) .. controls (313.54,25.42) and (313.57,25.58) .. (315,20) .. controls (316.43,14.42) and (324.67,15.06) .. (330,15) ;
\draw [shift={(340,40)}, rotate = 181.03] [color={rgb, 255:red, 191; green, 97; blue, 106 }  ,draw opacity=1 ][line width=0.75]    (4.37,-1.32) .. controls (2.78,-0.56) and (1.32,-0.12) .. (0,0) .. controls (1.32,0.12) and (2.78,0.56) .. (4.37,1.32)   ;
\draw [color={rgb, 255:red, 191; green, 97; blue, 106 }  ,draw opacity=1 ]   (330,45) .. controls (341.21,45.37) and (353.54,34.58) .. (355,30) .. controls (356.46,25.42) and (356.43,25.58) .. (355,20) .. controls (353.73,15.06) and (347.13,15) .. (341.94,15) ;
\draw [shift={(340,15)}, rotate = 0.66] [color={rgb, 255:red, 191; green, 97; blue, 106 }  ,draw opacity=1 ][line width=0.75]    (4.37,-1.32) .. controls (2.78,-0.56) and (1.32,-0.12) .. (0,0) .. controls (1.32,0.12) and (2.78,0.56) .. (4.37,1.32)   ;
\draw  [draw opacity=0][fill={rgb, 255:red, 255; green, 255; blue, 255 }  ,fill opacity=1 ] (342.5,40) .. controls (342.5,38.62) and (343.62,37.5) .. (345,37.5) .. controls (346.38,37.5) and (347.5,38.62) .. (347.5,40) .. controls (347.5,41.38) and (346.38,42.5) .. (345,42.5) .. controls (343.62,42.5) and (342.5,41.38) .. (342.5,40) -- cycle ;
\draw [color={rgb, 255:red, 191; green, 97; blue, 106 }  ,draw opacity=1 ]   (192.96,39.93) .. controls (187.06,39.68) and (186.32,39.14) .. (185,35) .. controls (183.54,30.42) and (183.57,25.58) .. (185,20) .. controls (186.43,14.42) and (189.67,15.06) .. (195,15) ;
\draw [shift={(195,40)}, rotate = 181.8] [color={rgb, 255:red, 191; green, 97; blue, 106 }  ,draw opacity=1 ][line width=0.75]    (4.37,-1.32) .. controls (2.78,-0.56) and (1.32,-0.12) .. (0,0) .. controls (1.32,0.12) and (2.78,0.56) .. (4.37,1.32)   ;
\draw [color={rgb, 255:red, 191; green, 97; blue, 106 }  ,draw opacity=1 ]   (207.04,14.86) .. controls (211.32,14.74) and (213.73,16.02) .. (215,20) .. controls (216.46,24.58) and (216.43,29.42) .. (215,35) .. controls (213.57,40.58) and (210.33,39.94) .. (205,40) ;
\draw [shift={(205,15)}, rotate = 354.07] [color={rgb, 255:red, 191; green, 97; blue, 106 }  ,draw opacity=1 ][line width=0.75]    (4.37,-1.32) .. controls (2.78,-0.56) and (1.32,-0.12) .. (0,0) .. controls (1.32,0.12) and (2.78,0.56) .. (4.37,1.32)   ;
\draw  [draw opacity=0][fill={rgb, 255:red, 191; green, 97; blue, 106 }  ,fill opacity=0.25 ] (197.5,40) .. controls (197.5,38.62) and (198.62,37.5) .. (200,37.5) .. controls (201.38,37.5) and (202.5,38.62) .. (202.5,40) .. controls (202.5,41.38) and (201.38,42.5) .. (200,42.5) .. controls (198.62,42.5) and (197.5,41.38) .. (197.5,40) -- cycle ;
\draw    (345,50) -- (345,35) ;
\draw   (140,35) -- (144.5,20) -- (155.5,20) -- (160,35) -- cycle ;
\draw    (150,20) -- (150,10) ;
\draw [color={rgb, 255:red, 191; green, 97; blue, 106 }  ,draw opacity=1 ]   (140,20) .. controls (136.32,28.79) and (134.61,39.97) .. (140,40) .. controls (145.39,40.03) and (154.61,39.97) .. (160,40) .. controls (165.06,40.02) and (165.31,30.69) .. (160.9,21.71) ;
\draw [shift={(160,20)}, rotate = 60.63] [color={rgb, 255:red, 191; green, 97; blue, 106 }  ,draw opacity=1 ][line width=0.75]    (4.37,-1.32) .. controls (2.78,-0.56) and (1.32,-0.12) .. (0,0) .. controls (1.32,0.12) and (2.78,0.56) .. (4.37,1.32)   ;
\draw  [draw opacity=0] (140,35) -- (160,35) -- (160,55) -- (140,55) -- cycle ;
\draw   (250,35) -- (254.5,20) -- (275.5,20) -- (280,35) -- cycle ;
\draw    (255,45) -- (255,35) ;
\draw    (275,45) -- (275,35) ;
\draw    (265,20) -- (265,10) ;
\draw [color={rgb, 255:red, 191; green, 97; blue, 106 }  ,draw opacity=1 ]   (248.13,39.28) .. controls (246.96,38.82) and (245.98,38.09) .. (245,35) .. controls (243.54,30.42) and (243.57,25.58) .. (245,20) .. controls (246.43,14.42) and (254.67,15.06) .. (260,15) ;
\draw [shift={(250,40)}, rotate = 205.96] [color={rgb, 255:red, 191; green, 97; blue, 106 }  ,draw opacity=1 ][line width=0.75]    (4.37,-1.32) .. controls (2.78,-0.56) and (1.32,-0.12) .. (0,0) .. controls (1.32,0.12) and (2.78,0.56) .. (4.37,1.32)   ;
\draw [color={rgb, 255:red, 191; green, 97; blue, 106 }  ,draw opacity=1 ]   (260,40) -- (268,40) ;
\draw [shift={(270,40)}, rotate = 180] [color={rgb, 255:red, 191; green, 97; blue, 106 }  ,draw opacity=1 ][line width=0.75]    (4.37,-1.32) .. controls (2.78,-0.56) and (1.32,-0.12) .. (0,0) .. controls (1.32,0.12) and (2.78,0.56) .. (4.37,1.32)   ;
\draw [color={rgb, 255:red, 191; green, 97; blue, 106 }  ,draw opacity=1 ]   (280,40) .. controls (282.43,39.87) and (283.54,39.58) .. (285,35) .. controls (286.46,30.42) and (286.43,25.58) .. (285,20) .. controls (283.73,15.06) and (277.13,15) .. (271.94,15) ;
\draw [shift={(270,15)}, rotate = 0.66] [color={rgb, 255:red, 191; green, 97; blue, 106 }  ,draw opacity=1 ][line width=0.75]    (4.37,-1.32) .. controls (2.78,-0.56) and (1.32,-0.12) .. (0,0) .. controls (1.32,0.12) and (2.78,0.56) .. (4.37,1.32)   ;

\draw (200,27.5) node  [font=\footnotesize]  {$a$};
\draw (100,27.5) node  [font=\footnotesize]  {$a$};
\draw (77.5,27.5) node  [font=\tiny,color={rgb, 255:red, 191; green, 97; blue, 106 }  ,opacity=1 ]  {$a_{0}$};
\draw (122.5,27.5) node  [font=\tiny,color={rgb, 255:red, 191; green, 97; blue, 106 }  ,opacity=1 ]  {$a_{1}$};
\draw (335,27.5) node  [font=\footnotesize]  {$a$};
\draw (307.5,27.5) node  [font=\tiny,color={rgb, 255:red, 191; green, 97; blue, 106 }  ,opacity=1 ]  {$a_{0}$};
\draw (362.5,27.5) node  [font=\tiny,color={rgb, 255:red, 191; green, 97; blue, 106 }  ,opacity=1 ]  {$a_{1}$};
\draw (222.5,27.5) node  [font=\tiny,color={rgb, 255:red, 191; green, 97; blue, 106 }  ,opacity=1 ]  {$a_{1}$};
\draw (177.5,27.5) node  [font=\tiny,color={rgb, 255:red, 191; green, 97; blue, 106 }  ,opacity=1 ]  {$a_{0}$};
\draw (150,27.5) node  [font=\footnotesize]  {$a$};
\draw (150,45) node  [font=\tiny,color={rgb, 255:red, 191; green, 97; blue, 106 }  ,opacity=1 ]  {$a_{0}$};
\draw (265,27.5) node  [font=\footnotesize]  {$a$};
\draw (237.5,27.5) node  [font=\tiny,color={rgb, 255:red, 191; green, 97; blue, 106 }  ,opacity=1 ]  {$a_{0}$};
\draw (265,47.5) node  [font=\tiny,color={rgb, 255:red, 191; green, 97; blue, 106 }  ,opacity=1 ]  {$a_{1}$};
\draw (292.5,27.5) node  [font=\tiny,color={rgb, 255:red, 191; green, 97; blue, 106 }  ,opacity=1 ]  {$a_{2}$};

\end{tikzpicture}
       \caption{Generators of the monoidal category of contours.}
      \label{ax:fig:monoidal-contour}
    \end{figure}

    Specifically, it is freely presented by
    \emph{(i)} a pair of morphisms $a_0 \in 𝓓𝔹(A^L; X^L)$, $a_1 \in 𝓓𝔹(X^R;A^R)$ for each morphism $a \in 𝔹(A;X)$; 
    \emph{(ii)} a morphism $a_0 \in 𝓓𝔹(A^L; A^R)$, for each \sequentialUnit{} $a \in ℂ(A;N)$; 
    \emph{(iii)} a pair of morphisms $a_0 \in 𝓓𝔹(A^L; I)$ and $a_0 \in 𝓓𝔹(I; A^R)$, for each \parallelUnit{} $a \in 𝔹(A;I)$;
    \emph{(iv)} a triple of morphisms 
    $a_0 \in 𝓓𝔹(A^L; X^L)$, $a_1 \in 𝓓𝔹(X^R; Y^L)$, $a_2 \in 𝓓𝔹(Y^R; A^R)$ for each \sequentialSplit{} $a \in 𝔹(A;X ◁ Y)$; and 
    \emph{(v)} a pair of morphisms $a_0 \in 𝓓𝔹(A^L; X^L ⊗ Y^L)$ and $a_1 \in 𝓓𝔹(X^R ⊗ Y^R; A^R)$ for each \parallelSplit{} $a \in 𝔹(A; X ⊗ Y)$, see \Cref{ax:fig:monoidal-contour}.   

    For each equality $a ⨾^2 b = c ⨾^1 d$, we impose the equations $a_0 = c_0 ⨾ d_0$; $a_1 ⨾ b_0 = d_1$ and $b_1 = d_2 ⨾ c_1$; $a_2 ⨾ b_2 = c_2$. For each equality $a ⨾^2 b = c = d ⨾^1 e$, we impose $a_0 = c_0 = d_0 ⨾ e_0 ⨾ d_1$ and $a_1 ⨾ b_0 ⨾ a_2 = c_1 = d_2$.
    These follow from \Cref{fig:promonoidalContourCompose}.

    For each application of associativity, $α(a ⨾_1 b) = c ⨾_2 d$, we impose the equations $a_0 ⨾ (b_0 ⊗ \im) = c_0 ⨾ (\im ⊗ d_0)$ and $(b_1 ⊗ \im) ⨾ a_1 = (\im ⊗ d_1) ⨾ c_1$. These follow from \Cref{ax:fig:monoidalContourAssociativity}.

    \begin{figure}[ht]
      \centering

\tikzset{every picture/.style={line width=0.75pt}} %

\begin{tikzpicture}[x=0.75pt,y=0.75pt,yscale=-1,xscale=1]
\draw  [color={rgb, 255:red, 0; green, 0; blue, 0 }  ,draw opacity=1 ][fill={rgb, 255:red, 191; green, 97; blue, 106 }  ,fill opacity=0.25 ] (35,75) -- (39.5,60) -- (60.5,60) -- (65,75) -- cycle ;
\draw    (40,90) -- (40,75) ;
\draw [color={rgb, 255:red, 191; green, 97; blue, 106 }  ,draw opacity=1 ]   (34.17,79.65) .. controls (32.5,79.04) and (31.26,78.95) .. (30,75) .. controls (28.54,70.42) and (28.57,65.58) .. (30,60) .. controls (31.43,54.42) and (39.67,55.06) .. (45,55) ;
\draw [shift={(35,80)}, rotate = 205.96] [color={rgb, 255:red, 191; green, 97; blue, 106 }  ,draw opacity=1 ][line width=0.75]    (4.37,-1.32) .. controls (2.78,-0.56) and (1.32,-0.12) .. (0,0) .. controls (1.32,0.12) and (2.78,0.56) .. (4.37,1.32)   ;
\draw [color={rgb, 255:red, 191; green, 97; blue, 106 }  ,draw opacity=1 ]   (65,85) .. controls (67.43,84.87) and (68.54,79.58) .. (70,75) .. controls (71.46,70.42) and (71.43,65.58) .. (70,60) .. controls (68.73,55.06) and (62.13,55) .. (56.94,55) ;
\draw [shift={(55,55)}, rotate = 0.66] [color={rgb, 255:red, 191; green, 97; blue, 106 }  ,draw opacity=1 ][line width=0.75]    (4.37,-1.32) .. controls (2.78,-0.56) and (1.32,-0.12) .. (0,0) .. controls (1.32,0.12) and (2.78,0.56) .. (4.37,1.32)   ;
\draw [color={rgb, 255:red, 191; green, 97; blue, 106 }  ,draw opacity=1 ]   (45,85) .. controls (56.21,85.37) and (68.54,74.58) .. (70,70) .. controls (71.46,65.42) and (71.43,65.58) .. (70,60) .. controls (68.73,55.06) and (62.13,55) .. (56.94,55) ;
\draw [shift={(55,55)}, rotate = 0.66] [color={rgb, 255:red, 191; green, 97; blue, 106 }  ,draw opacity=1 ][line width=0.75]    (4.37,-1.32) .. controls (2.78,-0.56) and (1.32,-0.12) .. (0,0) .. controls (1.32,0.12) and (2.78,0.56) .. (4.37,1.32)   ;
\draw  [draw opacity=0][fill={rgb, 255:red, 255; green, 255; blue, 255 }  ,fill opacity=1 ] (57.5,80) .. controls (57.5,78.62) and (58.62,77.5) .. (60,77.5) .. controls (61.38,77.5) and (62.5,78.62) .. (62.5,80) .. controls (62.5,81.38) and (61.38,82.5) .. (60,82.5) .. controls (58.62,82.5) and (57.5,81.38) .. (57.5,80) -- cycle ;
\draw    (60,90) -- (60,75) ;
\draw  [draw opacity=0][fill={rgb, 255:red, 255; green, 255; blue, 255 }  ,fill opacity=1 ] (37,78) .. controls (37,76.9) and (38.34,76) .. (40,76) .. controls (41.66,76) and (43,76.9) .. (43,78) .. controls (43,79.1) and (41.66,80) .. (40,80) .. controls (38.34,80) and (37,79.1) .. (37,78) -- cycle ;
\draw [color={rgb, 255:red, 191; green, 97; blue, 106 }  ,draw opacity=1 ]   (52.82,79.94) .. controls (33.45,79.26) and (31.4,74.42) .. (30,70) .. controls (28.54,65.42) and (28.57,65.58) .. (30,60) .. controls (31.43,54.42) and (39.67,55.06) .. (45,55) ;
\draw [shift={(55,80)}, rotate = 181.03] [color={rgb, 255:red, 191; green, 97; blue, 106 }  ,draw opacity=1 ][line width=0.75]    (4.37,-1.32) .. controls (2.78,-0.56) and (1.32,-0.12) .. (0,0) .. controls (1.32,0.12) and (2.78,0.56) .. (4.37,1.32)   ;
\draw [color={rgb, 255:red, 191; green, 97; blue, 106 }  ,draw opacity=1 ]   (67.82,39.94) .. controls (48.45,39.26) and (46.4,34.42) .. (45,30) .. controls (43.54,25.42) and (43.57,25.58) .. (45,20) .. controls (46.43,14.42) and (54.67,15.06) .. (60,15) ;
\draw [shift={(70,40)}, rotate = 181.03] [color={rgb, 255:red, 191; green, 97; blue, 106 }  ,draw opacity=1 ][line width=0.75]    (4.37,-1.32) .. controls (2.78,-0.56) and (1.32,-0.12) .. (0,0) .. controls (1.32,0.12) and (2.78,0.56) .. (4.37,1.32)   ;
\draw    (55,36) .. controls (55.03,50.73) and (50.03,44.73) .. (50,60) ;
\draw  [color={rgb, 255:red, 0; green, 0; blue, 0 }  ,draw opacity=1 ][fill={rgb, 255:red, 191; green, 97; blue, 106 }  ,fill opacity=0.25 ] (50,35) -- (54.5,20) -- (75.5,20) -- (80,35) -- cycle ;
\draw [color={rgb, 255:red, 191; green, 97; blue, 106 }  ,draw opacity=1 ]   (49.17,39.65) .. controls (47.5,39.04) and (46.26,38.95) .. (45,35) .. controls (43.54,30.42) and (43.57,25.58) .. (45,20) .. controls (46.43,14.42) and (54.67,15.06) .. (60,15) ;
\draw [shift={(50,40)}, rotate = 205.96] [color={rgb, 255:red, 191; green, 97; blue, 106 }  ,draw opacity=1 ][line width=0.75]    (4.37,-1.32) .. controls (2.78,-0.56) and (1.32,-0.12) .. (0,0) .. controls (1.32,0.12) and (2.78,0.56) .. (4.37,1.32)   ;
\draw [color={rgb, 255:red, 191; green, 97; blue, 106 }  ,draw opacity=1 ]   (80,45) .. controls (82.43,44.87) and (83.54,39.58) .. (85,35) .. controls (86.46,30.42) and (86.43,25.58) .. (85,20) .. controls (83.73,15.06) and (77.13,15) .. (71.94,15) ;
\draw [shift={(70,15)}, rotate = 0.66] [color={rgb, 255:red, 191; green, 97; blue, 106 }  ,draw opacity=1 ][line width=0.75]    (4.37,-1.32) .. controls (2.78,-0.56) and (1.32,-0.12) .. (0,0) .. controls (1.32,0.12) and (2.78,0.56) .. (4.37,1.32)   ;
\draw [color={rgb, 255:red, 191; green, 97; blue, 106 }  ,draw opacity=1 ]   (55,45) .. controls (75.36,44.69) and (83.54,34.58) .. (85,30) .. controls (86.46,25.42) and (86.43,25.58) .. (85,20) .. controls (83.73,15.06) and (77.13,15) .. (71.94,15) ;
\draw [shift={(70,15)}, rotate = 0.66] [color={rgb, 255:red, 191; green, 97; blue, 106 }  ,draw opacity=1 ][line width=0.75]    (4.37,-1.32) .. controls (2.78,-0.56) and (1.32,-0.12) .. (0,0) .. controls (1.32,0.12) and (2.78,0.56) .. (4.37,1.32)   ;
\draw    (65,20) -- (65,10) ;
\draw  [draw opacity=0][fill={rgb, 255:red, 255; green, 255; blue, 255 }  ,fill opacity=1 ] (52,38) .. controls (52,36.9) and (53.34,36) .. (55,36) .. controls (56.66,36) and (58,36.9) .. (58,38) .. controls (58,39.1) and (56.66,40) .. (55,40) .. controls (53.34,40) and (52,39.1) .. (52,38) -- cycle ;
\draw [color={rgb, 255:red, 191; green, 97; blue, 106 }  ,draw opacity=1 ]   (67.82,39.94) .. controls (48.45,39.26) and (46.4,34.42) .. (45,30) .. controls (43.54,25.42) and (43.57,25.58) .. (45,20) .. controls (46.43,14.42) and (54.67,15.06) .. (60,15) ;
\draw [shift={(70,40)}, rotate = 181.03] [color={rgb, 255:red, 191; green, 97; blue, 106 }  ,draw opacity=1 ][line width=0.75]    (4.37,-1.32) .. controls (2.78,-0.56) and (1.32,-0.12) .. (0,0) .. controls (1.32,0.12) and (2.78,0.56) .. (4.37,1.32)   ;
\draw    (75,35) .. controls (75.03,49.73) and (80.03,44.73) .. (80,60) ;
\draw    (80,90) -- (80,60) ;
\draw  [color={rgb, 255:red, 0; green, 0; blue, 0 }  ,draw opacity=1 ][fill={rgb, 255:red, 191; green, 97; blue, 106 }  ,fill opacity=0.25 ] (140,75) -- (144.5,60) -- (165.5,60) -- (170,75) -- cycle ;
\draw    (145,90) -- (145,75) ;
\draw [color={rgb, 255:red, 191; green, 97; blue, 106 }  ,draw opacity=1 ]   (139.17,79.65) .. controls (137.5,79.04) and (136.26,78.95) .. (135,75) .. controls (133.54,70.42) and (133.57,65.58) .. (135,60) .. controls (136.43,54.42) and (144.67,55.06) .. (150,55) ;
\draw [shift={(140,80)}, rotate = 205.96] [color={rgb, 255:red, 191; green, 97; blue, 106 }  ,draw opacity=1 ][line width=0.75]    (4.37,-1.32) .. controls (2.78,-0.56) and (1.32,-0.12) .. (0,0) .. controls (1.32,0.12) and (2.78,0.56) .. (4.37,1.32)   ;
\draw [color={rgb, 255:red, 191; green, 97; blue, 106 }  ,draw opacity=1 ]   (170,85) .. controls (172.43,84.87) and (173.54,79.58) .. (175,75) .. controls (176.46,70.42) and (176.43,65.58) .. (175,60) .. controls (173.73,55.06) and (167.13,55) .. (161.94,55) ;
\draw [shift={(160,55)}, rotate = 0.66] [color={rgb, 255:red, 191; green, 97; blue, 106 }  ,draw opacity=1 ][line width=0.75]    (4.37,-1.32) .. controls (2.78,-0.56) and (1.32,-0.12) .. (0,0) .. controls (1.32,0.12) and (2.78,0.56) .. (4.37,1.32)   ;
\draw [color={rgb, 255:red, 191; green, 97; blue, 106 }  ,draw opacity=1 ]   (150,85) .. controls (161.21,85.37) and (173.54,74.58) .. (175,70) .. controls (176.46,65.42) and (176.43,65.58) .. (175,60) .. controls (173.73,55.06) and (167.13,55) .. (161.94,55) ;
\draw [shift={(160,55)}, rotate = 0.66] [color={rgb, 255:red, 191; green, 97; blue, 106 }  ,draw opacity=1 ][line width=0.75]    (4.37,-1.32) .. controls (2.78,-0.56) and (1.32,-0.12) .. (0,0) .. controls (1.32,0.12) and (2.78,0.56) .. (4.37,1.32)   ;
\draw  [draw opacity=0][fill={rgb, 255:red, 255; green, 255; blue, 255 }  ,fill opacity=1 ] (162.5,80) .. controls (162.5,78.62) and (163.62,77.5) .. (165,77.5) .. controls (166.38,77.5) and (167.5,78.62) .. (167.5,80) .. controls (167.5,81.38) and (166.38,82.5) .. (165,82.5) .. controls (163.62,82.5) and (162.5,81.38) .. (162.5,80) -- cycle ;
\draw    (165,90) -- (165,75) ;
\draw  [draw opacity=0][fill={rgb, 255:red, 255; green, 255; blue, 255 }  ,fill opacity=1 ] (142,78) .. controls (142,76.9) and (143.34,76) .. (145,76) .. controls (146.66,76) and (148,76.9) .. (148,78) .. controls (148,79.1) and (146.66,80) .. (145,80) .. controls (143.34,80) and (142,79.1) .. (142,78) -- cycle ;
\draw [color={rgb, 255:red, 191; green, 97; blue, 106 }  ,draw opacity=1 ]   (157.82,79.94) .. controls (138.45,79.26) and (136.4,74.42) .. (135,70) .. controls (133.54,65.42) and (133.57,65.58) .. (135,60) .. controls (136.43,54.42) and (144.67,55.06) .. (150,55) ;
\draw [shift={(160,80)}, rotate = 181.03] [color={rgb, 255:red, 191; green, 97; blue, 106 }  ,draw opacity=1 ][line width=0.75]    (4.37,-1.32) .. controls (2.78,-0.56) and (1.32,-0.12) .. (0,0) .. controls (1.32,0.12) and (2.78,0.56) .. (4.37,1.32)   ;
\draw [color={rgb, 255:red, 191; green, 97; blue, 106 }  ,draw opacity=1 ]   (142.82,39.94) .. controls (123.45,39.26) and (121.4,34.42) .. (120,30) .. controls (118.54,25.42) and (118.57,25.58) .. (120,20) .. controls (121.43,14.42) and (129.67,15.06) .. (135,15) ;
\draw [shift={(145,40)}, rotate = 181.03] [color={rgb, 255:red, 191; green, 97; blue, 106 }  ,draw opacity=1 ][line width=0.75]    (4.37,-1.32) .. controls (2.78,-0.56) and (1.32,-0.12) .. (0,0) .. controls (1.32,0.12) and (2.78,0.56) .. (4.37,1.32)   ;
\draw    (130,36) .. controls (130.03,50.73) and (125.03,44.73) .. (125,60) ;
\draw  [color={rgb, 255:red, 0; green, 0; blue, 0 }  ,draw opacity=1 ][fill={rgb, 255:red, 191; green, 97; blue, 106 }  ,fill opacity=0.25 ] (125,35) -- (129.5,20) -- (150.5,20) -- (155,35) -- cycle ;
\draw [color={rgb, 255:red, 191; green, 97; blue, 106 }  ,draw opacity=1 ]   (124.17,39.65) .. controls (122.5,39.04) and (121.26,38.95) .. (120,35) .. controls (118.54,30.42) and (118.57,25.58) .. (120,20) .. controls (121.43,14.42) and (129.67,15.06) .. (135,15) ;
\draw [shift={(125,40)}, rotate = 205.96] [color={rgb, 255:red, 191; green, 97; blue, 106 }  ,draw opacity=1 ][line width=0.75]    (4.37,-1.32) .. controls (2.78,-0.56) and (1.32,-0.12) .. (0,0) .. controls (1.32,0.12) and (2.78,0.56) .. (4.37,1.32)   ;
\draw [color={rgb, 255:red, 191; green, 97; blue, 106 }  ,draw opacity=1 ]   (155,45) .. controls (157.43,44.87) and (158.54,39.58) .. (160,35) .. controls (161.46,30.42) and (161.43,25.58) .. (160,20) .. controls (158.73,15.06) and (152.13,15) .. (146.94,15) ;
\draw [shift={(145,15)}, rotate = 0.66] [color={rgb, 255:red, 191; green, 97; blue, 106 }  ,draw opacity=1 ][line width=0.75]    (4.37,-1.32) .. controls (2.78,-0.56) and (1.32,-0.12) .. (0,0) .. controls (1.32,0.12) and (2.78,0.56) .. (4.37,1.32)   ;
\draw [color={rgb, 255:red, 191; green, 97; blue, 106 }  ,draw opacity=1 ]   (130,45) .. controls (150.36,44.69) and (158.54,34.58) .. (160,30) .. controls (161.46,25.42) and (161.43,25.58) .. (160,20) .. controls (158.73,15.06) and (152.13,15) .. (146.94,15) ;
\draw [shift={(145,15)}, rotate = 0.66] [color={rgb, 255:red, 191; green, 97; blue, 106 }  ,draw opacity=1 ][line width=0.75]    (4.37,-1.32) .. controls (2.78,-0.56) and (1.32,-0.12) .. (0,0) .. controls (1.32,0.12) and (2.78,0.56) .. (4.37,1.32)   ;
\draw    (140,20) -- (140,10) ;
\draw  [draw opacity=0][fill={rgb, 255:red, 255; green, 255; blue, 255 }  ,fill opacity=1 ] (127,38) .. controls (127,36.9) and (128.34,36) .. (130,36) .. controls (131.66,36) and (133,36.9) .. (133,38) .. controls (133,39.1) and (131.66,40) .. (130,40) .. controls (128.34,40) and (127,39.1) .. (127,38) -- cycle ;
\draw [color={rgb, 255:red, 191; green, 97; blue, 106 }  ,draw opacity=1 ]   (142.82,39.94) .. controls (123.45,39.26) and (121.4,34.42) .. (120,30) .. controls (118.54,25.42) and (118.57,25.58) .. (120,20) .. controls (121.43,14.42) and (129.67,15.06) .. (135,15) ;
\draw [shift={(145,40)}, rotate = 181.03] [color={rgb, 255:red, 191; green, 97; blue, 106 }  ,draw opacity=1 ][line width=0.75]    (4.37,-1.32) .. controls (2.78,-0.56) and (1.32,-0.12) .. (0,0) .. controls (1.32,0.12) and (2.78,0.56) .. (4.37,1.32)   ;
\draw    (150,35) .. controls (150.03,49.73) and (155.03,44.73) .. (155,60) ;
\draw    (125,90) -- (125,60) ;
\draw  [draw opacity=0] (90,40) -- (115,40) -- (115,60) -- (90,60) -- cycle ;

\draw (50,67.5) node  [font=\footnotesize]  {$b$};
\draw (22.5,67.5) node  [font=\tiny,color={rgb, 255:red, 191; green, 97; blue, 106 }  ,opacity=1 ]  {$b_{0}$};
\draw (72.5,54.5) node  [font=\tiny,color={rgb, 255:red, 191; green, 97; blue, 106 }  ,opacity=1 ]  {$b_{1}$};
\draw (65,27.5) node  [font=\footnotesize]  {$a$};
\draw (37.5,27.5) node  [font=\tiny,color={rgb, 255:red, 191; green, 97; blue, 106 }  ,opacity=1 ]  {$a_{0}$};
\draw (92.5,25.5) node  [font=\tiny,color={rgb, 255:red, 191; green, 97; blue, 106 }  ,opacity=1 ]  {$a_{1}$};
\draw (155,67.5) node  [font=\footnotesize]  {$d$};
\draw (130.5,55.5) node  [font=\tiny,color={rgb, 255:red, 191; green, 97; blue, 106 }  ,opacity=1 ]  {$d_{0}$};
\draw (177.5,54.5) node  [font=\tiny,color={rgb, 255:red, 191; green, 97; blue, 106 }  ,opacity=1 ]  {$d_{1}$};
\draw (140,27.5) node  [font=\footnotesize]  {$c$};
\draw (112.5,27.5) node  [font=\tiny,color={rgb, 255:red, 191; green, 97; blue, 106 }  ,opacity=1 ]  {$c_{0}$};
\draw (167.5,25.5) node  [font=\tiny,color={rgb, 255:red, 191; green, 97; blue, 106 }  ,opacity=1 ]  {$c_{1}$};
\draw (102.5,50) node  [font=\footnotesize]  {$=$};

\end{tikzpicture}
       \caption{Equation from associativity.}
      \label{ax:fig:monoidalContourAssociativity}
    \end{figure}
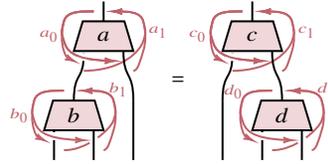

    For each application of unitality, $λ(a ⨾_1 b) = c = \rho(d ⨾_2 e)$, we impose the equations $a_0 ⨾ (b_0 ⊗ \im) = c_0 = d_0 ⨾ (\im ⊗ e_0)$ and $(b_1 ⊗ \im) ⨾ a_1 = c_1 = (\im ⊗ e_1) ⨾ d_1$. These follow from \Cref{ax:fig:monoidalContourUnitality}.

    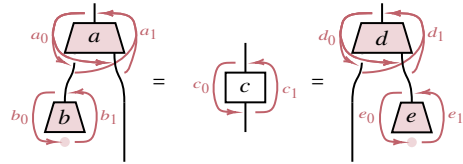
\begin{figure}[ht]
      \centering

\tikzset{every picture/.style={line width=0.75pt}} %

\begin{tikzpicture}[x=0.75pt,y=0.75pt,yscale=-1,xscale=1]
\draw [color={rgb, 255:red, 191; green, 97; blue, 106 }  ,draw opacity=1 ]   (67.82,39.94) .. controls (48.45,39.26) and (46.4,34.42) .. (45,30) .. controls (43.54,25.42) and (43.57,25.58) .. (45,20) .. controls (46.43,14.42) and (54.67,15.06) .. (60,15) ;
\draw [shift={(70,40)}, rotate = 181.03] [color={rgb, 255:red, 191; green, 97; blue, 106 }  ,draw opacity=1 ][line width=0.75]    (4.37,-1.32) .. controls (2.78,-0.56) and (1.32,-0.12) .. (0,0) .. controls (1.32,0.12) and (2.78,0.56) .. (4.37,1.32)   ;
\draw    (55,36) .. controls (55.03,50.73) and (50.03,44.73) .. (50,60) ;
\draw  [color={rgb, 255:red, 0; green, 0; blue, 0 }  ,draw opacity=1 ][fill={rgb, 255:red, 191; green, 97; blue, 106 }  ,fill opacity=0.25 ] (50,35) -- (54.5,20) -- (75.5,20) -- (80,35) -- cycle ;
\draw [color={rgb, 255:red, 191; green, 97; blue, 106 }  ,draw opacity=1 ]   (49.17,39.65) .. controls (47.5,39.04) and (46.26,38.95) .. (45,35) .. controls (43.54,30.42) and (43.57,25.58) .. (45,20) .. controls (46.43,14.42) and (54.67,15.06) .. (60,15) ;
\draw [shift={(50,40)}, rotate = 205.96] [color={rgb, 255:red, 191; green, 97; blue, 106 }  ,draw opacity=1 ][line width=0.75]    (4.37,-1.32) .. controls (2.78,-0.56) and (1.32,-0.12) .. (0,0) .. controls (1.32,0.12) and (2.78,0.56) .. (4.37,1.32)   ;
\draw [color={rgb, 255:red, 191; green, 97; blue, 106 }  ,draw opacity=1 ]   (80,45) .. controls (82.43,44.87) and (83.54,39.58) .. (85,35) .. controls (86.46,30.42) and (86.43,25.58) .. (85,20) .. controls (83.73,15.06) and (77.13,15) .. (71.94,15) ;
\draw [shift={(70,15)}, rotate = 0.66] [color={rgb, 255:red, 191; green, 97; blue, 106 }  ,draw opacity=1 ][line width=0.75]    (4.37,-1.32) .. controls (2.78,-0.56) and (1.32,-0.12) .. (0,0) .. controls (1.32,0.12) and (2.78,0.56) .. (4.37,1.32)   ;
\draw [color={rgb, 255:red, 191; green, 97; blue, 106 }  ,draw opacity=1 ]   (55,45) .. controls (75.36,44.69) and (83.54,34.58) .. (85,30) .. controls (86.46,25.42) and (86.43,25.58) .. (85,20) .. controls (83.73,15.06) and (77.13,15) .. (71.94,15) ;
\draw [shift={(70,15)}, rotate = 0.66] [color={rgb, 255:red, 191; green, 97; blue, 106 }  ,draw opacity=1 ][line width=0.75]    (4.37,-1.32) .. controls (2.78,-0.56) and (1.32,-0.12) .. (0,0) .. controls (1.32,0.12) and (2.78,0.56) .. (4.37,1.32)   ;
\draw    (65,20) -- (65,10) ;
\draw  [draw opacity=0][fill={rgb, 255:red, 255; green, 255; blue, 255 }  ,fill opacity=1 ] (52,38) .. controls (52,36.9) and (53.34,36) .. (55,36) .. controls (56.66,36) and (58,36.9) .. (58,38) .. controls (58,39.1) and (56.66,40) .. (55,40) .. controls (53.34,40) and (52,39.1) .. (52,38) -- cycle ;
\draw [color={rgb, 255:red, 191; green, 97; blue, 106 }  ,draw opacity=1 ]   (67.82,39.94) .. controls (48.45,39.26) and (46.4,34.42) .. (45,30) .. controls (43.54,25.42) and (43.57,25.58) .. (45,20) .. controls (46.43,14.42) and (54.67,15.06) .. (60,15) ;
\draw [shift={(70,40)}, rotate = 181.03] [color={rgb, 255:red, 191; green, 97; blue, 106 }  ,draw opacity=1 ][line width=0.75]    (4.37,-1.32) .. controls (2.78,-0.56) and (1.32,-0.12) .. (0,0) .. controls (1.32,0.12) and (2.78,0.56) .. (4.37,1.32)   ;
\draw    (75,35) .. controls (75.03,49.73) and (80.03,44.73) .. (80,60) ;
\draw    (80,90) -- (80,60) ;
\draw [color={rgb, 255:red, 191; green, 97; blue, 106 }  ,draw opacity=1 ]   (212.82,39.94) .. controls (193.45,39.26) and (191.4,34.42) .. (190,30) .. controls (188.54,25.42) and (188.57,25.58) .. (190,20) .. controls (191.43,14.42) and (199.67,15.06) .. (205,15) ;
\draw [shift={(215,40)}, rotate = 181.03] [color={rgb, 255:red, 191; green, 97; blue, 106 }  ,draw opacity=1 ][line width=0.75]    (4.37,-1.32) .. controls (2.78,-0.56) and (1.32,-0.12) .. (0,0) .. controls (1.32,0.12) and (2.78,0.56) .. (4.37,1.32)   ;
\draw    (200,36) .. controls (200.03,50.73) and (195.03,44.73) .. (195,60) ;
\draw  [color={rgb, 255:red, 0; green, 0; blue, 0 }  ,draw opacity=1 ][fill={rgb, 255:red, 191; green, 97; blue, 106 }  ,fill opacity=0.25 ] (195,35) -- (199.5,20) -- (220.5,20) -- (225,35) -- cycle ;
\draw [color={rgb, 255:red, 191; green, 97; blue, 106 }  ,draw opacity=1 ]   (194.17,39.65) .. controls (192.5,39.04) and (191.26,38.95) .. (190,35) .. controls (188.54,30.42) and (188.57,25.58) .. (190,20) .. controls (191.43,14.42) and (199.67,15.06) .. (205,15) ;
\draw [shift={(195,40)}, rotate = 205.96] [color={rgb, 255:red, 191; green, 97; blue, 106 }  ,draw opacity=1 ][line width=0.75]    (4.37,-1.32) .. controls (2.78,-0.56) and (1.32,-0.12) .. (0,0) .. controls (1.32,0.12) and (2.78,0.56) .. (4.37,1.32)   ;
\draw [color={rgb, 255:red, 191; green, 97; blue, 106 }  ,draw opacity=1 ]   (225,45) .. controls (227.43,44.87) and (228.54,39.58) .. (230,35) .. controls (231.46,30.42) and (231.43,25.58) .. (230,20) .. controls (228.73,15.06) and (222.13,15) .. (216.94,15) ;
\draw [shift={(215,15)}, rotate = 0.66] [color={rgb, 255:red, 191; green, 97; blue, 106 }  ,draw opacity=1 ][line width=0.75]    (4.37,-1.32) .. controls (2.78,-0.56) and (1.32,-0.12) .. (0,0) .. controls (1.32,0.12) and (2.78,0.56) .. (4.37,1.32)   ;
\draw [color={rgb, 255:red, 191; green, 97; blue, 106 }  ,draw opacity=1 ]   (200,45) .. controls (220.36,44.69) and (228.54,34.58) .. (230,30) .. controls (231.46,25.42) and (231.43,25.58) .. (230,20) .. controls (228.73,15.06) and (222.13,15) .. (216.94,15) ;
\draw [shift={(215,15)}, rotate = 0.66] [color={rgb, 255:red, 191; green, 97; blue, 106 }  ,draw opacity=1 ][line width=0.75]    (4.37,-1.32) .. controls (2.78,-0.56) and (1.32,-0.12) .. (0,0) .. controls (1.32,0.12) and (2.78,0.56) .. (4.37,1.32)   ;
\draw    (210,20) -- (210,10) ;
\draw  [draw opacity=0][fill={rgb, 255:red, 255; green, 255; blue, 255 }  ,fill opacity=1 ] (197,38) .. controls (197,36.9) and (198.34,36) .. (200,36) .. controls (201.66,36) and (203,36.9) .. (203,38) .. controls (203,39.1) and (201.66,40) .. (200,40) .. controls (198.34,40) and (197,39.1) .. (197,38) -- cycle ;
\draw [color={rgb, 255:red, 191; green, 97; blue, 106 }  ,draw opacity=1 ]   (212.82,39.94) .. controls (193.45,39.26) and (191.4,34.42) .. (190,30) .. controls (188.54,25.42) and (188.57,25.58) .. (190,20) .. controls (191.43,14.42) and (199.67,15.06) .. (205,15) ;
\draw [shift={(215,40)}, rotate = 181.03] [color={rgb, 255:red, 191; green, 97; blue, 106 }  ,draw opacity=1 ][line width=0.75]    (4.37,-1.32) .. controls (2.78,-0.56) and (1.32,-0.12) .. (0,0) .. controls (1.32,0.12) and (2.78,0.56) .. (4.37,1.32)   ;
\draw    (220,35) .. controls (220.03,49.73) and (225.03,44.73) .. (225,60) ;
\draw    (195,90) -- (195,60) ;
\draw  [draw opacity=0] (165,40) -- (190,40) -- (190,60) -- (165,60) -- cycle ;
\draw  [fill={rgb, 255:red, 191; green, 97; blue, 106 }  ,fill opacity=0.25 ] (40,75) -- (44.5,60) -- (55.5,60) -- (60,75) -- cycle ;
\draw  [draw opacity=0] (40,75) -- (60,75) -- (60,95) -- (40,95) -- cycle ;
\draw [color={rgb, 255:red, 191; green, 97; blue, 106 }  ,draw opacity=1 ]   (42.96,79.93) .. controls (37.06,79.68) and (36.32,79.14) .. (35,75) .. controls (33.54,70.42) and (33.57,65.58) .. (35,60) .. controls (36.43,54.42) and (39.67,55.06) .. (45,55) ;
\draw [shift={(45,80)}, rotate = 181.8] [color={rgb, 255:red, 191; green, 97; blue, 106 }  ,draw opacity=1 ][line width=0.75]    (4.37,-1.32) .. controls (2.78,-0.56) and (1.32,-0.12) .. (0,0) .. controls (1.32,0.12) and (2.78,0.56) .. (4.37,1.32)   ;
\draw [color={rgb, 255:red, 191; green, 97; blue, 106 }  ,draw opacity=1 ]   (57.04,54.86) .. controls (61.32,54.74) and (63.73,56.02) .. (65,60) .. controls (66.46,64.58) and (66.43,69.42) .. (65,75) .. controls (63.57,80.58) and (60.33,79.94) .. (55,80) ;
\draw [shift={(55,55)}, rotate = 354.07] [color={rgb, 255:red, 191; green, 97; blue, 106 }  ,draw opacity=1 ][line width=0.75]    (4.37,-1.32) .. controls (2.78,-0.56) and (1.32,-0.12) .. (0,0) .. controls (1.32,0.12) and (2.78,0.56) .. (4.37,1.32)   ;
\draw  [draw opacity=0][fill={rgb, 255:red, 191; green, 97; blue, 106 }  ,fill opacity=0.25 ] (47.5,80) .. controls (47.5,78.62) and (48.62,77.5) .. (50,77.5) .. controls (51.38,77.5) and (52.5,78.62) .. (52.5,80) .. controls (52.5,81.38) and (51.38,82.5) .. (50,82.5) .. controls (48.62,82.5) and (47.5,81.38) .. (47.5,80) -- cycle ;
\draw  [fill={rgb, 255:red, 191; green, 97; blue, 106 }  ,fill opacity=0.25 ] (215,75.15) -- (219.5,60.15) -- (230.5,60.15) -- (235,75.15) -- cycle ;
\draw  [draw opacity=0] (215,75.15) -- (235,75.15) -- (235,95.15) -- (215,95.15) -- cycle ;
\draw [color={rgb, 255:red, 191; green, 97; blue, 106 }  ,draw opacity=1 ]   (217.96,80.08) .. controls (212.06,79.83) and (211.32,79.29) .. (210,75.15) .. controls (208.54,70.57) and (208.57,65.73) .. (210,60.15) .. controls (211.43,54.57) and (214.67,55.21) .. (220,55.15) ;
\draw [shift={(220,80.15)}, rotate = 181.8] [color={rgb, 255:red, 191; green, 97; blue, 106 }  ,draw opacity=1 ][line width=0.75]    (4.37,-1.32) .. controls (2.78,-0.56) and (1.32,-0.12) .. (0,0) .. controls (1.32,0.12) and (2.78,0.56) .. (4.37,1.32)   ;
\draw [color={rgb, 255:red, 191; green, 97; blue, 106 }  ,draw opacity=1 ]   (232.04,55.01) .. controls (236.32,54.89) and (238.73,56.17) .. (240,60.15) .. controls (241.46,64.73) and (241.43,69.57) .. (240,75.15) .. controls (238.57,80.73) and (235.33,80.09) .. (230,80.15) ;
\draw [shift={(230,55.15)}, rotate = 354.07] [color={rgb, 255:red, 191; green, 97; blue, 106 }  ,draw opacity=1 ][line width=0.75]    (4.37,-1.32) .. controls (2.78,-0.56) and (1.32,-0.12) .. (0,0) .. controls (1.32,0.12) and (2.78,0.56) .. (4.37,1.32)   ;
\draw  [draw opacity=0][fill={rgb, 255:red, 191; green, 97; blue, 106 }  ,fill opacity=0.25 ] (222.5,80.15) .. controls (222.5,78.77) and (223.62,77.65) .. (225,77.65) .. controls (226.38,77.65) and (227.5,78.77) .. (227.5,80.15) .. controls (227.5,81.53) and (226.38,82.65) .. (225,82.65) .. controls (223.62,82.65) and (222.5,81.53) .. (222.5,80.15) -- cycle ;
\draw    (141.08,45) -- (141.08,30) ;
\draw   (131.08,45) -- (151.08,45) -- (151.08,60) -- (131.08,60) -- cycle ;
\draw    (141.08,75) -- (141.08,60) ;
\draw [color={rgb, 255:red, 191; green, 97; blue, 106 }  ,draw opacity=1 ]   (134.05,64.93) .. controls (128.14,64.68) and (127.4,64.14) .. (126.08,60) .. controls (124.63,55.42) and (124.65,50.58) .. (126.08,45) .. controls (127.52,39.42) and (130.76,40.06) .. (136.08,40) ;
\draw [shift={(136.08,65)}, rotate = 181.8] [color={rgb, 255:red, 191; green, 97; blue, 106 }  ,draw opacity=1 ][line width=0.75]    (4.37,-1.32) .. controls (2.78,-0.56) and (1.32,-0.12) .. (0,0) .. controls (1.32,0.12) and (2.78,0.56) .. (4.37,1.32)   ;
\draw [color={rgb, 255:red, 191; green, 97; blue, 106 }  ,draw opacity=1 ]   (148.13,39.86) .. controls (152.4,39.74) and (154.82,41.02) .. (156.08,45) .. controls (157.54,49.58) and (157.52,54.42) .. (156.08,60) .. controls (154.65,65.58) and (151.41,64.94) .. (146.08,65) ;
\draw [shift={(146.08,40)}, rotate = 354.07] [color={rgb, 255:red, 191; green, 97; blue, 106 }  ,draw opacity=1 ][line width=0.75]    (4.37,-1.32) .. controls (2.78,-0.56) and (1.32,-0.12) .. (0,0) .. controls (1.32,0.12) and (2.78,0.56) .. (4.37,1.32)   ;
\draw  [draw opacity=0] (85,40) -- (110,40) -- (110,60) -- (85,60) -- cycle ;

\draw (65,27.5) node  [font=\footnotesize]  {$a$};
\draw (37.5,27.5) node  [font=\tiny,color={rgb, 255:red, 191; green, 97; blue, 106 }  ,opacity=1 ]  {$a_{0}$};
\draw (92.5,25.5) node  [font=\tiny,color={rgb, 255:red, 191; green, 97; blue, 106 }  ,opacity=1 ]  {$a_{1}$};
\draw (210,27.5) node  [font=\footnotesize]  {$d$};
\draw (182.5,27.5) node  [font=\tiny,color={rgb, 255:red, 191; green, 97; blue, 106 }  ,opacity=1 ]  {$d_{0}$};
\draw (237.5,25.5) node  [font=\tiny,color={rgb, 255:red, 191; green, 97; blue, 106 }  ,opacity=1 ]  {$d_{1}$};
\draw (177.5,50) node  [font=\footnotesize]  {$=$};
\draw (50,67.5) node  [font=\footnotesize]  {$b$};
\draw (72.5,67.5) node  [font=\tiny,color={rgb, 255:red, 191; green, 97; blue, 106 }  ,opacity=1 ]  {$b_{1}$};
\draw (27.5,67.5) node  [font=\tiny,color={rgb, 255:red, 191; green, 97; blue, 106 }  ,opacity=1 ]  {$b_{0}$};
\draw (225,67.65) node  [font=\footnotesize]  {$e$};
\draw (247.5,67.65) node  [font=\tiny,color={rgb, 255:red, 191; green, 97; blue, 106 }  ,opacity=1 ]  {$e_{1}$};
\draw (202.5,67.65) node  [font=\tiny,color={rgb, 255:red, 191; green, 97; blue, 106 }  ,opacity=1 ]  {$e_{0}$};
\draw (141.08,52.5) node  [font=\footnotesize]  {$c$};
\draw (119,52.5) node  [font=\tiny,color={rgb, 255:red, 191; green, 97; blue, 106 }  ,opacity=1 ]  {$c_{0}$};
\draw (164,55) node  [font=\tiny,color={rgb, 255:red, 191; green, 97; blue, 106 }  ,opacity=1 ]  {$c_{1}$};
\draw (97.5,50) node  [font=\footnotesize]  {$=$};

\end{tikzpicture}
       \caption{Equations from unitality.}
      \label{ax:fig:monoidalContourUnitality}
    \end{figure}

    For each application of the laxator, $ψ_2(a \mathbin{|} b \mathbin{|} c) = (d \mathbin{|} e \mathbin{|} f)$, we impose the equation $a_0 ⨾ (b_0 ⊗ c_0) = d_0 ⨾ e_0$, the middle equation $b_1 ⊗ c_1 = e_1 ⨾ d_1 ⨾ f_0$, and $(b_2 ⊗ c_2) ⨾ a_1 = f_1 ⨾ d_2$. These follow \Cref{ax:fig:monoidal-contour-equation1}.

    For each application of the laxator, $ψ_0(a) = (b \mathbin{|}_1 c \mathbin{|}_2 d)$, we impose an equation $a_0 = b_0 ⨾ c_0$, an equation $\im = c_1 ⨾ b_1 ⨾ d_0$, and an equation $a_1 = d_1 ⨾ b_2$. This follows \Cref{ax:fig:monoidal-contour-equation2}.

    For each application of the laxator, $\varphi_2(a \mathbin{|}_1 b \mathbin{|}_2 c) = d$, we impose an equation $a_0 ⨾ (b_0 ⊗ c_0) ⨾ a_1 = d_0$. This follows \Cref{ax:fig:monoidal-contour-equation3}.

    For each application of the laxator, $\varphi_0(a) = b$, we impose an equation $a_0 ⨾ a_1 = b_0$. This follows \Cref{ax:fig:monoidal-contour-equation4}.

    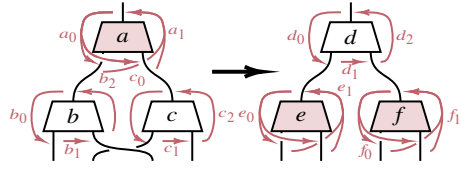
\begin{figure}[ht!]
      \centering

\tikzset{every picture/.style={line width=0.75pt}} %

\begin{tikzpicture}[x=0.75pt,y=0.75pt,yscale=-1,xscale=1]
\draw [color={rgb, 255:red, 191; green, 97; blue, 106 }  ,draw opacity=1 ]   (187.82,39.94) .. controls (168.45,39.26) and (166.4,34.42) .. (165,30) .. controls (163.54,25.42) and (163.57,25.58) .. (165,20) .. controls (166.43,14.42) and (174.67,15.06) .. (180,15) ;
\draw [shift={(190,40)}, rotate = 181.03] [color={rgb, 255:red, 191; green, 97; blue, 106 }  ,draw opacity=1 ][line width=0.75]    (4.37,-1.32) .. controls (2.78,-0.56) and (1.32,-0.12) .. (0,0) .. controls (1.32,0.12) and (2.78,0.56) .. (4.37,1.32)   ;
\draw    (175,35) .. controls (175.03,49.73) and (160.03,44.73) .. (160,60) ;
\draw  [color={rgb, 255:red, 0; green, 0; blue, 0 }  ,draw opacity=1 ][fill={rgb, 255:red, 191; green, 97; blue, 106 }  ,fill opacity=0.25 ] (170,35) -- (174.5,20) -- (195.5,20) -- (200,35) -- cycle ;
\draw [color={rgb, 255:red, 191; green, 97; blue, 106 }  ,draw opacity=1 ]   (169.17,39.65) .. controls (167.5,39.04) and (166.26,38.95) .. (165,35) .. controls (163.54,30.42) and (163.57,25.58) .. (165,20) .. controls (166.43,14.42) and (174.67,15.06) .. (180,15) ;
\draw [shift={(170,40)}, rotate = 205.96] [color={rgb, 255:red, 191; green, 97; blue, 106 }  ,draw opacity=1 ][line width=0.75]    (4.37,-1.32) .. controls (2.78,-0.56) and (1.32,-0.12) .. (0,0) .. controls (1.32,0.12) and (2.78,0.56) .. (4.37,1.32)   ;
\draw [color={rgb, 255:red, 191; green, 97; blue, 106 }  ,draw opacity=1 ]   (200,45) .. controls (202.43,44.87) and (203.54,39.58) .. (205,35) .. controls (206.46,30.42) and (206.43,25.58) .. (205,20) .. controls (203.73,15.06) and (197.13,15) .. (191.94,15) ;
\draw [shift={(190,15)}, rotate = 0.66] [color={rgb, 255:red, 191; green, 97; blue, 106 }  ,draw opacity=1 ][line width=0.75]    (4.37,-1.32) .. controls (2.78,-0.56) and (1.32,-0.12) .. (0,0) .. controls (1.32,0.12) and (2.78,0.56) .. (4.37,1.32)   ;
\draw [color={rgb, 255:red, 191; green, 97; blue, 106 }  ,draw opacity=1 ]   (175,45) .. controls (195.36,44.69) and (203.54,34.58) .. (205,30) .. controls (206.46,25.42) and (206.43,25.58) .. (205,20) .. controls (203.73,15.06) and (197.13,15) .. (191.94,15) ;
\draw [shift={(190,15)}, rotate = 0.66] [color={rgb, 255:red, 191; green, 97; blue, 106 }  ,draw opacity=1 ][line width=0.75]    (4.37,-1.32) .. controls (2.78,-0.56) and (1.32,-0.12) .. (0,0) .. controls (1.32,0.12) and (2.78,0.56) .. (4.37,1.32)   ;
\draw  [draw opacity=0][fill={rgb, 255:red, 255; green, 255; blue, 255 }  ,fill opacity=1 ] (192.5,40) .. controls (192.5,38.62) and (193.62,37.5) .. (195,37.5) .. controls (196.38,37.5) and (197.5,38.62) .. (197.5,40) .. controls (197.5,41.38) and (196.38,42.5) .. (195,42.5) .. controls (193.62,42.5) and (192.5,41.38) .. (192.5,40) -- cycle ;
\draw    (185,20) -- (185,10) ;
\draw   (195,75) -- (199.5,60) -- (220.5,60) -- (225,75) -- cycle ;
\draw [color={rgb, 255:red, 191; green, 97; blue, 106 }  ,draw opacity=1 ]   (193.13,79.28) .. controls (191.96,78.82) and (190.98,78.09) .. (190,75) .. controls (188.54,70.42) and (188.57,65.58) .. (190,60) .. controls (191.43,54.42) and (199.67,55.06) .. (205,55) ;
\draw [shift={(195,80)}, rotate = 205.96] [color={rgb, 255:red, 191; green, 97; blue, 106 }  ,draw opacity=1 ][line width=0.75]    (4.37,-1.32) .. controls (2.78,-0.56) and (1.32,-0.12) .. (0,0) .. controls (1.32,0.12) and (2.78,0.56) .. (4.37,1.32)   ;
\draw [color={rgb, 255:red, 191; green, 97; blue, 106 }  ,draw opacity=1 ]   (205,80) -- (213,80) ;
\draw [shift={(215,80)}, rotate = 180] [color={rgb, 255:red, 191; green, 97; blue, 106 }  ,draw opacity=1 ][line width=0.75]    (4.37,-1.32) .. controls (2.78,-0.56) and (1.32,-0.12) .. (0,0) .. controls (1.32,0.12) and (2.78,0.56) .. (4.37,1.32)   ;
\draw [color={rgb, 255:red, 191; green, 97; blue, 106 }  ,draw opacity=1 ]   (225,80) .. controls (227.43,79.87) and (228.54,79.58) .. (230,75) .. controls (231.46,70.42) and (231.43,65.58) .. (230,60) .. controls (228.73,55.06) and (222.13,55) .. (216.94,55) ;
\draw [shift={(215,55)}, rotate = 0.66] [color={rgb, 255:red, 191; green, 97; blue, 106 }  ,draw opacity=1 ][line width=0.75]    (4.37,-1.32) .. controls (2.78,-0.56) and (1.32,-0.12) .. (0,0) .. controls (1.32,0.12) and (2.78,0.56) .. (4.37,1.32)   ;
\draw    (220,90) -- (220,75) ;
\draw   (145,75) -- (149.5,60) -- (170.5,60) -- (175,75) -- cycle ;
\draw [color={rgb, 255:red, 191; green, 97; blue, 106 }  ,draw opacity=1 ]   (143.13,79.28) .. controls (141.96,78.82) and (140.98,78.09) .. (140,75) .. controls (138.54,70.42) and (138.57,65.58) .. (140,60) .. controls (141.43,54.42) and (149.67,55.06) .. (155,55) ;
\draw [shift={(145,80)}, rotate = 205.96] [color={rgb, 255:red, 191; green, 97; blue, 106 }  ,draw opacity=1 ][line width=0.75]    (4.37,-1.32) .. controls (2.78,-0.56) and (1.32,-0.12) .. (0,0) .. controls (1.32,0.12) and (2.78,0.56) .. (4.37,1.32)   ;
\draw [color={rgb, 255:red, 191; green, 97; blue, 106 }  ,draw opacity=1 ]   (155,80) -- (163,80) ;
\draw [shift={(165,80)}, rotate = 180] [color={rgb, 255:red, 191; green, 97; blue, 106 }  ,draw opacity=1 ][line width=0.75]    (4.37,-1.32) .. controls (2.78,-0.56) and (1.32,-0.12) .. (0,0) .. controls (1.32,0.12) and (2.78,0.56) .. (4.37,1.32)   ;
\draw [color={rgb, 255:red, 191; green, 97; blue, 106 }  ,draw opacity=1 ]   (180,75) .. controls (181.46,70.42) and (181.43,65.58) .. (180,60) .. controls (178.73,55.06) and (172.13,55) .. (166.94,55) ;
\draw [shift={(165,55)}, rotate = 0.66] [color={rgb, 255:red, 191; green, 97; blue, 106 }  ,draw opacity=1 ][line width=0.75]    (4.37,-1.32) .. controls (2.78,-0.56) and (1.32,-0.12) .. (0,0) .. controls (1.32,0.12) and (2.78,0.56) .. (4.37,1.32)   ;
\draw    (150,90) -- (150,75) ;
\draw  [draw opacity=0][fill={rgb, 255:red, 255; green, 255; blue, 255 }  ,fill opacity=1 ] (172,38) .. controls (172,36.9) and (173.34,36) .. (175,36) .. controls (176.66,36) and (178,36.9) .. (178,38) .. controls (178,39.1) and (176.66,40) .. (175,40) .. controls (173.34,40) and (172,39.1) .. (172,38) -- cycle ;
\draw  [draw opacity=0][fill={rgb, 255:red, 255; green, 255; blue, 255 }  ,fill opacity=1 ] (197,45) .. controls (197,43.9) and (198.34,43) .. (200,43) .. controls (201.66,43) and (203,43.9) .. (203,45) .. controls (203,46.1) and (201.66,47) .. (200,47) .. controls (198.34,47) and (197,46.1) .. (197,45) -- cycle ;
\draw [color={rgb, 255:red, 191; green, 97; blue, 106 }  ,draw opacity=1 ]   (187.82,39.94) .. controls (168.45,39.26) and (166.4,34.42) .. (165,30) .. controls (163.54,25.42) and (163.57,25.58) .. (165,20) .. controls (166.43,14.42) and (174.67,15.06) .. (180,15) ;
\draw [shift={(190,40)}, rotate = 181.03] [color={rgb, 255:red, 191; green, 97; blue, 106 }  ,draw opacity=1 ][line width=0.75]    (4.37,-1.32) .. controls (2.78,-0.56) and (1.32,-0.12) .. (0,0) .. controls (1.32,0.12) and (2.78,0.56) .. (4.37,1.32)   ;
\draw    (200,75) .. controls (199.96,89.32) and (170.14,79.77) .. (170,90) ;
\draw    (195,35) .. controls (195.03,49.73) and (210.03,44.73) .. (210,60) ;
\draw  [draw opacity=0][fill={rgb, 255:red, 255; green, 255; blue, 255 }  ,fill opacity=1 ] (178,84) .. controls (178,81.79) and (181.13,80) .. (185,80) .. controls (188.87,80) and (192,81.79) .. (192,84) .. controls (192,86.21) and (188.87,88) .. (185,88) .. controls (181.13,88) and (178,86.21) .. (178,84) -- cycle ;
\draw    (170,75) .. controls (170.14,89.16) and (200.14,80.08) .. (200,90) ;
\draw [color={rgb, 255:red, 0; green, 0; blue, 0 }  ,draw opacity=1 ][line width=1.5]    (230,45) -- (252,45) ;
\draw [shift={(255,45)}, rotate = 180] [color={rgb, 255:red, 0; green, 0; blue, 0 }  ,draw opacity=1 ][line width=1.5]    (8.53,-2.57) .. controls (5.42,-1.09) and (2.58,-0.23) .. (0,0) .. controls (2.58,0.23) and (5.42,1.09) .. (8.53,2.57)   ;
\draw  [color={rgb, 255:red, 0; green, 0; blue, 0 }  ,draw opacity=1 ][fill={rgb, 255:red, 191; green, 97; blue, 106 }  ,fill opacity=0.25 ] (260,75) -- (264.5,60) -- (285.5,60) -- (290,75) -- cycle ;
\draw    (265,90) -- (265,75) ;
\draw [color={rgb, 255:red, 191; green, 97; blue, 106 }  ,draw opacity=1 ]   (259.17,79.65) .. controls (257.5,79.04) and (256.26,78.95) .. (255,75) .. controls (253.54,70.42) and (253.57,65.58) .. (255,60) .. controls (256.43,54.42) and (264.67,55.06) .. (270,55) ;
\draw [shift={(260,80)}, rotate = 205.96] [color={rgb, 255:red, 191; green, 97; blue, 106 }  ,draw opacity=1 ][line width=0.75]    (4.37,-1.32) .. controls (2.78,-0.56) and (1.32,-0.12) .. (0,0) .. controls (1.32,0.12) and (2.78,0.56) .. (4.37,1.32)   ;
\draw [color={rgb, 255:red, 191; green, 97; blue, 106 }  ,draw opacity=1 ]   (290,85) .. controls (292.43,84.87) and (293.54,79.58) .. (295,75) .. controls (296.46,70.42) and (296.43,65.58) .. (295,60) .. controls (293.73,55.06) and (287.13,55) .. (281.94,55) ;
\draw [shift={(280,55)}, rotate = 0.66] [color={rgb, 255:red, 191; green, 97; blue, 106 }  ,draw opacity=1 ][line width=0.75]    (4.37,-1.32) .. controls (2.78,-0.56) and (1.32,-0.12) .. (0,0) .. controls (1.32,0.12) and (2.78,0.56) .. (4.37,1.32)   ;
\draw [color={rgb, 255:red, 191; green, 97; blue, 106 }  ,draw opacity=1 ]   (270,85) .. controls (281.21,85.37) and (293.54,74.58) .. (295,70) .. controls (296.46,65.42) and (296.43,65.58) .. (295,60) .. controls (293.73,55.06) and (287.13,55) .. (281.94,55) ;
\draw [shift={(280,55)}, rotate = 0.66] [color={rgb, 255:red, 191; green, 97; blue, 106 }  ,draw opacity=1 ][line width=0.75]    (4.37,-1.32) .. controls (2.78,-0.56) and (1.32,-0.12) .. (0,0) .. controls (1.32,0.12) and (2.78,0.56) .. (4.37,1.32)   ;
\draw  [draw opacity=0][fill={rgb, 255:red, 255; green, 255; blue, 255 }  ,fill opacity=1 ] (282.5,80) .. controls (282.5,78.62) and (283.62,77.5) .. (285,77.5) .. controls (286.38,77.5) and (287.5,78.62) .. (287.5,80) .. controls (287.5,81.38) and (286.38,82.5) .. (285,82.5) .. controls (283.62,82.5) and (282.5,81.38) .. (282.5,80) -- cycle ;
\draw    (285,90) -- (285,75) ;
\draw   (285,35) -- (289.5,20) -- (310.5,20) -- (315,35) -- cycle ;
\draw    (300,20) -- (300,10) ;
\draw [color={rgb, 255:red, 191; green, 97; blue, 106 }  ,draw opacity=1 ]   (283.13,39.28) .. controls (281.96,38.82) and (280.98,38.09) .. (280,35) .. controls (278.54,30.42) and (278.57,25.58) .. (280,20) .. controls (281.43,14.42) and (289.67,15.06) .. (295,15) ;
\draw [shift={(285,40)}, rotate = 205.96] [color={rgb, 255:red, 191; green, 97; blue, 106 }  ,draw opacity=1 ][line width=0.75]    (4.37,-1.32) .. controls (2.78,-0.56) and (1.32,-0.12) .. (0,0) .. controls (1.32,0.12) and (2.78,0.56) .. (4.37,1.32)   ;
\draw [color={rgb, 255:red, 191; green, 97; blue, 106 }  ,draw opacity=1 ]   (295,40) -- (303,40) ;
\draw [shift={(305,40)}, rotate = 180] [color={rgb, 255:red, 191; green, 97; blue, 106 }  ,draw opacity=1 ][line width=0.75]    (4.37,-1.32) .. controls (2.78,-0.56) and (1.32,-0.12) .. (0,0) .. controls (1.32,0.12) and (2.78,0.56) .. (4.37,1.32)   ;
\draw [color={rgb, 255:red, 191; green, 97; blue, 106 }  ,draw opacity=1 ]   (315,40) .. controls (317.43,39.87) and (318.54,39.58) .. (320,35) .. controls (321.46,30.42) and (321.43,25.58) .. (320,20) .. controls (318.73,15.06) and (312.13,15) .. (306.94,15) ;
\draw [shift={(305,15)}, rotate = 0.66] [color={rgb, 255:red, 191; green, 97; blue, 106 }  ,draw opacity=1 ][line width=0.75]    (4.37,-1.32) .. controls (2.78,-0.56) and (1.32,-0.12) .. (0,0) .. controls (1.32,0.12) and (2.78,0.56) .. (4.37,1.32)   ;
\draw    (290,35) .. controls (290.03,49.73) and (275.03,44.73) .. (275,60) ;
\draw  [color={rgb, 255:red, 0; green, 0; blue, 0 }  ,draw opacity=1 ][fill={rgb, 255:red, 191; green, 97; blue, 106 }  ,fill opacity=0.25 ] (310,75) -- (314.5,60) -- (335.5,60) -- (340,75) -- cycle ;
\draw    (315,90) -- (315,75) ;
\draw [color={rgb, 255:red, 191; green, 97; blue, 106 }  ,draw opacity=1 ]   (309.17,79.65) .. controls (307.5,79.04) and (306.26,78.95) .. (305,75) .. controls (303.54,70.42) and (303.57,65.58) .. (305,60) .. controls (306.43,54.42) and (314.67,55.06) .. (320,55) ;
\draw [shift={(310,80)}, rotate = 205.96] [color={rgb, 255:red, 191; green, 97; blue, 106 }  ,draw opacity=1 ][line width=0.75]    (4.37,-1.32) .. controls (2.78,-0.56) and (1.32,-0.12) .. (0,0) .. controls (1.32,0.12) and (2.78,0.56) .. (4.37,1.32)   ;
\draw [color={rgb, 255:red, 191; green, 97; blue, 106 }  ,draw opacity=1 ]   (340,85) .. controls (342.43,84.87) and (343.54,79.58) .. (345,75) .. controls (346.46,70.42) and (346.43,65.58) .. (345,60) .. controls (343.73,55.06) and (337.13,55) .. (331.94,55) ;
\draw [shift={(330,55)}, rotate = 0.66] [color={rgb, 255:red, 191; green, 97; blue, 106 }  ,draw opacity=1 ][line width=0.75]    (4.37,-1.32) .. controls (2.78,-0.56) and (1.32,-0.12) .. (0,0) .. controls (1.32,0.12) and (2.78,0.56) .. (4.37,1.32)   ;
\draw [color={rgb, 255:red, 191; green, 97; blue, 106 }  ,draw opacity=1 ]   (320,85) .. controls (331.21,85.37) and (343.54,74.58) .. (345,70) .. controls (346.46,65.42) and (346.43,65.58) .. (345,60) .. controls (343.73,55.06) and (337.13,55) .. (331.94,55) ;
\draw [shift={(330,55)}, rotate = 0.66] [color={rgb, 255:red, 191; green, 97; blue, 106 }  ,draw opacity=1 ][line width=0.75]    (4.37,-1.32) .. controls (2.78,-0.56) and (1.32,-0.12) .. (0,0) .. controls (1.32,0.12) and (2.78,0.56) .. (4.37,1.32)   ;
\draw  [draw opacity=0][fill={rgb, 255:red, 255; green, 255; blue, 255 }  ,fill opacity=1 ] (332.5,80) .. controls (332.5,78.62) and (333.62,77.5) .. (335,77.5) .. controls (336.38,77.5) and (337.5,78.62) .. (337.5,80) .. controls (337.5,81.38) and (336.38,82.5) .. (335,82.5) .. controls (333.62,82.5) and (332.5,81.38) .. (332.5,80) -- cycle ;
\draw    (310,35) .. controls (310.03,49.73) and (325.03,44.73) .. (325,60) ;
\draw    (335,90) -- (335,75) ;
\draw  [draw opacity=0][fill={rgb, 255:red, 255; green, 255; blue, 255 }  ,fill opacity=1 ] (262,78) .. controls (262,76.9) and (263.34,76) .. (265,76) .. controls (266.66,76) and (268,76.9) .. (268,78) .. controls (268,79.1) and (266.66,80) .. (265,80) .. controls (263.34,80) and (262,79.1) .. (262,78) -- cycle ;
\draw  [draw opacity=0][fill={rgb, 255:red, 255; green, 255; blue, 255 }  ,fill opacity=1 ] (312,78) .. controls (312,76.9) and (313.34,76) .. (315,76) .. controls (316.66,76) and (318,76.9) .. (318,78) .. controls (318,79.1) and (316.66,80) .. (315,80) .. controls (313.34,80) and (312,79.1) .. (312,78) -- cycle ;
\draw [color={rgb, 255:red, 191; green, 97; blue, 106 }  ,draw opacity=1 ]   (277.82,79.94) .. controls (258.45,79.26) and (256.4,74.42) .. (255,70) .. controls (253.54,65.42) and (253.57,65.58) .. (255,60) .. controls (256.43,54.42) and (264.67,55.06) .. (270,55) ;
\draw [shift={(280,80)}, rotate = 181.03] [color={rgb, 255:red, 191; green, 97; blue, 106 }  ,draw opacity=1 ][line width=0.75]    (4.37,-1.32) .. controls (2.78,-0.56) and (1.32,-0.12) .. (0,0) .. controls (1.32,0.12) and (2.78,0.56) .. (4.37,1.32)   ;
\draw [color={rgb, 255:red, 191; green, 97; blue, 106 }  ,draw opacity=1 ]   (327.82,79.94) .. controls (308.45,79.26) and (306.4,74.42) .. (305,70) .. controls (303.54,65.42) and (303.57,65.58) .. (305,60) .. controls (306.43,54.42) and (314.67,55.06) .. (320,55) ;
\draw [shift={(330,80)}, rotate = 181.03] [color={rgb, 255:red, 191; green, 97; blue, 106 }  ,draw opacity=1 ][line width=0.75]    (4.37,-1.32) .. controls (2.78,-0.56) and (1.32,-0.12) .. (0,0) .. controls (1.32,0.12) and (2.78,0.56) .. (4.37,1.32)   ;

\draw (185,27.5) node  [font=\footnotesize]  {$a$};
\draw (157.5,27.5) node  [font=\tiny,color={rgb, 255:red, 191; green, 97; blue, 106 }  ,opacity=1 ]  {$a_{0}$};
\draw (212.5,25.5) node  [font=\tiny,color={rgb, 255:red, 191; green, 97; blue, 106 }  ,opacity=1 ]  {$a_{1}$};
\draw (210,67.5) node  [font=\footnotesize]  {$c$};
\draw (192.5,49.5) node  [font=\tiny,color={rgb, 255:red, 191; green, 97; blue, 106 }  ,opacity=1 ]  {$c_{0}$};
\draw (210,85) node  [font=\tiny,color={rgb, 255:red, 191; green, 97; blue, 106 }  ,opacity=1 ]  {$c_{1}$};
\draw (237.5,67.5) node  [font=\tiny,color={rgb, 255:red, 191; green, 97; blue, 106 }  ,opacity=1 ]  {$c_{2}$};
\draw (160,67.5) node  [font=\footnotesize]  {$b$};
\draw (132.5,67.5) node  [font=\tiny,color={rgb, 255:red, 191; green, 97; blue, 106 }  ,opacity=1 ]  {$b_{0}$};
\draw (160,85) node  [font=\tiny,color={rgb, 255:red, 191; green, 97; blue, 106 }  ,opacity=1 ]  {$b_{1}$};
\draw (177.5,49.5) node  [font=\tiny,color={rgb, 255:red, 191; green, 97; blue, 106 }  ,opacity=1 ]  {$b_{2}$};
\draw (275,67.5) node  [font=\footnotesize]  {$e$};
\draw (247.5,67.5) node  [font=\tiny,color={rgb, 255:red, 191; green, 97; blue, 106 }  ,opacity=1 ]  {$e_{0}$};
\draw (297.5,54.5) node  [font=\tiny,color={rgb, 255:red, 191; green, 97; blue, 106 }  ,opacity=1 ]  {$e_{1}$};
\draw (300,27.5) node  [font=\footnotesize]  {$d$};
\draw (272.5,27.5) node  [font=\tiny,color={rgb, 255:red, 191; green, 97; blue, 106 }  ,opacity=1 ]  {$d_{0}$};
\draw (300,45) node  [font=\tiny,color={rgb, 255:red, 191; green, 97; blue, 106 }  ,opacity=1 ]  {$d_{1}$};
\draw (327.5,27.5) node  [font=\tiny,color={rgb, 255:red, 191; green, 97; blue, 106 }  ,opacity=1 ]  {$d_{2}$};
\draw (325,67.5) node  [font=\footnotesize]  {$f$};
\draw (307.5,85.5) node  [font=\tiny,color={rgb, 255:red, 191; green, 97; blue, 106 }  ,opacity=1 ]  {$f_{0}$};
\draw (352.5,67.5) node  [font=\tiny,color={rgb, 255:red, 191; green, 97; blue, 106 }  ,opacity=1 ]  {$f_{1}$};

\end{tikzpicture}
       \caption{Equations for the first laxator.}
      \label{ax:fig:monoidal-contour-equation1}
    \end{figure}
    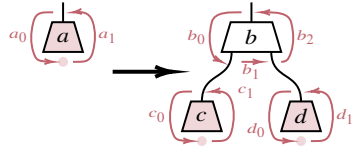
\begin{figure}[ht!]
      \centering

\tikzset{every picture/.style={line width=0.75pt}} %

\begin{tikzpicture}[x=0.75pt,y=0.75pt,yscale=-1,xscale=1]
\draw  [fill={rgb, 255:red, 191; green, 97; blue, 106 }  ,fill opacity=0.25 ] (265,75) -- (269.5,60) -- (280.5,60) -- (285,75) -- cycle ;
\draw  [draw opacity=0] (265,50) -- (285,50) -- (285,70) -- (265,70) -- cycle ;
\draw [color={rgb, 255:red, 191; green, 97; blue, 106 }  ,draw opacity=1 ]   (267.96,79.93) .. controls (262.06,79.68) and (261.32,79.14) .. (260,75) .. controls (258.54,70.42) and (258.57,65.58) .. (260,60) .. controls (261.43,54.42) and (264.67,55.06) .. (270,55) ;
\draw [shift={(270,80)}, rotate = 181.8] [color={rgb, 255:red, 191; green, 97; blue, 106 }  ,draw opacity=1 ][line width=0.75]    (4.37,-1.32) .. controls (2.78,-0.56) and (1.32,-0.12) .. (0,0) .. controls (1.32,0.12) and (2.78,0.56) .. (4.37,1.32)   ;
\draw [color={rgb, 255:red, 191; green, 97; blue, 106 }  ,draw opacity=1 ]   (282.04,54.86) .. controls (286.32,54.74) and (288.73,56.02) .. (290,60) .. controls (291.46,64.58) and (291.43,69.42) .. (290,75) .. controls (288.57,80.58) and (285.33,79.94) .. (280,80) ;
\draw [shift={(280,55)}, rotate = 354.07] [color={rgb, 255:red, 191; green, 97; blue, 106 }  ,draw opacity=1 ][line width=0.75]    (4.37,-1.32) .. controls (2.78,-0.56) and (1.32,-0.12) .. (0,0) .. controls (1.32,0.12) and (2.78,0.56) .. (4.37,1.32)   ;
\draw  [draw opacity=0][fill={rgb, 255:red, 191; green, 97; blue, 106 }  ,fill opacity=0.25 ] (272.5,80) .. controls (272.5,78.62) and (273.62,77.5) .. (275,77.5) .. controls (276.38,77.5) and (277.5,78.62) .. (277.5,80) .. controls (277.5,81.38) and (276.38,82.5) .. (275,82.5) .. controls (273.62,82.5) and (272.5,81.38) .. (272.5,80) -- cycle ;
\draw [color={rgb, 255:red, 0; green, 0; blue, 0 }  ,draw opacity=1 ][line width=1.5]    (230,45) -- (257,45) ;
\draw [shift={(260,45)}, rotate = 180] [color={rgb, 255:red, 0; green, 0; blue, 0 }  ,draw opacity=1 ][line width=1.5]    (8.53,-2.57) .. controls (5.42,-1.09) and (2.58,-0.23) .. (0,0) .. controls (2.58,0.23) and (5.42,1.09) .. (8.53,2.57)   ;
\draw   (285,35) -- (289.5,20) -- (310.5,20) -- (315,35) -- cycle ;
\draw    (300,20) -- (300,10) ;
\draw [color={rgb, 255:red, 191; green, 97; blue, 106 }  ,draw opacity=1 ]   (283.13,39.28) .. controls (281.96,38.82) and (280.98,38.09) .. (280,35) .. controls (278.54,30.42) and (278.57,25.58) .. (280,20) .. controls (281.43,14.42) and (289.67,15.06) .. (295,15) ;
\draw [shift={(285,40)}, rotate = 205.96] [color={rgb, 255:red, 191; green, 97; blue, 106 }  ,draw opacity=1 ][line width=0.75]    (4.37,-1.32) .. controls (2.78,-0.56) and (1.32,-0.12) .. (0,0) .. controls (1.32,0.12) and (2.78,0.56) .. (4.37,1.32)   ;
\draw [color={rgb, 255:red, 191; green, 97; blue, 106 }  ,draw opacity=1 ]   (295,40) -- (303,40) ;
\draw [shift={(305,40)}, rotate = 180] [color={rgb, 255:red, 191; green, 97; blue, 106 }  ,draw opacity=1 ][line width=0.75]    (4.37,-1.32) .. controls (2.78,-0.56) and (1.32,-0.12) .. (0,0) .. controls (1.32,0.12) and (2.78,0.56) .. (4.37,1.32)   ;
\draw [color={rgb, 255:red, 191; green, 97; blue, 106 }  ,draw opacity=1 ]   (315,40) .. controls (317.43,39.87) and (318.54,39.58) .. (320,35) .. controls (321.46,30.42) and (321.43,25.58) .. (320,20) .. controls (318.73,15.06) and (312.13,15) .. (306.94,15) ;
\draw [shift={(305,15)}, rotate = 0.66] [color={rgb, 255:red, 191; green, 97; blue, 106 }  ,draw opacity=1 ][line width=0.75]    (4.37,-1.32) .. controls (2.78,-0.56) and (1.32,-0.12) .. (0,0) .. controls (1.32,0.12) and (2.78,0.56) .. (4.37,1.32)   ;
\draw    (290,35) .. controls (290.03,49.73) and (275.03,44.73) .. (275,60) ;
\draw    (310,35) .. controls (310.03,49.73) and (325.03,44.73) .. (325,60) ;
\draw  [fill={rgb, 255:red, 191; green, 97; blue, 106 }  ,fill opacity=0.25 ] (315,75) -- (319.5,60) -- (330.5,60) -- (335,75) -- cycle ;
\draw  [draw opacity=0] (315,50) -- (335,50) -- (335,70) -- (315,70) -- cycle ;
\draw [color={rgb, 255:red, 191; green, 97; blue, 106 }  ,draw opacity=1 ]   (317.96,79.93) .. controls (312.06,79.68) and (311.32,79.14) .. (310,75) .. controls (308.54,70.42) and (308.57,65.58) .. (310,60) .. controls (311.43,54.42) and (314.67,55.06) .. (320,55) ;
\draw [shift={(320,80)}, rotate = 181.8] [color={rgb, 255:red, 191; green, 97; blue, 106 }  ,draw opacity=1 ][line width=0.75]    (4.37,-1.32) .. controls (2.78,-0.56) and (1.32,-0.12) .. (0,0) .. controls (1.32,0.12) and (2.78,0.56) .. (4.37,1.32)   ;
\draw [color={rgb, 255:red, 191; green, 97; blue, 106 }  ,draw opacity=1 ]   (332.04,54.86) .. controls (336.32,54.74) and (338.73,56.02) .. (340,60) .. controls (341.46,64.58) and (341.43,69.42) .. (340,75) .. controls (338.57,80.58) and (335.33,79.94) .. (330,80) ;
\draw [shift={(330,55)}, rotate = 354.07] [color={rgb, 255:red, 191; green, 97; blue, 106 }  ,draw opacity=1 ][line width=0.75]    (4.37,-1.32) .. controls (2.78,-0.56) and (1.32,-0.12) .. (0,0) .. controls (1.32,0.12) and (2.78,0.56) .. (4.37,1.32)   ;
\draw  [draw opacity=0][fill={rgb, 255:red, 191; green, 97; blue, 106 }  ,fill opacity=0.25 ] (322.5,80) .. controls (322.5,78.62) and (323.62,77.5) .. (325,77.5) .. controls (326.38,77.5) and (327.5,78.62) .. (327.5,80) .. controls (327.5,81.38) and (326.38,82.5) .. (325,82.5) .. controls (323.62,82.5) and (322.5,81.38) .. (322.5,80) -- cycle ;
\draw  [fill={rgb, 255:red, 191; green, 97; blue, 106 }  ,fill opacity=0.25 ] (195,35) -- (199.5,20) -- (210.5,20) -- (215,35) -- cycle ;
\draw    (205,20) -- (205,10) ;
\draw  [draw opacity=0] (195,35) -- (215,35) -- (215,55) -- (195,55) -- cycle ;
\draw [color={rgb, 255:red, 191; green, 97; blue, 106 }  ,draw opacity=1 ]   (197.96,39.93) .. controls (192.06,39.68) and (191.32,39.14) .. (190,35) .. controls (188.54,30.42) and (188.57,25.58) .. (190,20) .. controls (191.43,14.42) and (194.67,15.06) .. (200,15) ;
\draw [shift={(200,40)}, rotate = 181.8] [color={rgb, 255:red, 191; green, 97; blue, 106 }  ,draw opacity=1 ][line width=0.75]    (4.37,-1.32) .. controls (2.78,-0.56) and (1.32,-0.12) .. (0,0) .. controls (1.32,0.12) and (2.78,0.56) .. (4.37,1.32)   ;
\draw [color={rgb, 255:red, 191; green, 97; blue, 106 }  ,draw opacity=1 ]   (212.04,14.86) .. controls (216.32,14.74) and (218.73,16.02) .. (220,20) .. controls (221.46,24.58) and (221.43,29.42) .. (220,35) .. controls (218.57,40.58) and (215.33,39.94) .. (210,40) ;
\draw [shift={(210,15)}, rotate = 354.07] [color={rgb, 255:red, 191; green, 97; blue, 106 }  ,draw opacity=1 ][line width=0.75]    (4.37,-1.32) .. controls (2.78,-0.56) and (1.32,-0.12) .. (0,0) .. controls (1.32,0.12) and (2.78,0.56) .. (4.37,1.32)   ;
\draw  [draw opacity=0][fill={rgb, 255:red, 191; green, 97; blue, 106 }  ,fill opacity=0.25 ] (202.5,40) .. controls (202.5,38.62) and (203.62,37.5) .. (205,37.5) .. controls (206.38,37.5) and (207.5,38.62) .. (207.5,40) .. controls (207.5,41.38) and (206.38,42.5) .. (205,42.5) .. controls (203.62,42.5) and (202.5,41.38) .. (202.5,40) -- cycle ;

\draw (275,67.5) node  [font=\footnotesize]  {$c$};
\draw (297.5,55.5) node  [font=\tiny,color={rgb, 255:red, 191; green, 97; blue, 106 }  ,opacity=1 ]  {$c_{1}$};
\draw (252.5,67.5) node  [font=\tiny,color={rgb, 255:red, 191; green, 97; blue, 106 }  ,opacity=1 ]  {$c_{0}$};
\draw (300,27.5) node  [font=\footnotesize]  {$b$};
\draw (272.5,27.5) node  [font=\tiny,color={rgb, 255:red, 191; green, 97; blue, 106 }  ,opacity=1 ]  {$b_{0}$};
\draw (300,45) node  [font=\tiny,color={rgb, 255:red, 191; green, 97; blue, 106 }  ,opacity=1 ]  {$b_{1}$};
\draw (327.5,27.5) node  [font=\tiny,color={rgb, 255:red, 191; green, 97; blue, 106 }  ,opacity=1 ]  {$b_{2}$};
\draw (325,67.5) node  [font=\footnotesize]  {$d$};
\draw (347.5,67.5) node  [font=\tiny,color={rgb, 255:red, 191; green, 97; blue, 106 }  ,opacity=1 ]  {$d_{1}$};
\draw (302.5,74.5) node  [font=\tiny,color={rgb, 255:red, 191; green, 97; blue, 106 }  ,opacity=1 ]  {$d_{0}$};
\draw (205,27.5) node  [font=\footnotesize]  {$a$};
\draw (227.5,27.5) node  [font=\tiny,color={rgb, 255:red, 191; green, 97; blue, 106 }  ,opacity=1 ]  {$a_{1}$};
\draw (182.5,27.5) node  [font=\tiny,color={rgb, 255:red, 191; green, 97; blue, 106 }  ,opacity=1 ]  {$a_{0}$};

\end{tikzpicture}
       \caption{Equations for the second laxator.}
      \label{ax:fig:monoidal-contour-equation2}
    \end{figure}
    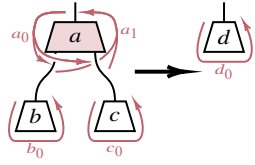
\begin{figure}[ht!]
      \centering

\tikzset{every picture/.style={line width=0.75pt}} %

\begin{tikzpicture}[x=0.75pt,y=0.75pt,yscale=-1,xscale=1]
\draw [color={rgb, 255:red, 191; green, 97; blue, 106 }  ,draw opacity=1 ]   (187.82,39.94) .. controls (168.45,39.26) and (166.4,34.42) .. (165,30) .. controls (163.54,25.42) and (163.57,25.58) .. (165,20) .. controls (166.43,14.42) and (174.67,15.06) .. (180,15) ;
\draw [shift={(190,40)}, rotate = 181.03] [color={rgb, 255:red, 191; green, 97; blue, 106 }  ,draw opacity=1 ][line width=0.75]    (4.37,-1.32) .. controls (2.78,-0.56) and (1.32,-0.12) .. (0,0) .. controls (1.32,0.12) and (2.78,0.56) .. (4.37,1.32)   ;
\draw    (175,35) .. controls (175.03,49.73) and (165.03,44.73) .. (165,60) ;
\draw  [color={rgb, 255:red, 0; green, 0; blue, 0 }  ,draw opacity=1 ][fill={rgb, 255:red, 191; green, 97; blue, 106 }  ,fill opacity=0.25 ] (170,35) -- (174.5,20) -- (195.5,20) -- (200,35) -- cycle ;
\draw [color={rgb, 255:red, 191; green, 97; blue, 106 }  ,draw opacity=1 ]   (169.17,39.65) .. controls (167.5,39.04) and (166.26,38.95) .. (165,35) .. controls (163.54,30.42) and (163.57,25.58) .. (165,20) .. controls (166.43,14.42) and (174.67,15.06) .. (180,15) ;
\draw [shift={(170,40)}, rotate = 205.96] [color={rgb, 255:red, 191; green, 97; blue, 106 }  ,draw opacity=1 ][line width=0.75]    (4.37,-1.32) .. controls (2.78,-0.56) and (1.32,-0.12) .. (0,0) .. controls (1.32,0.12) and (2.78,0.56) .. (4.37,1.32)   ;
\draw [color={rgb, 255:red, 191; green, 97; blue, 106 }  ,draw opacity=1 ]   (200,45) .. controls (202.43,44.87) and (203.54,39.58) .. (205,35) .. controls (206.46,30.42) and (206.43,25.58) .. (205,20) .. controls (203.73,15.06) and (197.13,15) .. (191.94,15) ;
\draw [shift={(190,15)}, rotate = 0.66] [color={rgb, 255:red, 191; green, 97; blue, 106 }  ,draw opacity=1 ][line width=0.75]    (4.37,-1.32) .. controls (2.78,-0.56) and (1.32,-0.12) .. (0,0) .. controls (1.32,0.12) and (2.78,0.56) .. (4.37,1.32)   ;
\draw [color={rgb, 255:red, 191; green, 97; blue, 106 }  ,draw opacity=1 ]   (175,45) .. controls (195.36,44.69) and (203.54,34.58) .. (205,30) .. controls (206.46,25.42) and (206.43,25.58) .. (205,20) .. controls (203.73,15.06) and (197.13,15) .. (191.94,15) ;
\draw [shift={(190,15)}, rotate = 0.66] [color={rgb, 255:red, 191; green, 97; blue, 106 }  ,draw opacity=1 ][line width=0.75]    (4.37,-1.32) .. controls (2.78,-0.56) and (1.32,-0.12) .. (0,0) .. controls (1.32,0.12) and (2.78,0.56) .. (4.37,1.32)   ;
\draw  [draw opacity=0][fill={rgb, 255:red, 255; green, 255; blue, 255 }  ,fill opacity=1 ] (192.5,40) .. controls (192.5,38.62) and (193.62,37.5) .. (195,37.5) .. controls (196.38,37.5) and (197.5,38.62) .. (197.5,40) .. controls (197.5,41.38) and (196.38,42.5) .. (195,42.5) .. controls (193.62,42.5) and (192.5,41.38) .. (192.5,40) -- cycle ;
\draw    (185,20) -- (185,10) ;
\draw  [draw opacity=0][fill={rgb, 255:red, 255; green, 255; blue, 255 }  ,fill opacity=1 ] (172,38) .. controls (172,36.9) and (173.34,36) .. (175,36) .. controls (176.66,36) and (178,36.9) .. (178,38) .. controls (178,39.1) and (176.66,40) .. (175,40) .. controls (173.34,40) and (172,39.1) .. (172,38) -- cycle ;
\draw  [draw opacity=0][fill={rgb, 255:red, 255; green, 255; blue, 255 }  ,fill opacity=1 ] (197,45) .. controls (197,43.9) and (198.34,43) .. (200,43) .. controls (201.66,43) and (203,43.9) .. (203,45) .. controls (203,46.1) and (201.66,47) .. (200,47) .. controls (198.34,47) and (197,46.1) .. (197,45) -- cycle ;
\draw [color={rgb, 255:red, 191; green, 97; blue, 106 }  ,draw opacity=1 ]   (187.82,39.94) .. controls (168.45,39.26) and (166.4,34.42) .. (165,30) .. controls (163.54,25.42) and (163.57,25.58) .. (165,20) .. controls (166.43,14.42) and (174.67,15.06) .. (180,15) ;
\draw [shift={(190,40)}, rotate = 181.03] [color={rgb, 255:red, 191; green, 97; blue, 106 }  ,draw opacity=1 ][line width=0.75]    (4.37,-1.32) .. controls (2.78,-0.56) and (1.32,-0.12) .. (0,0) .. controls (1.32,0.12) and (2.78,0.56) .. (4.37,1.32)   ;
\draw    (195,35) .. controls (195.03,49.73) and (205.03,44.73) .. (205,60) ;
\draw [color={rgb, 255:red, 0; green, 0; blue, 0 }  ,draw opacity=1 ][line width=1.5]    (215,45) -- (237,45) ;
\draw [shift={(240,45)}, rotate = 180] [color={rgb, 255:red, 0; green, 0; blue, 0 }  ,draw opacity=1 ][line width=1.5]    (8.53,-2.57) .. controls (5.42,-1.09) and (2.58,-0.23) .. (0,0) .. controls (2.58,0.23) and (5.42,1.09) .. (8.53,2.57)   ;
\draw   (155,75) -- (159.5,60) -- (170.5,60) -- (175,75) -- cycle ;
\draw [color={rgb, 255:red, 191; green, 97; blue, 106 }  ,draw opacity=1 ]   (155,60) .. controls (151.32,68.79) and (149.61,79.97) .. (155,80) .. controls (160.39,80.03) and (169.61,79.97) .. (175,80) .. controls (180.06,80.02) and (180.31,70.69) .. (175.9,61.71) ;
\draw [shift={(175,60)}, rotate = 60.63] [color={rgb, 255:red, 191; green, 97; blue, 106 }  ,draw opacity=1 ][line width=0.75]    (4.37,-1.32) .. controls (2.78,-0.56) and (1.32,-0.12) .. (0,0) .. controls (1.32,0.12) and (2.78,0.56) .. (4.37,1.32)   ;
\draw  [draw opacity=0] (155,75) -- (175,75) -- (175,95) -- (155,95) -- cycle ;
\draw   (195,75) -- (199.5,60) -- (210.5,60) -- (215,75) -- cycle ;
\draw [color={rgb, 255:red, 191; green, 97; blue, 106 }  ,draw opacity=1 ]   (195,60) .. controls (191.32,68.79) and (189.61,79.97) .. (195,80) .. controls (200.39,80.03) and (209.61,79.97) .. (215,80) .. controls (220.06,80.02) and (220.31,70.69) .. (215.9,61.71) ;
\draw [shift={(215,60)}, rotate = 60.63] [color={rgb, 255:red, 191; green, 97; blue, 106 }  ,draw opacity=1 ][line width=0.75]    (4.37,-1.32) .. controls (2.78,-0.56) and (1.32,-0.12) .. (0,0) .. controls (1.32,0.12) and (2.78,0.56) .. (4.37,1.32)   ;
\draw  [draw opacity=0] (195,75) -- (215,75) -- (215,95) -- (195,95) -- cycle ;
\draw   (250,35) -- (254.5,20) -- (265.5,20) -- (270,35) -- cycle ;
\draw    (260,20) -- (260,10) ;
\draw [color={rgb, 255:red, 191; green, 97; blue, 106 }  ,draw opacity=1 ]   (250,20) .. controls (246.32,28.79) and (244.61,39.97) .. (250,40) .. controls (255.39,40.03) and (264.61,39.97) .. (270,40) .. controls (275.06,40.02) and (275.31,30.69) .. (270.9,21.71) ;
\draw [shift={(270,20)}, rotate = 60.63] [color={rgb, 255:red, 191; green, 97; blue, 106 }  ,draw opacity=1 ][line width=0.75]    (4.37,-1.32) .. controls (2.78,-0.56) and (1.32,-0.12) .. (0,0) .. controls (1.32,0.12) and (2.78,0.56) .. (4.37,1.32)   ;
\draw  [draw opacity=0] (250,35) -- (270,35) -- (270,55) -- (250,55) -- cycle ;

\draw (185,27.5) node  [font=\footnotesize]  {$a$};
\draw (157.5,27.5) node  [font=\tiny,color={rgb, 255:red, 191; green, 97; blue, 106 }  ,opacity=1 ]  {$a_{0}$};
\draw (212.5,25.5) node  [font=\tiny,color={rgb, 255:red, 191; green, 97; blue, 106 }  ,opacity=1 ]  {$a_{1}$};
\draw (165,67.5) node  [font=\footnotesize]  {$b$};
\draw (165,85) node  [font=\tiny,color={rgb, 255:red, 191; green, 97; blue, 106 }  ,opacity=1 ]  {$b_{0}$};
\draw (205,67.5) node  [font=\footnotesize]  {$c$};
\draw (205,85) node  [font=\tiny,color={rgb, 255:red, 191; green, 97; blue, 106 }  ,opacity=1 ]  {$c_{0}$};
\draw (260,27.5) node  [font=\footnotesize]  {$d$};
\draw (260,45) node  [font=\tiny,color={rgb, 255:red, 191; green, 97; blue, 106 }  ,opacity=1 ]  {$d_{0}$};

\end{tikzpicture}
       \caption{Equations for the third laxator.}
      \label{ax:fig:monoidal-contour-equation3}
    \end{figure}
    \begin{figure}[ht!] 
      \centering

\tikzset{every picture/.style={line width=0.75pt}} %

\begin{tikzpicture}[x=0.75pt,y=0.75pt,yscale=-1,xscale=1]
\draw [color={rgb, 255:red, 0; green, 0; blue, 0 }  ,draw opacity=1 ][line width=1.5]    (215,25) -- (242,25) ;
\draw [shift={(245,25)}, rotate = 180] [color={rgb, 255:red, 0; green, 0; blue, 0 }  ,draw opacity=1 ][line width=1.5]    (8.53,-2.57) .. controls (5.42,-1.09) and (2.58,-0.23) .. (0,0) .. controls (2.58,0.23) and (5.42,1.09) .. (8.53,2.57)   ;
\draw   (250,35) -- (254.5,20) -- (265.5,20) -- (270,35) -- cycle ;
\draw    (260,20) -- (260,10) ;
\draw [color={rgb, 255:red, 191; green, 97; blue, 106 }  ,draw opacity=1 ]   (250,20) .. controls (246.32,28.79) and (244.61,39.97) .. (250,40) .. controls (255.39,40.03) and (264.61,39.97) .. (270,40) .. controls (275.06,40.02) and (275.31,30.69) .. (270.9,21.71) ;
\draw [shift={(270,20)}, rotate = 60.63] [color={rgb, 255:red, 191; green, 97; blue, 106 }  ,draw opacity=1 ][line width=0.75]    (4.37,-1.32) .. controls (2.78,-0.56) and (1.32,-0.12) .. (0,0) .. controls (1.32,0.12) and (2.78,0.56) .. (4.37,1.32)   ;
\draw  [draw opacity=0] (250,35) -- (270,35) -- (270,55) -- (250,55) -- cycle ;
\draw  [fill={rgb, 255:red, 191; green, 97; blue, 106 }  ,fill opacity=0.25 ] (175,35) -- (179.5,20) -- (190.5,20) -- (195,35) -- cycle ;
\draw    (185,20) -- (185,10) ;
\draw  [draw opacity=0] (175,35) -- (195,35) -- (195,55) -- (175,55) -- cycle ;
\draw [color={rgb, 255:red, 191; green, 97; blue, 106 }  ,draw opacity=1 ]   (177.96,39.93) .. controls (172.06,39.68) and (171.32,39.14) .. (170,35) .. controls (168.54,30.42) and (168.57,25.58) .. (170,20) .. controls (171.43,14.42) and (174.67,15.06) .. (180,15) ;
\draw [shift={(180,40)}, rotate = 181.8] [color={rgb, 255:red, 191; green, 97; blue, 106 }  ,draw opacity=1 ][line width=0.75]    (4.37,-1.32) .. controls (2.78,-0.56) and (1.32,-0.12) .. (0,0) .. controls (1.32,0.12) and (2.78,0.56) .. (4.37,1.32)   ;
\draw [color={rgb, 255:red, 191; green, 97; blue, 106 }  ,draw opacity=1 ]   (192.04,14.86) .. controls (196.32,14.74) and (198.73,16.02) .. (200,20) .. controls (201.46,24.58) and (201.43,29.42) .. (200,35) .. controls (198.57,40.58) and (195.33,39.94) .. (190,40) ;
\draw [shift={(190,15)}, rotate = 354.07] [color={rgb, 255:red, 191; green, 97; blue, 106 }  ,draw opacity=1 ][line width=0.75]    (4.37,-1.32) .. controls (2.78,-0.56) and (1.32,-0.12) .. (0,0) .. controls (1.32,0.12) and (2.78,0.56) .. (4.37,1.32)   ;
\draw  [draw opacity=0][fill={rgb, 255:red, 191; green, 97; blue, 106 }  ,fill opacity=0.25 ] (182.5,40) .. controls (182.5,38.62) and (183.62,37.5) .. (185,37.5) .. controls (186.38,37.5) and (187.5,38.62) .. (187.5,40) .. controls (187.5,41.38) and (186.38,42.5) .. (185,42.5) .. controls (183.62,42.5) and (182.5,41.38) .. (182.5,40) -- cycle ;

\draw (260,27.5) node  [font=\footnotesize]  {$b$};
\draw (260,45) node  [font=\tiny,color={rgb, 255:red, 191; green, 97; blue, 106 }  ,opacity=1 ]  {$b_{0}$};
\draw (185,27.5) node  [font=\footnotesize]  {$a$};
\draw (207.5,27.5) node  [font=\tiny,color={rgb, 255:red, 191; green, 97; blue, 106 }  ,opacity=1 ]  {$a_{1}$};
\draw (162.5,27.5) node  [font=\tiny,color={rgb, 255:red, 191; green, 97; blue, 106 }  ,opacity=1 ]  {$a_{0}$};

\end{tikzpicture}
       \caption{Equations for the fourth laxator.}
      \label{ax:fig:monoidal-contour-equation4}
    \end{figure}
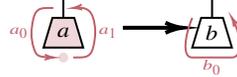  
  \end{definition}

  \begin{proposition}[From \Cref{prop:monoidalContourFunctor}] \label{ax:prop:monoidalContourFunctor}
    \MonoidalContour{} gives a functor $𝓓 : \pDuo → \Mon$.
  \end{proposition}
  \begin{proof}
    \Cref{def:monoidalContour} defines the action on \produoidalCategories{}. We define the action on \produoidalFunctors{}. Given a \produoidalFunctor{} $F: 𝕍 → 𝕎$, define the strict monoidal functor $𝓓F : 𝓓𝕍 → 𝓓𝕎$ by the following morphism of presentations:
    \begin{itemize}
      \item the objects $X^{L}$ and $X^{R}$ are mapped to $F(X)^{L}$ and $F(X)^{R}$;
      \item for each $a \in 𝕍(A;N)$, the morphism $a_0 : A^L \to A$ is mapped to $F_N(a)_0$;
      \item for each $b \in 𝕍(A;I)$, both $b_0 : A^L \to I$ and $b_1 : I → A^R$ are mapped to $F_I(b)_0$ and $F_I(b)_1$;
      \item for each $c \in 𝕍(X;B)$, the morphisms $c_0 : B^L \to X^L, c_1 : X^R \to B^R$ are mapped to $F(c)_0$ and $F(c)_1$;
      \item for each $d \in 𝕍(C; Y ⊗ Z)$, the morphisms $d_0 : C^L \to Y^L$, $d_1 : Y^R \to Z^L$ and $d_2 : Z^R \to C^R$ are mapped to $F_{◁}(d)_0$, $F_{◁}(d)_1$ and $F_{◁}(d)_2$;
      \item for each $e \in 𝕍(C; Y ◁ Z)$, the morphisms $e_0 : C^L \to Y^L$, $e_1 : Y^R \to Z^L$ and $e_2 : Z^R \to C^R$ are mapped to $F_{◁}(e)_0$, $F_{◁}(e)_1$ and $F_{◁}(e)$.
    \end{itemize}
  It follows from $F: 𝕍 \to 𝕎$ being a \produoidalFunctor{} that the contour equations of \Cref{defn:contour} hold between the images of generators, so this assignment extends freely to a strict monoidal functor. In particular when $\idFun{_𝕍} : 𝕍 \to 𝕍$ is an identity, it is an identity functor. Let $G : 𝕌 → 𝕍$ be another \produoidalFunctor{}, then $𝓒(G⨾F) = 𝓒(G)⨾𝓒(F)$ follows from the composition of \produoidalFunctors{}.
  \end{proof}

\subsection{Spliced Monoidal Arrows} \label{ax:sec:produoidalSplice}

  \begin{proposition}[From \Cref{prop:spliceIsProduoidal}]
    \label{ax:prop:spliceIsProduoidal}
    \SplicedMonoidalArrows{} form a \produoidalCategory{} with their sequential and parallel splits, units, and suitable coherence morphisms and laxators.
  \end{proposition}
  \begin{proof}
    We use the laxators constructed in \Cref{lemma:produoidalFirstlaxator,lemma:produoidalSecondlaxator,lemma:produoidalThirdlaxator,lemma:produoidalFourthlaxator}. Because these laxators are constructed out of compositions and Yoneda lemma, they do satisfy all formal coherence equations.
  \end{proof}

  \begin{lemma}[Produoidal splice, first laxator]
    \label{lemma:produoidalFirstlaxator}
    We can construct a natural transformation,
    $$\psi_2 \colon
    𝓣{ℂ} \left(
        \biobj{A}{B} ;
        \left( \biobj{X}{Y} ◁ \biobj{X'}{Y'} \right) ⊗
        \left( \biobj{U}{V} ◁ \biobj{U'}{V'} \right)
    \right) \to
    𝓣{ℂ} \left(
        \biobj{A}{B} ;
        \left( \biobj{X}{Y}   ⊗ \biobj{U}{V}   \right) ◁
        \left( \biobj{X'}{Y'} ⊗ \biobj{U'}{V'} \right)
    \right),$$
    exclusively from compositions and Yoneda isomorphisms. This laxator is defined by $$ψ_2(\bisplice{f_0}{f_1} \mathbin{|} \trisplice{h_0}{h_1}{h_2} \mathbin{|} \trisplice{k_0}{k_1}{k_2}) = \trisplice{g_0}{g_1}{g_2}|\bisplice{p_0}{p_1}|\bisplice{q_0}{q_1}$$ if and only if
    $$\trisplice{f_0 ⨾ (h_0 ⊗ k_0)}{h_1 ⊗ k_1}{(h_2 ⊗ k_2) ⨾ f_1} = \trisplice{g_0 ⨾ p_0}{p_1 ⨾ g_1 ⨾ q_0}{q_1 ⨾ g_2}.$$
  \end{lemma}
  \begin{proof}
    We will show that the right hand side is isomorphic to the following set.
    Then, we construct a map from the left hand to this same set,
    $$ℂ(A ; X ⊗ X') × ℂ(Y ⊗ Y' ; U ⊗ U') × ℂ(V ⊗ V' ; B).$$
    Indeed the following coend derivation constructs an isomorphism.
    \begin{align*}
        & 𝓣{ℂ} \left(
            \biobj{A}{B} ;
            \left( \biobj{X}{Y}   ⊗ \biobj{X'}{Y'}   \right) ◁
            \left( \biobj{U}{V} ⊗ \biobj{U'}{V'} \right)
        \right) \phantom{\int} &\bydef \\
        & \int^{\biobj{Z}{W}, \biobj{Z'}{W'} \in 𝓣{ℂ}}
        𝓣{ℂ} \left( \biobj{A}{B} ; \biobj{Z}{W} ◁ \biobj{Z'}{W'} \right) ×
        𝓣{ℂ} \left( \biobj{Z}{W} ; \biobj{X}{Y} ⊗ \biobj{X'}{Y'} \right) ×
        𝓣{ℂ} \left( \biobj{Z'}{W'} ;  \biobj{U}{V}⊗ \biobj{U'}{V'} \right) &\bydef \\
        & \int^{\biobj{Z}{W}, \biobj{Z'}{W'} \in 𝓣{ℂ}}
        𝓣{ℂ} \left( \biobj{A}{B} ; \biobj{Z}{W} ◁ \biobj{Z'}{W'} \right) ×
        ℂ(Z ; X ⊗ X') × ℂ(Y ⊗ Y' ; W) × ℂ(Z' ; U ⊗ U') × ℂ(V ⊗ V' ; W') &\bydef \\
        & \int^{\biobj{Z}{W}, \biobj{Z'}{W'} \in 𝓣{ℂ}}
        𝓣{ℂ} \left( \biobj{A}{B} ; \biobj{Z}{W} ◁ \biobj{Z'}{W'} \right) ×
        𝓣{ℂ}(\biobj{Z}{W} ; \biobj{X ⊗ X'}{Y ⊗ Y'}) × 𝓣{ℂ}(\biobj{Z'}{W'} ; \biobj{U ⊗ U'}{V ⊗ V'}) & \yo1\\
        &
        𝓣{ℂ} \left( \biobj{A}{B} ; \biobj{X ⊗ X'}{Y ⊗ Y'} ◁ \biobj{U ⊗ U'}{V ⊗ V'} \right) \phantom{∫} & \bydef\\
        &  ℂ(A ; X ⊗ X') × ℂ(Y ⊗ Y'; U ⊗ U') × ℂ(V ⊗ V' ; B).\phantom{\int}
    \end{align*}
    The isomorphism sends the triple $(\trisplice{g_0}{g_1}{g_2}|\bisplice{p_0}{p_1}|\bisplice{q_0}{q_1})$ to $\trisplice{g_0 ⨾ p_0}{p_1 ⨾ g_1 ⨾ q_0}{q_1 ⨾ g_2}$.
    On the other hand, we define a map from the left hand side of the equation to this set, given by
    $$\bisplice{f_0}{f_1} \mathbin{|} \trisplice{h_0}{h_1}{h_2} \mathbin{|} \trisplice{k_0}{k_1}{k_2}
    \mapsto
    \trisplice{f_0 ⨾ (h_0 ⊗ k_0)}{h_1 ⊗ k_1}{(h_2 ⊗ k_2) ⨾ f_1}.$$
    In conclusion, composing both the isomorphism and the map, $ψ_2(\bisplice{f_0}{f_1} \mathbin{|} \trisplice{h_0}{h_1}{h_2} \mathbin{|} \trisplice{k_0}{k_1}{k_2}) = \trisplice{g_0}{g_1}{g_2}|\bisplice{p_0}{p_1}|\bisplice{q_0}{q_1}$ if and only if
    $$\trisplice{f_0 ⨾ (h_0 ⊗ k_0)}{(h_1 ⊗ k_1)}{(h_2 ⊗ k_2) ⨾ f_1} = \trisplice{g_0 ⨾ p_0}{p_1 ⨾ g_1 ⨾ q_0}{q_1 ⨾ g_2},$$
    which is what we wanted to prove.
  \end{proof}

  \begin{lemma}[Produoidal splice, second laxator]
    \label{lemma:produoidalSecondlaxator}
    We can construct a natural transformation,
    $$\psi_0 \colon
    𝓣{ℂ} \left(\biobj{A}{B} ; I \right) \to
    𝓣{ℂ} \left(\biobj{A}{B} ; I ◁ I\right),$$
    exclusively from compositions and \YonedaIsomorphisms{}. 
    This laxator is defined by $ψ_0(\bisplice{f_0}{f_1} \mathbin{|} \trisplice{h_0}{h_1}{h_2} \mathbin{|} \trisplice{k_0}{k_1}{k_2}) = \trisplice{g_0}{g_1}{g_2}|\bisplice{p_0}{p_1}|\bisplice{q_0}{q_1}$ if and only if
    $$\trisplice{f_0 ⨾ (h_0 ⊗ k_0)}{(h_1 ⊗ k_1)}{(h_2 ⊗ k_2) ⨾ f_1} = \trisplice{g_0 ⨾ p_0}{p_1 ⨾ g_1 ⨾ q_0}{q_1 ⨾ g_2}.$$
  \end{lemma}
  \begin{proof}
    We will show that the right hand side is isomorphic to the following set.
    Then, we construct a map from the left hand to this same set,
    $ℂ(A;I) × ℂ(I;I) × ℂ(I;B)$. Indeed the following coend derivation constructs an isomorphism.
    \begin{align*}
        & 𝓣{ℂ} \left(
            \biobj{A}{B} ;
            I ◁ I
        \right) \phantom{\int} & \bydef \\
        & \int^{\biobj{Z}{W}, \biobj{Z'}{W'} \in 𝓣{ℂ}}
        𝓣{ℂ} \left( \biobj{A}{B} ; \biobj{Z}{W} ◁ \biobj{Z'}{W'} \right) ×
        𝓣{ℂ} \left( \biobj{Z}{W} ; I \right) ×
        𝓣{ℂ} \left( \biobj{Z'}{W'} ;  I \right) & \bydef \\
        & \int^{\biobj{Z}{W}, \biobj{Z'}{W'} \in 𝓣{ℂ}}
        𝓣{ℂ} \left( \biobj{A}{B} ; \biobj{Z}{W} ◁ \biobj{Z'}{W'} \right) ×
        𝓣{ℂ} \left( \biobj{Z}{W} ; \biobj{I}{I} \right) ×
        𝓣{ℂ} \left( \biobj{Z'}{W'} ; \biobj{I}{I} \right) & \yo1 \\
        &
        𝓣{ℂ} \left( \biobj{A}{B} ; \biobj{I}{I} ◁ \biobj{I}{I} \right) \phantom{\int} & \bydef \\
        &
        ℂ(A;I) × ℂ(I;I) × ℂ(I; B).  \phantom{\int}\\
    \end{align*}
    This isomorphism sends the triple $\trisplice{b_0}{b_1}{b_2} \mathbin{|} \bisplice{c_0}{c_1} \mathbin{|} \bisplice{d_0}{d_1}$ to $\trisplice{b_0 ⨾ c_0}{c_1 ⨾ b_1 ⨾ d_0}{d_1 ⨾ b_2}$. On the other hand, we define a map from the left hand side of the equation to this set, given by
    $$\bisplice{a_0}{a_1} \mapsto \trisplice{a_0}{\im_I}{a_1}.$$
    In conclusion, composing both the isomorphism and this function, we get that $\psi_0\bisplice{a_0}{a_1} = \trisplice{b_0}{b_1}{b_2} \mathbin{|} \bisplice{c_0}{c_1} \mathbin{|} \bisplice{d_0}{d_1}$ if and only if $\trisplice{a_0}{\im_I}{a_1} = \trisplice{b_0 ⨾ c_0}{c_1 ⨾ b_1 ⨾ d_0}{d_1 ⨾ b_2}$.
  \end{proof}

  \begin{lemma}[Produoidal splice, third laxator]
    \label{lemma:produoidalThirdlaxator}
    We can construct a natural transformation,
    $$\varphi_2 \colon
    𝓣{ℂ} \left(\biobj{A}{B} ; N ⊗ N \right) \to
    𝓣{ℂ} \left(\biobj{A}{B} ; N \right),$$
    exclusively from compositions and Yoneda isomorphisms. This laxator is defined by $\varphi_2(\bisplice{f_0}{f_1} \mathbin{|} h_0 \mathbin{|} h_1) = f_0 ⨾ (h_0 ⊗ h_1 ⨾ f_1)$.
  \end{lemma}
  \begin{lemma}[Produoidal splice, fourth laxator]
    \label{lemma:produoidalFourthlaxator}
    We can construct a natural transformation,
    $$\varphi_0 \colon
    𝓣{ℂ} \left(\biobj{A}{B} ; I \right) \to
    𝓣{ℂ} \left(\biobj{A}{B} ; N \right),$$
    exclusively from compositions and Yoneda isomorphisms. This laxator is defined by $\varphi_0\bisplice{a_0}{a_1} = a_0 ⨾ a_1$.
  \end{lemma}

  \begin{proposition}[From \Cref{prop:monoidalSpliceFunctor}] \label{ax:prop:monoidalSpliceFunctor}
    Monoidal splice gives a functor $𝓣{} : \Mon → \pDuo$. %
  \end{proposition}
  \begin{proof}
    \Cref{def:monoidalSplice} defines the action on \monoidalCategories{}. We define the action on monoidal functors. Given a monoidal functor $F : ℂ \to 𝔻$, define the \produoidalFunctor{} $𝓣F : 𝓣ℂ → 𝓣𝔻$ by
    \begin{align*}
      &\biobj{A}{B} ↦ \biobj{FA}{FB} \\
      &𝓣F := F_{A,X} × F_{Y,B} : 𝓣{ℂ}(\biobj{A}{B}, \biobj{X}{Y}) → 𝓣{𝔻}(\biobj{FA}{FB}, \biobj{FX}{FY}) \\
      &𝓣{F}_{◁} := F_{A,X} × F_{Y,X'} × F_{Y',B} : 𝓣{ℂ}(\biobj{A}{B}, \biobj{X}{Y} ◁ \biobj{X'}{Y'}) → 𝓣{𝔻}(\biobj{FA}{FB}, \biobj{FX}{FY} ◁ \biobj{FX'}{FY'}) \\
      &𝓣{F}_{⊗} := F_{A,X⊗Y} × F_{X'⊗Y',B} : 𝓣{ℂ}(\biobj{A}{B}, \biobj{X}{Y} ⊗ \biobj{X'}{Y'}) → 𝓣{𝔻}(\biobj{FA}{FB}, \biobj{FX}{FY} ⊗ \biobj{FX'}{FY'}) \\
      &𝓣{F}_{N} := F_{A,B} : 𝓣{ℂ}(\biobj{A}{B}, N) → 𝓣{𝔻}(\biobj{FA}{FB}, N) \\
      &𝓣{F}_{I} := F_{A,I} × F_{I,B} : 𝓣{ℂ}(\biobj{A}{B}, I) → 𝓣{𝔻}(\biobj{FA}{FB}, I). \\
    \end{align*}

    It follows from the produoidal structure on \splicedMonoidalArrows{} (\Cref{ax:prop:spliceIsProduoidal}) that this preserves coherence maps. If $\idFun{_ℂ} : ℂ → ℂ$ is an identity functor, then it defines the identity $\idFun{_{𝓣{ℂ}}}$, which has underlying functor the identity and identity natural transformations. If $G : 𝔹 → ℂ$ is another monoidal functor, then $\Splice{(G ⨾ F)} = \Splice{G} ⨾ \Splice{F}$ follows from composition of monoidal functors.
  \end{proof}

\begin{theorem}[From \Cref{prop:produoidalSpliceContour}]
  \label{ax:prop:produoidalSpliceContour}
  There exists an adjunction between monoidal categories (and strict monoidal functors) and \produoidalCategories{} (and produoidal functors), where the \monoidalContour{} is the left adjoint, and the \produoidalSplice{} category is the right adjoint.
\end{theorem}
\begin{proof}
  As in \Cref{ax:th:catpromadj}, we again have that $𝓓{𝔹}$ is presented by generators and equations; so, to specify a strict monoidal functor $𝓓{𝔹} → 𝕄$, it is enough to specify images of the generators satisfying the equations. Let $(𝕄,⊗_M,I_M)$ be a monoidal category. Then a strict monoidal functor $𝓓{𝔹} \to 𝕄$ amounts to the following data.
  \begin{itemize}
  \item For each object $X \in \obj{𝔹}$, a pair of objects $X^{L}, X^{R} \in \obj{𝕄}$;
  \item for each element $f \in 𝔹(X;N)$, a morphism $f_0 \in 𝕄(X^L;X^R)$;
  \item for each unit $f \in 𝔹(X;I)$, a choice of morphisms $f_0 \in 𝕄(X^L;I_M)$, $g_0 \in 𝕄(I_M;X^R)$;
  \item for each morphism $f \in 𝔹(A;X)$, a choice of morphisms $f_0 \in 𝕄(A^L; X^L)$ and $f_1 \in 𝕄(X^R; A^R)$;
  \item for each sequential split $f \in 𝔹(A; X ◁ Y)$, a choice of morphisms $f_0 \in 𝕄(A^L;X^L), f_1 \in 𝕄(X^L;X^R),\mbox{ and } f_2 \in 𝕄(X^R,A^R)$;
  \item for each parallel split $f \in 𝔹(A; X ⊗ Y)$, a choice of morphisms $f_0 \in 𝕄(A^L;X^L ⊗ Y^L)\mbox{ and }f_1 \in 𝕄(X^R ⊗ Y^R; A^R)$.
  \end{itemize}

  \noindent Such that for each promonoidal structure
  \begin{itemize} 
  \item $α(a ⨾_1 b) = (c ⨾_2 d)$ in $𝔹$ $\implies$ $a_0 ⨾ (b_0 ⊗ \im) = c_0 ⨾ (\im ⊗ d_0)$ and $(b_1 ⊗ \im) ⨾ a_1 = (\im ⊗ d_1) ⨾ c_1$ in $𝕄$;
  \item $λ(a ⨾_1 b) = c = ρ(d ⨾_2 e)$ in $𝔹$ $\implies$ $a_0 ⨾ (b_0 ⊗ \im) = c_0 = d_0 ⨾ (\im ⊗ e_0)$ and  $(b_1 ⊗ \im) ⨾ a_1 = c_1 = (\im ⊗ e_1) ⨾ d_1$ in $𝕄$;
  \end{itemize}
  \noindent and such that
  \begin{itemize}
  \item $ψ_2(a ｜ b ｜ c) = (d ｜ e ｜ f)$ in $𝔹$ $\implies$ $a_0 ⨾ (b_0 ⊗ c_0)  = d_0 ⨾ e_0, b_1 ⊗ c_1 = e_1 ⨾ d_1 ⨾ f_0$ and $(b_2 ⊗ c_2) ⨾ a_1 = f_1 ⨾ d_2$ in $𝕄$;
  \item $ψ_0(a) = (b ｜ c ｜ d)$ in $𝔹$ $\implies$ $a_0 = b_0 ⨾ c_0$, $\im = c_1 ⨾ b_1 ⨾ d_0$, and $a_1 = d_1 ⨾ b_2$ in $𝕄$;
  \item $φ_2(a ｜ b ｜ c) = d$ in $𝔹$ $\implies$ $a_0 ⨾ (b_0 ⊗ c_0) ⨾ a_1 = d_0$ in $𝕄$;
  \item $φ_0(a) = b$ in $𝔹$ $\implies$ $a_0 ⨾ a_1 = b_0$ in $𝕄$.
  \end{itemize}

  \noindent On the other hand, a \produoidalFunctor{} $F : 𝔹 \to 𝓣{𝕄}$, also amounts to the following data. For each
  \begin{itemize}
  \item $X \in \obj{𝔹}$ an object $F(X) = (X^L, X^R) \in \obj{𝓣{𝕄}}$;
  \item $f \in 𝔹(X;N)$, an element $F(f) = f_0 \in 𝓣{𝕄}(\biobj{X^L}{X^R};N)$;
  \item $f \in 𝔹(X;I)$, a unit $F(f) = ⟨ f \parallel g ⟩ \in 𝓣{𝕄}(\biobj{X^L}{X^R};I_M)$
  \item $f \in 𝔹(A;X)$, a spliced arrow $F(f) = \bisplice{f_0}{f_1} \in 𝓣{𝕄}(\biobj{A}{B}, \biobj{X}{Y})$;
  \item $f \in 𝔹(A; X ◁ Y)$, a spliced arrow $F(f) = \trisplice{f_0}{f_1}{f_2} \in 𝓣{𝕄}(\biobj{A^L}{A^R}, \biobj{X^L}{X^R} ◁ \biobj{Y^L}{Y^R})$;
  \item $f \in 𝔹(A; X ⊗ Y)$, a spliced monoidal arrow $F(f) = ⟨ f_0 ⨾ □ ⊗ □ ⨾ f_1 ⟩ \in 𝓣{𝕄}(\biobj{A^L}{A^R}, \biobj{X^L}{X^R} ⊗ \biobj{Y^L}{Y^R})$;
\end{itemize}
Such that for each promonoidal structure
\begin{itemize}
  \item $α(a ｜ b) = (c ｜d)$ in $𝔹$ $\implies$ $α(Fa ｜Fb) = (Fc ｜ Fd)$ in $𝓣{𝕄}$;
  \item $λ(a ｜ b) = c = ρ(d ｜ e)$ in $𝔹$ $\implies$ $λ(Fa ｜ Fb) = Fc = ρ(Fd ｜ Fe)$ in $𝓣{𝕄}$;
  \item[] and such that
  \item $ψ_2(a ｜ b ｜ c) = (d ｜ e ｜ f)$ in $𝔹$ $\implies$ $ψ_2(Fa ｜ Fb ｜ Fc) = (Fd ｜Fe ｜Ff)$ in $𝓣{𝕄}$;
  \item $ψ_0(a) = (b ｜ c ｜ d)$ in $𝔹$ $\implies$ $ψ_0(Fa) = (Fa ｜ Fc ｜ Fd)$ in $𝓣{𝕄}$;
  \item $φ_2(a ｜ b ｜ c) = d$ in $𝔹$ $\implies$ $φ_2(Fa ｜Fb ｜ Fc) = Fd$ in $𝓣{𝕄}$;
  \item $φ_0(a) = b$ in $𝔹$ $\implies$ $φ_0(Fa) = Fb$ in $𝓣{𝕄}$.
  \end{itemize}
\noindent Each of these points is exactly equal by definition, which establishes the desired adjunction.
\end{proof}

\clearpage %
\section{Normalization}

\subsection{Normalization}

  \begin{theorem}[From \Cref{th:normalizationProduoidal}]
    \label{ax:th:normalizationProduoidal}
    Let $𝕍_{⊗,I,◁,N}$ be a \produoidal{} category. The profunctor
    $𝓝𝕍(•; •) =
    𝕍(•; N ⊗ • ⊗ N)$
    forms a \promonad{}. Moreover, the Kleisli category of this promonad is a normal \produoidalCategory{} with the following sequential and parallel splits and units.
    $$\begin{gathered}
        𝓝𝕍(A;B) = 𝕍(A ; N ⊗ B ⊗ N); \\
        𝓝𝕍(A;B ⊗ C) = 𝕍(A ; N ⊗ B ⊗ N ⊗ C ⊗ N); \\
        𝓝𝕍(A;B ◁ C) = 𝕍(A ; (N ⊗ B ⊗ N) ◁ (N ⊗ C ⊗ N)); \\
        𝓝𝕍(A;I) = 𝕍(A;N); \\
        𝓝𝕍(A;N) = 𝕍(A;N).
    \end{gathered}$$
\end{theorem}
\begin{proof}
We define the following multiplication and unit for the \promonad{}, $𝓝𝕍$. They are constructed out of laxators of the produoidal category $𝕍$ and Yoneda isomorphisms; thus, they must be associative and unital by coherence. The unit is defined by
$$\begin{aligned}
& 𝕍(A; B) &\quad
\cong & \quad\mbox{(by unitality of $𝕍$)} \\
& 𝕍(A; I ⊗ B ⊗ I) &\quad
\to & \quad\mbox{(by the laxators of $𝕍$)} \\
& 𝕍(A; N ⊗ B ⊗ N) &\quad
= & \quad\mbox{(by definition)} \\
& 𝓝𝕍(A; B).
\end{aligned}$$

The multiplication is defined by,
$$\begin{aligned}
& ∫^{B ∈ 𝕍} 𝓝𝕍(A; B) × 𝓝𝕍(B; C) &\quad
= & \quad\mbox{(by definition)} \\
& ∫^{B ∈ 𝕍} 𝕍(A; N ⊗ B ⊗ N) × 𝕍(B; N ⊗ C ⊗ N) &\quad
≅ & \quad\mbox{(\byYonedaReduction{})} \\
& 𝕍(A; N ⊗ N ⊗ C ⊗ N ⊗ N) \phantom{\int} &\quad
→ & \quad\mbox{(by laxators of $𝕍$)} \\
& 𝕍(A; N ⊗ C ⊗ N) \phantom{\int} &\quad
= & \quad\mbox{(by definition)} \\
& 𝓝𝕍(A;C). \phantom{\int}
\end{aligned}$$

Let us now construct the unitors and the associators. Again, they are constructed out of laxators of the \produoidalCategory{} $𝕍$, the associators and unitors of $𝕍$, and Yoneda isomorphisms. We first consider the right unitor.
\begin{align*}
& ∫^{X ∈ 𝓝𝕍} 𝓝𝕍(A; B ⊗ X) × 𝓝𝕍(X; N) &\quad
= & \quad\mbox{(by definition)} \\
& ∫^{X ∈ 𝓝𝕍} 𝕍(A; N ⊗ B ⊗ N ⊗ X ⊗ N) × 𝓝𝕍(X; N) &\quad
≅ & \quad\mbox{(by associativity of $𝕍$)} \\
& ∫^{X ∈ 𝓝𝕍, P ∈ 𝕍} 𝕍(A; N ⊗ B ⊗ P) × 𝕍(P; N ⊗ X ⊗ N) × 𝓝𝕍(X; N) &\quad
= & \quad\mbox{(by definition)} \\
& ∫^{X ∈ 𝓝𝕍, P ∈ 𝕍} 𝕍(A; N ⊗ B ⊗ P) × 𝓝𝕍(P;X) × 𝓝𝕍(X; N) &\quad
≅ & \quad\mbox{(\byYonedaReduction{})} \\
& ∫^{P ∈ 𝕍} 𝕍(A; N ⊗ B ⊗ P) × 𝓝𝕍(P; N) &\quad
= & \quad\mbox{(by definition)} \\
& ∫^{P ∈ 𝕍} 𝕍(A; N ⊗ B ⊗ P) × 𝕍(P; N) &\quad
≅ & \quad\mbox{(by unitality)} \\
& ∫^{P ∈ 𝕍} 𝕍(A; N ⊗ B ⊗ N).
\end{align*}
We now consider the left unitor.
\begin{align*}
& ∫^{X ∈ 𝓝𝕍} 𝓝𝕍(A; X ⊗ B) × 𝓝𝕍(X; N) &\quad
= & \quad\mbox{(by definition)} \\
& ∫^{X ∈ 𝓝𝕍} 𝕍(A; N ⊗ X ⊗ N ⊗ B ⊗ N) × 𝓝𝕍(X; N) &\quad
≅ & \quad\mbox{(by associativity of $𝕍$)} \\
& ∫^{X ∈ 𝓝𝕍, P ∈ 𝕍} 𝕍(A; P ⊗ B ⊗ N) × 𝕍(P; N ⊗ X ⊗ N) × 𝓝𝕍(X; N) &\quad
= & \quad\mbox{(by definition)} \\
& ∫^{X ∈ 𝓝𝕍, P ∈ 𝕍} 𝕍(A; P ⊗ B ⊗ N) × 𝓝𝕍(P;X) × 𝓝𝕍(X; N) &\quad
≅ & \quad\mbox{(\byYonedaReduction{})} \\
& ∫^{P ∈ 𝕍} 𝕍(A; P ⊗ B ⊗ N) × 𝓝𝕍(P; N) &\quad
= & \quad\mbox{(by definition)} \\
& ∫^{P ∈ 𝕍} 𝕍(A; P ⊗ B ⊗ N) × 𝕍(P; N) &\quad
≅ & \quad\mbox{(by unitality)} \\
& ∫^{P ∈ 𝕍} 𝕍(A; N ⊗ B ⊗ N).
\end{align*}
Finally, we consider the associator. We can do so in two steps, showing that both sides of the equation
$$∫^{X ∈ 𝓝𝕍} 𝓝𝕍(A; B ⊗ X) × 𝓝𝕍(X; C ⊗ D) ≅ ∫^{Y ∈ 𝓝𝕍} 𝓝𝕍(A; Y ⊗ D) × 𝓝𝕍(Y; B ⊗ C)$$
are isomorphic to $𝕍(A; N ⊗ B ⊗ N ⊗ C ⊗ N ⊗ D ⊗ N)$. The first side by
\begin{align*}
& ∫^{X ∈ 𝓝𝕍} 𝓝𝕍(A; B ⊗ X) × 𝓝𝕍(X; C ⊗ D) &\quad
= & \quad\mbox{(by definition)} \\
& ∫^{X ∈ 𝓝𝕍} 𝕍(A; N ⊗ B ⊗ N ⊗ X ⊗ N) × 𝓝𝕍(X; C ⊗ D) &\quad
≅ & \quad\mbox{(by associativity)} \\
& ∫^{X ∈ 𝓝𝕍, P ∈ 𝕍} 𝕍(A; N ⊗ B ⊗ P) × 𝕍(P; N ⊗ X ⊗ N) × 𝓝𝕍(X; C ⊗ D) &\quad
= & \quad\mbox{(by definition)} \\
& ∫^{X ∈ 𝓝𝕍, P ∈ 𝕍} 𝕍(A; N ⊗ B ⊗ P) × 𝓝𝕍(P; X) × 𝓝𝕍(X; C ⊗ D) &\quad
≅ & \quad\mbox{(\byYonedaReduction{})} \\
& ∫^{P ∈ 𝕍} 𝕍(A; N ⊗ B ⊗ P) × 𝓝𝕍(P; C ⊗ D) &\quad
= & \quad\mbox{(by definition)} \\
& ∫^{P ∈ 𝕍} 𝕍(A; N ⊗ B ⊗ P) × 𝕍(P; N ⊗ C ⊗ N ⊗ D ⊗ N) &\quad
≅ & \quad\mbox{(by associativity)} \\
& 𝕍(A; N ⊗ B ⊗ N ⊗ C ⊗ N ⊗ D ⊗ N),\phantom{\int}
\end{align*}
and the second side by
\begin{align*}
& ∫^{Y ∈ 𝓝𝕍} 𝓝𝕍(A; Y ⊗ D) × 𝓝𝕍(Y; B ⊗ C) &\quad
= & \quad\mbox{(by definition)} \\
& ∫^{Y ∈ 𝓝𝕍} 𝕍(A; N ⊗ Y ⊗ N ⊗ D ⊗ N) × 𝓝𝕍(Y; B ⊗ C) &\quad
≅ & \quad\mbox{(by associativity)} \\
& ∫^{Y ∈ 𝓝𝕍, P ∈ 𝕍} 𝕍(A; P ⊗ D ⊗ N) × 𝕍(P; N ⊗ Y ⊗ N) × 𝓝𝕍(Y; B ⊗ C) &\quad
= & \quad\mbox{(by definition)} \\
& ∫^{Y ∈ 𝓝𝕍, P ∈ 𝕍} 𝕍(A; P ⊗ D ⊗ N) × 𝓝𝕍(P; X) × 𝓝𝕍(X; B ⊗ C) &\quad
≅ & \quad\mbox{(\byYonedaReduction{})} \\
& ∫^{P ∈ 𝕍} 𝕍(A; P ⊗ D ⊗ N) × 𝓝𝕍(P; B ⊗ C) &\quad
= & \quad\mbox{(by definition)} \\
& ∫^{P ∈ 𝕍} 𝕍(A; P ⊗ D ⊗ N) × 𝕍(P; N ⊗ B ⊗ N ⊗ C ⊗ N) &\quad
≅ & \quad\mbox{(by associativity)} \\
& 𝕍(A; N ⊗ B ⊗ N ⊗ C ⊗ N ⊗ D ⊗ N).\phantom{\int}
\end{align*}
Precisely because they are constructed out of coherence morphisms for the base produoidal category $𝕍$, we know that these satisfy the pentagon and triangle equations and define a promonoidal category. The unitors and associators for the sequential promonoidal structure are defined similarly. Finally, we define the laxators of $𝓝𝕍$, making it into a \produoidalCategory{}.

The first laxator,
$$
ψ_2 : 𝓝𝕍(A; (B_1 ◁ C_1) ⊗ (B_2 ◁ C_2)) \longrightarrow 𝓝𝕍(A; (B_1 ⊗ B_2) ◁ (C_1 ⊗ C_2)),
$$
is defined by the following reasoning.
\begin{align*}
& 𝓝𝕍(A; (B_1 ◁ C_1) ⊗ (B_2 ◁ C_2)) &
\\ & \quad= \quad\mbox{(by definition)} \\
& 𝕍(A; N ⊗ ((N ⊗ B_1 ⊗ N) ◁ (N ⊗ C_1 ⊗ N)) ⊗ N ⊗ ((N ⊗ B_2 ⊗ N) ◁ (N ⊗ C_2 ⊗ N)) ⊗ N) &
\\ & \quad \to \quad\mbox{(by $ψ_2$ of $𝕍$)} \\
& 𝕍(A; ((N ⊗ N ⊗ B_1 ⊗ N) ◁ (N ⊗ N ⊗ B_2 ⊗ N)) ⊗ ((N ⊗ N ⊗ B_2 ⊗ N ⊗ N) ◁ (N ⊗ C_2 ⊗ N ⊗ N))) &
\\ & \quad \to \quad\mbox{(by $ψ_2$ of $𝕍$)} \\
& 𝕍(A; (N ⊗ N ⊗ B_1 ⊗ N ⊗ N ⊗ N ⊗ B_2 ⊗ N ⊗ N) ◁ (N ⊗ N ⊗ C_1 ⊗ N ⊗ N ⊗ N ⊗ C_2 ⊗ N ⊗ N)) &
\\ & \quad \to \quad\mbox{(by $φ_2$ of $𝕍$)} \\
& 𝕍(A; (N ⊗ N ⊗ B_1 ⊗ N ⊗ B_2 ⊗ N ⊗ N) ◁ (N ⊗ N ⊗ C_1 ⊗ N ⊗ C_2 ⊗ N ⊗ N)) &
\\ & \quad= \quad\mbox{(by definition)} \\
& 𝓝𝕍(A; (B_1 ⊗ B_2) ◁ (C_1 ⊗ C_2)).
\end{align*}
The remaining laxators are isomorphisms that arise from applications of unitality or just as identities.
\begin{align*}
  & ψ_0 : 𝓝𝕍(A, I) \overset{≅}\longrightarrow 𝓝𝕍(A; I ◁ I) \\
  & φ_2 : 𝓝𝕍(A ; N ⊗ N) \overset{≅}\longrightarrow 𝓝𝕍(A ; N) \\
  & φ_0 : 𝓝𝕍(A ; I) \overset{id}\longrightarrow 𝓝𝕍(A ; N)
\end{align*}

This has shown that the resulting category is also a \emph{normal} produoidal category.
\end{proof}

\begin{proposition} \label{ax:prop:normalizationFunctor}
  Normalization extends to a endofunctor of \produoidalCategories{} $𝓝 : \Produo \to \Produo$.
\end{proposition}
\begin{proof}
  Let $𝕍_{⊗,I,◁,N}$ and $𝕎_{⊘,J,◀,M}$ be produoidal categories. $𝓝$ sends $𝕍$ to its normalization $𝓝𝕍$. Let $(F, F_⊗, F_I, F_{◁}, F_N) : 𝕍 \to 𝕎$ be a produoidal functor. Then $𝓝F : 𝓝𝕍 \to 𝓝𝕎$ has underlying functor defined by $F$ on objects and on morphisms by

  $$\begin{aligned}
    & 𝕍(A; N ⊗ B ⊗ N) &\quad
    = & \quad\mbox{(by definition)} \\
    & ∫^{X,Y \in 𝕍} 𝕍(A; X ⊗ B ⊗ Y) × 𝕍(X;N) × 𝕍(Y;N) &\quad
    \to & \quad\mbox{(induced by $F_⊗, F_N$)} \\
    & ∫^{X,Y \in 𝕍} 𝕎(FA; FX ⊘ FB ⊘ FY) × 𝕎(FX; M) × 𝕎(FY; M) &\quad
    \to & \quad\mbox{(inclusion, universal prop. of coend)} \\
    & ∫^{P,Q \in 𝕎} 𝕎(FA; P ⊘ FB ⊘ Q) × 𝕎(P; M) × 𝕎(Q; M) &\quad
    = & \quad\mbox{(by definition)} \\
    & 𝕎(FA; M ⊘ FB ⊘ M).
  \end{aligned}$$

  $𝓝F_⊗$ and $𝓝F_{◁}$ are defined similarly, and $𝓝F_N$ is $F_N$. We have $𝓝\idFun{_𝕍} = \idFun{_{𝓝𝕍}}$, since all the data of the left hand side is given by identity maps on $𝓝𝕍$, and if $G : 𝕌 → 𝕍$ is another produoidal functor, then $𝓝(G⨾F) = 𝓝G ⨾ 𝓝F$ follows from the naturality of the components of $F$ and $G$.
\end{proof}

\begin{theorem}[From \Cref{th:normalizationIdempotent}]
  \label{ax:th:normalizationIdempotent}
  The functor $𝓝 : \Produo \to \Produo$ from \Cref{ax:prop:normalizationFunctor} is an idempotent monad.
\end{theorem}
\begin{proof}
  Let $𝕍_{⊗,I,◁,N}$ be a \produoidal{} category and let $◁_{N}, ⊗_N, N$ denote the sequential splits, parallel splits, and unit in its normalization $𝓝𝕍$.
  
  The monad has unit $\eta$ with component at $𝕍$ the following \produoidalFunctor{} $\eta_{𝕍} : 𝕍 \to 𝓝𝕍$. The underlying functor is the functor induced by the \promonad{} \cite[Lemma 3.8]{roman22}: it is identity on objects, and acts on morphisms by the unit of the promonad. The following components of the produoidal functor preserve laxators and coherence maps since they are constructed only from laxators and coherence maps.
  $$\begin{aligned}
  & η_{⊗} : 𝕍(A; B ⊗ C) \overset{\lambda\,,\,\rho}{\to} 𝕍(A; I ⊗ B ⊗ I ⊗ C ⊗ I) \xrightarrow[]{φ_0} 𝕍(A; N ⊗ B ⊗ N ⊗ C ⊗ N), \\
  & η_{I} : 𝕍(A;I) \xrightarrow[]{φ_0} 𝕍(A;N), \\
  & η_{◁} : 𝕍(A; B ◁ C) \overset{\lambda\,,\,\rho}{\to} 𝕍(A; (I ⊗ B ⊗ I) ◁ (I ⊗ C ⊗ I)) \xrightarrow[]{φ_0} 𝕍(A; (N ⊗ B ⊗ N) ◁ (N ⊗ C ⊗ N)), \\
  & η_{N} : 𝕍(A;N) \xrightarrow[]{\mathrm{id}} 𝕍(A;N).
  \end{aligned}$$

  The monad has multiplication $\mu$ with component at $𝕍$ the following \emph{isomorphism} $\mu_{𝕍} : 𝓝𝓝𝕍 \cong 𝓝𝕍$ of produoidal categories (witnessing that the monad is idempotent). The underlying functor is identity on objects, and acts on morphisms by $$𝓝𝓝𝕍(A;B) = 𝓝𝕍(A; N ⊗_{N} B ⊗_{N} N) \overset{\lambda\,,\,\rho}{\cong} 𝓝𝕍(A;B).$$
  The following natural transformations make this a produoidal functor:
  $$\begin{aligned}
  & \mu_{⊗} : 𝓝𝓝𝕍(A; B ⊗_{NN} C) = 𝓝𝕍(A; N ⊗_N B ⊗_N N ⊗_N C ⊗_N N) \overset{\lambda\,,\,\rho}{\cong} 𝓝𝕍(A; B ⊗_N C), \\
  & \mu_{N} = \mu_{I} : 𝓝𝓝𝕍(A;N) = 𝓝𝕍(A;N), \\
  & \mu_{◁} : 𝓝𝓝𝕍(A; B ◁_{NN} C) =  𝓝𝕍(A ; (N ⊗_N B ⊗_N N)  ◁_N (N ⊗_N C ⊗_N N)) \overset{\lambda\,,\,\rho}{\cong} 𝓝𝕍(A ; B ◁_N C).
\end{aligned}$$.

Finally we verify the monad laws. $η_{𝓝𝕍} ⨾ \mu_{𝕍}$ is identity on objects and on morphisms applies left and right unitors followed by their inverses, thus has underlying functor equal to the identity. The components of the natural transformations are also identities, since the laxator $φ_0$ is an identity for $𝓝𝕍$, and they are otherwise composed of unitors followed by their inverses, and similarly for the other unit law (using the unitality coherence equations of \Cref{cd:duoidal-coherence-nandi}). $\mu_{𝓝𝕍} ⨾ \mu_𝕍$ and $𝓝\mu_{𝕍} ⨾ \mu_𝕍$ are identity on objects and amount to applying left and right unitors twice on morphisms, and similarly for their components.
\end{proof}

\begin{lemma}
  \label{lemma:ax:exactlyOneAlgebra}
  A produoidal category $𝕍$ has exactly one algebra structure for the normalization monad when it is normal, and none otherwise.
\end{lemma}
\begin{proof}
  Let $(f_\map, f_{⊗}, f_I, f_{◁}, f_N) \colon 𝓝𝕍 \to 𝕍$ be an algebra. This means that the following commutative diagrams with the unit and multiplication of the normalization monad must commute.
  \begin{center}
    \begin{tikzcd}
      𝕍 \arrow[r, "η"] \arrow[rd, "\mathrm{id}"', no head] & 𝓝𝕍 \arrow[d, "f"] &  & 𝓝𝓝𝕍 \arrow[r, "μ"] \arrow[d, "𝓝f"'] & 𝓝𝕍 \arrow[d, "f"] \\
                                                   & 𝕍                  &  & 𝓝𝕍 \arrow[r, "f"']                    & 𝕍                 
      \end{tikzcd}  
  \end{center}
  Now, consider how the laxator $ψ_0 \colon 𝕍(•; I) \to 𝕍(•; N)$ is transported by these maps.
  \begin{center}
    \begin{tikzcd}
        & 𝕍(•;N) \arrow[rd, "f_I"] \arrow[d, "\mathrm{id}"] &                          \\
𝕍(•;I) \arrow[ru, "η_I"] \arrow[d, "ψ_0"'] \arrow[rr, "\mathrm{id}"', bend right] & 𝕍(•;N) \arrow[rd, "f_N"']                         & 𝕍(•;I) \arrow[d, "ψ_0"] \\
𝕍(•;N) \arrow[ru, "\mathrm{id}"'] \arrow[rr, "\mathrm{id}"', bend right]                   &                                                    & 𝕍(•;N)
    \end{tikzcd}
  \end{center}
  We conclude that $η_I = ψ_0$, but also that $f_N = \mathrm{id}$. As a consequence, $ψ_0$ is invertible and $f_I$ must be its inverse. We have shown that the produoidal category $𝕍$ must be normal.

  We will now show that this already determines all of the functor $f$.  We know that $η_{⊗}, η_{◁}, η_\map$ are isomorphisms because they are constructed from the unitors, associators, and the laxator $ψ_0$, which is an isomorphism in this case. This determines that $f_{⊗}, f_{◁}, f_\map$ must be their inverses. By construction, these satisfy all coherence morphisms.
\end{proof}

\begin{theorem}[From \Cref{th:freeNormalProduoidal}]
  \label{ax:th:freeNormalProduoidal}
  Normalization determines an adjunction between \produoidalCategories{} and normal produoidal categories,
  $$𝓝 \colon \Produo \rightleftharpoons \npDuo \colon 𝓤$$
  That is, $𝓝𝕍$ is the free normal produoidal category over $𝕍$.
\end{theorem}
\begin{proof}
  We know that the algebras for the normalization monad are exactly the normal produoidal categories (\Cref{lemma:ax:exactlyOneAlgebra}). We also know that the normalization monad is idempotent (\Cref{th:normalizationIdempotent}). 
  This implies that the forgetful functor from its category of algebras is fully faithful, and thus, the algebra morphisms are exactly the produoidal functors. As a consequence, the canonical adjunction to the category of algebras of the monad is exactly an adjunction to the category of normal produoidal categories. 
\end{proof}

\subsection{Symmetric Normalization}

\begin{theorem}[From \Cref{th:symNormalizationProduoidal}]
  \label{ax:th:symNormalizationProduoidal}
  Let $𝕍_{⊗,I,◁,N}$ be a \symmetricProduoidal{} category. The profunctor
    $𝓝_{σ}𝕍(•; •) =
    𝕍(•; N ⊗ •)$
    forms a \promonad{}. Moreover, the Kleisli category of this promonad is a normal symmetric produoidal category with the following sequential and parallel splits and units.
    \begin{align*}
        𝓝_{σ}𝕍(A;B) & = 𝕍(A ; N ⊗ B); \\
        𝓝_{σ}𝕍(A;B ⊗_N C) & = 𝕍(A ; N ⊗ B ⊗ C); \\
        𝓝_{σ}𝕍(A;B ◁_N C) & = 𝕍(A ; (N ⊗ B) ◁ (N ⊗ C)); \\
        𝓝_{σ}𝕍(A;N) & = 𝕍(A ; N); \\
        𝓝_{σ}𝕍(A;I) & = 𝕍(A ; N).
    \end{align*}
\end{theorem}
\begin{proof}
  The unit and multiplication of the \promonad{} are given in essentially the same way as in the proof of \Cref{ax:th:normalizationProduoidal}.
  Likewise the associators, unitors and laxators of $𝓝_{σ}𝕍$ are given in essentially the same way, though one must use the fact that $𝕍$ is symmetric. We need additionally a symmetry natural isomorphism for $𝓝_{σ}𝕍$. Its components are defined by,
  \begin{align*}
    & 𝓝_{σ}𝕍(A; B ⊗ C) &\quad
    = &\quad\mbox{(by definition)}\phantom{∫} \\
    & 𝕍(A; N ⊗ B ⊗ C) &\quad
    ≅ &\quad\mbox{(by associativity)} \phantom{∫} \\
    & ∫^{X ∈ 𝕍} 𝕍(A; N ⊗ X) × 𝕍(X; B ⊗ C) &\quad
    ≅ & \quad\mbox{(by symmetry of 𝕍)} \phantom{∫} \\
    & ∫^{X ∈ 𝕍} 𝕍(A; N ⊗ X) × 𝕍(X; C ⊗ B) &\quad
    ≅ & \quad\mbox{(by associativity)} \phantom{∫} \\
    & 𝕍(A; N ⊗ C ⊗ B) &\quad
    = & \quad\mbox{(by definition)} \phantom{∫} \\
    & 𝓝_{σ}𝕍(A;C ⊗ B). \phantom{\int}
  \end{align*}

  These satisfy hexagon and symmetry identities because these are satisfied by $𝕍$, and we only use symmetries and coherences of $𝕍$. Thus we have a normal symmetric produoidal category $𝓝_{σ}𝕍$.
\end{proof}

\begin{definition}[Symmetric produoidal functor]
  A \emph{symmetric produoidal functor} is a \produoidalFunctor{} $F \colon 𝕍 \to 𝕎$ that moreover preserves the symmetry, in that $F_\otimes ⨾ σ_{𝕍} = σ_{𝕎} ⨾ F_\otimes$.
  We denote by $\symProduo$ the category of \symmetricProduoidal{} categories and symmetric produoidal functors.
\end{definition}

\begin{proposition} \label{ax:prop:symNormalizationFunctor}
  Symmetric normalization extends to a endofunctor of \symmetricProduoidal{} categories $𝓝_σ : \symProduo \to \symProduo$.
\end{proposition}
\begin{proof}
  The construction is essentially the same as in \Cref{ax:prop:normalizationFunctor}.
  The only thing left to check is that $𝓝_σF$ there constructed preserves symmetries whenever $F$ does (see \Cref{ax:th:symNormalizationProduoidal}).
  This is because the symmetry of $𝓝_σ𝕍$ is constructed out of associativity and symmetries of $𝕍$, which $𝓝_σF_\otimes$, constructed itself out of $F_\otimes$, associativity, and symmetries of $𝕍$, must preserve. 
\end{proof}

\begin{theorem}
  \label{ax:th:symm:normalizationIdempotent}
  The functor $𝓝_σ : \symProduo \to \symProduo$ from \Cref{ax:prop:normalizationFunctor} is an idempotent monad.
\end{theorem}
\begin{proof}
  The construction is again essentially the same as in
  \Cref{ax:th:normalizationIdempotent}. 
  It is left to check that the unit and multiplication constructed in this way preserve the symmetries. Indeed, $η_σ \colon 𝕍 \to 𝓝_σ𝕍$ is \symmetricProduoidal{} because $η_\otimes$ is constructed out of natural associators and laxators that commute with the symmetry.
\end{proof}

\begin{lemma}
  \label{lemma:sym:exactlyOneAlgebra}
  A \symmetricProduoidal{} category $𝕍$ has exactly one algebra structure for the symmetric normalization monad when it is normal, and none otherwise.
\end{lemma}
\begin{proof}
  The proof essentially follows the same reasoning as \Cref{lemma:ax:exactlyOneAlgebra}, replacing the construction with the symmetric version and the previous lemmas.
\end{proof}

\begin{theorem}[From \Cref{th:sym:freeNormalProduoidal}]
  \label{ax:sym:th:freeNormalProduoidal} 
  Symmetric normalization determines an adjunction between \symmetricProduoidal{} categories and normal symmetric \produoidal{} categories,
  $$𝓝_σ \colon \symProduo \rightleftharpoons \nSymProduo \colon 𝓤$$
  Where we define the category of normal symmetric \produoidal{} categories, $\nSymProduo$, to use as functors the \symmetricProduoidal{} functors, adquiring a full forgetful functor $𝓤$.
  
  That is, $𝓝_σ𝕍$ is the free symmetric normal produoidal category over the \symmetricProduoidal{} category $𝕍$.
\end{theorem}
\begin{proof}
  The proof essentially follows
  \Cref{ax:th:freeNormalProduoidal}, now using the previous lemmas and \Cref{lemma:sym:exactlyOneAlgebra}.
\end{proof}

\subsection{Normalization of duoidals and normalization of produoidals} \label{ax:sec:garnernorm}
We conjecture that the normalization of a \produoidalCategory{} could still be seen to arise from the normalization procedure for duoidal categories outlined by Garner and López Franco \cite{garner16}.
Every \produoidalCategory{} $\mathbb{V}$ induces a closed duoidal structure on its presheaf category $\hat{\mathbb{V}}:=[\mathbb{V}^\text{op},\mathbf{Set}]$: indeed, by a result of Day, any promonoidal structure induces a closed monoidal structure on the presheaf category \cite{day,day_thesis}; furthermore, one can confirm that the two closed monoidal structures on $\hat{\mathbb{V}}$ interact in such a way as to make the category duoidal (\Cref{th:produoidalInduceDuoidal}).

Normalizing the duoidal $\hat{\mathbb{V}}$ yields the category of algebras $\text{EM}(𝓝\mathbb{V})$ for the \promonad{} $𝓝\mathbb{V}$ -- or, equivalently, the category of algebras for the cocontinuous monad induced by $𝓝\mathbb{V}$ on $\hat{\mathbb{V}}$.
$\text{EM}(𝓝\mathbb{V})$ is now normal duoidal, and furthermore the closure of the tensors on $\hat{\mathbb{V}}$ carries across to make $\text{EM}(𝓝\mathbb{V})$ also closed.
Now, one notes that we have the following isomorphism $\text{EM}(𝓝\mathbb{V})\cong[𝓝\mathbb{V}^\text{op},\mathbf{Set}]$, that is, the category of algebras is the presheaf category of the Kleisli object $𝓝\mathbb{V}$ of the \promonad{} in $\mathbf{Prof}$.
Therefore, the closed monoidal structures of $\text{EM}(𝓝\mathbb{V})$ must correspond to promonoidal structures of $𝓝\mathbb{V}$ and these interact so as to make $𝓝\mathbb{V}$ produoidal.

\clearpage

\section{Monoidal Context}
\label{sec:ax:monoidalContext}

\begin{figure}[ht]
  \centering

\tikzset{every picture/.style={line width=0.75pt}} %

\begin{tikzpicture}[x=0.75pt,y=0.75pt,yscale=-1,xscale=1]
\draw [color={rgb, 255:red, 0; green, 0; blue, 0 }  ,draw opacity=1 ]   (49.99,110) -- (49.99,120) ;
\draw  [color={rgb, 255:red, 191; green, 97; blue, 106 }  ,draw opacity=1 ][fill={rgb, 255:red, 191; green, 97; blue, 106 }  ,fill opacity=0.2 ][dash pattern={on 4.5pt off 4.5pt}] (40,70) -- (60,70) -- (60,85) -- (40,85) -- cycle ;
\draw  [color={rgb, 255:red, 0; green, 0; blue, 0 }  ,draw opacity=1 ][fill={rgb, 255:red, 255; green, 255; blue, 255 }  ,fill opacity=1 ] (35,45) -- (65,45) -- (65,60) -- (35,60) -- cycle ;
\draw [color={rgb, 255:red, 0; green, 0; blue, 0 }  ,draw opacity=1 ]   (49.99,35) -- (49.99,45) ;
\draw [color={rgb, 255:red, 0; green, 0; blue, 0 }  ,draw opacity=1 ]   (30,75) -- (30,80) ;
\draw  [color={rgb, 255:red, 0; green, 0; blue, 0 }  ,draw opacity=1 ][fill={rgb, 255:red, 255; green, 255; blue, 255 }  ,fill opacity=1 ] (35,95) -- (65,95) -- (65,110) -- (35,110) -- cycle ;
\draw [color={rgb, 255:red, 0; green, 0; blue, 0 }  ,draw opacity=1 ]   (49.99,60) -- (50,70) ;
\draw [color={rgb, 255:red, 0; green, 0; blue, 0 }  ,draw opacity=1 ]   (50,85) -- (50,95) ;
\draw [color={rgb, 255:red, 0; green, 0; blue, 0 }  ,draw opacity=1 ]   (39.99,60) .. controls (39.59,69.2) and (30.2,65.8) .. (30,75) ;
\draw [color={rgb, 255:red, 0; green, 0; blue, 0 }  ,draw opacity=1 ]   (60,60) .. controls (59.6,69.2) and (70.2,65.8) .. (70,75) ;
\draw [color={rgb, 255:red, 0; green, 0; blue, 0 }  ,draw opacity=1 ]   (60,95) .. controls (59.6,85.8) and (70.2,89.2) .. (70,80) ;
\draw [color={rgb, 255:red, 0; green, 0; blue, 0 }  ,draw opacity=1 ]   (39.99,95) .. controls (39.59,85.8) and (30.2,89.2) .. (30,80) ;
\draw [color={rgb, 255:red, 0; green, 0; blue, 0 }  ,draw opacity=1 ]   (70,75) -- (70,80) ;
\draw [color={rgb, 255:red, 0; green, 0; blue, 0 }  ,draw opacity=1 ]   (294.99,110) -- (294.99,120) ;
\draw  [color={rgb, 255:red, 0; green, 0; blue, 0 }  ,draw opacity=1 ][fill={rgb, 255:red, 255; green, 255; blue, 255 }  ,fill opacity=1 ] (279.99,20) -- (309.99,20) -- (309.99,35) -- (279.99,35) -- cycle ;
\draw [color={rgb, 255:red, 0; green, 0; blue, 0 }  ,draw opacity=1 ]   (294.99,10) -- (294.99,20) ;
\draw [color={rgb, 255:red, 0; green, 0; blue, 0 }  ,draw opacity=1 ]   (270,50) -- (270,105) ;
\draw  [color={rgb, 255:red, 0; green, 0; blue, 0 }  ,draw opacity=1 ][fill={rgb, 255:red, 255; green, 255; blue, 255 }  ,fill opacity=1 ] (280,95) -- (310,95) -- (310,110) -- (280,110) -- cycle ;
\draw [color={rgb, 255:red, 0; green, 0; blue, 0 }  ,draw opacity=1 ]   (294.99,35) -- (295,45) ;
\draw [color={rgb, 255:red, 0; green, 0; blue, 0 }  ,draw opacity=1 ]   (284.99,35) .. controls (284.59,44.2) and (270.2,40.8) .. (270,50) ;
\draw [color={rgb, 255:red, 0; green, 0; blue, 0 }  ,draw opacity=1 ]   (304.99,35) .. controls (304.59,44.2) and (320.2,40.8) .. (320,50) ;
\draw  [color={rgb, 255:red, 0; green, 0; blue, 0 }  ,draw opacity=1 ][fill={rgb, 255:red, 255; green, 255; blue, 255 }  ,fill opacity=1 ] (280,44.5) -- (310,44.5) -- (310,59.5) -- (280,59.5) -- cycle ;
\draw  [color={rgb, 255:red, 191; green, 97; blue, 106 }  ,draw opacity=1 ][fill={rgb, 255:red, 191; green, 97; blue, 106 }  ,fill opacity=0.2 ][dash pattern={on 4.5pt off 4.5pt}] (284.99,70) -- (304.99,70) -- (304.99,85) -- (284.99,85) -- cycle ;
\draw [color={rgb, 255:red, 0; green, 0; blue, 0 }  ,draw opacity=1 ]   (295,60) -- (295,70) ;
\draw [color={rgb, 255:red, 0; green, 0; blue, 0 }  ,draw opacity=1 ]   (315,75) -- (315,80) ;
\draw [color={rgb, 255:red, 0; green, 0; blue, 0 }  ,draw opacity=1 ]   (294.99,135) -- (294.99,145) ;
\draw  [color={rgb, 255:red, 0; green, 0; blue, 0 }  ,draw opacity=1 ][fill={rgb, 255:red, 255; green, 255; blue, 255 }  ,fill opacity=1 ] (280,120) -- (310,120) -- (310,135) -- (280,135) -- cycle ;
\draw [color={rgb, 255:red, 0; green, 0; blue, 0 }  ,draw opacity=1 ]   (305,120) .. controls (304.6,110.8) and (320.2,114.2) .. (320,105) ;
\draw [color={rgb, 255:red, 0; green, 0; blue, 0 }  ,draw opacity=1 ]   (285,120) .. controls (284.6,110.8) and (270.2,114.2) .. (270,105) ;
\draw [color={rgb, 255:red, 0; green, 0; blue, 0 }  ,draw opacity=1 ]   (320,50) -- (320,105) ;
\draw [color={rgb, 255:red, 0; green, 0; blue, 0 }  ,draw opacity=1 ]   (275,75) -- (275,80) ;
\draw [color={rgb, 255:red, 0; green, 0; blue, 0 }  ,draw opacity=1 ]   (294.99,85) -- (294.99,95) ;
\draw [color={rgb, 255:red, 0; green, 0; blue, 0 }  ,draw opacity=1 ]   (305,95) .. controls (304.6,85.8) and (315.2,89.2) .. (315,80) ;
\draw [color={rgb, 255:red, 0; green, 0; blue, 0 }  ,draw opacity=1 ]   (285,95) .. controls (284.6,85.8) and (275.2,89.2) .. (275,80) ;
\draw [color={rgb, 255:red, 0; green, 0; blue, 0 }  ,draw opacity=1 ]   (305,60) .. controls (304.6,69.2) and (315.2,65.8) .. (315,75) ;
\draw [color={rgb, 255:red, 0; green, 0; blue, 0 }  ,draw opacity=1 ]   (285,60) .. controls (284.6,69.2) and (275.2,65.8) .. (275,75) ;
\draw  [draw opacity=0] (120,65) -- (140,65) -- (140,90) -- (120,90) -- cycle ;
\draw  [draw opacity=0] (320,65) -- (340,65) -- (340,90) -- (320,90) -- cycle ;
\draw [color={rgb, 255:red, 0; green, 0; blue, 0 }  ,draw opacity=1 ]   (169.99,110) -- (169.99,120) ;
\draw  [color={rgb, 255:red, 191; green, 97; blue, 106 }  ,draw opacity=1 ][fill={rgb, 255:red, 191; green, 97; blue, 106 }  ,fill opacity=0.2 ][dash pattern={on 4.5pt off 4.5pt}] (160,70) -- (180,70) -- (180,85) -- (160,85) -- cycle ;
\draw  [color={rgb, 255:red, 0; green, 0; blue, 0 }  ,draw opacity=1 ][fill={rgb, 255:red, 255; green, 255; blue, 255 }  ,fill opacity=1 ] (154.99,45) -- (184.99,45) -- (184.99,60) -- (154.99,60) -- cycle ;
\draw [color={rgb, 255:red, 0; green, 0; blue, 0 }  ,draw opacity=1 ]   (169.99,35) -- (169.99,45) ;
\draw  [color={rgb, 255:red, 0; green, 0; blue, 0 }  ,draw opacity=1 ][fill={rgb, 255:red, 255; green, 255; blue, 255 }  ,fill opacity=1 ] (154.99,95) -- (184.99,95) -- (184.99,110) -- (154.99,110) -- cycle ;
\draw [color={rgb, 255:red, 0; green, 0; blue, 0 }  ,draw opacity=1 ]   (169.99,60) -- (170,70) ;
\draw [color={rgb, 255:red, 0; green, 0; blue, 0 }  ,draw opacity=1 ]   (170,85) -- (170,95) ;
\draw [color={rgb, 255:red, 0; green, 0; blue, 0 }  ,draw opacity=1 ]   (159.99,60) .. controls (159.59,69.2) and (150.19,65.8) .. (149.99,75) ;
\draw [color={rgb, 255:red, 0; green, 0; blue, 0 }  ,draw opacity=1 ]   (179.99,60) .. controls (179.59,69.2) and (190.2,65.8) .. (190,75) ;
\draw [color={rgb, 255:red, 0; green, 0; blue, 0 }  ,draw opacity=1 ]   (179.99,95) .. controls (179.59,85.8) and (190.2,89.2) .. (190,80) ;
\draw [color={rgb, 255:red, 0; green, 0; blue, 0 }  ,draw opacity=1 ]   (159.99,95) .. controls (159.59,85.8) and (150.19,89.2) .. (149.99,80) ;
\draw [color={rgb, 255:red, 0; green, 0; blue, 0 }  ,draw opacity=1 ]   (189.99,75) -- (190,80) ;
\draw  [draw opacity=0] (190,65) -- (210,65) -- (210,90) -- (190,90) -- cycle ;
\draw [color={rgb, 255:red, 0; green, 0; blue, 0 }  ,draw opacity=1 ]   (149.99,75) -- (150,80) ;
\draw [color={rgb, 255:red, 0; green, 0; blue, 0 }  ,draw opacity=1 ]   (229.99,110) -- (229.99,120) ;
\draw  [color={rgb, 255:red, 191; green, 97; blue, 106 }  ,draw opacity=1 ][fill={rgb, 255:red, 191; green, 97; blue, 106 }  ,fill opacity=0.2 ][dash pattern={on 4.5pt off 4.5pt}] (220,70) -- (240,70) -- (240,85) -- (220,85) -- cycle ;
\draw  [color={rgb, 255:red, 0; green, 0; blue, 0 }  ,draw opacity=1 ][fill={rgb, 255:red, 255; green, 255; blue, 255 }  ,fill opacity=1 ] (214.99,45) -- (244.99,45) -- (244.99,60) -- (214.99,60) -- cycle ;
\draw [color={rgb, 255:red, 0; green, 0; blue, 0 }  ,draw opacity=1 ]   (229.99,35) -- (229.99,45) ;
\draw  [color={rgb, 255:red, 0; green, 0; blue, 0 }  ,draw opacity=1 ][fill={rgb, 255:red, 255; green, 255; blue, 255 }  ,fill opacity=1 ] (215,95) -- (245,95) -- (245,110) -- (215,110) -- cycle ;
\draw [color={rgb, 255:red, 0; green, 0; blue, 0 }  ,draw opacity=1 ]   (229.99,60) -- (230,70) ;
\draw [color={rgb, 255:red, 0; green, 0; blue, 0 }  ,draw opacity=1 ]   (230,85) -- (230,95) ;
\draw [color={rgb, 255:red, 0; green, 0; blue, 0 }  ,draw opacity=1 ]   (219.99,60) .. controls (219.59,69.2) and (210.19,65.8) .. (209.99,75) ;
\draw [color={rgb, 255:red, 0; green, 0; blue, 0 }  ,draw opacity=1 ]   (239.99,60) .. controls (239.59,69.2) and (250.2,65.8) .. (250,75) ;
\draw [color={rgb, 255:red, 0; green, 0; blue, 0 }  ,draw opacity=1 ]   (239.99,95) .. controls (239.59,85.8) and (250.2,89.2) .. (250,80) ;
\draw [color={rgb, 255:red, 0; green, 0; blue, 0 }  ,draw opacity=1 ]   (219.99,95) .. controls (219.59,85.8) and (210.19,89.2) .. (209.99,80) ;
\draw [color={rgb, 255:red, 0; green, 0; blue, 0 }  ,draw opacity=1 ]   (209.99,75) -- (210,80) ;
\draw [color={rgb, 255:red, 0; green, 0; blue, 0 }  ,draw opacity=1 ]   (249.99,75) -- (250,80) ;
\draw  [draw opacity=0] (250,65) -- (270,65) -- (270,90) -- (250,90) -- cycle ;
\draw [color={rgb, 255:red, 0; green, 0; blue, 0 }  ,draw opacity=1 ]   (110,60) -- (110,70) ;
\draw  [color={rgb, 255:red, 191; green, 97; blue, 106 }  ,draw opacity=1 ][fill={rgb, 255:red, 191; green, 97; blue, 106 }  ,fill opacity=0.2 ][dash pattern={on 4.5pt off 4.5pt}] (100,70) -- (120,70) -- (120,85) -- (100,85) -- cycle ;
\draw [color={rgb, 255:red, 0; green, 0; blue, 0 }  ,draw opacity=1 ]   (110,85) -- (110,95) ;
\draw  [draw opacity=0] (70,65) -- (90,65) -- (90,90) -- (70,90) -- cycle ;
\draw  [draw opacity=0] (30,120) -- (70,120) -- (70,145) -- (30,145) -- cycle ;
\draw  [draw opacity=0] (210,120) -- (250,120) -- (250,145) -- (210,145) -- cycle ;
\draw  [draw opacity=0] (90,120) -- (130,120) -- (130,145) -- (90,145) -- cycle ;

\draw (50,52.5) node  [font=\footnotesize]  {$f$};
\draw (50,102.5) node  [font=\footnotesize]  {$g$};
\draw (294.99,27.5) node  [font=\footnotesize]  {$f$};
\draw (295,102.5) node  [font=\footnotesize]  {$k$};
\draw (295,52) node  [font=\footnotesize]  {$h$};
\draw (295,127.5) node  [font=\footnotesize]  {$g$};
\draw (130,77.5) node    {$;$};
\draw (330,77.5) node    {$;$};
\draw (169.99,52.5) node  [font=\footnotesize]  {$f$};
\draw (169.99,102.5) node  [font=\footnotesize]  {$g$};
\draw (200,77.5) node    {$≺$};
\draw (229.99,52.5) node  [font=\footnotesize]  {$h$};
\draw (229.99,102.5) node  [font=\footnotesize]  {$k$};
\draw (260,77.5) node    {$=$};
\draw (80,77.5) node    {$;$};
\draw (50,132.5) node    {$( i)$};
\draw (230,132.5) node    {$( iii)$};
\draw (110,132.5) node    {$( ii)$};

\end{tikzpicture}

   \caption{Generic monoidal context (i), identity (ii) and composition (iii).}
  \label{fig:monoidalcontexts}
\end{figure}
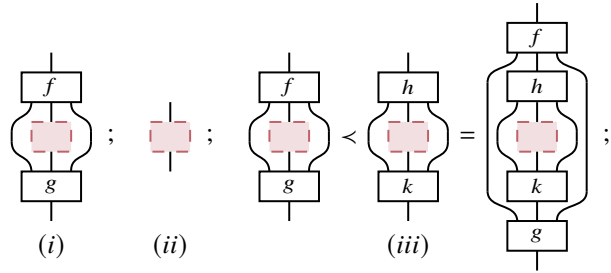

\begin{figure}[t]
  \centering

\tikzset{every picture/.style={line width=0.75pt}} %

\begin{tikzpicture}[x=0.75pt,y=0.75pt,yscale=-1,xscale=1]
\draw  [color={rgb, 255:red, 0; green, 0; blue, 0 }  ,draw opacity=1 ][fill={rgb, 255:red, 255; green, 255; blue, 255 }  ,fill opacity=1 ] (270,15) -- (310,15) -- (310,30) -- (270,30) -- cycle ;
\draw [color={rgb, 255:red, 0; green, 0; blue, 0 }  ,draw opacity=1 ]   (290,5) -- (290,15) ;
\draw  [color={rgb, 255:red, 0; green, 0; blue, 0 }  ,draw opacity=1 ][fill={rgb, 255:red, 255; green, 255; blue, 255 }  ,fill opacity=1 ] (270,75) -- (310,75) -- (310,90) -- (270,90) -- cycle ;
\draw  [color={rgb, 255:red, 0; green, 0; blue, 0 }  ,draw opacity=1 ][fill={rgb, 255:red, 255; green, 255; blue, 255 }  ,fill opacity=1 ] (270,134.5) -- (315,134.5) -- (315,149.5) -- (270,149.5) -- cycle ;
\draw [color={rgb, 255:red, 0; green, 0; blue, 0 }  ,draw opacity=1 ]   (290,150) -- (290,160) ;
\draw  [color={rgb, 255:red, 191; green, 97; blue, 106 }  ,draw opacity=1 ][fill={rgb, 255:red, 191; green, 97; blue, 106 }  ,fill opacity=0.2 ][dash pattern={on 4.5pt off 4.5pt}] (205,55) -- (225,55) -- (225,70) -- (205,70) -- cycle ;
\draw  [color={rgb, 255:red, 0; green, 0; blue, 0 }  ,draw opacity=1 ][fill={rgb, 255:red, 255; green, 255; blue, 255 }  ,fill opacity=1 ] (195,15) -- (235,15) -- (235,30) -- (195,30) -- cycle ;
\draw [color={rgb, 255:red, 0; green, 0; blue, 0 }  ,draw opacity=1 ]   (215,5) -- (215,15) ;
\draw [color={rgb, 255:red, 0; green, 0; blue, 0 }  ,draw opacity=1 ]   (200,50) -- (200,75) ;
\draw  [color={rgb, 255:red, 0; green, 0; blue, 0 }  ,draw opacity=1 ][fill={rgb, 255:red, 255; green, 255; blue, 255 }  ,fill opacity=1 ] (195,75) -- (235,75) -- (235,90) -- (195,90) -- cycle ;
\draw [color={rgb, 255:red, 0; green, 0; blue, 0 }  ,draw opacity=1 ]   (215,30) -- (215.01,55) ;
\draw [color={rgb, 255:red, 0; green, 0; blue, 0 }  ,draw opacity=1 ]   (215,70) -- (215,75) ;
\draw  [color={rgb, 255:red, 0; green, 0; blue, 0 }  ,draw opacity=1 ][fill={rgb, 255:red, 255; green, 255; blue, 255 }  ,fill opacity=1 ] (195,134.5) -- (235,134.5) -- (235,149.5) -- (195,149.5) -- cycle ;
\draw [color={rgb, 255:red, 0; green, 0; blue, 0 }  ,draw opacity=1 ]   (215,150) -- (215,160) ;
\draw  [color={rgb, 255:red, 0; green, 0; blue, 0 }  ,draw opacity=1 ][fill={rgb, 255:red, 255; green, 255; blue, 255 }  ,fill opacity=1 ] (190,35) -- (210,35) -- (210,50) -- (190,50) -- cycle ;
\draw  [color={rgb, 255:red, 0; green, 0; blue, 0 }  ,draw opacity=1 ][fill={rgb, 255:red, 255; green, 255; blue, 255 }  ,fill opacity=1 ] (220,35) -- (240,35) -- (240,50) -- (220,50) -- cycle ;
\draw [color={rgb, 255:red, 0; green, 0; blue, 0 }  ,draw opacity=1 ]   (230,50) -- (230,75) ;
\draw  [color={rgb, 255:red, 191; green, 97; blue, 106 }  ,draw opacity=1 ][fill={rgb, 255:red, 191; green, 97; blue, 106 }  ,fill opacity=0.2 ][dash pattern={on 4.5pt off 4.5pt}] (205,115) -- (225,115) -- (225,130) -- (205,130) -- cycle ;
\draw [color={rgb, 255:red, 0; green, 0; blue, 0 }  ,draw opacity=1 ]   (200,110) -- (200,135) ;
\draw [color={rgb, 255:red, 0; green, 0; blue, 0 }  ,draw opacity=1 ]   (215,90) -- (215.01,115) ;
\draw [color={rgb, 255:red, 0; green, 0; blue, 0 }  ,draw opacity=1 ]   (215,130) -- (215,135) ;
\draw  [color={rgb, 255:red, 0; green, 0; blue, 0 }  ,draw opacity=1 ][fill={rgb, 255:red, 255; green, 255; blue, 255 }  ,fill opacity=1 ] (190,95) -- (210,95) -- (210,110) -- (190,110) -- cycle ;
\draw  [color={rgb, 255:red, 0; green, 0; blue, 0 }  ,draw opacity=1 ][fill={rgb, 255:red, 255; green, 255; blue, 255 }  ,fill opacity=1 ] (220,95) -- (240,95) -- (240,110) -- (220,110) -- cycle ;
\draw [color={rgb, 255:red, 0; green, 0; blue, 0 }  ,draw opacity=1 ]   (230,110) -- (230,135) ;
\draw  [color={rgb, 255:red, 191; green, 97; blue, 106 }  ,draw opacity=1 ][fill={rgb, 255:red, 191; green, 97; blue, 106 }  ,fill opacity=0.2 ][dash pattern={on 4.5pt off 4.5pt}] (280,50) -- (300,50) -- (300,35) -- (280,35) -- cycle ;
\draw [color={rgb, 255:red, 0; green, 0; blue, 0 }  ,draw opacity=1 ]   (290,75) -- (290.01,50) ;
\draw [color={rgb, 255:red, 0; green, 0; blue, 0 }  ,draw opacity=1 ]   (290,35) -- (290,30) ;
\draw  [color={rgb, 255:red, 0; green, 0; blue, 0 }  ,draw opacity=1 ][fill={rgb, 255:red, 255; green, 255; blue, 255 }  ,fill opacity=1 ] (265,70) -- (285,70) -- (285,55) -- (265,55) -- cycle ;
\draw  [color={rgb, 255:red, 0; green, 0; blue, 0 }  ,draw opacity=1 ][fill={rgb, 255:red, 255; green, 255; blue, 255 }  ,fill opacity=1 ] (295,70) -- (315,70) -- (315,55) -- (295,55) -- cycle ;
\draw  [color={rgb, 255:red, 191; green, 97; blue, 106 }  ,draw opacity=1 ][fill={rgb, 255:red, 191; green, 97; blue, 106 }  ,fill opacity=0.2 ][dash pattern={on 4.5pt off 4.5pt}] (280,110) -- (300,110) -- (300,95) -- (280,95) -- cycle ;
\draw [color={rgb, 255:red, 0; green, 0; blue, 0 }  ,draw opacity=1 ]   (275,115) -- (275,90) ;
\draw [color={rgb, 255:red, 0; green, 0; blue, 0 }  ,draw opacity=1 ]   (290,135) -- (290.01,110) ;
\draw [color={rgb, 255:red, 0; green, 0; blue, 0 }  ,draw opacity=1 ]   (290,95) -- (290,90) ;
\draw  [color={rgb, 255:red, 0; green, 0; blue, 0 }  ,draw opacity=1 ][fill={rgb, 255:red, 255; green, 255; blue, 255 }  ,fill opacity=1 ] (265,130) -- (285,130) -- (285,115) -- (265,115) -- cycle ;
\draw  [color={rgb, 255:red, 0; green, 0; blue, 0 }  ,draw opacity=1 ][fill={rgb, 255:red, 255; green, 255; blue, 255 }  ,fill opacity=1 ] (295,130) -- (315,130) -- (315,115) -- (295,115) -- cycle ;
\draw [color={rgb, 255:red, 0; green, 0; blue, 0 }  ,draw opacity=1 ]   (305,115) -- (305,90) ;
\draw    (200,30) -- (200,35) ;
\draw    (230,30) -- (230,35) ;
\draw [color={rgb, 255:red, 0; green, 0; blue, 0 }  ,draw opacity=1 ]   (275,30) -- (275,55) ;
\draw [color={rgb, 255:red, 0; green, 0; blue, 0 }  ,draw opacity=1 ]   (305,30) -- (305,55) ;
\draw [color={rgb, 255:red, 0; green, 0; blue, 0 }  ,draw opacity=1 ]   (275,70) -- (275,75) ;
\draw [color={rgb, 255:red, 0; green, 0; blue, 0 }  ,draw opacity=1 ]   (305,70) -- (305,75) ;
\draw [color={rgb, 255:red, 0; green, 0; blue, 0 }  ,draw opacity=1 ]   (230,90) -- (230,95) ;
\draw [color={rgb, 255:red, 0; green, 0; blue, 0 }  ,draw opacity=1 ]   (275,130) -- (275,135) ;
\draw [color={rgb, 255:red, 0; green, 0; blue, 0 }  ,draw opacity=1 ]   (305,130) -- (305,135) ;
\draw [color={rgb, 255:red, 0; green, 0; blue, 0 }  ,draw opacity=1 ]   (200,90) -- (200,95) ;
\draw  [draw opacity=0] (235,70) -- (270,70) -- (270,95) -- (235,95) -- cycle ;

\draw (290,22.5) node  [font=\footnotesize]  {$f$};
\draw (290,82.5) node  [font=\footnotesize]  {$g$};
\draw (290,142) node  [font=\footnotesize]  {$h$};
\draw (215,22.5) node  [font=\footnotesize]  {$f$};
\draw (215,82.5) node  [font=\footnotesize]  {$g$};
\draw (215,142) node  [font=\footnotesize]  {$h$};
\draw (200,42.5) node  [font=\footnotesize]  {$m$};
\draw (200,102.5) node  [font=\footnotesize]  {$m'$};
\draw (230,42.5) node  [font=\footnotesize]  {$n$};
\draw (230,102.5) node  [font=\footnotesize]  {$n'$};
\draw (275,62.5) node  [font=\footnotesize]  {$m$};
\draw (305,62.5) node  [font=\footnotesize]  {$n$};
\draw (275,122.5) node  [font=\footnotesize]  {$m'$};
\draw (305,122.5) node  [font=\footnotesize]  {$n'$};
\draw (252.5,82.5) node    {$=$};

\end{tikzpicture}
   \caption{Dinaturality of sequential splits of monoidal contexts.}
  \label{fig:sequentialprotensor}
\end{figure}
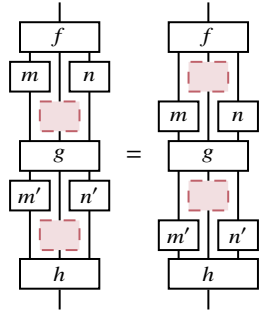
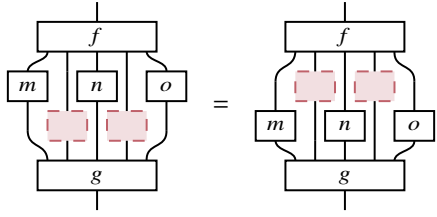
\begin{figure}[h]
  \centering

\tikzset{every picture/.style={line width=0.75pt}} %

\begin{tikzpicture}[x=0.75pt,y=0.75pt,yscale=-1,xscale=1]
\draw [color={rgb, 255:red, 0; green, 0; blue, 0 }  ,draw opacity=1 ]   (55,105) -- (55,115) ;
\draw  [color={rgb, 255:red, 191; green, 97; blue, 106 }  ,draw opacity=1 ][fill={rgb, 255:red, 191; green, 97; blue, 106 }  ,fill opacity=0.2 ][dash pattern={on 4.5pt off 4.5pt}] (30,65) -- (50,65) -- (50,80) -- (30,80) -- cycle ;
\draw  [color={rgb, 255:red, 0; green, 0; blue, 0 }  ,draw opacity=1 ][fill={rgb, 255:red, 255; green, 255; blue, 255 }  ,fill opacity=1 ] (24.99,20) -- (85,20) -- (85,35) -- (24.99,35) -- cycle ;
\draw [color={rgb, 255:red, 0; green, 0; blue, 0 }  ,draw opacity=1 ]   (55,10) -- (55,20) ;
\draw [color={rgb, 255:red, 0; green, 0; blue, 0 }  ,draw opacity=1 ]   (20,60) -- (20,75) ;
\draw  [color={rgb, 255:red, 0; green, 0; blue, 0 }  ,draw opacity=1 ][fill={rgb, 255:red, 255; green, 255; blue, 255 }  ,fill opacity=1 ] (25,90) -- (85,90) -- (85,105) -- (25,105) -- cycle ;
\draw [color={rgb, 255:red, 0; green, 0; blue, 0 }  ,draw opacity=1 ]   (39.99,35) -- (40,45) ;
\draw [color={rgb, 255:red, 0; green, 0; blue, 0 }  ,draw opacity=1 ]   (40,80) -- (40,90) ;
\draw [color={rgb, 255:red, 0; green, 0; blue, 0 }  ,draw opacity=1 ]   (29.99,35) .. controls (30.27,40.09) and (20.27,39.52) .. (20,45) ;
\draw [color={rgb, 255:red, 0; green, 0; blue, 0 }  ,draw opacity=1 ]   (80,90) .. controls (79.6,80.8) and (90.2,84.2) .. (90,75) ;
\draw [color={rgb, 255:red, 0; green, 0; blue, 0 }  ,draw opacity=1 ]   (30,90) .. controls (29.6,80.8) and (20.2,84.2) .. (20,75) ;
\draw [color={rgb, 255:red, 0; green, 0; blue, 0 }  ,draw opacity=1 ]   (40,45) -- (40,65) ;
\draw  [color={rgb, 255:red, 191; green, 97; blue, 106 }  ,draw opacity=1 ][fill={rgb, 255:red, 191; green, 97; blue, 106 }  ,fill opacity=0.2 ][dash pattern={on 4.5pt off 4.5pt}] (60,65) -- (80,65) -- (80,80) -- (60,80) -- cycle ;
\draw [color={rgb, 255:red, 0; green, 0; blue, 0 }  ,draw opacity=1 ]   (69.99,35) -- (70,65) ;
\draw [color={rgb, 255:red, 0; green, 0; blue, 0 }  ,draw opacity=1 ]   (70,80) -- (70,90) ;
\draw [color={rgb, 255:red, 0; green, 0; blue, 0 }  ,draw opacity=1 ]   (90,60) -- (90,75) ;
\draw  [color={rgb, 255:red, 0; green, 0; blue, 0 }  ,draw opacity=1 ][fill={rgb, 255:red, 255; green, 255; blue, 255 }  ,fill opacity=1 ] (45,45) -- (65,45) -- (65,60) -- (45,60) -- cycle ;
\draw [color={rgb, 255:red, 0; green, 0; blue, 0 }  ,draw opacity=1 ]   (55,35) -- (55,45) ;
\draw  [color={rgb, 255:red, 0; green, 0; blue, 0 }  ,draw opacity=1 ][fill={rgb, 255:red, 255; green, 255; blue, 255 }  ,fill opacity=1 ] (10,45) -- (30,45) -- (30,60) -- (10,60) -- cycle ;
\draw  [color={rgb, 255:red, 0; green, 0; blue, 0 }  ,draw opacity=1 ][fill={rgb, 255:red, 255; green, 255; blue, 255 }  ,fill opacity=1 ] (80,45) -- (100,45) -- (100,60) -- (80,60) -- cycle ;
\draw [color={rgb, 255:red, 0; green, 0; blue, 0 }  ,draw opacity=1 ]   (55,60) -- (55,90) ;
\draw [color={rgb, 255:red, 0; green, 0; blue, 0 }  ,draw opacity=1 ]   (79.99,35) .. controls (80.27,40.09) and (90.27,39.52) .. (90,45) ;
\draw [color={rgb, 255:red, 0; green, 0; blue, 0 }  ,draw opacity=1 ]   (180,105) -- (180,115) ;
\draw  [color={rgb, 255:red, 191; green, 97; blue, 106 }  ,draw opacity=1 ][fill={rgb, 255:red, 191; green, 97; blue, 106 }  ,fill opacity=0.2 ][dash pattern={on 4.5pt off 4.5pt}] (155,60) -- (175,60) -- (175,45) -- (155,45) -- cycle ;
\draw  [color={rgb, 255:red, 0; green, 0; blue, 0 }  ,draw opacity=1 ][fill={rgb, 255:red, 255; green, 255; blue, 255 }  ,fill opacity=1 ] (149.99,20) -- (210,20) -- (210,35) -- (149.99,35) -- cycle ;
\draw [color={rgb, 255:red, 0; green, 0; blue, 0 }  ,draw opacity=1 ]   (180,10) -- (180,20) ;
\draw [color={rgb, 255:red, 0; green, 0; blue, 0 }  ,draw opacity=1 ]   (145,65) -- (145,50) ;
\draw  [color={rgb, 255:red, 0; green, 0; blue, 0 }  ,draw opacity=1 ][fill={rgb, 255:red, 255; green, 255; blue, 255 }  ,fill opacity=1 ] (150,90) -- (210,90) -- (210,105) -- (150,105) -- cycle ;
\draw [color={rgb, 255:red, 0; green, 0; blue, 0 }  ,draw opacity=1 ]   (164.99,90) -- (165,80) ;
\draw [color={rgb, 255:red, 0; green, 0; blue, 0 }  ,draw opacity=1 ]   (165,45) -- (165,35) ;
\draw [color={rgb, 255:red, 0; green, 0; blue, 0 }  ,draw opacity=1 ]   (154.99,90) .. controls (155.27,84.91) and (145.27,85.48) .. (145,80) ;
\draw [color={rgb, 255:red, 0; green, 0; blue, 0 }  ,draw opacity=1 ]   (205,35) .. controls (204.6,44.2) and (215.2,40.8) .. (215,50) ;
\draw [color={rgb, 255:red, 0; green, 0; blue, 0 }  ,draw opacity=1 ]   (155,35) .. controls (154.6,44.2) and (145.2,40.8) .. (145,50) ;
\draw [color={rgb, 255:red, 0; green, 0; blue, 0 }  ,draw opacity=1 ]   (165,80) -- (165,60) ;
\draw  [color={rgb, 255:red, 191; green, 97; blue, 106 }  ,draw opacity=1 ][fill={rgb, 255:red, 191; green, 97; blue, 106 }  ,fill opacity=0.2 ][dash pattern={on 4.5pt off 4.5pt}] (185,60) -- (205,60) -- (205,45) -- (185,45) -- cycle ;
\draw [color={rgb, 255:red, 0; green, 0; blue, 0 }  ,draw opacity=1 ]   (194.99,90) -- (195,60) ;
\draw [color={rgb, 255:red, 0; green, 0; blue, 0 }  ,draw opacity=1 ]   (195,45) -- (195,35) ;
\draw [color={rgb, 255:red, 0; green, 0; blue, 0 }  ,draw opacity=1 ]   (215,65) -- (215,50) ;
\draw  [color={rgb, 255:red, 0; green, 0; blue, 0 }  ,draw opacity=1 ][fill={rgb, 255:red, 255; green, 255; blue, 255 }  ,fill opacity=1 ] (170,80) -- (190,80) -- (190,65) -- (170,65) -- cycle ;
\draw [color={rgb, 255:red, 0; green, 0; blue, 0 }  ,draw opacity=1 ]   (180,90) -- (180,80) ;
\draw  [color={rgb, 255:red, 0; green, 0; blue, 0 }  ,draw opacity=1 ][fill={rgb, 255:red, 255; green, 255; blue, 255 }  ,fill opacity=1 ] (135,80) -- (155,80) -- (155,65) -- (135,65) -- cycle ;
\draw  [color={rgb, 255:red, 0; green, 0; blue, 0 }  ,draw opacity=1 ][fill={rgb, 255:red, 255; green, 255; blue, 255 }  ,fill opacity=1 ] (205,80) -- (225,80) -- (225,65) -- (205,65) -- cycle ;
\draw [color={rgb, 255:red, 0; green, 0; blue, 0 }  ,draw opacity=1 ]   (180,65) -- (180,35) ;
\draw [color={rgb, 255:red, 0; green, 0; blue, 0 }  ,draw opacity=1 ]   (204.99,90) .. controls (205.27,84.91) and (215.27,85.48) .. (215,80) ;
\draw  [draw opacity=0] (100,50) -- (135,50) -- (135,75) -- (100,75) -- cycle ;

\draw (54.99,27.5) node  [font=\footnotesize]  {$f$};
\draw (55,98.5) node  [font=\footnotesize]  {$g$};
\draw (179.99,27.5) node  [font=\footnotesize]  {$f$};
\draw (180,98.5) node  [font=\footnotesize]  {$g$};
\draw (20,52.5) node  [font=\footnotesize]  {$m$};
\draw (55,52.5) node  [font=\footnotesize]  {$n$};
\draw (90,52.5) node  [font=\footnotesize]  {$o$};
\draw (215,72.5) node  [font=\footnotesize]  {$o$};
\draw (180,72.5) node  [font=\footnotesize]  {$n$};
\draw (145,72.5) node  [font=\footnotesize]  {$m$};
\draw (117.5,62.5) node    {$=$};

\end{tikzpicture}
   \caption{Parallel splits for monoidal contexts.}
  \label{fig:parallelprotensor}
\end{figure}

\begin{remark}[Algebra of monoidal contexts]
  We explicitly state all the operations that form the normal produoidal algebra of \monoidalContexts{}. We do so using 1-dimensional notation for compactness, but we do believe the conceptual picture is clearer when they are translated into 2-dimensional string diagrams.
  \begin{align*}
    & \quad \mbox{(Identity)} \\
    & (\im_A ⨾ ■ ⨾ \im_B)  \\
    \\
    & \quad \mbox{(Composition)}\\
    & \nmc{f}{M}{N}{g} ≺ \nmc{h}{P}{Q}{k} = \\
    & \qquad 
    \nmc{f ⨾ (\im_M ⊗ h ⊗ \im_N)}
    {M ⊗ P}{Q ⊗ N}
    {(\im_M ⊗ k ⊗ \im_N) ⨾ g}, \\
    \\
    & \quad \mbox{(Unit action)} \\
    & \nmc{f}{M}{N}{g} ≺ h = f ⨾ (\im_M ⊗ h ⊗ \im_N) ⨾ g,
    \\
    \\
    & \quad \mbox{(Seq. split first action)}  \\
    & \nmcs{f}{M}{N}{g}{K}{L}{h} ≺_1 
    (u ⨾ (\im_P ⊗ ■ ⊗ \im_Q) ⨾ v) = 
    \\ 
    & \qquad f ⨾ (\im_M ⊗ u ⊗ \im_N) ⨾ (\im_{M ⊗ P} ⊗ ■ ⊗ \im_{Q ⊗ N}) 
    ⨾ (\im_M ⊗ v ⊗ \im_N) ⨾ g ⨾ (\im_K ⊗ ■ ⊗ \im_L) ⨾ h, 
    \\
    \\
    & \quad \mbox{(Seq. split second action)}\\
    & \nmcs{f}{M}{N}{g}{K}{L}{h} ≺_2 \nmc{u}{P}{Q}{v} = \\
    & \qquad f ⨾ (\im_M ⊗ ■ ⊗ \im_N) ⨾ g ⨾ (\im_K ⊗ u ⊗ \im_L) 
    ⨾ (\im_{K ⊗ P} ⊗ ■ ⊗ \im_{Q ⊗ L}) ⨾ (\im_K ⊗ v ⊗ \im_L) ⨾ h, 
    \\
    \\
    & \quad \mbox{(Seq. split third action)}  \\
    & \nmc{f}{M}{N}{g} ≺ (u ⨾ (\im_P ⊗ ■ ⊗ \im_Q) ⨾ v ⨾
     (\im_R ⊗ ■ ⊗ \im_S) ⨾ w) = \\
    & \qquad f ⨾ (\im_M ⊗ u ⊗ \im_N) ⨾ (\im_{M ⊗ P} ⊗ ■ ⊗ \im_{Q ⊗ N}) ⨾ (\im_M ⊗ v ⊗ \im_N) ⨾ (\im_{M ⊗ R} ⊗ ■ ⊗ \im_{S ⊗ N}) 
    \\ & \qquad\quad⨾ (\im_M ⊗ w ⊗ \im_N) ⨾ g, \\
    \\
    & \quad \mbox{(Seq. left associativity)} \\
    & \nmcs{f}{M}{N}{g}{K}{L}{h} ≺^α_1 
    \nmcs{u}{P}{Q}{v}{R}{S}{w} = \\
    & \qquad f ⨾ (\im_M ⊗ u ⊗ \im_N) ⨾ (\im_{M ⊗ P} ⊗ ■ ⊗ \im_{Q ⊗ N}) ⨾ (\im_M ⊗ v ⊗ \im_N) ⨾ (\im_{M ⊗ R} ⊗ ■ ⊗ \im_{S ⊗ N}) \\
    & \qquad\quad ⨾ (\im_M ⊗ w ⊗ \im_N) ⨾ g ⨾ (\im_K ⊗ ■ ⊗ \im_L) ⨾ h, 
    \\
    \\
    & \quad \mbox{(Seq. right associativity)}\\
    & \nmcs{f}{M}{N}{g}{K}{L}{h} ≺^α_2 \nmcs{u}{P}{Q}{v}{R}{S}{w} = \\
    & \qquad  f ⨾ (\im_M ⊗ ■ ⊗ \im_N) ⨾ g ⨾(\im_K ⊗ u ⊗ \im_L) ⨾ (\im_{K ⊗ P} ⊗ ■ ⊗ \im_{R ⊗ L}) ⨾ (\im_K ⊗ v ⊗ \im_L) \\
    & \qquad\quad⨾ (\im_{K ⊗ R} ⊗ ■ ⊗ \im_{S ⊗ L}) ⨾ 
    (\im_L ⊗ w ⊗ \im_L) ⨾ h, 
    \\
    \\
    & \quad \mbox{(Seq. left unitor)} \\
    & \nmcs{f}{M}{N}{g}{K}{L}{h} ≺^\lambda u = \\
    & \qquad f ⨾ (\im_M ⊗ u ⊗ \im_N) ⨾ g ⨾ (\im_K ⊗ ■ ⊗ \im_L) ⨾ h, 
    \\
    \\
    & \quad \mbox{(Seq. right unitor)} \\ 
    & \nmcs{f}{M}{N}{g}{K}{L}{h} ≺^\rho u = \\
    & \qquad f ⨾ (\im_M ⊗ ■ ⊗ \im_N) ⨾ g ⨾ (\im_K ⊗ u ⊗ \im_L) ⨾ h, \\
    \\
    & \quad \mbox{(Par. split first action)} \\ 
    & \nmcp{f}{M}{N}{O}{g} ≺₁ \nmc{u}{P}{Q}{v} = \\
    & \quad f ⨾ (\im_M ⊗ u ⊗ \im_{N ⊗ X' ⊗ O}) 
    ⨾ (\im_{M ⊗ P} ⊗ ■ ⊗ \im_{Q ⊗ N} ⊗ ■ ⊗ \im_O) 
    ⨾ (\im_M ⊗ v ⊗ \im_{N ⊗ Y' ⊗ O}) ⨾ g, \\
    \\
    & \quad \mbox{(Par. split second action)}\\ 
    & \nmcp{f}{M}{N}{O}{g} ≺₂ \nmc{u}{P}{Q}{v} =  \\
    & \qquad f ⨾ (\im_{M ⊗ X ⊗ N} ⊗ u ⊗ \im_O) ⨾ (\im_M ⊗ ■ ⊗ \im_{N ⊗ P} ⊗ ■ ⊗ \im_{O ⊗ Q}) ⨾ (\im_{M ⊗ Y ⊗ N} ⊗ v ⊗ \im_O) ⨾ g, \\
    \\
    & \quad \mbox{(Par. split third action)} \\
    & \nmc{f}{M}{N}{g} ≺ 
    \nmcp{u}{P}{Q}{R}{v} = \\
    & \qquad f ⨾ (\im_M ⊗ u ⊗ \im_N) ⨾ (\im_{M ⊗ P} ⊗ ■ ⊗ \im_Q ⊗ ■ ⊗ \im_{R ⊗ N}) ⨾ (\im_M ⊗ w ⊗ \im_N) ⨾ g, \\
    \\
    & \mbox{(Par. left associativity)} \\
    & \nmcp{f}{M}{N}{O}{g} ≺^α_1 \nmcp{u}{P}{Q}{R}{v} = \\
    & \qquad f ⨾ (\im_M ⊗ u ⊗ \im_{N ⊗ X ⊗ O}) 
    ⨾ (\im_{M ⊗ P} ⊗ ■ ⊗ \im_{Q} ⊗ ■ ⊗ \im_{R ⊗ N} ⊗ ■ ⊗ \im_{O}) \\
    & \qquad\quad ⨾ (\im_M ⊗ v ⊗ \im_{N ⊗ Y ⊗ O}) ⨾ g, \\
    \\
    & \mbox{(Par. right associativity)} \\
    & \nmcp{f}{M}{N}{O}{g} ≺^α_2 \nmcp{u}{P}{Q}{R}{v} = \\
    & \qquad f ⨾ (\im_{M ⊗ X ⊗ N} ⊗ u ⊗ \im_O) 
    ⨾ (\im_{M} ⊗ ■ ⊗ \im_{N ⊗ P} ⊗ ■ ⊗ \im_{Q} ⊗ ■ ⊗ \im_{R ⊗ O}) \\ 
    & \qquad\quad ⨾ (\im_{M ⊗ Y ⊗ N} ⊗ v ⊗ \im_O) ⨾ g, \\
    \\
    & \quad \mbox{(Par. left unitor)} \\
    & \nmcp{f}{M}{N}{O}{g} ≺^\lambda u = \\
    &\qquad f ⨾ (\im_M ⊗ u ⊗ \im_{N ⊗ X' ⊗ O}) 
    ⨾ (\im_{M ⊗ Y ⊗ N} ⊗ ■ ⊗ \im_O) ⨾ g = \\
    &\qquad f  ⨾ (\im_{M ⊗ Y ⊗ N} ⊗ ■ ⊗ \im_O)  
    ⨾ (\im_M ⊗ u ⊗ \im_{N ⊗ X' ⊗ O}) ⨾  g,\\
    \\
    &\quad  \mbox{(Par. right unitor)} \\  
    & \nmcp{f}{M}{N}{O}{g} ≺^\rho v = \\
    & \qquad f ⨾ (\im_{M ⊗ X ⊗ N} ⊗ v ⊗ \im_O) 
    ⨾  (\im_{M} ⊗ ■ ⊗ \im_{N ⊗ Y' ⊗ O}) ⨾ g = \\
    & \qquad  f ⨾ (\im_{M} ⊗ ■ ⊗ \im_{N ⊗ Y' ⊗ O}) 
    ⨾  (\im_{M ⊗ X ⊗ N} ⊗ v ⊗ \im_O) ⨾ g, \\
    \\
    & \quad \mbox{(Laxator, left side)}  \\
    & \nmcp{f}{M}{N}{O}{g} & \\
    & \qquad    ≺₁^{ψ} \nmcs{j₀}{U}{V}{j₁}{U'}{V'}{j₂} \\
    & \qquad    ≺₂^{ψ} \nmcs{k₀}{W}{T}{k₁}{W'}{T'}{k₂} = \\ 
    & \quad f ⨾ 
    (\im_M ⊗ j₀ ⊗ \im_N ⊗ k₀ ⊗ \im_O) 
    ⨾ (\im_{M ⊗ U} ⊗ ■ ⊗ \im_{V ⊗ N ⊗ U'} ⊗ ■ ⊗ \im_{V' ⊗ O}) \\
    & \qquad\quad
    ⨾ (\im_M ⊗ j₁ ⊗ \im_N ⊗ k₁ ⊗ \im_O) 
    ⨾ (\im_{M ⊗ W} ⊗ ■ ⊗ \im_{T ⊗ N ⊗ W'} ⊗ ■ ⊗ \im_{T' ⊗ O}) \\
    & \qquad\quad
    ⨾ (\im_M ⊗ j₂ ⊗ \im_N ⊗ k₂ ⊗ \im_O) ⨾ g, \\
    \\
    & \quad \mbox{(Laxator, right side)} \\
    & \nmcs{f}{M}{N}{g}{K}{L}{h} \\
    & \qquad   ≺₁^{ψ} \nmcp{j₀}{P}{Q}{R}{j₁} \\ 
    & \qquad   ≺₂^{ψ} \nmcp{k₀}{P'}{Q'}{R'}{k₁} = & \\ 
    & \quad f ⨾ (\im_{M} ⊗ j₀ ⊗ \im_N) 
    ⨾ (\im_{M ⊗ P} ⊗ ■ ⊗ \im_Q ⊗ ■ ⊗ \im_{R ⊗ N}) \\
    & \qquad\quad
    ⨾ (\im_M ⊗ j₁ ⊗ \im_N) ⨾ g ⨾ (\im_K ⊗ k₀ ⊗ \im_L) 
    ⨾ (\im_{K ⊗ P'} ⊗ ■ ⊗ \im_{Q'} ⊗ ■ ⊗ \im_{R' ⊗ L}) \\
    & \qquad\quad
    ⨾ (\im_K ⊗ k₁ ⊗ \im_L) ⨾ h.
    \\
  \end{align*}
\end{remark}

\begin{remark}
  In the following derivations, we understand that an isolated $(\blacksquare)$ actually means $(\im_I \otimes \blacksquare \otimes \im_I)$.
\end{remark}

\begin{proposition}[From \Cref{prop:contextCategory}]
    \label{ax:prop:contextCategory}
    \MonoidalContexts{} form a category. Composition of \monoidalContexts{} is associative and unital.
\end{proposition}
\begin{proof}
  We first check that the composition of \monoidalContexts{} is associative.
  \begin{align*}
    & (\nmc{f}{M}{N}{g} ≺ \nmc{f'}{M'}{N'}{g'}) ≺  \nmc{f''}{M''}{N''}{g''}
    & \quad = \\
    & \nmc{f ⨾ (\im_M ⊗ f' ⊗ \im_N)}{M ⊗ M'}{N' ⊗ N}{(\im_M ⊗ g' ⊗ \im_N) ⨾ g} ≺ 
    \\ & \qquad \nmc{f''}{M''}{N''}{g''}
    & \quad = \\
    & f ⨾ (\im_M ⊗ f' ⊗ \im_N) ⨾ (\im_{M ⊗ M'} ⊗ f'' ⊗ \im_{N ⊗ N'}) \\
    & \qquad ⨾ (\im_{M ⊗ M' ⊗ M''}⊗ \blacksquare ⊗ \im_{N'' ⊗ N' ⊗ N}) 
    ⨾ (\im_{M ⊗ M'} ⊗ g'' ⊗ \im_{N' ⊗ N}) ⨾ (\im_N ⊗ g' ⊗ \im_N) ⨾ g
    & \quad = \\
    & f ⨾ (\im_{M} ⊗ (f' ⨾ (\im_{M'} ⊗ f'' ⊗ \im_{N'})) ⊗ \im_{N}) \\
    & \qquad ⨾ (\im_{M ⊗ M' ⊗ M''} ⊗ \blacksquare ⊗ \im_{N'' ⊗ N' ⊗ N}) ⨾ (\im_{M} ⊗ ((\im_{M'} ⊗ g'' ⊗ \im_{N'}) ⨾ g') ⊗ \im_{N}) ⨾ g
    & \quad = \\
    & \nmc{f}{M}{N}{g} ≺ 
      (f' ⨾ (\im_{M'} ⊗ f'' ⊗ \im_{N'}) ⨾ \\ 
    & \qquad (\im_{M' ⊗ M''} \otimes \blacksquare \otimes \im_{N'' ⊗ N'}) ⨾ (\im_{M'} ⊗ g'' ⊗ \im_{N'}) ⨾ g')
    & \quad = \\
    & \nmc{f}{M}{N}{g} ≺ (\nmc{f'}{M'}{N'}{g'} ≺ \nmc{f''}{M''}{N''}{g''})
  \end{align*}
  We now check left unitality of the identities,
  \begin{align*}
    & \nmc{f}{M}{N}{g} ≺ \nmc{\im_X}{I}{I}{\im_Y} & = \\
    & \qquad \nmc{f ⨾ (\im_M ⊗ \im_X ⊗ \im_N)}{M}{N}{(\im_M ⊗ \im_X ⊗ \im_N) ⨾ g} & = \\
    & \qquad \nmc{f}{M}{N}{g},
  \end{align*}
  and right unitality,
  \begin{align*}
    & \nmc{\im_A}{I}{I}{\im_B} ≺ \nmc{f}{M}{N}{g} & = \\
    & \qquad \nmc{\im_A ⨾ (\im_I ⊗ f ⊗ \im_I)}{M}{N}{(\im_I ⊗ g ⊗ \im_I) ⨾ \im_B} & = \\
    & \qquad \nmc{f}{M}{N}{g}.
  \end{align*}
  This concludes the proof.
\end{proof}

\begin{proposition}[From \Cref{prop:MonoidalContextProtensor}]
  \label{ax:prop:MonoidalContextProtensor}
  The category of \monoidalContexts{} forms a normal \produoidalCategory{} with its units, sequential and parallel splits.
\end{proposition}
\begin{proof}
  \Cref{lemma:mc-assocseq,lemma:mc-leftunitseq,lemma:mc-rightunitseq} construct the associators and unitors for the sequential promonoidal structure, and \Cref{lemma:mc-assocpar,lemma:mc-leftunitpar,lemma:mc-rightunitpar} define the associators and unitors for the parallel promonoidal structure. As they are all constructed with \YonedaIsomorphisms{}, they must satisfy the coherence equations. \Cref{lemma:mc-laxators} defines the laxators, again using only \YonedaIsomorphisms{} and composition in $ℂ$. For concision, our proofs freely elide the tensor product of objects, writing $XY$ for $X ⊗ Y$.
\end{proof}

\begin{lemma}[Monoidal contexts sequential associator] \label{lemma:mc-assocseq}
  We construct a natural isomorphism
  $$(≺^α_2) : \intr{{\biobj{U}{V} \in \Mctx{ℂ}}} \MC{\biobj{A}{B}}{\biobj{X}{Y} ◁ \biobj{U}{V}} × \MC{\biobj{U}{V}}{\biobj{X'}{Y'} ◁ \biobj{X''}{Y''}} ≅ \intr{{\biobj{U}{V} \in \Mctx{ℂ}}} \MC{\biobj{A}{B}}{\biobj{U}{V} ◁ \biobj{X''}{Y''}} × \MC{\biobj{U}{V}}{\biobj{X}{Y} ◁ \biobj{X'}{Y'}} : (≺^α_1),$$
  satisfying the coherence equations of \produoidalCategories{}. This isomorphism is defined on representatives of the equivalence class as
  $$\begin{aligned}
    &\nmcs{f_0}{}{}{f_1}{}{}{f_2} ≺^α_2 \nmcs{g_0}{}{}{g_1}{}{}{g_2} = \\
    &(\nmcs{h_0}{}{}{h_1}{}{}{h_2} ｜ \nmcs{k_0}{}{}{k_1}{}{}{k_2})
    \end{aligned}$$
  if and only if
  $$\begin{aligned}
    &f_0 ⨾ (\im ⊗ ■ ⊗ \im) ⨾ f_1 ⨾ (\im ⊗ g_0 ⊗ \im) ⨾ (\im ⊗ ■ ⊗ \im) ⨾ (\im ⊗ g_1 ⊗ \im) ⨾ (\im ⊗ ■ ⊗ \im) ⨾ (\im ⊗ g_2 ⊗ \im) ⨾ f_2 = \\
    &h_0 ⨾ (\im ⊗ k_0 ⊗ \im) ⨾ (\im ⊗ ■ ⊗ \im) ⨾ (\im ⊗ k_1 ⊗ \im) ⨾ (\im ⊗ ■ ⊗ \im) ⨾ (\im ⊗ k_2 ⊗ \im) ⨾ h_1 ⨾ (\im ⊗  ■ ⊗ \im) ⨾ h_2.
  \end{aligned}$$
\end{lemma}
\begin{proof}
  Firstly, we construct an isomorphism between the left hand side and a set of quadruples of morphisms. This isomorphism sends the pair
  \begin{align*}
    &\nmcs{f_0}{}{}{f_1}{}{}{f_2} ｜ \nmcs{g_0}{}{}{g_1}{}{}{g_2} \\
    \text{to}\quad &(f_0 ⨾ (\im ⊗ ■ ⊗ \im) ⨾ f_1 ⨾ (\im ⊗ g_0 ⊗ \im) ⨾ (\im ⊗ ■ ⊗ \im) ⨾ g_1 ⨾ (\im ⊗ ■ ⊗ \im) ⨾ (\im ⊗ g_2 ⊗ \im) ⨾ f_2).
  \end{align*}
  The isomorphism is constructed by the following coend derivation.
  \begin{align*}
    & \intr{\biobj{U}{V} \in \Mctx{ℂ}} \MC{\biobj{A}{B}}{\biobj{X}{Y} ◁ \biobj{U}{V}} × \MC{\biobj{U}{V}}{\biobj{X'}{Y'} ◁ \biobj{X''}{Y''}} &\quad
    \bydef \\ %
    & \intr{\biobj{U}{V} \in \Mctx{ℂ}, \M,\N,\O,\P \in ℂ}
    ℂ(A; \M X \N) × ℂ(\M  Y \N; \O U \P)  × ℂ(\O V \P; B) × \MC{\biobj{U}{V}}{\biobj{X'}{Y'} ◁ \biobj{X''}{Y''}} &\quad
    \bydef \\%
    & \intr{\biobj{U}{V} \in \Mctx{ℂ}, \M,\N,\O,\P,\Q,\R \in ℂ}
    ℂ(A; \M X \N) × \MC{\biobj{\M Y \N}{B}}{\biobj{U}{V}}  × ℂ(U;  \O X' \P) × ℂ(\O Y' \P; \Q  X''  \R) × ℂ(\Q  Y'' \R; V) &\quad
    \yo2 \\%
    & \intr{\M,\N,\O,\P,\Q,\R \in ℂ}
    ℂ(A; \M  X \N) × ℂ(\M  Y \N; \O  X' \P) × ℂ(\O  Y'  \P;  \Q  X''  \R) × ℂ(\Q  Y''  \R; B).
  \end{align*}

  Now we construct an isomorphism between the right hand side and the same set of quadruples of morphisms.
  This isomorphism sends the pair
  \begin{align*}
    &\nmcs{h_0}{}{}{h_1}{}{}{h_2} ｜ \nmcs{k_0}{}{}{k_1}{}{}{k_2}) \\
    \text{to}\quad &(h_0 ⨾ (\im ⊗ k_0 ⊗ \im) ⨾ (\im ⊗ ■ ⊗ \im) ⨾ k_1 ⨾ (\im ⊗ ■ ⊗ \im) ⨾ (\im ⊗ k_2 ⊗ \im) ⨾ h_1 ⨾ (\im ⊗ ■ ⊗ \im) ⨾ h_2).
  \end{align*}
  \begin{align*}
    & \intr{\biobj{U}{V} \in \Mctx{ℂ}} \MC{\biobj{A}{B}}{\biobj{U}{V} ◁ \biobj{X''}{Y''}} × \MC{\biobj{U}{V}}{\biobj{X}{Y} ◁ \biobj{X'}{Y'}} &\quad
    \bydef \\ %
    & \intr{\biobj{U}{V} \in \Mctx{ℂ}, \M,\N,\O,\P \in ℂ} ℂ(A; \M U \N) × ℂ(\M  V \N; \O X'' \P)  × ℂ(\O Y'' \P; B) ×  \MC{\biobj{U}{V}}{\biobj{X}{Y} ◁ \biobj{X'}{Y'}} &\quad
    \bydef \\%
    & \intr{\biobj{U}{V} \in \Mctx{ℂ}, \M,\N,\O,\P,\Q,\R \in ℂ} \MC{\biobj{A}{\O X'' \P}}{\biobj{U}{V}}  × ℂ(\O  Y''  \P; B) × ℂ(U; \M  X  \N) × ℂ(\M  Y  \N; \Q  X'  \R) ×  ℂ(\Q  Y'  \R; V) &\quad
    \yo2 \\%
    & \intr{\M,\N,\O,\P,\Q,\R \in ℂ} ℂ(A; \M  X \N) × ℂ(\M  Y \N; \Q  X' \R) × ℂ(\Q  Y'  \R;  \O  X''  \P) × ℂ(\O  Y''  \P; B).
  \end{align*}
  Composing both isomorphisms, we obtain the desired associator. Since it is composed exclusively from \YonedaIsomorphisms{}, it must satisfy the coherence equations of \produoidalCategories{} (\Cref{ax:def:produoidal}).
\end{proof}

\begin{lemma}[Monoidal contexts sequential left unitor] \label{lemma:mc-leftunitseq}
  We construct a natural isomorphism
  $$(≺^{λ}) : ∫^{\biobj{U}{V} \in \Mctx{ℂ}} \MC{\biobj{A}{B}}{\biobj{U}{V} ◁ \biobj{X}{Y}} × \MC{\biobj{U}{V}}{N} ≅ \MC{\biobj{A}{B}}{\biobj{X}{Y}},$$
  satisfying the coherence equations of \produoidalCategories{}. This isomorphism is defined on representatives of the equivalence class as
  \begin{align*}
    & \nmcs{f_0}{M}{N}{f_1}{K}{L}{f_2} ≺^\lambda g & = \\ 
    & f_0 ⨾ (\im_M ⊗ g ⊗ \im_N) ⨾ f_1 ⨾ (\im_K ⊗ ■ ⊗ \im_L) ⨾ f_2.
  \end{align*}
\end{lemma}
\begin{proof}
  We need to prove that this function is well-defined and does indeed induce an isomorphism after quotienting. We show this by constructing the isomorphism using coend calculus.
  \begin{align*}
    & \intr{\biobj{U}{V} \in \Mctx{ℂ}} \MC{\biobj{A}{B}}{\biobj{U}{V} ◁ \biobj{X}{Y}} × \MC{\biobj{U}{V}}{N} &\quad
    \bydef \\
    & \intr{\biobj{U}{V} \in \Mctx{ℂ}, \P,\Q,\R,\S \in ℂ}
    ℂ(A; \P  U  \Q) × ℂ(\P  V  \Q; \R  X \S) × ℂ(\R Y \S; B) × \MC{\biobj{U}{V}}{N} &\quad
    \bydef \\
    & \intr{\biobj{U}{V} \in \Mctx{ℂ}, \R,\S \in ℂ}
    \MC{\biobj{A}{\R  X  \S}}{\biobj{U}{V}} × ℂ(\R  Y  \S; B) × \MC{\biobj{U}{V}}{N} &\quad
    \yo2 \\
    & \intr{\R, \S \in ℂ} 
    ℂ(A; \R X \S) × ℂ(\R  Y  \S; B) &\quad
    \bydef \\
    & \MC{\biobj{A}{B}}{\biobj{X}{Y}}. \phantom{\int}
  \end{align*}

  Since it is composed exclusively from \YonedaIsomorphisms{}, it must satisfy the coherence equations of \produoidalCategories{} (\Cref{ax:def:produoidal}).
\end{proof}

\begin{lemma}[Monoidal contexts sequential right unitor] \label{lemma:mc-rightunitseq}
  We construct a natural isomorphism
  $$(≺^ρ) : ∫^{\biobj{U}{V} \in \Mctx{ℂ}} \MC{\biobj{A}{B}}{\biobj{X}{Y} ◁ \biobj{U}{V}} × \MC{\biobj{U}{V}}{N} ≅ \MC{\biobj{A}{B}}{\biobj{X}{Y}}$$
  satisfying the coherence equations of \produoidalCategories{}. This isomorphism is defined on representatives of the equivalence class as
  \begin{align*}
    & \nmcs{f_0}{M}{N}{f_1}{K}{L}{f_2} ≺^\rho g & = \\ 
    & f_0 ⨾ (\im_M ⊗ ■ ⊗ \im_N) ⨾ f_1 ⨾ (\im_K ⊗ g ⊗ \im_L) ⨾ f_2.
  \end{align*}
\end{lemma}
\begin{proof} As above, we do this by coend calculus:
  \begin{align*}
    & \intr{\biobj{U}{V} \in \Mctx{ℂ}} \MC{\biobj{A}{B}}{\biobj{X}{Y} ◁ \biobj{U}{V}} × \MC{\biobj{U}{V}}{N} &\quad
    \bydef \\
    & \intr{\biobj{U}{V} \in \Mctx{ℂ}, \P,\Q,\R,\S \in ℂ} ℂ(A; \P  X \Q) × ℂ(\P Y \Q; \R U \S) × ℂ(\R V  \S; B) × \MC{\biobj{U}{V}}{N} &\quad
    \bydef \\
    & \intr{\biobj{U}{V} \in \Mctx{ℂ}, \P,\Q,\R,\S \in ℂ} ℂ(A; \P X \Q) × \MC{\biobj{\P Y \Q}{B}}{\biobj{U}{V}} × \MC{\biobj{U}{V}}{N} &\quad
    \yo2 \\
    & \intr{\R,\S \in ℂ} ℂ(A; \R  X \S) × ℂ(\R Y \S; B) &\quad
    \bydef \\
    & \MC{\biobj{A}{B}}{\biobj{X}{Y}}. \phantom{\int}
  \end{align*}

  Since it is composed exclusively from \YonedaIsomorphisms{}, it must satisfy the coherence equations of \produoidalCategories{} (\Cref{ax:def:produoidal}).
\end{proof}

\begin{lemma}[Monoidal contexts parallel associator] \label{lemma:mc-assocpar}
  We construct a natural isomorphism
  $$(≺^α_2) : ∫^{\biobj{U}{V} \in \Mctx{ℂ}} \MC{\biobj{A}{B}}{\biobj{X}{Y} ⊗ \biobj{U}{V}} × \MC{\biobj{U}{V}}{\biobj{X'}{Y'} ⊗ \biobj{X''}{Y''}} ≅ ∫^{\biobj{U}{V} \in \Mctx{ℂ}} \MC{\biobj{A}{B}}{\biobj{U}{V} ⊗ \biobj{X''}{Y''}} × \MC{\biobj{U}{V}}{\biobj{X}{Y} ⊗ \biobj{X'}{Y'}} : (≺^α_1) $$
  exclusively from \YonedaIsomorphisms{}. This isomorphism is defined on representatives of the equivalence class as
  \begin{align*}
    & \nmcp{f_0}{}{}{}{f_1} ≺^α_1 \nmcp{g_0}{}{}{}{g_1} & =  \\
    & \nmcp{h_0}{}{}{}{h_1} \mathbin{|} \nmcp{j_0}{}{}{}{j_1}
  \end{align*}
  if and only if
  $$\begin{aligned}
    &f_0 ⨾ (\im_M ⊗ g_0 ⊗ \im_{N ⊗ X ⊗ O}) ⨾ (\im_{M ⊗ P} ⊗ ■ ⊗ \im_{Q} ⊗ ■ ⊗ \im_{R ⊗ N} ⊗ ■ ⊗ \im_{O}) ⨾ (\im_M ⊗ g_1 ⊗ \im_{N ⊗ Y ⊗ O}) ⨾ f_1 = \\
    &h_0 ⨾ (\im_{M ⊗ X ⊗ N} ⊗ j_0 ⊗ \im_O) ⨾ (\im_{M} ⊗ ■ ⊗ \im_{N ⊗ P} ⊗ ■ ⊗ \im_{Q} ⊗ ■ ⊗ \im_{R ⊗ O}) ⨾ (\im_{M ⊗ Y ⊗ N} ⊗ j_1 ⊗ \im_O) ⨾ h_1,
  \end{aligned}$$
\end{lemma}
\begin{proof}
  The left hand side is isomorphic to the following set,
  \begin{align*}
    & \intr{\biobj{U}{V} \in \Mctx{ℂ}} \MC{\biobj{A}{B}}{\biobj{X}{Y} ⊗ \biobj{U}{V}} × \MC{\biobj{U}{V}}{\biobj{X'}{Y'} ⊗ \biobj{X''}{Y''}} & \bydef \\
    & \intr{\biobj{U}{V} \in \Mctx{ℂ}, \M, \N, \O \in ℂ} 
        ℂ(A; \M ⊗ X ⊗ \N ⊗ U ⊗ \O) × 
        ℂ(\M ⊗ Y ⊗ \N ⊗ V ⊗ \O; B) × 
        \MC{\biobj{U}{V}}{\biobj{X'}{Y'} ⊗ \biobj{X''}{Y''}} & \yo2 \\
    & \intr{\biobj{U}{V} \in \Mctx{ℂ},\M,\N,\O,\P,\Q \in ℂ} 
        ℂ(A; \M ⊗ X ⊗ \P) × ℂ(\P; \N ⊗ U ⊗ \O) × ℂ(\M ⊗ Y ⊗ \Q; B) \\
    & \qquad\qquad × ℂ(\N ⊗ V ⊗ \O; \Q) × \MC{\biobj{U}{V}}{\biobj{X'}{Y'} ⊗ \biobj{X''}{Y''}} & \bydef \\
    & \intr{\biobj{U}{V} \in \Mctx{ℂ}, \M,\M',\N',\O',\P,\Q\in ℂ} 
        ℂ(A; \M ⊗ X ⊗ \P) × \MC{\biobj{\P}{\Q}}{\biobj{U}{V}} × ℂ(\M ⊗ Y ⊗ \Q; B)
     × ℂ(U ; \M' ⊗ X' ⊗ \N' ⊗ X'' ⊗ \O') \\
    & \qquad\qquad × ℂ(\M' ⊗ Y' ⊗ \N' ⊗ Y'' ⊗ \O'; V) & \yo2\\
    & \intr{\M,\M',\N',\O' \in ℂ} 
        ℂ(A; \M ⊗ X ⊗ \M' ⊗ X' ⊗ \N' ⊗ X'' ⊗ \O') × 
        ℂ(\M ⊗ Y ⊗ \M' ⊗ Y' ⊗ \N' ⊗ Y'' ⊗ \O' ; B).
  \end{align*}

  \noindent In the same way, the right hand side is isomorphic to the following set,
  \begin{align*}
    & \intr{\biobj{U}{V} \in \Mctx{ℂ}} \MC{\biobj{A}{B}}{\biobj{U}{V} ⊗ \biobj{X''}{Y''}} × \MC{\biobj{U}{V}}{\biobj{X}{Y} ⊗ \biobj{X'}{Y'}} & \bydef \\
    & \intr{\biobj{U}{V} \in \Mctx{ℂ}, \M,\N,\O \in ℂ} 
        ℂ(A; \M ⊗ U ⊗ \N ⊗ X'' ⊗ \O) × 
        ℂ(\M ⊗ V ⊗ \N ⊗ Y'' ⊗ \O; B) × \MC{\biobj{U}{V}}{\biobj{X}{Y} ⊗ \biobj{X'}{Y'}} & \yo1 \\
    & \intr{\biobj{U}{V} \in \Mctx{ℂ}, \M,\N,\O,\P,\Q \in ℂ} 
        ℂ(\P; \M ⊗ U ⊗ \N) × ℂ(A; \P ⊗ X'' ⊗ \O) × ℂ(\M ⊗ V ⊗ \N ; \Q) & \bydef \\
    & \qquad × ℂ(\Q ⊗ Y'' ⊗ \O; B) × \MC{\biobj{U}{V}}{\biobj{X}{Y} ⊗ \biobj{X'}{Y'}} & \yo1 \\
    & \intr{\biobj{U}{V} \in \Mctx{ℂ},\M',\N',\O',\O,\P,\Q \in ℂ} 
        \MC{\biobj{P}{Q}}{\biobj{U}{V}} × ℂ(A; \P ⊗ X'' ⊗ \O) × ℂ(\Q ⊗ Y'' ⊗ \O; B) & 
        \bydef \\
    &\qquad × ℂ(U ; \M' ⊗ X ⊗ \N' ⊗ X' ⊗ \O') × ℂ(\M' ⊗ Y ⊗ \N' ⊗ Y' ⊗ \O'; V) & \yo1 \\
    & \intr{\M',\N',\O',\O,\P,\Q \in ℂ} 
        ℂ(A; \P ⊗ X'' ⊗ \O) × ℂ(\Q ⊗ Y'' ⊗ \O; B) \\
    &\qquad × ℂ(\P ; \M' ⊗ X ⊗ \N' ⊗ X' ⊗ \O') × 
        ℂ(\M' ⊗ Y ⊗ \N' ⊗ Y' ⊗ \O'; \Q) & \yo1 \\
    & \intr{\M',\N',\O',\O \in ℂ} 
        ℂ(A; \M' ⊗ X ⊗ \N' ⊗ X' ⊗ \O' ⊗ X'' ⊗ \O) × 
        ℂ(\M' ⊗ Y ⊗ \N' ⊗ Y' ⊗ \O' ⊗ Y'' ⊗ \O ; B). &\quad
  \end{align*}
  Composing both isomorphisms, we obtain the desired associator.
\end{proof}

\begin{lemma}[Monoidal contexts parallel left unitor] \label{lemma:mc-leftunitpar}
  We construct a natural isomorphism
  $$(≺^{\lambda}) : ∫^{\biobj{U}{V} \in \Mctx{ℂ}} \MC{\biobj{A}{B}}{\biobj{U}{V} ⊗ \biobj{X}{Y}} × \MC{\biobj{U}{V}}{N} ≅ \MC{\biobj{A}{B}}{\biobj{X}{Y}}$$
  exclusively from \YonedaIsomorphisms{}. This isomorphism is defined by
  $$\nmcp{f_0}{M}{N}{H}{f_1} ≺^{λ} g = f_0 ⨾ (\im_M ⊗ g ⊗ \im_{N⊗X'⊗O}) ⨾ (\im_{M⊗Y⊗N} ⊗ ■ ⊗ \im_{O}) ⨾ f_1.$$
\end{lemma}
\begin{proof}
  \begin{align*}
    & \intr{\biobj{U}{V} \in \Mctx{ℂ}} 
    \MC{\biobj{A}{B}}{\biobj{U}{V} ⊗ \biobj{X}{Y}} × \MC{\biobj{U}{V}}{N} & \bydef \\
      & \intr{\biobj{U}{V} \in \Mctx{ℂ}, \P,\Q,\R \in ℂ} 
      ℂ(A; \P ⊗ U ⊗ \Q ⊗ X ⊗ \R) × ℂ(\P ⊗ V ⊗ \Q ⊗ Y ⊗ \R; B) × ℂ(U; V) & \yo1 \\
    & \intr{\biobj{U}{V} \in \Mctx{ℂ}, \P,\Q,\R,\S,\T \in ℂ} 
        ℂ(A; \S ⊗ X ⊗ \R) × ℂ(\S; \P ⊗ U ⊗ \Q) × ℂ(\P ⊗ V ⊗ \Q; \T) \\
    &\qquad × ℂ(\T ⊗ Y ⊗ \R; B) × ℂ(U; V) & \bydef \\
    & \intr{\biobj{U}{V} \in \Mctx{ℂ}, \R,\S,\T \in ℂ} 
      ℂ(A; \S ⊗ X ⊗ \R) × \MC{\biobj{S}{T}}{\biobj{U}{V}} × ℂ(\T ⊗ Y ⊗ \R; B) ×
      ℂ(U; V) &\yo1 \\
    & \intr{\S,\R \in ℂ} ℂ(A; \S ⊗ X ⊗ \R) × ℂ(\S ⊗ Y ⊗ \R; B) & \bydef \\
        & \MC{\biobj{A}{B}}{\biobj{X}{Y}}. \phantom{\int}
  \end{align*}
\end{proof}

\begin{lemma}[Monoidal contexts parallel right unitor] \label{lemma:mc-rightunitpar}
  We construct a natural isomorphism
  $$(≺^ρ) : ∫^{\biobj{U}{V} \in \Mctx{ℂ}} \MC{\biobj{A}{B}}{\biobj{X}{Y} ⊗ \biobj{U}{V}} × \MC{\biobj{U}{V}}{N} ≅ \MC{\biobj{A}{B}}{\biobj{X}{Y}}$$
  exclusively from \YonedaIsomorphisms{}. This isomorphism is defined by
  $$\nmcp{f_0}{M}{N}{H}{f_1} ≺^{ρ} g = f_0 ⨾ (\im_M ⊗ ■ ⊗ \im_{N⊗Y'⊗O})  ⨾ (\im_{M⊗X⊗N} ⊗ g ⊗ \im_O) ⨾ f_1.$$
\end{lemma}
\begin{proof}
  We construct the isomorphism by the following coend derivation,
  \begin{align*}
    & \intr{\biobj{U}{V} \in \Mctx{ℂ}} \MC{\biobj{A}{B}}{\biobj{X}{Y} ⊗ \biobj{U}{V}} × \MC{\biobj{U}{V}}{N} &\bydef \\
    & \intr{\biobj{U}{V} \in \Mctx{ℂ}, \P,\Q,\R \in ℂ} ℂ(A; \P ⊗ X ⊗ \Q ⊗ U ⊗ \R) × ℂ(\P ⊗ Y ⊗ \Q ⊗ V ⊗ \R; B) × ℂ(U; V) &\yo1 \\
    & \intr{\biobj{U}{V} \in \Mctx{ℂ}, \P,\Q,\R,\S,\T \in ℂ} 
      ℂ(A; \P ⊗ X ⊗ \S) × ℂ(\S; \Q ⊗ U ⊗ \R) × ℂ(\Q ⊗ V ⊗ \R; \T) \\
    &\qquad × ℂ(\P ⊗ Y ⊗ \T;  B) × ℂ(U; V) &\bydef \\
    & \intr{\biobj{U}{V} \in \Mctx{ℂ}, \P,\S,\T \in ℂ} 
        ℂ(A; \P ⊗ X ⊗ \S) × \MC{\biobj{S}{T}}{\biobj{U}{V}} × ℂ(\P ⊗ Y ⊗ \T;  B) × ℂ(U; V) & \yo1 \\
    & \intr{\P,\T \in ℂ} ℂ(A; \P ⊗ X ⊗ \T) × ℂ(\P ⊗ Y ⊗ \T; B) &\bydef \\
    & \MC{\biobj{A}{B}}{\biobj{X}{Y}}. \phantom{\int}
  \end{align*}
  This concludes the proof.
\end{proof}

\begin{lemma}[Monoidal contexts laxators] \label{lemma:mc-laxators}
  We construct the following morphisms
  $$\begin{aligned}
    & ψ_2 : \MC{\biobj{A}{B}}{\left(\biobj{X}{Y} ◁ \biobj{X'}{Y'}\right) ⊗ \left(\biobj{U}{V} ◁ \biobj{U'}{V'}\right)} → \MC{\biobj{A}{B}}{\left(\biobj{X}{Y} ⊗ \biobj{U}{V}\right) ◁ \left(\biobj{X'}{Y'} ⊗ \biobj{U'}{V'}\right)} \\
    & ψ_0 : \MC{\biobj{A}{B}}{I} → \MC{\biobj{A}{B}}{I ◁ I} \\
    & φ_2 : \MC{\biobj{A}{B}}{N ⊗ N} → \MC{\biobj{A}{B}}{N} \\
    & φ_0 : \MC{\biobj{A}{B}}{I} → \MC{\biobj{A}{B}}{N}.
  \end{aligned}$$
  exclusively from composition in $ℂ$ and \YonedaIsomorphisms{}. The laxator $ψ_2$ is defined by stating that the following equation holds
  $$\begin{aligned}
    &\nmcp{f_0}{}{}{}{f_1} \\
    &\qquad ≺₁^{ψ} \nmcs{g_0}{}{}{g_1}{}{}{g_2}  \\
    &\qquad ≺₂^{ψ} \nmcs{h₀}{}{}{h₁}{}{}{h₂} = \\
    &\nmcs{j_0}{}{}{j_1}{}{}{j_2} ｜ \\
    &\qquad \nmcp{k_0}{}{}{}{k_1} ｜ \\
    &\qquad \nmcp{l_0}{}{}{}{l_1}
  \end{aligned}$$
  if and only if
  $$\begin{aligned}
    & f_0 ⨾ (\im ⊗ g_0 ⊗ \im ⊗ h_0 ⊗ \im) ⨾ (\im ⊗ ■ ⊗ \im ⊗ ■ ⊗ \im) ⨾ (\im ⊗ g_1 ⊗ \im ⊗ h_1 ⊗ \im) ⨾ \\
    &\qquad (\im ⊗ ■ ⊗ \im ⊗ ■ ⊗ \im) ⨾ (\im ⊗ g_2 ⊗ \im ⊗ h_2 ⊗ \im) ⨾ f_1 = \\
    & j_0 ⨾ (\im ⊗ k_0 ⊗ \im) ⨾ (\im ⊗ ■ ⊗ \im ⊗ ■ ⊗ \im) ⨾ (\im ⊗ k_1 ⊗ \im)⨾ j_1 ⨾ (\im ⊗ l_0 ⊗ \im) ⨾ \\
    &\qquad (\im ⊗ ■ ⊗ \im ⊗ ■ ⊗ \im) ⨾ (\im ⊗ l_1 ⊗ \im) ⨾ j_2.
  \end{aligned}$$

  Furthermore, since $\Mctx{ℂ}$ is normal, $ψ_0, φ_2$, and $φ_0$ are isomorphisms.
\end{lemma}
\begin{proof}
  Consider the right hand side of $ψ_2$. It is isomorphic to the following
  \begin{align*}
    & \intr{\biobj{P}{Q}, \biobj{P'}{Q'} \in \Mctx{ℂ}} \MC{\biobj{A}{B}}{\biobj{P}{Q} ◁ \biobj{P'}{Q'}} × \MC{\biobj{P}{Q}}{\biobj{X}{Y} ⊗ \biobj{U}{V}} × \MC{\biobj{P'}{Q'}}{\biobj{X'}{Y'} ⊗ \biobj{U'}{V'}} & \bydef \\
    & \intr{\biobj{P}{Q}, \biobj{P'}{Q'} \in \Mctx{ℂ}, M,N,O ∈ ℂ} 
      ℂ(A; M⊗P⊗N) × ℂ(M⊗Q⊗N; M'⊗P'⊗N') × ℂ(M'⊗Q'⊗N';B) \\
    & \qquad × \MC{\biobj{P}{Q}}{\biobj{X}{Y} ⊗ \biobj{U}{V}} × \MC{\biobj{P'}{Q'}}{\biobj{X'}{Y'} ⊗ \biobj{U'}{V'}} & \bydef \\
    & \intr{\biobj{P}{Q}, \biobj{P'}{Q'} \in \Mctx{ℂ}, M,N,O ∈ ℂ} ℂ(A; M⊗P⊗N) × \MC{\biobj{M⊗Q⊗N}{B}}{\biobj{P'}{Q'}} × \MC{\biobj{P}{Q}}{\biobj{X}{Y} ⊗ \biobj{U}{V}} × \MC{\biobj{P'}{Q'}}{\biobj{X'}{Y'} ⊗ \biobj{U'}{V'}} & \yo1 \\
    & \intr{\biobj{P}{Q} \in \Mctx{ℂ}, M,N,O,C,D,E,F,G,H ∈ ℂ} ℂ(A; M⊗P⊗N) × ℂ(P;C⊗X⊗D⊗U⊗E) × ℂ(C⊗Y⊗D⊗V⊗E;Q) × \\ & \qquad 
    ℂ(M⊗Q⊗N; F⊗X'⊗G⊗U'⊗H) × ℂ(F⊗Y'⊗G⊗V'⊗H;B) & \bydef \\
    & \intr{\biobj{P}{Q}, \in \Mctx{ℂ}, C,D,E,F,G,H ∈ ℂ} \MC{\biobj{A}{F⊗X'⊗G⊗U''⊗H}}{\biobj{P}{Q}} × ℂ(P;C⊗X⊗D⊗U⊗E) \\
    & \qquad × ℂ(C⊗Y⊗D⊗V⊗E;Q) × ℂ(F⊗Y'⊗G⊗V'⊗H;B) & \yo1 \\
    & \intr{C,D,E,F,G,H ∈ ℂ} ℂ(A;C⊗X⊗D⊗U⊗E) × ℂ(C⊗Y⊗D⊗V⊗E;F⊗X'⊗G⊗U'⊗H) \\
    & \qquad × ℂ(F⊗Y'⊗G⊗V'⊗H;B). &\quad
  \end{align*}

  This isomorphism sends an element $(j_0 ⨾ (\im ⊗ ■ ⊗ \im) ⨾ j_1 ⨾ (\im ⊗ ■ ⊗ \im) ⨾ j_2 ｜ k_0 ⨾ (\im ⊗ ■ ⊗ \im ⊗ ■ ⊗ \im) ⨾ k_1 ｜l_0 ⨾ (\im ⊗ ■ ⊗ \im ⊗ ■ ⊗ \im) ⨾ l_1)$ to $⟨ j_0 ⨾ (\im ⊗ k_0 ⊗ \im) ⨾  (\im ⊗ ■ ⊗ \im ⊗ ■ ⊗ \im) ⨾ (\im ⊗ k_1 ⊗ \im) ⨾ j_1 ⨾ (\im ⊗ l_0 ⊗ \im) ⨾ (\im ⊗ ■ ⊗ \im ⊗ ■ ⊗ \im) ⨾ (\im ⊗ l_1 ⊗ \im) ⨾ j_2 ⟩$. Define a map from the left hand side of $ψ_2$ to this set, sending a triple
  $$\begin{aligned}
    & (f_0 ⨾ (\im ⊗ ■ ⊗ \im ⊗ ■ ⊗ \im) ⨾ f_1｜ \\
    & g_0 ⨾ (\im ⊗ ■ ⊗ \im) ⨾ g_1 ⨾ (\im ⊗ ■ ⊗ \im) ⨾ g_2｜\\
    & h_0 ⨾ (\im ⊗ ■ ⊗ \im) ⨾ h_1 ⨾ (\im ⊗ ■ ⊗ \im) ⨾ h_2) \\ 
    &\qquad ↦ \\
    & f_0 ⨾ (\im ⊗ g_0 ⊗ \im ⊗ h_0 ⊗ \im) ⨾ (\im ⊗ ■ ⊗ \im ⊗ ■ ⊗ \im) ⨾ (\im ⊗ g_1 ⊗ \im ⊗ h_1 ⊗ \im) ⨾ \\
    &\qquad (\im ⊗ ■ ⊗ \im ⊗ ■ ⊗ \im) ⨾ (\im ⊗ g_2 ⊗ \im ⊗ h_2 ⊗ \im) ⨾ f_1.
  \end{aligned}$$
  Now composing this map with the isomorphism yields the desired morphism $ψ_2$. The remaining laxators $ψ_0, φ_2,$ and $φ_0$ are isomorphisms that arise from applications of unitality or just as identities.
\end{proof}

\begin{theorem}[From \Cref{th:monoidalContextsAreANormalization}]
  \label{ax:th:monoidalContextsAreANormalization}
  \MonoidalContexts{} are the free normalization of the cofree \produoidalCategory{} over a category.
\end{theorem}
\begin{proof} 
  We already know that the normalization procedure yields the free normalization over a produoidal category. It is only left to note that this is exactly the category we have explicitly constructed in this section.

  This amounts to proving that the \produoidalCategory{} of \monoidalContexts{} is precisely the normalization of the \hyperlink{linkProduoidalSplice}{produoidal category of spliced arrows}.
  We do so for morphisms, the rest of the proof is similar.
  \begin{align*}
    & 𝓝𝓢{ℂ} \left( {\biobj{A}{B}};{\biobj{X}{Y}} \right) \phantom{\int}
    & \bydef \\
    & 𝓢{ℂ} \left( {\biobj{A}{B}}; N \otimes {\biobj{X}{Y} \otimes N} \right)  
    & \bydef \\
    & ∫^{\biobj{U}{V}, \biobj{U'}{V;}\in 𝓢{ℂ}} 
      𝓢ℂ\left({\biobj{A}{B}};{\biobj{U}{V} \otimes \biobj{X}{Y} \otimes \biobj{U'}{V'}}\right) × 
      𝓢ℂ\left({\biobj{U}{V}}; N \right) ×
      𝓢ℂ\left({\biobj{U'}{V'}}; N \right)
    & \bydef \\
    & ∫^{\biobj{U}{V}, \biobj{U'}{V'} \in 𝓢{ℂ}} 
      ℂ(A ; U \otimes X \otimes U') ×
      ℂ(V \otimes Y \otimes V' ; B) × 
      ℂ\left(U ; V \right) ×
      ℂ\left(U' ; V' \right)
    & \bydef \\
    & ∫^{\U, \V, \U', \V' \in ℂ} 
      ℂ(A ; \U \otimes X \otimes \U') ×
      ℂ(\V \otimes Y \otimes \V' ; B) × 
      ℂ\left(\U ; \V \right) ×
      ℂ\left(\U' ; \V' \right)
    & \yo1 \\
    & ∫^{\U,\U' \in ℂ} 
      ℂ\left( A ; \U \otimes X \otimes \U'\right) × 
      ℂ(\U \otimes Y \otimes \U' ; B)
    & \bydef \\
    & 𝓜ℂ\left( {\biobj{A}{B}};{\biobj{X}{Y}} \right)
  \end{align*}
  The rest of the profunctors follow a similar reasoning.
\end{proof}

\clearpage %
\section{Monoidal Lenses}
\begin{proposition}[From \Cref{prop:monoidalLensesProduoidal}] \label{ax:prop:monoidalLensesProduoidal}
  \MonoidalLenses{} form a normal symmetric produoidal category with the following morphisms, units, sequential and parallel splits.
  $$\begin{aligned}
    \LC{\biobj{A}{B}}{\biobj{X}{Y}} =\ &ℂ(A ; • ⊗ X) \diamond ℂ(• ⊗ Y ; B); \\
    \LC{\biobj{A}{B}}{N} =\ &ℂ(A ; B); \\
    \LC{\biobj{A}{B}}{\biobj{X}{Y} \triangleleft \biobj{X'}{Y'}} =\ &
      ℂ(A ; •^1 ⊗ X)\ \diamond 
      ℂ(•^1 ⊗ Y; •^2 ⊗ X') \diamond 
      ℂ(•^2 ⊗ Y' ; B); \\
    \LC{\biobj{A}{B}}{\biobj{X}{Y} ⊗ \biobj{X'}{Y'}} =\ &
      ℂ(A ; •^1 ⊗ X ⊗ X')\diamond
      ℂ(•^1 ⊗ Y ⊗ Y'; B).
  \end{aligned}$$
\end{proposition}
\begin{proof}
  \Cref{lemma:mlens-assocseq,lemma:mlens-assocpar} construct the associators, and \Cref{lemma:mlens-unitseq,lemma:mlens-unitpar} define the unitors. \Cref{lemma:mlens-sym} constructs the symmetry. As they are all constructed with \YonedaIsomorphisms{} and symmetries, they must satisfy the coherence equations. Finally, the laxators are constructed in much the same way as in \Cref{lemma:mc-laxators}. %
\end{proof}

\begin{lemma}[Monoidal lenses sequential associator] \label{lemma:mlens-assocseq}
  We construct a natural isomorphism
  $$(≺^α_2) : ∫^{\biobj{U}{V} \in 𝓛{ℂ}} \LC{\biobj{A}{B}}{\biobj{X}{Y} ◁ \biobj{U}{V}} × \LC{\biobj{U}{V}}{\biobj{X'}{Y'} ◁ \biobj{X''}{Y''}} ≅ 
  ∫^{\biobj{U}{V} \in \Mctx{ℂ}} \LC{\biobj{A}{B}}{\biobj{U}{V} ◁ \biobj{X''}{Y''}} × \LC{\biobj{U}{V}}{\biobj{X}{Y} ◁ \biobj{X'}{Y'}} : (≺^α_1)$$
  exclusively from \YonedaIsomorphisms{}. %
\end{lemma}
\begin{proof}
  Out of Yoneda reductions, we construct an isomorphism between the left hand side and a set of quadruples of morphisms. %
    \begin{align*}
    & \intr{\biobj{U}{V} \in 𝓛ℂ} \LC{\biobj{A}{B}}{\biobj{X}{Y} ◁ \biobj{U}{V}} × \LC{\biobj{U}{V}}{\biobj{X'}{Y'} ◁ \biobj{X''}{Y''}} 
    & \bydef \\
    & \intr{\biobj{U}{V} \in 𝓛ℂ, \P, \Q \in ℂ} 
      ℂ(A; \P ⊗ X) × ℂ(\P ⊗ Y; \Q ⊗ U) × ℂ(\Q ⊗ V; B) × 
      \LC{\biobj{U}{V}}{\biobj{X'}{Y'} ◁ \biobj{X''}{Y''}}
    & \bydef \\
    & \intr{\biobj{U}{V} \in 𝓛ℂ, \P, \Q, \R \in ℂ} 
      ℂ(A; \P ⊗ X) × \LC{\biobj{\P ⊗ Y}{B}}{\biobj{U}{V}} ×
      ℂ(U; \Q ⊗ X') × ℂ(\Q ⊗ Y'; \R ⊗ X'') × ℂ(\R ⊗ Y''; V)
    & \yo2 \\
    & \intr{\P,\Q,\R \in ℂ} 
      ℂ(A; \P ⊗ X) × ℂ(\P ⊗ Y; \Q ⊗ X') × 
      ℂ(\Q ⊗ Y'; \R ⊗ X'') × ℂ(\R ⊗ Y''; B).
  \end{align*}

    Out of Yoneda reductions, we construct an isomorphism between the right hand side and the same set of quadruples of morphisms. 
  \begin{align*}
    & \intr{\biobj{U}{V} \in 𝓛ℂ} \LC{\biobj{A}{B}}{\biobj{U}{V} ◁ \biobj{X''}{Y''}} × \LC{\biobj{U}{V}}{\biobj{X}{Y} ◁ \biobj{X'}{Y'}}
    & \bydef \\
    & \intr{\biobj{U}{V} \in 𝓛ℂ, \P, \Q \in ℂ} 
      ℂ(A; \Q ⊗ U) × ℂ(\Q ⊗ V; \P ⊗ X'') × ℂ(\P ⊗ Y''; B) × 
      \LC{\biobj{U}{V}}{\biobj{X}{Y} ◁ \biobj{X'}{Y'}}
    & \bydef \\
    & \intr{\biobj{U}{V} \in 𝓛ℂ, \P,\Q,\R \in ℂ} 
      \LC{\biobj{A}{\P ⊗ X''}}{\biobj{U}{V}} × ℂ(\P ⊗ Y''; B) × 
      ℂ(U; \Q ⊗ X) × ℂ(\Q ⊗ Y; \R ⊗ X') × ℂ(\R ⊗ Y'; V)
    & \yo2 \\
    & \intr{\P, \Q, \R \in ℂ} 
      ℂ(A; \Q ⊗ X) × ℂ(\Q ⊗ Y; \R ⊗ X') × ℂ(\R ⊗ Y'; \P ⊗ X'') × 
      ℂ(\P ⊗ Y''; B).
  \end{align*}
  Composing both isomorphisms, we obtain the desired associator. It gets defined by the following operations,
  \begin{align*}
    & (f₀ ⨾ (\im_M ⊗ ▪) ⨾ f₁ ⨾ (\im_N ⊗ ▪) ⨾ f₂) ≺^α_1 
    (g₀ ⨾ (\im_P ⊗ ▪) ⨾ g₁ ⨾ (\im_Q ⊗ ▪) ⨾ g₂) &= \\
    & f₀ ⨾ (\im_M ⊗ g₀) ⨾ (\im_{M ⊗ P} ⊗ ▪) ⨾ 
    (\im_M ⊗ g₁) ⨾ (\im_{M ⊗ Q} ⊗ ▪) ⨾ (\im_M ⊗ g₂) ⨾ f₁ ⨾
    (\im_N ⊗ ▪) ⨾ f₂. \\
    & (f₀ ⨾ (\im_M ⊗ ▪) ⨾ f₁ ⨾ (\im_N ⊗ ▪) ⨾ f₂) ≺^α_2
    (h₀ ⨾ (\im_P ⊗ ▪) ⨾ h₁ ⨾ (\im_Q ⊗ ▪) ⨾ h₂) &= \\
    & f₀ ⨾ (\im_{M} ⊗ ▪) ⨾ f₁ ⨾
    (\im_{N} ⊗ h₀) ⨾ (\im_{N ⊗ P} ⨾ ▪) ⨾ (\im_N ⊗ h₁) ⨾
    (\im_{N ⊗ Q} ⊗ ▪) ⨾ (\im_N ⊗ h₂) ⨾ f₂. \\
  \end{align*}
\end{proof}

\begin{lemma}[Monoidal lenses parallel associator] \label{lemma:mlens-assocpar}
   We construct a natural isomorphism
  $$(≺^α_2) : ∫^{\biobj{U}{V} \in \Mctx{ℂ}} \LC{\biobj{A}{B}}{\biobj{X}{Y} ⊗ \biobj{U}{V}} × \LC{\biobj{U}{V}}{\biobj{X'}{Y'} ⊗ \biobj{X''}{Y''}} ≅ ∫^{\biobj{U}{V} \in \Mctx{ℂ}} \LC{\biobj{A}{B}}{\biobj{U}{V} ⊗ \biobj{X''}{Y''}} × \LC{\biobj{U}{V}}{\biobj{X}{Y} ⊗ \biobj{X'}{Y'}} : (≺^α_1)$$
  exclusively from \YonedaIsomorphisms{}. %
\end{lemma}
\begin{proof}
  The left hand side is isomorphic to:
  \begin{align*} %
    &∫^{\biobj{U}{V} \in 𝓛ℂ} \LC{\biobj{A}{B}}{\biobj{X}{Y} ⊗ \biobj{U}{V}} × \LC{\biobj{U}{V}}{\biobj{X'}{Y'} ⊗ \biobj{X''}{Y''}} &\quad
    = & \quad\mbox{(by \hyperlink{linklensrep}{representability})} \\
    &∫^{\biobj{U}{V} \in 𝓛ℂ}  \LC{\biobj{A}{B}}{\biobj{X}{Y} ⊗ \biobj{U}{V}} × \LC{\biobj{U}{V}}{\biobj{X' ⊗ X''}{Y' ⊗ Y''}} &\quad
    ≅ & \quad\mbox{(\byYonedaReduction{})} \\
    &\LC{\biobj{A}{B}}{\biobj{X}{Y} ⊗ \biobj{X' ⊗ X''}{Y' ⊗ Y''}} &\phantom{\int}\quad
    ≅ & \quad\mbox{(by \hyperlink{linklensrep}{representability})} \\
    &\LC{\biobj{A}{B}}{\biobj{X ⊗ X' ⊗ X''}{Y ⊗ Y' ⊗ Y''}},
  \end{align*}

  and the right hand side is isomorphic to the same:
  \begin{align*} %
    &∫^{\biobj{U}{V} \in 𝓛ℂ} \LC{\biobj{A}{B}}{\biobj{U}{V} ⊗ \biobj{X''}{Y''}} × \LC{\biobj{U}{V}}{\biobj{X}{Y} ⊗ \biobj{X'}{Y'}} &\quad
    = & \quad\mbox{(by \hyperlink{linklensrep}{representability})} \\
    &∫^{\biobj{U}{V} \in 𝓛ℂ}  \LC{\biobj{A}{B}}{\biobj{U}{V} ⊗ \biobj{X''}{Y''}} × \LC{\biobj{U}{V}}{\biobj{X ⊗ X'}{Y ⊗ Y'}} &\quad
    ≅ & \quad\mbox{(\byYonedaReduction{})} \\
    &\LC{\biobj{A}{B}}{\biobj{X ⊗ X'}{Y ⊗ Y'} ⊗ \biobj{X''}{Y''}} & \phantom{\int} \quad
    ≅ & \quad\mbox{(by \hyperlink{linklensrep}{representability})} \\
    &\LC{\biobj{A}{B}}{\biobj{X ⊗ X' ⊗ X''}{Y ⊗ Y' ⊗ Y''}}.
  \end{align*}
  Composing both isomorphisms, we obtain the desired associator,
  \begin{align*}
    & (f₀ ⨾ (\im_M ⊗ ▪ ⊗ ▪) ⨾ f₁) ≺^α_1 
    (g₀ ⨾ (\im_P ⊗ ▪ ⊗ ▪) ⨾ g₁) &= \\
    & f₀ ⨾ (\im_M ⊗ g₀ ⊗ \im_{X''}) ⨾ 
    (\im_{M ⊗ P} ⊗ ▪ ⊗ ▪ ⊗ ▪) ⨾ 
    (\im_M ⊗ g₁ ⊗ \im_{Y''}) ⨾ f₁. \\
    & (f₀ ⨾ (\im_M ⊗ ▪ ⊗ ▪) ⨾ f₁) ≺^α_2
    (h₀ ⨾ (\im_Q ⊗ ▪ ⊗ ▪) ⨾ h₁) &= \\
    & f₀ ⨾ σ ⨾ (\im_M ⊗ h₀ ⊗ \im_{X}) ⨾  σ ⨾
    (\im_{M ⊗ P} ⊗ ▪ ⊗ ▪ ⊗ ▪) ⨾ σ ⨾ 
    (\im_M ⊗ h₁ ⊗ \im_{Y}) ⨾ σ ⨾ f₁.
  \end{align*}
  This concludes the proof.
  \end{proof}

\begin{lemma}[Monoidal lenses sequential right unitor] \label{lemma:mlens-unitseq}
  We construct a natural isomorphism
  $$(≺^{ρ}) : ∫^{\biobj{X'}{Y'} \in 𝓛{ℂ}} \LC{\biobj{A}{B}}{\biobj{X}{Y} ◁ \biobj{X'}{Y'}} × \LC{\biobj{X'}{Y'}}{N} ≅ \LC{\biobj{A}{B}}{\biobj{X}{Y}}$$
  exclusively from \YonedaIsomorphisms{}. %
\end{lemma}
\begin{proof}
  We construct the isomorphism with the following coend calculus derivation.
  \begin{align*} %
    & \intr{\biobj{X'}{Y'} \in 𝓛ℂ} \LC{\biobj{A}{B}}{\biobj{X}{Y} ◁ \biobj{X'}{Y'}} × \LC{\biobj{X'}{Y'}}{N} & \bydef \\
    & \intr{\biobj{X'}{Y'} \in 𝓛ℂ, P,Q \in ℂ} ℂ(A; P ⊗ X) × ℂ(P ⊗ Y; Q ⊗ X') × ℂ(Q ⊗ Y'; B) × ℂ(X'; Y') &\bydef \\
    & \intr{\biobj{X'}{Y'} \in 𝓛ℂ, P \in ℂ} ℂ(A; P ⊗ X) × \LC{\biobj{P ⊗ Y}{B}}{\biobj{X'}{Y'}} × ℂ(X'; Y') &\yo2 \\
    & \intr{P \in ℂ} ℂ(A; P ⊗ X) × ℂ(P ⊗ Y; B).
  \end{align*}
  We obtain the following right unitor.
  \begin{align*}
    & (f₀ ⨾ (\im_M ⊗ ▪) ⨾ f₁ ⨾ (\im_N ⊗ ▪) ⨾ f₂) ≺^{ρ} g 
    &= \\
    & f₀ ⨾ (\im_M ⊗ ▪) ⨾ f₁ ⨾ (\im_N ⊗ g) ⨾ f₂.
  \end{align*}
  The left unitor is defined similarly.
\end{proof}

\begin{lemma}[Monoidal lenses parallel right unitor] \label{lemma:mlens-unitpar}
    We construct a natural isomorphism
  $$(≺^{ρ}) : ∫^{\biobj{X'}{Y'} \in 𝓛{ℂ}} 
      \LC{\biobj{A}{B}}{\biobj{X}{Y} ⊗ \biobj{X'}{Y'}} × \LC{\biobj{X'}{Y'}}{N} ≅ 
    \LC{\biobj{A}{B}}{\biobj{X}{Y}}$$
  exclusively from \YonedaIsomorphisms{} and symmetry of $ℂ$.
\end{lemma}
\begin{proof}
  We construct the isomorphism with the following coend calculus derivations.
  \begin{align*} %
    & \intr{\biobj{X'}{Y'} \in 𝓛ℂ} \LC{\biobj{A}{B}}{\biobj{X}{Y} ⊗ \biobj{X'}{Y'}} × \LC{\biobj{X'}{Y'}}{N} \\
    & \quad = \qquad \quad\mbox{(by definition)} \\
    & \intr{\biobj{X'}{Y'} \in 𝓛ℂ, P \in ℂ} ℂ(A; P ⊗ X ⊗ X') × ℂ(P ⊗ Y ⊗ Y'; B) × ℂ(X'; Y') \\
    & \quad ≅ \qquad \quad\mbox{(by symmetry of ℂ)} \\
    & \intr{\biobj{X'}{Y'} \in 𝓛ℂ, P \in ℂ} ℂ(A; P ⊗ X' ⊗ X) × ℂ(P ⊗ Y' ⊗ Y; B) × ℂ(X'; Y') \\
    & \quad ≅ \qquad \quad\mbox{(\byYonedaReduction{})} \\
    & \intr{\biobj{X'}{Y'} \in 𝓛ℂ, P,Q,R \in ℂ} ℂ(A; Q ⊗ X) × ℂ(Q; P ⊗ X') × ℂ(P ⊗ Y'; R) × ℂ(R ⊗ Y; B) × C(X'; Y') \\
    & \quad = \qquad \quad\mbox{(by definition)} \\
    & \intr{\biobj{X'}{Y'} \in 𝓛ℂ, P,Q,R \in ℂ} ℂ(A; Q ⊗ X) × \LC{\biobj{Q}{R}}{\biobj{X'}{Y'}} × ℂ(R ⊗ Y; B) × C(X'; Y') \\
    & \quad ≅ \qquad \quad\mbox{(\byYonedaReduction{})} \\
    & \intr{Q \in ℂ} ℂ(A; Q ⊗ X) × ℂ(Q ⊗ Y; B). 
  \end{align*}
  We obtain the following right unitor.
  \begin{align*}
    & (f₀ ⨾ (\im_M ⊗ ▪ ⊗ ▪) ⨾ f₁) ≺^{ρ} g  
    & = \\
    & f₀ ⨾ (\im_M ⊗ ▪ ⊗ ▪) ⨾ (\im_M ⊗ (σ ⨾ (g ⊗ \im_X) ⨾ σ))⨾ f₁ 
    &= \\
    & f₀ ⨾ (\im_M ⊗ (σ ⨾ (g ⊗ \im_X) ⨾ σ)) ⨾ (\im_M ⊗ ▪ ⊗ ▪) ⨾ f₁.
  \end{align*}
  The left unitor is defined similarly.
\end{proof}

\begin{lemma}[Monoidal lenses symmetry] \label{lemma:mlens-sym}
  We construct the symmetries $\LC{\biobj{A}{B}}{\biobj{X}{Y} ⊗ \biobj{X'}{Y'}} ≅ \LC{\biobj{A}{B}}{\biobj{X'}{Y'} ⊗ \biobj{X}{Y}}$.
\end{lemma}
\begin{proof}
  These follow from the symmetries of $ℂ$ and \hyperlink{linklensrep}{representability of $⊗$ for monoidal lenses}.
  $$\LC{\biobj{A}{B}}{\biobj{X}{Y} ⊗ \biobj{X'}{Y'}} ≅ \LC{\biobj{A}{B}}{\biobj{X ⊗ X'}{Y ⊗ Y'}} ≅ \LC{\biobj{A}{B}}{\biobj{X' ⊗ X}{Y' ⊗ Y}} ≅ \LC{\biobj{A}{B}}{\biobj{X'}{Y'} ⊗ \biobj{X}{Y}}.$$
  This concludes the proof.
\end{proof}

\begin{proposition}[From \Cref{prop:sessionNotation}]
    \label{ax:prop:sessionNotation}
    Let $(ℂ,\otimes,I)$ be a symmetric \monoidalCategory{}. There exist monoidal functors $(\Send{}) \colon ℂ \to 𝓛ℂ$ and $(\Get{}) \colon ℂ^{op} \to 𝓛ℂ$. 
\end{proposition}
\begin{proof}
  This proof appears with a different language in the work of Riley \cite[Proposition 2.0.14]{riley2018categories}. In fact, there, the combined identity-on-objects functor $(\Send{} × \Get{}) \colon ℂ × ℂ^{op} \to 𝓛ℂ$ is shown to be monoidal.
  In our case, we can define $\Send{f} = (f ⨾ \blacksquare ⨾ \im_I)$ and $\Get{g} = (\im_I ⨾ \blacksquare ⨾ g)$, and then check that compositions and tensoring of morphisms are compatible with composition and tensoring of \monoidalLenses{}, this is straightforward.
  Moreover, as we comment in the text, we can see that, by definition, $\Send{(A \otimes B)} = \left(\biobj{A \otimes B}{I}\right) = \left(\biobj{A}{I}\right) \otimes \left(\biobj{B}{I}\right) = \Send{A} \otimes \Send{B}$ and $\Get{(A \otimes B)} = \left(\biobj{I}{A \otimes B}\right) = \left(\biobj{I}{A}\right) \otimes \left(\biobj{I}{B}\right) = \Get{A} \otimes \Get{B}$.
\end{proof}

\begin{proposition}[From \Cref{prop:cartesianlenses}]
  \label{ax:prop:cartesianlenses}
  Let $(ℂ,×,1)$ be a cartesian monoidal category.
  Its \produoidalCategory{} of lenses is given by the following \profunctors{}.
  $$\begin{aligned}
    \mathbf{Lens}\left(\biobj{A}{B}; \biobj{X}{Y} \right) &=
    ℂ(A;X) × ℂ(A × Y; B). \\
    \mathbf{Lens}\left(\biobj{A}{B}; \biobj{X}{Y} ◁ \biobj{X'}{Y'} \right) &=
    ℂ(A;X) × ℂ(A × Y; X') × ℂ(A × Y × Y' ; B). \\
    \mathbf{Lens}\left(\biobj{A}{B}; \biobj{X}{Y} ⊗ \biobj{X'}{Y'} \right) &=
    ℂ(A;X × X') × ℂ(A × Y × Y'; B). \\
    \mathbf{Lens}\left(\biobj{A}{B}\right) &=
    ℂ(A;B). \\
  \end{aligned}$$
\end{proposition}

\begin{proof}
  We employ coend calculus. The derivation of the morphisms of cartesian lenses is very well-known \cite{riley2018categories,ClarkeRoman20:ProfunctorOptics}; we derive the sequential and parallel splits. Indeed, the sequential split reduces as
  \begin{align*}
    & \int^{M,N} ℂ(A; M × X) × ℂ(M × Y; N × X') × ℂ(N × Y'; B) \\
    ≅ & \quad \mbox{ (Universal property of the product) }\\
    & \int^{M,N} ℂ(A; M) × ℂ(A; X) × ℂ(M × Y; N) × ℂ(M × Y; X') × ℂ(N × Y'; B) \\
    ≅ & \quad \mbox{ (\byYonedaReduction{}) }\\
    & ℂ(A; X) × ℂ(A × Y; X') × ℂ(A × Y × Y'; B).
  \end{align*}
  And the parallel split reduces as
  \begin{align*}
    & \int^{M} ℂ(A; M × X × X') × ℂ(M × Y × Y'; B)  \\
    ≅ & \quad \mbox{ (Universal property of the product) }\\
    & \int^{M} ℂ(A; M) × ℂ(A; X × X') × ℂ(M × Y × Y'; B)  \\
    ≅ & \quad \mbox{ (\byYonedaReduction{}) }\\
    & ℂ(A; X × X') × ℂ(A × Y × Y'; B).
  \end{align*}
  The unit is just the same as in the general monoidal case.
\end{proof}

\begin{theorem}[From \Cref{th:lensesuniversal}]
  \label{ax:th:lensesuniversal}
  Monoidal lenses are the free symmetric normalization of the cofree \symmetricProduoidal{} category over a monoidal category.
\end{theorem}
\begin{proof}
  We have already proven that the symmetric normalization procedure yields the free symmetric normalization over a symmetric produoidal category (\Cref{th:sym:freeNormalProduoidal}).

  The rest of the proof amounts to show that the normal \symmetricProduoidal{}{} category of \monoidalLenses{} is precisely the symmetric normalization of the \produoidalCategory{} of \splicedArrows{}. We do so for morphisms, the rest of the proof is similar.
  \begin{align*}
    & 𝓝_σ𝓢{ℂ} \left( {\biobj{A}{B}};{\biobj{X}{Y}} \right) \phantom{\int}
    & \bydef \\
    & 𝓢{ℂ} \left( {\biobj{A}{B}}; N \otimes \biobj{X}{Y} \right)  
    & \bydef \\
    & ∫^{\biobj{U}{V} \in 𝓢{ℂ}} 
      𝓢ℂ\left({\biobj{A}{B}};\biobj{U}{V} \otimes \biobj{X}{Y} \right) × 
      𝓢ℂ\left({\biobj{U}{V}}; N \right)
    & \bydef \\
    & ∫^{\U, \V \in ℂ} 
      ℂ(A ; \U \otimes X) ×
      ℂ(\V \otimes Y ; B) × 
      ℂ\left(\U ; \V\right)
    & \yo1 \\
    & ∫^{\U \in ℂ} 
      ℂ\left( A ; \U \otimes X\right) × ℂ(\U \otimes Y ; B)
    & \bydef \\
    & 𝓛ℂ\left( {\biobj{A}{B}};{\biobj{X}{Y}} \right)
  \end{align*}
  The rest of the profunctors follow a similar reasoning.
\end{proof}

\clearpage %
\section{Further Work}

\begin{theorem}[From \Cref{th:virtualdependence}]
\label{ax:th:virtualdependence}
Let $𝕍$ be a normal and $⊗$-symmetric \produoidalCategory{} with coends over $𝕍$ commuting with finite connected limits.
Then, $[𝕍\op, \mathbf{Set}]$ is a dependence category in the sense of Shapiro and Spivak \cite{shapiro22:duoidal}.
\end{theorem}
\begin{proof}
  Whenever $𝕍$ is produoidal, $[𝕍\op, \mathbf{Set}]$, its category of presheaves is duoidal, with the structure given by convolution
  (\Cref{th:produoidalInduceDuoidal}).

  At the same time,  $[𝕍\op, \mathbf{Set}]$ is a locally cartesian closed category will all limits because it is a presheaf category. Whenever finite connected limits are preserved by $⊗, \triangleleft$, we obtain a dependence category \cite[Theorem 4.8]{shapiro22:duoidal}.
  This means we only need the following isomorphism,

  \[\begin{aligned}
    & \int^{U,V} 𝕍(X; U ⊗ V) × \lim\nolimits_i P_i(U) × \lim\nolimits_j Q_j(V) \\
    ≅ & \quad \mbox{ (Commutation of limits) }\\
    & \int^{U,V}  \lim\nolimits_{i,j} 𝕍(X; U ⊗ V) × P_i(U) ×  Q_j(V)\\
    ≅ &  \quad \mbox{ (Coends commute with finite connected limits) } \\
    & \lim\nolimits_{i,j} \int^{U,V}  𝕍(X; U ⊗ V) × P_i(U) ×  Q_j(V)
  \end{aligned}
    \]

  Where we use our hypothesis on the last step. We conjecture this can be extended to an arbitrary $𝕍$ with minor constraints.
\end{proof}

\clearpage %
\section{Duoidal and Produoidal Categories}

  By the Eckmann-Hilton argument, each time we have two monoids $(\ast,\circ)$ such that one is a monoid homomorphism over the other, $(a ∘ b) \ast (c ∘ d) = (a \ast c) ∘ (b \ast d)$, we know that both monoids coincide into a single commutative monoid.
  
  However, an extra dimension helps us side-step the Eckmann-Hilton argument. If, instead of equalities or isomorphisms, we use directed morphisms, both monoids (which now may become 2-monoids) do not necessarily coincide, and the resulting structure is that of a duoidal category.
  
  \begin{definition}[Duoidal category]
    \defining{linkduoidal}{}
      A \emph{duoidal category} \cite{aguiar10:monoidal} is a category $ℂ$ with two monoidal structures, $(ℂ,⊗,I,α,λ,ρ)$ and $(ℂ,◁,N, β, κ, ν)$ such that the latter distribute over the former.  In other words, it is endowed with a duoidal tensor, $(◁) \colon ℂ × ℂ → ℂ$, together with natural distributors
      $$ψ_2﹕ (X ◁ Z) ⊗ (Y ◁ W) → (X ⊗ Y) ◁ (Z ⊗ W), \qquad
      ψ_0﹕ I → I ◁ I, \qquad
      φ_2﹕ N ⊗ N → N, \quad\mbox{and}\quad
      φ_0 ﹕ I → N,$$
      satisfying the following coherence equations (\Cref{cd:duoidal-coherence-assoc,cd:duoidal-coherence-unit,cd:duoidal-coherence-unit2,cd:duoidal-coherence-bimonoids,cd:duoidal-coherence-nandi}).
      \begin{figure}[ht]
        \centering
        \begin{tikzcd}
          ((A ◁ B) ⊗ (C ◁ D)) ⊗ (E ◁ F) 
          \rar{\alpha} \dar[swap]{ψ₂ ⊗ id} &
          (A ◁ B) ⊗ ((C ◁ D) ⊗ (E ◁ F))
          \dar{id ⊗ ψ₂} \\
          ((A ⊗ C) ◁ (B ⊗ D)) ⊗ (E ◁ F)
          \dar[swap]{ψ₂} &
          (A ◁ B) ⊗ ((C ⊗ E) ◁ (D ⊗ F))
          \dar{ψ₂} \\
          ((A ⊗ C) ⊗ E) ◁ ((B ⊗ D) ⊗ F)
          \rar{α ◁ α} &
          (A ⊗ (C ⊗ E)) ◁ (B ⊗ (D ⊗ F))
        \end{tikzcd}
        \begin{tikzcd}
          ((A ◁ B) ◁ C) ⊗ ((D ◁ E) ◁ F) 
          \rar{\beta ⊗ \beta} \dar[swap]{ψ₂} &
          (A ◁ (B ◁ C)) ⊗ (D ◁ (E ◁ F))
          \dar{ψ₂} \\
          ((A ◁ B) ⊗ (D ◁ E)) ◁ (C ⊗ F)
          \dar[swap]{ψ₂ ⊗ id} &
          (A ⊗ D) ◁ ((B ◁ C) ⊗ (E ◁ F))
          \dar{id ⊗ ψ₂} \\
          ((A ⊗ D) ◁ (B ⊗ E)) ◁ (C ⊗ F)
          \rar{\beta} &
          (A ⊗ D) ◁ ((B ⊗ E) ◁ (C ⊗ F))
        \end{tikzcd}
        \caption{Coherence diagrams for associativity of a duoidal category.}
        \vspace*{5mm}
        \label{cd:duoidal-coherence-assoc}
        \begin{tikzcd}
          I ⊗ (A ◁ B) 
          \rar{ψ₀ ⊗ id} \dar[swap]{\lambda} &
          (I ◁ I) ⊗ (A ◁ B) 
          \dar{ψ₂} \\
          A ◁ B
          &
          (I ⊗ A) ◁ (I ⊗ B)
          \lar{\lambda ◁ \lambda}
        \end{tikzcd}
        \begin{tikzcd}
          (A ◁ B) ⊗ I
          \rar{ψ₀ ⊗ id} \dar[swap]{\rho} &
          (A ◁ B) ⊗ (I ◁ I)
          \dar{ψ₂} \\
          A ◁ B
          &
          (A ⊗ I) ◁ (B ⊗ I)
          \lar{\rho ◁ \rho}
        \end{tikzcd}
        \caption{Coherence diagrams for $⊗$-unitality of a duoidal category.}
        \vspace*{5mm}
        \label{cd:duoidal-coherence-unit}
        \begin{tikzcd}
          N ◁ (A ⊗ B)
          \dar[swap]{κ} &
          (N ⊗ N) ◁ (A ⊗ B)
          \lar[swap]{\varphi_2 ◁ id} \dar{ψ₂} \\
          A ⊗ B
          &
          (N ◁ A) ⊗ (N ◁ B)
          \lar{κ ⊗ κ}
        \end{tikzcd}
        \begin{tikzcd}
          (A ⊗ B) ◁ N
          \dar[swap]{ν} &
          (A ⊗ B)  ◁ (N ⊗ N)
          \lar[swap]{id ◁ \varphi_2} \dar{ψ₂} \\
          A ⊗ B
          &
          (A ◁ N) ⊗ (B ◁ N)
          \lar{ν ⊗ ν}
        \end{tikzcd}
        \caption{Coherence diagrams for $◁$-unitality of a duoidal category.}\label{cd:duoidal-coherence-unit2}
        \vspace*{5mm}
        \begin{tikzcd}
          (N ⊗ N) ⊗ N
          \ar{rr}{\alpha} \dar[swap]{\varphi_2 ⊗ id}&&
          N ⊗ (N ⊗ N)
          \dar{id ⊗ \varphi_2}\\
          N ⊗ N
          \rar[swap]{\varphi_2} &
          N
          &
          N ⊗ N
          \lar{\varphi_2}
        \end{tikzcd}
        \begin{tikzcd}
          I ◁ I
          \dar[swap]{\psi_0 ⊗ id} &
          I 
          \lar[swap]{\psi_0} \rar{\psi_0} &
          I ◁ I
          \dar{id ⊗ \psi_0} \\
          (I ◁ I) ◁ I
          \ar{rr}[swap]{\beta} &&
          I ◁ (I ◁ I)
        \end{tikzcd}
        \caption{Associativity and coassociativity for $N$ and $I$ in a duoidal category.}\label{cd:duoidal-coherence-bimonoids}
        \vspace*{5mm}
        \begin{tikzcd}
          N ⊗ I
          \rar{ρ} \dar[swap]{id ⊗ \varphi_0} &
          N
          \\
          N ⊗ N 
          \urar[swap]{\varphi_2} &
        \end{tikzcd}
        \begin{tikzcd}
          I ⊗ N
          \rar{\lambda} \dar[swap]{\varphi_0 ⊗ id} &
          N
          \\
          N ⊗ N 
          \urar[swap]{\varphi_2} &
        \end{tikzcd}
        \begin{tikzcd}
          I ◁ N
          \rar{id ⊗ \varphi_0} \dar[swap]{ν} &
          I ◁ I
          \\
          I 
          \urar[swap]{\psi_0} &
        \end{tikzcd}
        \begin{tikzcd}
          N ◁ I
          \rar{id ⊗ \varphi_0} \dar[swap]{κ} &
          I ◁ I
          \\
          I 
          \urar[swap]{\psi_0} &
        \end{tikzcd}
        \caption{Unitality and counitality for $N$ and $I$ in a duoidal category.}\label{cd:duoidal-coherence-nandi}
      \end{figure}
  \end{definition}

  \begin{remark}
      In other words, the duoidal tensor and unit are lax monoidal functors for the first monoidal structure, which means that the laxators must satisfy the following equations.
      \begin{enumerate}
      \item $(ψ_2 ⊗ id) ⨾ ψ_2 ⨾ (α ◁ α) = α ⨾ (id ⊗ ψ_2) ⨾ ψ_2$, for the associator;
      \item $(ψ_0 ⊗ id) ⨾ ψ_2 ⨾ (λ ◁ λ) = λ$, for the left unitor; and
      \item $(id ⊗ ψ_0) ⨾ ψ_2 ⨾ (ρ ◁ ρ) = ρ$, for the right unitor;
      \item $α ⨾ (id ⊗ φ_2) ⨾ φ_2 = (φ_2 ⊗ id) ⨾ φ_2$, for the associator;
      \item $(φ_0 ⊗ id) ⨾ φ_2 = λ$, for the left unitor; and
      \item $(id ⊗ φ_0) ⨾ φ_2 = ρ$, for the right unitor.
      \end{enumerate}
  \end{remark}

  \begin{theorem}[{{Coherence, \cite{aguiar10:monoidal}}}]
    Any two parallel morphisms constructed out of the coherence isomorphisms and laxators of a \duoidalCategory{} coincide. 
  \end{theorem}

  \clearpage

  \subsection{Normalization of duoidal categories}

  Garner and López Franco \cite{garner16} introduce a procedure for normalizing a sufficiently well-behaved duoidal category, based in the construction of a new duoidal category of \emph{bimodules}.
  In this text, we introduce a normalization procedure for an arbitrary \produoidal{} category. For completeness, let us recall first the original procedure \cite{garner16}.

  Let $M$ be a bimonoid in the \duoidalCategory{} $(𝕍,⊗,I,◁,N)$, with maps $e ﹕ I → M$ and $m ﹕ M ⊗ M → M$; and with maps $u ﹕ M → N$ and $d ﹕ M → M ◁ M$. Consider now the category of $M^⊗$-bimodules. This category has a monoidal structure lifted from $(𝕍,◁,N)$:
  \begin{enumerate}
      \item the unit, $N$, has a bimodule structure with $$M⊗N⊗M \overset{u ⊗ \im ⊗ u}\longrightarrow  N⊗N⊗N \longrightarrow N;$$
      \item the sequencing of two $M^⊗$-bimodules is a $M^⊗$-bimodule with
      $$\begin{aligned}
          M&⊗(A◁B)⊗M \\
          &→ (M◁M)⊗(A◁B)⊗(M◁M) \\
          &→ (M⊗A⊗M)◁(M⊗B⊗M) → A◁B.
      \end{aligned}$$
  \end{enumerate}
  Moreover, whenever $𝕍$ admits reflexive coequalizers preserved by $(⊗)$, the category of $M^{⊗}$-bimodules is monoidal with the tensor of bimodules: the coequalizer
  $$A ⊗ M ⊗ B \rightrightarrows A ⊗ B \twoheadrightarrow A ⊗_M B.$$
  In this case $(\mathbf{Bimod}^{⊗}_M, ⊗_M, M, ◁, N)$ is a duoidal category.
  
  \begin{theorem}[Normalization of a duoidal category]
    \label{thm:normalizationDuoidal}
      Let $(𝕍,⊗,I,◁,N)$ be a \duoidalCategory{} with reflexive coequalizers preserved by $(⊗)$. The category of $N$-bimodules is then a normal duoidal category,
      $$\mathcal{N}(𝕍) = (\mathbf{Bimod}^{⊗}_N, ⊗_N, N, ◁, N).$$
      We call this category the \emph{normalization} \cite{garner16} of the duoidal category $𝕍$.
  \end{theorem}

\subsection{Produoidal Categories}

\begin{definition}[Produoidal category, from \Cref{def:produoidal}]
  \label{ax:def:produoidal}
  A \emph{produoidal category} is a category $𝕍$ endowed with two \promonoidal{} structures,
  $$\begin{gathered}
    𝕍(• ; • ⊗ •) \colon 𝕍 × 𝕍 ⧑ 𝕍, \mbox{ and } 𝕍(•; I) \colon 1 ⧑ 𝕍, \\
    𝕍(• ;• ◁ •) \colon 𝕍 × 𝕍 ⧑ 𝕍, \mbox{ and } 𝕍(•; N) \colon 1 ⧑ 𝕍,
  \end{gathered}$$
  such that one laxly distributes over the other. 
  This is to say that it is endowed with the following natural \emph{laxators},
  \begin{align*}
  ψ_2 \colon 𝕍(•;(X◁Y)⊗(Z◁W)) &→ 𝕍(•;(X⊗Z)◁(Y⊗W)),\\
  ψ_0 \colon 𝕍(•;I) &→ 𝕍(•;I◁I),\\
  φ_2 \colon 𝕍(•;N⊗N) &→ 𝕍(•;N),\\
  φ_0 \colon 𝕍(•;I) &→ 𝕍(•;N).
  \end{align*}
  Laxators, together with unitors and associators must satisfy the coherence conditions in the following diagrams (\Cref{pcd:duoidal-coherence-assoc,pcd:duoidal-coherence-unit,pcd:duoidal-coherence-unit2,pcd:duoidal-coherence-bimonoids,pcd:duoidal-coherence-nandi}).
  \begin{figure}[ht]
    \centering
    \begin{tikzcd}
      𝕍(\bullet, ((A ◁ B) ⊗ (C ◁ D)) ⊗ (E ◁ F))
      \rar{\alpha} \dar[swap]{ψ₂ ⊗ id} &
      𝕍(\bullet, (A ◁ B) ⊗ ((C ◁ D) ⊗ (E ◁ F)))
      \dar{id ⊗ ψ₂} \\
      𝕍(\bullet, ((A ⊗ C) ◁ (B ⊗ D)) ⊗ (E ◁ F))
      \dar[swap]{ψ₂} &
      𝕍(\bullet, (A ◁ B) ⊗ ((C ⊗ E) ◁ (D ⊗ F)))
      \dar{ψ₂} \\
      𝕍(\bullet, ((A ⊗ C) ⊗ E) ◁ ((B ⊗ D) ⊗ F))
      \rar{α ◁ α} &
      𝕍(\bullet, (A ⊗ (C ⊗ E)) ◁ (B ⊗ (D ⊗ F)))
    \end{tikzcd}
    \begin{tikzcd}
      𝕍(\bullet, ((A ◁ B) ◁ C) ⊗ ((D ◁ E) ◁ F) )
      \rar{\beta ⊗ \beta} \dar[swap]{ψ₂} &
      𝕍(\bullet, (A ◁ (B ◁ C)) ⊗ (D ◁ (E ◁ F)))
      \dar{ψ₂} \\
      𝕍(\bullet, ((A ◁ B) ⊗ (D ◁ E)) ◁ (C ⊗ F))
      \dar[swap]{ψ₂ ⊗ id} &
      𝕍(\bullet, (A ⊗ D) ◁ ((B ◁ C) ⊗ (E ◁ F)))
      \dar{id ⊗ ψ₂} \\
      𝕍(\bullet, ((A ⊗ D) ◁ (B ⊗ E)) ◁ (C ⊗ F))
      \rar{\beta} &
      𝕍(\bullet, (A ⊗ D) ◁ ((B ⊗ E) ◁ (C ⊗ F)))
    \end{tikzcd}
    \caption{Coherence diagrams for associativity of a produoidal category.}
    \vspace*{5mm}
    \label{pcd:duoidal-coherence-assoc}
    \begin{tikzcd}
      𝕍(\bullet, I ⊗ (A ◁ B) )
      \rar{ψ₀ ⊗ id} \dar[swap]{\lambda} &
      𝕍(\bullet, (I ◁ I) ⊗ (A ◁ B) )
      \dar{ψ₂} \\
      𝕍(\bullet, A ◁ B)
      &
      𝕍(\bullet, (I ⊗ A) ◁ (I ⊗ B))
      \lar{\lambda ◁ \lambda}
    \end{tikzcd}
    \begin{tikzcd}
      𝕍(\bullet, (A ◁ B) ⊗ I)
      \rar{ψ₀ ⊗ id} \dar[swap]{\rho} &
      𝕍(\bullet, (A ◁ B) ⊗ (I ◁ I))
      \dar{ψ₂} \\
      𝕍(\bullet, A ◁ B)
      &
      𝕍(\bullet, (A ⊗ I) ◁ (B ⊗ I))
      \lar{\rho ◁ \rho}
    \end{tikzcd}
    \caption{Coherence diagrams for $⊗$-unitality of a produoidal category.}
    \vspace*{5mm}
    \label{pcd:duoidal-coherence-unit}
    \begin{tikzcd}
      𝕍(\bullet, N ◁ (A ⊗ B))
      \dar[swap]{κ} &
      𝕍(\bullet, (N ⊗ N) ◁ (A ⊗ B))
      \lar[swap]{\varphi_2 ◁ id} \dar{ψ₂} \\
      𝕍(\bullet, A ⊗ B)
      &
      𝕍(\bullet, (N ◁ A) ⊗ (N ◁ B))
      \lar{κ ⊗ κ}
    \end{tikzcd}
    \begin{tikzcd}
      𝕍(\bullet, (A ⊗ B) ◁ N)
      \dar[swap]{ν} &
      𝕍(\bullet, (A ⊗ B)  ◁ (N ⊗ N))
      \lar[swap]{id ◁ \varphi_2} \dar{ψ₂} \\
      𝕍(\bullet, A ⊗ B)
      &
      𝕍(\bullet, (A ◁ N) ⊗ (B ◁ N))
      \lar{ν ⊗ ν}
    \end{tikzcd}
    \caption{Coherence diagrams for $◁$-unitality of a produoidal category.}\label{pcd:duoidal-coherence-unit2}
    \vspace*{5mm}
    \begin{tikzcd}
      𝕍(\bullet, (N ⊗ N) ⊗ N)
      \ar{rr}{\alpha} \dar[swap]{\varphi_2 ⊗ id}&&
      𝕍(\bullet, N ⊗ (N ⊗ N))
      \dar{id ⊗ \varphi_2}\\
      𝕍(\bullet, N ⊗ N)
      \rar[swap]{\varphi_2} &
      𝕍(\bullet, N)
      &
      𝕍(\bullet, N ⊗ N)
      \lar{\varphi_2}
    \end{tikzcd}
    \begin{tikzcd}
      𝕍(\bullet, I ◁ I)
      \dar[swap]{\psi_0 ⊗ id} &
      𝕍(\bullet, I )
      \lar[swap]{\psi_0} \rar{\psi_0} &
      𝕍(\bullet, I ◁ I)
      \dar{id ⊗ \psi_0} \\
      𝕍(\bullet, (I ◁ I) ◁ I)
      \ar{rr}[swap]{\beta} &&
      𝕍(\bullet, I ◁ (I ◁ I))
    \end{tikzcd}
    \caption{Associativity and coassociativity for $N$ and $I$ in a produoidal category.}\label{pcd:duoidal-coherence-bimonoids}
    \vspace*{5mm}
    \begin{tikzcd}
      𝕍(\bullet, N ⊗ I)
      \rar{ρ} \dar[swap]{id ⊗ \varphi_0} &
      𝕍(\bullet, N)
      \\
      𝕍(\bullet, N ⊗ N )
      \urar[swap]{\varphi_2} &
    \end{tikzcd}
    \begin{tikzcd}
      𝕍(\bullet, I ⊗ N)
      \rar{\lambda} \dar[swap]{\varphi_0 ⊗ id} &
      𝕍(\bullet, N)
      \\
      𝕍(\bullet, N ⊗ N )
      \urar[swap]{\varphi_2} &
    \end{tikzcd}
    \begin{tikzcd}
      𝕍(\bullet, I ◁ N)
      \rar{id ⊗ \varphi_0} \dar[swap]{ν} &
      𝕍(\bullet, I ◁ I)
      \\
      𝕍(\bullet, I )
      \urar[swap]{\psi_0} &
    \end{tikzcd}
    \begin{tikzcd}
      𝕍(\bullet, N ◁ I)
      \rar{id ⊗ \varphi_0} \dar[swap]{κ} &
      𝕍(\bullet, I ◁ I)
      \\
      𝕍(\bullet, I )
      \urar[swap]{\psi_0} &
    \end{tikzcd}
    \caption{Unitality and counitality for $N$ and $I$ in a produoidal category.}\label{pcd:duoidal-coherence-nandi}
  \end{figure}
\end{definition}

\clearpage

\subsection{Produoidals induce duoidals} \label{ax:sec:pduo-to-duo}

\begin{theorem}
  \label{th:produoidalInduceDuoidal}
  Let $𝕍$ be a \produoidalCategory{}, then its category of presheaves, $[𝕍\op, \mathbf{Set}]$, is duoidal with the structure given by convolution
  \cite{bookerstreet13}.
\end{theorem}
\begin{proof}
  Let $P$ and $Q$ be presheaves in $𝕍$. We define the following tensor products on presheaves by convolution of the tensor products in $𝕍$.
  $$(P ⊗ Q)(A) = \int^{U,V} \hom(A, U ⊗ V) × P(U) × Q(V),$$
  $$(P ◁ Q)(A) = \int^{U,V} \hom(A, U ◁ V) × P(U) × Q(V).$$
  These tensor products can be shown in a straightforward way to form a duoidal category, inheriting the laxators from those of $𝕍$.
\end{proof}
\section{Tambara modules}

\begin{definition}[Tambara module, \cite{pastro07}]
  \label{def:tambaramodule}
  \defining{linkTambara}{}
  Let $(𝔸, ⊗, I)$ be a strict monoidal category. A \emph{Tambara module} is a \profunctor{} $T ﹕ 𝔸\op × 𝔸 → \mathbf{Set}$ endowed with natural transformations
  $$t_l^M \colon T(X;Y) →  T(M ⊗ X, M ⊗ Y),$$
  $$t_r^M \colon T(X;Y) →  T(X ⊗ M, Y ⊗ M),$$
  that are natural in both $X$ and $Y$, but also dinatural on $M$. These must moreover satisfy the following axioms:
  \begin{itemize}
      \item $t_l^I = id$ and $t^I_r = id$, unitality;
      \item $t_l^M ⨾ t_l^N = t_l^{N ⊗ M}$ and $t_r^M ⨾ t_r^N = t_l^{M ⊗ N}$, multiplicativity;
      \item $t_l^M ⨾ t_r^N = t_r^N ⨾ t_l^M$, and compatibility.
  \end{itemize}
\end{definition}

Tambara modules are the algebras of a monad. We start by noting that the $\hom$ profunctor is a monoid with respect to Day convolution. This makes the following functor a monad on endoprofunctors, the so-called Pastro-Street monad \cite{pastro07},
  $$Φ(P) = hom \circledast P \circledast hom;$$
  where $Φ﹕[ℂ\op × ℂ,\Set] → [ℂ\op × ℂ,\Set]$.
  
  \begin{theorem}
      The algebras of the Pastro-Street monad, the {$Φ$-algebras}, are precisely \tambaraModules{} \cite{pastro07}. As a consequence,
      the \emph{free Tambara module} over a profunctor $H ﹕ ℂ\op × ℂ → \Set$ is $Φ(H)$.
  \end{theorem}
  
  \begin{example}
      Consider the \profunctor{} $よ(A;B) ﹕ 𝔸\op × 𝔸 → \Set$ that produes a hole of types $A$ and $B$. That is, let $よ(A;B) = \hom(•,A) × \hom(B,•)$. The free Tambara module over it is the \monoidalContext{} with a hole of type $A$ and $B$,
      $$Φ(よ^A_B) = ∫^{M,N} \hom(•,M⊗A⊗N) × \hom(M⊗B⊗N,•).$$
  \end{example}

  \subsection{Normalization of profunctors}

  Let $(ℂ,⊗,I)$ be a monoidal category. The category of endoprofunctors $ℂ\op × ℂ → \mathbf{Set}$ is then duoidal with composition $(◁)$ and Day convolution $(⊛)$.
  $$(ℂ\op × ℂ,\mathbf{Set}, ⊛, I, ◁, \hom).$$
  Moreover, we can also construct its normalization: the category of endoprofunctors, $[ℂ\op × ℂ,\mathbf{Set}]$, has reflexive coequalisers; thus, we are in the conditions of \Cref{thm:normalizationDuoidal}.
  The normal duoidal category of $\hom^⊛$-bimodules has been traditionally called the category of \emph{Tambara modules}.
  $$\mathcal{N}(ℂ\op × ℂ,\mathbf{Set}, ⊛, I, ◁, \hom) = (\mathbf{Tamb}, ⊛_{\hom}, \hom, ◁, \hom).$$
  
  \begin{theorem}
    The category of Tambara modules is a normal duoidal category and, in fact, it is the normalization of the duoidal category of endoprofunctors.
  \end{theorem}

\section{Monoidal Categories}

\subsection{Monoidal categories.}
Endowed with the notion of isomorphism, we can now relax our definition of theory of processes by substituting strict equalities by isomorphism.

\begin{definition}
    \defining{linkmonoidalcategory}{}
    A \{symmetric\} monoidal category \cite{maclane78} $(ℂ,\otimes,I)$ is a tuple
    \[(ℂ_\mathrm{obj}, ℂ_\mathrm{mor}, (\comp), \im, (\tensor)_\mathrm{obj}, (\tensor)_\mathrm{mor}, I, \alpha, λ, \rho, 
    \{\sigma\}),\]
    specifying a set of objects, or resource types, $ℂ_\mathrm{obj}$; a set of morphisms, or processes, $ℂ_\mathrm{mor}$; a composition operation; a family of identity morphisms; a tensor operation on objects and morphisms; a unit object and families of associator, left unitor, right unitor \{and swapping morphisms\}.

    The families of associator, left unitor and right unitor morphisms have the following types.
    \[\begin{aligned}
        \alpha_{A,B,C} \colon & (A \tensor B) \tensor C \to A \tensor (B \tensor C), \\
        λ_A \colon & I \tensor A \to A, \\
        \rho_A \colon & A \tensor I \to A.
    \end{aligned}\]

    They must satisfy the following non-strict versions of the axioms.
    \begin{align}\setcounter{equation}{0}
        A \tensor (B \tensor C) &\cong (A \tensor B) \tensor C, \\
        A \tensor I &\cong A \cong I \tensor A, \\
        (f \comp g) \comp h &= f \comp (g \comp h), \\
        \im_B \comp f &= f = f \comp \im_B, \\
        (f \tensor (g \tensor h)) \comp \alpha &= \alpha \comp ((f \tensor g) \tensor h), \\
        (f \tensor \im_I) \comp \rho &= \rho \comp f, \\
        (f \tensor g) \comp (h \tensor k) &= (f \comp h) \tensor (g \comp k), \\
        \sigma_{A,B \tensor C} \comp \alpha &= \alpha \comp (\sigma_{A,B} \comp \im_C) \comp (\im_B \tensor \sigma_{A,C}), \\
        \sigma_{A,B \tensor C} \comp \alpha &= \alpha \comp (\sigma_{A,B} \comp \im_C) \comp (\im_B \tensor \sigma_{A,C}), \\
        \sigma_{A,A'} \comp (g \tensor f) &= (f \tensor g) \comp \sigma_{B,B'}, \\
        \sigma_{A,B} \comp \sigma_{B,A} &= \im_{A \tensor B}.
    \end{align}
    \{Additionally\}, they must satisfy the following axioms, whenever they are formally well-typed.
    \begin{align}
        \alpha \comp \alpha &= (\alpha \tensor \im) \comp \alpha \comp (\im \tensor \alpha), \\
        \rho &= \alpha \comp (\im \tensor λ), \\
        \alpha \comp \sigma \comp \alpha &= (\sigma \tensor \im) \comp \alpha \comp (\im \tensor \sigma).
    \end{align}
\end{definition}

String diagrams \cite{joyal91} are a sound and complete syntax for monoidal categories.

\begin{construction}
    Let $ℂ$ be a monoidal category.
    Its strictification, $\Strict(ℂ)$, is a monoidal category where
    \begin{itemize}
        \item objects are cliques: for each list of objects of $ℂ$, say, $[A_0, \dots, A_n] \in \List(ℂ)$, we form the clique containing all possible parenthesizations and coherence isomorphisms between them;
        \item morphisms are clique morphisms: a morphism between any two components of the clique, which determines a morphism between all of them.
    \end{itemize}
    The tensor product is concatenation, which makes it a strict monoidal category.
\end{construction}

\begin{remark}
    There is a strong monoidal functor $ℂ \to \Strict(ℂ)$, this makes an object $A$ into an object $[A]$; this is fully-faithful but, moreover, it is essentially surjective, giving a monoidal equivalence.
\end{remark}
\begin{theorem}
    Every monoidal category is monoidally equivalent to its strictification.
\end{theorem}

\end{document}